\documentclass[12pt]{article}

\usepackage[margin=1.25in]{geometry}
\usepackage{setspace}
\usepackage{amsmath,amsthm,amssymb,mathtools}
\usepackage{booktabs,adjustbox}
\usepackage{natbib}
\usepackage{subfig}
\usepackage[inline]{enumitem}
\usepackage{xcolor}
\usepackage{lmodern}
\usepackage{microtype}
\usepackage{hyperref}
\usepackage[bottom]{footmisc}
\usepackage{xr}
\hypersetup{
	colorlinks,
	citecolor=black,
	linkcolor=black,
	urlcolor=black
}

\makeatletter
\newtheoremstyle{mytheorem}
  {3pt}
  {3pt}
  {\itshape}
  {}
  {\itshape\bfseries}
  {.}
  {.5em}
  {\thmname{#1}\thmnumber{\@ifnotempty{#1}{ }#2}%
   \thmnote{ {\the\thm@notefont(#3)}}}
\makeatother
\theoremstyle{mytheorem}
\newtheorem{prop}{Proposition}
\newtheorem{cor}{Corollary}
\newtheorem{lem}{Lemma}

\newtheorem{assume}{Assumption}

\makeatletter
\newcommand{\mylabel}[2]{#2\def\@currentlabel{#2}\label{#1}}
\makeatother

\DeclareMathOperator{\eG}{\mathbb G}

\DeclareMathOperator{\Exp}{\mathbf E}

\DeclareMathOperator{\argmin}{argmin}

\let\Pr\relax
\DeclareMathOperator{\Pr}{\mathbf{Pr}}

\DeclareMathOperator{\TV}{TV}
\DeclareMathOperator{\KS}{KS}
\DeclareMathOperator{\DR}{DR}
\DeclareMathOperator{\LP}{LP}
\DeclareMathOperator{\MD}{MD}
\DeclareMathOperator{\MKW}{MK\smash /W}
\DeclareMathOperator{\Lip}{Lip}

\DeclareMathOperator{\Var}{\mathbf{Var}}
\DeclareMathOperator{\Cov}{\mathbf{Cov}}
\DeclareMathOperator{\diam}{diam}
\DeclareMathOperator{\Vol}{Vol}
\DeclareMathOperator*{\minimize}{min.}
\DeclareMathOperator{\Gr}{Gr}

\DeclareMathOperator{\sgn}{sgn}

\usepackage{datetime}
\newdateformat{monthyear}{\monthname[\THEMONTH] \THEYEAR}
\date{\monthyear\today}

\makeatletter
\renewcommand{\maketitle}{
	\renewcommand{\thefootnote}{\fnsymbol{footnote}}
	\setcounter{footnote}{0}
	\vspace*{-\baselineskip}
	\begin{center}
		{\LARGE \@title \par}
		\vskip \baselineskip
		{\large \@author}
		\vskip \baselineskip
		{\large \@date}
	\end{center}
	\setcounter{footnote}{0}
	\renewcommand{\thefootnote}{\arabic{footnote}}
}
\makeatother

\allowdisplaybreaks
\usepackage{minitoc}

\title{
On the role of the design phase in a linear regression}
\author{Junho Choi\footnote{The author acknowledges financial support from the College of Social Sciences of Seoul National University (SNU). Earlier versions of this paper were presented at the Applied Micro Brown Bag (2022, 2023) in the Department of Economics at SNU, the 6th International Conference on Econometrics and Statistics,
the Lunch Seminar in the Department of Economics at the University of Wisconsin-Madison (2024),
and the 13th International Association for Applied Econometrics Conference.
The present version was presented at the KER International Conference 2025 and the 2025 World Congress of the Econometric Society.
Special thanks are extended to Seojeong Lee for his continuous and invaluable guidance throughout this project, and also to Ryo Okui, Jungmin Lee, Myung Hwan Seo, Jack Porter, Harold Chiang, Kohei Yata, Bruce Hansen, Xiaoxia Shi, Yong Cai, Hyunseung Kang, and Jeffrey Smith, whose feedback contributed to the development of this paper.
The author thanks anonymous referees whose comments were found very helpful.}\\
University of Wisconsin-Madison\\
junho.choi@wisc.edu}

\begin{document}
\maketitle
\begin{abstract}
The ``design phase'' refers to a stage in observational studies, during which a researcher constructs a subsample that achieves a better balance in covariate distributions between the treated and untreated units. In this paper, we study the role of this preliminary phase in the context of linear regression, offering a justification for its utility. To that end, we first formalize the design phase as a process of estimand adjustment via selecting a subsample. Then, we show that covariate balance of a subsample is indeed a justifiable criterion for guiding the selection: it informs on the maximum degree of model misspecification that can be allowed for a subsample, when a researcher wishes to restrict the bias of the estimand for the parameter of interest within a target level of precision. In this sense, the pursuit of a balanced subsample in the design phase is interpreted as identifying an estimand that is less susceptible to bias in the presence of model misspecification. Also, we demonstrate that covariate imbalance can serve as a sensitivity measure in regression analysis, and illustrate how it can structure a communication between a researcher and the readers of her report.

\end{abstract}

\begin{mtchideinmaintoc}
\newpage
\begin{quote}
	``It's fairly clear that the first things you would report are some estimate of the average effect of the intervention, and some measure of uncertainty, standard error. $\ldots$ and we don't really have much agreement on what the third number should be, or the fourth, or the fifth $\ldots$ clearly one of the next few numbers has to be something to do with sensitivity analysis.''\\
	---Guido Imbens \citep{Imbens2022}
\end{quote}

\section{Introduction}
In their influential book, \citet[Chapter 15, p.337]{ImbensRubin2015} define the term ``design phase'' to refer to the process of constructing a subsample within which the treated and untreated units are more balanced in their covariate distributions than in the full sample, so that, ``within this selected sample, inferences are more robust and credible.'' 

In this paper, we validate these authors' claim about the design phase, specifically in the context of linear regression with binary treatment. For that, the design phase is first formalized as a process of estimand adjustment through subsample selection.\footnote{``Subsample selection'' here should not be confused with ``subsampling'' as an inference method.} Then, we establish covariate balance as a justifiable criterion for guiding this selection: For a given subsample, it informs on the maximum degree of misspecification we may allow for our regression model, when pursuing a target closeness between the estimand and the parameter of interest. Thus, selecting a subsample with better balance can be understood as finding an estimand that is less prone to bias---i.e., ``credible''---despite possible misspecification of the regression model---i.e., ``robust.''

To fix ideas, 
let $\mathcal S$ be a finite random sample given to a researcher, each unit $i\in\mathcal S$ characterized by a scalar outcome $Y_i$, $p$-dimensional controls $X_i$, and an indicator $D_i$ for treated status. Let $\eG$ denote the empirical distribution of $(X_i',D_i)'$ within $\mathcal S$, $\eG^d$ the empirical conditional distribution of $X_i$ given $D_i=d$ within $\mathcal S$, where $d\in\{0,1\}$, and $F$ the population conditional distribution $\mathcal L_{Y|X,D}$ of $Y_i$ given $(X_i',D_i)'$.
 
 We study a scenario in which a researcher targets a parameter $\tau_{\eG^1,F}$---determined by $\eG^1$ and $F$---that identifies the average treatment effect on the treated units within $\mathcal S$---i.e., those with $D_i=1$---under the standard conditional independence assumption, and employs a least squares estimator $\hat\beta$ derived from $\mathcal S$ to estimate $\tau_{\eG^1,F}$. Due to the misspecification of her regression model, however, the estimand $\beta_{\eG,F}$ that $\hat\beta$ identifies in general differs from $\tau_{\eG^1,F}$, preventing valid inference for $\tau_{\eG^1,F}$ using $\hat\beta$. Nevertheless, there may still exist a \emph{subsample} $\mathcal S^*\subseteq \mathcal S$ within which the estimand $\beta_{\eG^*,F}$ and the parameter $\tau_{\eG^{*1},F}$ are close in value, i.e., $|\beta_{\eG^*,F}-\tau_{\eG^{*1},F}|\approx 0$.\footnote{$\eG^*$ denotes the empirical distribution of $(X_i',D_i)'$ within $\mathcal S^*$, and $\eG^{*1}$ the empirical conditional distribution of $X_i$ given $D_i=1$ within $\mathcal S^*$.} If so, the researcher can carry out valid inference for $\tau_{\eG^{*1},F}$ using the same least squares estimator $\hat\beta^*$, but now derived from $\mathcal S^*$.\footnote{We will show that $\hat\beta^*$ identifies $\beta_{\eG^*,F}$.}

This is our formalization of the design phase. That is, it is a process of selecting a subsample $\mathcal S^*$ whose $\beta_{\eG^*,F}$ and $\tau_{\eG^{*1},F}$ are close enough that we can conduct inference for $\tau_{\eG^{*1},F}$ through $\hat\beta^*$. The problem is, how one can know beforehand whether the two objects, both of which depend on the population,\footnote{Both $\beta_{\eG^*,F}$ and $\tau_{\eG^{*1},F}$ depend on $F$, which is unknown.} are close in finite samples. During the design phase, a researcher explores subsamples, compares their covariate balance, and selects the one that exhibits better balance. In light of our formalization, this can be justified only if the covariate balance of a subsample informs us of their closeness, i.e., $|\beta_{\eG^*,F}-\tau_{\eG^{*1},F}|$, which we call the ``bias'' of a linear regression.

In this paper, we formally show that this is indeed the case. To that end, we first establish a general representation of their difference, i.e., $\beta_{G,F}-\tau_{G^1,F}$, for an arbitrary distribution $G$ of $(X_i',D_i)'$.\footnote{$G^1$ denotes the conditional distribution of $X_i$ given $D_i=1$.} To be specific, in Section \ref{sec:representation-bias}, we represent their difference as an inner product of two basic functions, each resulting from misspecification of the regression model and covariate imbalance. Next, using this representation, in Section \ref{sec:bias-bounds}, we bound the bias, i.e., $|\beta_{G,F}-\tau_{G^1,F}|$, by the product of two objects---denoted by $m_{G,F}$ and $c_G$---each of which captures the overall degree of model misspecification and covariate imbalance.

These results
are then applied in Section \ref{sec:role-design-phase} when formalizing the aim of assessing covariate balance of subsamples in the design phase. Because the bias of the estimand $\beta_{\eG^*,F}$ within a subsample $\mathcal S^*$, i.e., $|\beta_{\eG^*,F}-\tau_{\eG^{*1},F}|$, is bounded by $m_{\eG^*,F}$ times $c_{\eG^*}$, the ``balance'' in the distributions of the controls, i.e., $1/c_{\eG^*}$---which can be \emph{computed} from $\mathcal S^*$---becomes the order of model misspecification $m_{\eG^*,F}$ a researcher may compromise, when she aims to restrict the bias below a pre-specified tolerance $\epsilon>0$. To be specific, if she wishes to be confident of $\tau_{\eG^{*1},F}$ using $\beta_{\eG^*,F}$ at the precision of $\epsilon$, $m_{\eG^*,F}$ can be allowed for up to $\epsilon/c_{\eG^*}$. In this regard, constructing a subsample with better balance in the design phase is equivalent to searching for an estimand, with which a researcher can be assured of less bias resulting from the misspecification of her regression model.

Our justification of the design phase can be approached from another perspective---as a framework for a structured communication between a researcher and the readers of her report. To explain, the bias bound takes various forms depending on the metric used to measure covariate imbalance $c_{\eG^*}$, which naturally defines a unique metric for model misspecification $m_{\eG^*,F}$. In turn, two measurements of imbalance under different metrics contain separate information on the robustness of an estimand. A researcher may therefore decide to report multiple forms of imbalance measurements to offer her readers rough ideas of \emph{how} her estimand $\beta_{\eG^*,F}$ is robust. The readers can then initiate perturbations of the true state $f$ starting from the researcher's model $l_{\theta_{\eG^*,F}}^s$,\footnote{$f$ and $l_{\theta_{\eG,F}}^s$ denote the conditional expectation and regression functions, respectively. Detailed definitions are provided in later sections.} compute their deviations $m_{\eG^*,F}$ under different misspecification metrics, and multiply them by the corresponding reported $c_{\eG^*}$'s to determine the maximum possible bias under the adversarial scenario they propose, i.e., the perturbed $f$. The smaller the $c_{\eG^*}$ achieved during the design phase, the larger the audience with less favorable $f$'s the researcher can persuade. Section \ref{sec:application} provides an empirical illustration of this communication.

\subsection{Related literature}

This paper pertains to the conditional estimand literature, as $\beta_{\eG^*,F}$ is the conditional linear projection of the outcome given the in-subsample empirical distribution of the controls and treated status.\footnote{assuming $Y_i$ has a finite second moment.} We extend the full-sample results of \citet{AbadieImbensZheng2014} to arbitrary subsamples in Section \ref{sec:role-design-phase}, where we establish the asymptotic properties of $\hat\beta^*-\beta_{\eG^*,F}$. In fact, our formalization of the design phase offers another motivation to use conditional estimands: A researcher can assess their population properties based on statistics in finite samples. 

The importance of securing covariate balance has been extensively emphasized in the literature. \citep[among others.]{CrumpHotzImbensMitnik2009, Imbens2015, ImbensRubin2015,HoImaiKingStuart2007,Stuart2010} The literature proposes various methods for that purpose---translated into our terms, systematic procedures for selecting $\mathcal S^*$---such as matching or trimming. In this paper, however, we do not limit ourselves to a specific balancing method. That is, we do not restrict the functional form of $\mathcal S^*$. Our focus is rather on formalizing and justifying the use of a general balancing process,
	especially when we use linear regression.

This study departs from the design-based inference literature \citep[among others]{AbadieAtheyImbensWooldridge2020} in that our inference conditions on the realizations of treated status as well as controls, so that statistical uncertainty only arises from the model side, i.e., $F$. By the same token, it also departs from \citet{AbadieSpiess2022}, whose focus is on unconditional inference.\footnote{However, it should be noted that they discuss conditional inference in their online appendix.}

This study contributes to the literature on sensitivity analysis \citep[among others]{BonhommeWeidner,AndrewsGentzkowShapiro2017}. Our results show that covariate imbalance serves as a sensitivity measure in regression analysis. We recommend that researchers report multiple imbalance metrics, as each computes the maximum possible bias along a different unit of perturbation.\footnote{This empirical practice can be particularly useful in settings where perfect balance is, in principle, unattainable, as in regression discontinuity designs.} This study also connects to the bias-aware inference literature \citep[among others]{ArmstrongKolesar2021}. We will detail how imbalance metrics can be used to robustify an existing confidence interval.

\section{Framework}\label{sec:framework}
We observe a finite random sample $\mathcal S$, consisting of $|\mathcal S|$ units, where each unit $i\in\mathcal S$ is characterized by a scalar outcome $Y_i\in\mathbb R$, a $p$-dimensional vector $X_i\in\mathbb R^p$ of controls, and an indicator $D_i\in\{0,1\}$ for treated status.\footnote{The notation $|\cdot|$ serves two meanings in this paper. When the argument is a set, it denotes its cardinality; when it is a vector, it indicates its absolute value.}
To simplify notation, we henceforth omit the subscript $i$ on variables unless required by context.

The parameter of our interest is
\begin{align}
	\tau\equiv \Exp[Y|D=1] - \Exp[\Exp[Y|X,D=0]|D=1],
\end{align}
which identifies the average treatment effect on the treated under the assumption that $Y(0)$ and $D$ are conditionally independent given $X$, where $Y(0)$ denotes the potential outcome when untreated.

Let $\mathcal X^d\subseteq \mathbb R^p$ denote the support of the conditional distribution $\mathcal L_{X|D=d}$ of $X$ given $D=d$, which we abbreviate by $G^d$. Assume that each $G^d$ is dominated by a common measure $\mu$, which may be, for example, a counting measure or a Lebesgue measure.\footnote{Such a measure always exists; $\sum_{d\in\{0,1\}}G^d$ is one example.} Then, we can rewrite $\tau$ as
\begin{align}
	\int (\Exp[Y|X=x,D=1] - \Exp[Y|X=x,D=0])\Pr[X=x|D=1]\mu(\mathrm dx),
\end{align}
where $\Pr[X=\cdot|D=d]\equiv(dG^d/d\mu)(\cdot)$ denotes the Radon-Nikodym derivative of $G^d$ with respect to $\mu$.
This representation is immediate but yields a pivotal observation. Note that, in the present setup, (i) the joint distribution between $X$ and $D$, which we denote by $G$, and (ii) the conditional distribution $\mathcal L_{Y|X,D}$ of $Y$ given $(X',D)'$, which we abbreviate by $F$, provide a complete description of the population. Thus, equation (2) reveals how each component of the population $(G,F)$ interacts with the parameter. In this regard, whenever there is a need to make explicit the dependence of an object on the components of a population, we append the corresponding subscript to signify that specific relationship; for instance, $\tau_{G^1,F}$, $\Pr_F$, or $\Exp_G$.

Now, consider a linear regression model
\begin{align}
	Y = \alpha + \beta D + s(X)'\gamma + E, \label{eq:reg-model}
\end{align}
where $s(\cdot)$ maps $X$ to a vector $s(X)$ of covariates, and $E$ denotes the regression error. We anchor the interpretation of $\beta$ by imposing the following conditions.
\begin{assume} \label{ass:regularity}
Let $Z\equiv (1,D,s(X)')'$. $\Exp[\|Z\|^2]$ and $\Exp[\|ZE\|]$ are finite, $\Exp[ZZ']$ is positive definite, and $\Exp[ZE]=\mathbf 0$.
\end{assume}
\noindent Let $\theta\equiv(\alpha,\beta,\gamma')'$. Under Assumption \ref{ass:regularity},
\begin{align}
	\theta=(\Exp[ZZ'])^{-1}\Exp[Z\Exp[Y|X,D]],
\end{align}
which is a vector of linear projection coefficients provided $\Exp[Y^2]$ exists. In particular,
\begin{align}
	\beta = \frac{\Exp[\dot D\Exp[Y|X,D]]}{\Exp[\dot D^2]}, \label{eq:beta}
\end{align}
where $\dot D$ denotes the population residual from the linear projection of $D$ on $(1,s(X)')'$.
We refer to $\theta$---specifically, $\beta$---as the ``regression estimand,'' as it is the limit to which the least squares estimator converges to.  

\paragraph{Example 1} We consider a toy example in which $Y$ is generated as
\begin{align}
	Y = D + DX + X + U,	
\end{align}
where $D$ and $X$ are binary variables such that $\Pr[D=d,X=x]=(1/2-p)\mathbf 1\{d=x\}+p\mathbf 1\{d\neq x\}$ for some $p\in(0,1/2)$ and $U$ is an arbitrary random variable that has a finite first moment, $\Exp[U]=0$, and is independent from $(X,D)'$. Because $X$ is binary conditional on $D$, $\mu=\delta_0 + \delta_1$ can be chosen as a counting measure, where $\delta_x$ denotes a Dirac measure defined as $\delta_x[B]\equiv \mathbf 1\{x\in B\}$ for a Borel set $B$. $G^1$ and $G^0$ are both Bernoulli, where $\Pr[X=1|D=1]=(1/2-p)/(1/2)$ and $\Pr[X=1|D=0]=p/(1/2)$. Note that $p$ parametrizes the dependence between $D$ and $X$. The farther it deviates from $1/4$, the more dependent they become. Also, in this setting, 
\begin{align*}
	\tau = \Exp[Y|D=1]-\Exp[\Exp[Y|X,D=0]|D=1]=(1+2(1-2p))-(1-2p)=2-2p.
\end{align*}
The following two specifications are considered for the regression model:
\begin{description}[itemsep=0em,topsep=3pt]
	\item[Specification A:] $Y = \alpha_A + \beta_A D + E$ and
	\item[Specification B:] $Y = \alpha_B + \beta_B D + \gamma_B X + E$.
\end{description}
Both of them are ``misspecified'' in the sense that their functional forms are different from the functional form of the conditional expectation, which includes the interaction term between $X$ and $D$. We will use this example to illustrate our results. \qed

\section{Bias representation} \label{sec:representation-bias}
We establish a general relationship between $\beta$ and $\tau$. Specifically, we provide a novel representation of their difference $\beta-\tau$, the absolute value of which we refer to as the ``bias'' of a linear regression.

Let $(x',d)'\in\cup_{d\in\{0,1\}}\mathcal X^d\times \{0,1\}$. Abbreviate by $f(x,d)$ the conditional expectation $\Exp[Y|X=x,D=d]$ of $Y$ given $(X',D)'=(x',d)'$, and by $g^d(x)$ the conditional density $\Pr[X=x|D=d]$ of $X$ at $x$ given $D=d$. Define $l_\theta^s(x,d) \equiv \alpha + \beta d + s(x)'\gamma$ to be the value of the regression function at $(x',d)'$.

The deviation of $f$ from $l_\theta^s$ results from the misspecification of the regression model \eqref{eq:reg-model}: Let $\Theta$ denote the parameter space for $\theta$. Then, $f\in\{l_\theta^s:\theta\in\Theta\}$ implies $f-l_\theta^s=0$. The discrepancy between the conditional densities $g^d$ is due to the covariate imbalance in the population: The greater the overlap between $G^d$'s, the closer $g^1-g^0$ is to zero.

We now state our result:
\begin{prop} \label{thm:representation}
Suppose that Assumption \ref{ass:regularity} holds. Then,
\begin{align}
\beta - \tau
= \int \underbrace{(f(x,0)-l_\theta^s(x,0))}_{\text{model misspecification}}\overbrace{(g^1(x)-g^0(x))}^{\text{covariate imbalance}}\mu(\mathrm dx), \label{eq:representation}
\end{align}
where we note that $(g^1(x)-g^0(x))\mu(\mathrm dx)=(G^1-G^0)(\mathrm dx)$.
\end{prop}

In Proposition \ref{thm:representation}, the bias of
linear regression is represented as an inner product of two basic functions $f(\cdot,0)-l_\theta^s(\cdot,0)$ and $g^1-g^0$. It recovers the existing results for the regression estimand $\beta$ under an ideal scenario in which either the regression model is correctly specified or the treatment is randomly assigned: Under correct specification, $f(\cdot,d)-l_\theta^s(\cdot,d)=0$, and the representation implies $\beta=\tau$. Under random assignment, 
	$g^1-g^0=0$, and again, by the representation, $\beta=\tau$.\footnote{This particular result has been shown by \citet[Chapter 7]{ImbensRubin2015}.} 
	
However, it also allows us to explore the middle ground where neither condition is exactly satisfied. For instance, the representation shows that, if the two basic functions are close to zero, $\beta$ should also be close to $\tau$. In other words, the bias of a regression is continuous with respect to model misspecification and covariate imbalance. 
	
Recall that we introduced the term ``bias'' to refer to the \emph{absolute} value of $\beta-\tau$. Thus, there is a subtlety in referring to 
	equation \eqref{eq:representation} as a bias representation. 
	Here, it is assumed that the sign of the outcome $Y$ is pre-adjusted based on that of $\beta-\tau$ by multiplying it with $\sgn(\beta-\tau)$.\footnote{$\sgn(\cdot)$ denotes the sign function.} This would then ensure the positivity of $\beta-\tau$.
	
For each $d\in\{0,1\}$, let $\mathcal T^d$ denote the support of the push-forward $s\sharp G^d$ through the covariate function $s$ of $G^d$. Abbreviate $s\sharp G^d$ by $G_s^d$. For each $(t',d)'\in\cup_{d\in\{0,1\}}\mathcal T^d\times \{0,1\}$, let $\Exp_{F_s}[Y|s(X)=t,D=d]$ denote the conditional expectation 
\begin{align*}
	\Exp_G[\Exp_F[Y|X,D]|s(X)=t,D=d]
\end{align*}
of $Y$ given $(s(X)',D)'=(t',d)'$. Define
\begin{align}
	\tau_{G_s^1,F_s^{}}\equiv \Exp[Y|D=1] - \Exp[\Exp[Y|s(X),D=0]|D=1]
\end{align}
to be a parameter that treats the covariates $s(X)$ as controls.

Let $(t',d)'\in\cup_{d\in\{0,1\}}\mathcal T^d\times \{0,1\}$. Abbreviate $\Exp_{F_s}[Y|s(X)=t,D=d]$ by $f^s(t,d)$, and let $l_\theta(t,d)\equiv \alpha + \beta d + t'\gamma$, so that $l_\theta$ defines the regression function at the level of the covariates $s(X)$.\footnote{$l_\theta(s(x),d)=l_\theta^s(x,d)$.} Proposition \ref{thm:representation} then yields
\begin{align}
	\beta - \tau_{G_s^1,F_s^{}} = \int (f^s(t,d)- l_\theta(t,d))(G_s^1-G_s^0)(\mathrm dt). \label{eq:representation-covariates}
\end{align}
Although this is immediate, as will be illustrated in Section \ref{sec:application}, it can be a useful device for reducing the original problem of assessing the bias to a simpler one, notably when $s(X)$ is a balancing score so that $\tau=\tau_{G_s^1,F_s^{}}$, and $X$ is of higher dimension than $s(X)$. 

In the appendix, we establish similar bias representations for the regression model extended to include interaction terms between $X$ and $D$, and for other parameters.

\paragraph{Example 1 (Cont')} We illustrate how equation \eqref{eq:representation} operates in our toy example. Let $l^A$ and $l^B$ denote the population regression functions for each specification. Then, it can be shown that $l^A(x,d)=2p + (3-6p)d$ and $l^B(x,d)=-p+(3/2)x + (3/2)d$. Hence, for Specification A, equation \eqref{eq:representation} holds in the form of
\begin{align*}
	1- 4p = \beta_A - \tau
	&= \sum_{x\in\{0,1\}}(f(x,0)-l^A(x,0))(g^1(x)-g^0(x))\\
	&= (1-2p)((1-2p)-2p) + (0-2p)(2p-(1-2p)) 
	= 1 - 4p, 
\end{align*}
and for Specification B, it takes the form of
\begin{align*}
	-1/2 + 2p = \beta_B - \tau
	&= \sum_{x\in\{0,1\}}(f(x,0)-l^B(x,0))(g^1(x)-g^0(x))\\
	&= (1-(-p+3/2))((1-2p)-2p) + (0-(-p))(2p-(1-2p))\\[0.5em]
	&= -1/2 + 2p.\tag*{\qed}
\end{align*}


\section{Bias bounds} \label{sec:bias-bounds}

The inner product structure of the bias representation \eqref{eq:representation}
yields useful bounds for the bias of a linear regression. All of them share the same structure of the multiplication of the two objects $m$ and $c$, each of which measures the degree of misspecification of the regression model and covariate imbalance, that is, 
\begin{align}
	|\beta-\tau|\leq mc. \label{eq:bias-bound}
\end{align}
This common structure of the bias bounds shows that misspecification of a regression model can be allowed for on the order of $1/c$. If a researcher wishes to be confident of $\tau$ using her $\beta$ at the precision of $\epsilon$, she may not concern much about misspecification of her regression model up to $\epsilon/c$. In an extreme case where $c=0$, the specification of a regression model becomes irrelevant to the identification of $\tau$.\footnote{However, we note that even when $c=0$, the specification of a regression model can be important in terms of the ``estimation.'' In some cases, we can attain higher efficiency. See \citet{Freedman2008a,Freedman2008b}, \citet{Lin2013}, and \citet{NegiWooldridge2021} for relevant discussion.}

An important observation will be that $c$ depends solely on the joint distribution of $X$ and $D$, i.e., $G$. This plays a critical role in justifying the design phase in Section \ref{sec:role-design-phase}.

In this section, we offer three forms of bias bound \eqref{eq:bias-bound}, which differ in how model misspecification and covariate imbalance are measured.
	Taken together, they yield a rich description of the robustness of a regression. Three additional ones are provided in the appendix: total variation, density ratio, and L\'evy-Prokhorov bounds.

\subsection{Kolmogorov-Smirnov bound}
The bound of the first form is referred to as the ``Kolmogorov-Smirnov bound,'' since it uses the Kolmogorov-Smirnov distance to measure the degree of covariate imbalance. To streamline the discussion, here we assume that $X$ is one-dimensional, i.e., $p=1$.\footnote{Nevertheless, one can always operationalize Corollary \ref{thm:bias-bound-KS} through the estimated indices, as detailed in Section \ref{sec:application}.}
\begin{assume} \label{ass:continuity}
	$f(\cdot,0)-l_\theta^s(\cdot,0)$ is continuous and bounded on $\cup_{d\in\{0,1\}}\mathcal X^d$.
\end{assume}
Let $\mathcal H$ be the class of all continuous and bounded extensions $h$ of $f(\cdot,0)-l_\theta^s(\cdot,0)$ on $\mathbb R$. That is, if $h\in\mathcal H$, then $h(x)=f(x,0)-l_\theta^s(x,0)$, for every $x\in\cup_{d\in\{0,1\}}\mathcal X^d$. $\mathcal H$ is then non-empty by Assumption \ref{ass:continuity}. Define
\begin{align}
	m^{\KS} &\equiv \inf_{h\in\mathcal H}V_{-\infty}^\infty[h] \text{ and}\\
	c^{\KS} &\equiv \sup_{x\in\mathbb R}|G^1(x)-G^0(x)|,
\end{align}
where $V_{-\infty}^\infty[\cdot]$ denotes the total variation of the argument on $\mathbb R$, and $G^d(x)$ the value of the distribution function of $G^d$ at $x$. $m^{\KS}$ measures the total variation of $f(\cdot,0)-l_\theta^s(\cdot,0)$ on $\cup_{d\in\{0,1\}}\mathcal X^d$, and $c^{\KS}$ is the Kolmogorov-Smirnov distance between $G^1$ and $G^0$.
\begin{cor} \label{thm:bias-bound-KS}
Suppose that the conditions of Proposition \ref{thm:representation} are satisfied. Also, suppose that Assumption \ref{ass:continuity} holds. Then,
\begin{align}
	|\beta - \tau| \leq m^{\KS}c^{\KS}.	\label{eq:bias-bound-KS}
\end{align}
\end{cor}

\paragraph{Example 1 (Cont')} We show how equation \eqref{eq:bias-bound-KS} works in our toy example. Since $\mathcal X^d=\{0,1\}$, $c^{\KS}=|G^1(0)-G^0(0)|=|1-4p|$. For Specification A, $m_A^{\KS}=|(f(1,0)-l^A(1,0))-(f(0,0)-l^A(0,0))|=|(1-2p)-(0-2p)|=1$, and similarly for Specification B, $m_B^{\KS}=|(1-(-p+3/2))-(0-(-p))|=1/2$. Then, equation \eqref{eq:bias-bound-KS} becomes equality in the forms of
\begin{align*}
	|1-4p|=|\beta_A - \tau|&= m_A^{\KS}c^{\KS}=1\times |1-4p| \text{ and }\\
	|-1/2+2p|=|\beta_B - \tau|&= m_B^{\KS}c^{\KS}=(1/2)\times |1-4p|. \tag*{\qed}
\end{align*}

\subsection{Monge-Kantorovich/Wasserstein bound}
The bound of the second form is called the ``Monge-Kantorovich/Wasserstein bound,'' as it uses the Monge-Kantorovich or Wasserstein distance to measure covariate imbalance. Let $\Pi(G^1,G^0)$ be the set of all probability measures $\pi$ on $\cup_{d\in\{0,1\}}\mathcal X^d\times \cup_{d\in\{0,1\}}\mathcal X^d$ such that its push-forward onto the first coordinate is $G^1$, and onto the second $G^0$.

Define 
\begin{align}
	c^{\MKW}&\equiv \inf_{\pi\in\Pi(G^1,G^0)}\int\|x_1-x_2\|\pi(\mathrm dx_1\mathrm dx_2)\text{ and}\\
	m^{\MKW}&\equiv	\|f(\cdot,0)-l_\theta^s(\cdot,0)\|_{\Lip},
\end{align}
where $\|\cdot\|_{\Lip}$ denotes the Lipschitz seminorm of the argument, i.e.,
\begin{align*}
	\|f(\cdot,0)-l_\theta^s(\cdot,0)\|_{\Lip} \equiv \sup_{x_1\neq x_2'\in\cup_{d\in\{0,1\}}\mathcal X^d}\frac{|f(x_1,0)-l_\theta^s(x_1,0)-(f(x_2,0)-l_\theta^s(x_2,0))|}{\|x_1-x_2\|}.
\end{align*}
Then, we have the following result:
\begin{cor}\label{thm:bias-bound-MKW}
	Suppose that the conditions of Proposition \ref{thm:representation} hold. Then,
	\begin{align}
		|\beta-\tau|\leq m^{\MKW}c^{\MKW}. \label{eq:bias-bound-MKW}
	\end{align}
\end{cor}

\paragraph{Example 1 (Cont')} Equation \eqref{eq:bias-bound-MKW} applies to the preceding example in the following manner. 
Because $(\mathrm d/\mathrm dx)(f(x,0)-l^A(x,0))=(\mathrm d/\mathrm dx)(x-2p)=1$ and $(\mathrm d/\mathrm dx)(f(x,0)-l^B(x,0))=(\mathrm d/\mathrm dx)(x-(-p+(3/2)x))=-1/2$, we have $m_A^{\MKW}=1$ and $m_B^{\MKW}=|-1/2|=1/2$. Also, by the theorem of \citet{Vallender1974}, 
\begin{align*}
	c^{\MKW}=\int_{-\infty}^\infty |G^1(x)-G^0(x)|\mathrm dx = (1-0)\times |G^1(0)-G^0(0)| = |1-4p|.
\end{align*}
Therefore, equation \eqref{eq:bias-bound-MKW} takes the forms of
\begin{align*}
	|1-4p|=|\beta_A-\tau|&=m_A^{\MKW}c^{\MKW}=1\times |1-4p|\text{ and}\\
	|-1/2+2p|=|\beta_B-\tau|&=m_B^{\MKW}c^{\MKW}=(1/2)\times |1-4p|.\tag*{\qed}
\end{align*}

A slightly weaker definition of $m^{\MKW}$ could have been considered, for instance,
\begin{align}
	\sup_{\pi\in\Pi(G^1,G^0)}	\biggl(\int_{x_1\neq x_2} \biggl(\frac{|f(x_1,0)-l_\theta^s(x_1,0)-(f(x_2,0)-l_\theta^s(x_2,0))|}{\|x_1-x_2\|}\biggr)^2\pi(\mathrm dx_1\mathrm dx_2)\biggr)^{\frac{1}{2}} \label{eq:bias-bound-MKW-m}
\end{align}
by strengthening that of $c^{\MKW}$ to
\begin{align}
	\inf_{\pi\in \Pi(G^1,G^0)} \biggl(\int \|x_1-x_2\|^2\pi(\mathrm dx_1\mathrm dx_2)\biggr)^{\frac{1}{2}}, \label{eq:bias-bound-MKW-c}
\end{align} 
which is called the quadratic Wasserstein distance.

\subsection{Mean difference bound}\label{sec:bias-bound-MD}
In practice, it is common to compare the means of several ``summaries'' $r_j(X),j\in\mathcal J$ of the controls $X$, i.e., $\Exp_{G^d}[r_j(X)|D=d]$, between the treated and untreated---rather than their distributions $G^d$'s---to assess covariate balance. Practices span from simply comparing the means of each control---which amounts to setting $\mathcal J=\{1,\dots,p\}$ and $r_j(X)=X_j$---to comparing those of more compact summaries, such as the propensity scores $\Pr[D=1|X]$. \citet[Section 14.2]{ImbensRubin2015}, in particular, recommend using the normalized difference in the means.

The bound of the final form, which we call the ``mean difference bound,'' explains how the mean differences in summaries are informative for the bias of the regression estimand. Here, we limit our considerations to the cases where both $G^d$'s are finitely supported, i.e., when $\mathcal X^d$'s are finite.

Let $r(X)\equiv (r_j(X))_{j\in\mathcal J}$ denote a $|\mathcal J|$-dimensional vector which collects the summaries in column. Fix a non-negative number $L\in[0,\infty]$. Consider a pair of convex optimization problems, defined for each $\sigma=\pm 1$ in the following manner:
\begin{align}
	\minimize_{\zeta_\sigma\in\mathbb R^{\mathcal J},\xi_\sigma^1\in\mathbb R_+^{\mathcal X^1},\xi_\sigma^0\in\mathbb R_+^{\mathcal X^0}} \|\zeta_\sigma\|^2 + L\sum_{d\in\{0,1\}}\sum_{x\in\mathcal X^d}\xi_\sigma^d(x)
\end{align}
subject to the inequality constraints
\begin{equation}
\begin{aligned}
	\sigma(f(x,0)-l_\theta^s(x,0)) &\leq \sigma\cdot r(x)'\zeta_\sigma  + \xi_\sigma^1(x) \text{ for every }x\in\mathcal X^1\text{ and}\\
	\sigma(f(x,0)-l_\theta^s(x,0)) &\geq \sigma\cdot r(x)'\zeta_\sigma -\xi_\sigma^0(x) \text{ for every }x\in\mathcal X^0,
\end{aligned}
\end{equation}
where $\mathbb R_+$ denotes the set of all non-negative real numbers. Let $(\zeta_\sigma^\star,\xi_\sigma^{1\star},\xi_\sigma^{0\star})$ denote a solution for each problem.

The values of $\xi_1^{d\star}(x)$'s reflect how well the graphs
\begin{align*}
	\Gr_{f(\cdot,0)-l_\theta^s(\cdot,0)}^{\mathcal X^d}\equiv \{(x',f(x,0)-l_\theta^s(x,0))'\}_{x\in\mathcal X^d}
	\subseteq \mathcal X^d\times \mathbb R	
\end{align*}
of the function $f(\cdot,0)-l_\theta^s(\cdot,0)$ on $\mathcal X^d$'s are separated by the functional
\begin{align*}
	(x',y)'\in\cup_{d\in\{0,1\}}\mathcal X^d\times \mathbb R\mapsto r(x)'\zeta_1^\star-y.
\end{align*}
If $\xi_1^{d\star}(x)$'s are all zero, then the level set $\{(x',y)'\in\cup_{d\in\{0,1\}}\mathcal X^d\times \mathbb R:r(x)'\zeta_1^\star-y=0\}$ sharply separates $\Gr_{f(\cdot,0)-l_\theta^s(\cdot,0)}^{\mathcal X^d}$'s such that $\Gr_{f(\cdot,0)-l_\theta^s(\cdot,0)}^{\mathcal X^1}$ is located below it. However, when some of them are positive, the functional achieves only a noisy separation, which allows for the corresponding errors, i.e., $\xi_1^{d\star}(x)$'s. 

A similar discussion can be made for $\sigma=-1$, with one critical difference being that $\Gr_{f(\cdot,0)-l_\theta^s(\cdot,0)}^{\mathcal X^1}$ now lies \emph{above} the level set $\{(x',y)'\in\cup_{d\in\{0,1\}}\mathcal X^d\times \mathbb R:r(x)'\zeta_{-1}^\star-y=0\}$ of the functional 
\begin{align*}
	(x',y)'\in\cup_{d\in\{0,1\}}\mathcal X^d\times \mathbb R\mapsto r(x)'\zeta_{-1}^\star - y
\end{align*}
provided it sharply separates the graphs. Thus, taken together, we may interpret the supremum $|\zeta_1^\star|\vee|\zeta_{-1}^\star|$ of the absolute values $|\zeta_\sigma^{d\star}|$ as the degree of misspecification; if both $r(\cdot)'\zeta_\sigma^\star$'s attain sharp separation, it follows that 
\begin{align}
	|f(\cdot,0)-l_\theta^s(\cdot,0)|\leq (|\zeta_1^\star|\vee|\zeta_{-1}^\star|)'|r(\cdot)|.
\end{align}

Now, define
\begin{align}
	m_\sigma^{\MD}&\equiv |\zeta_\sigma^\star|,\text{ where $\sigma=\pm 1$, and} \label{eq:m-sigma}\\
	c^{\MD}&\equiv |\Exp_{G^1}[r(\cdot)]-\Exp_{G^0}[r(\cdot)]|;
\end{align}
$c^{\MD}$ is a $|\mathcal J|$-dimensional vector collecting the absolute values of the mean differences of the summaries $r_j(X)$ between the treated and untreated.\footnote{$\Exp_{G^d}[r(\cdot)]$ is shorthand for $\Exp[r(X)|X=d]$.} Let $m^{\MD}\equiv m_1^{\MD}\vee m_{-1}^{\MD}$.

\begin{cor}\label{thm:bias-bound-MD}
	Suppose that the conditions of Proposition \ref{thm:representation} are satisfied. 
	Then,
	\begin{align}
		-m_{-1}^{\MD\prime}c^{\MD}-\sum_{d\in\{0,1\}}\Exp_{G^d}[\xi_{-1}^{d\star}(\cdot)]\leq \beta - \tau\leq m_1^{\MD\prime}c^{\MD} + \sum_{d\in\{0,1\}}\Exp_{G^d}[\xi_1^{d\star}(\cdot)].\label{eq:bias-bound-MD}
	\end{align}
	In particular, $|\beta-\tau|\leq m^{\MD\prime}c^{\MD} + \sum_{\sigma=\pm 1}\sum_{d\in\{0,1\}}\Exp_{G^d}[\xi_\sigma^{d\star}(\cdot)]$.
\end{cor}
The presence of the $\sum_{d\in\{0,1\}}\Exp_{G^d}[\xi_\sigma^{d\star}(\cdot)]$'s in the previous equations formalizes the notion that marginal alignment of the conditional distributions $G^d$, i.e., a small value of $c^{\MD}$, does not ensure a small bias. The informativeness of $c^{\MD}$ hinges on the ability of the summaries to separate $\Gr_{f(\cdot,0)-l_\theta^s(\cdot,0)}^{\mathcal X^d}$'s, i.e., small values of the slacks $\xi_\sigma^{d\star}$.

In fact, the optimalities of $(\zeta_\sigma^\star,\xi_\sigma^{1\star},\xi_\sigma^{0\star})$'s are not used in the proof of Corollary \ref{thm:bias-bound-MD}. Their use is limited to selecting an element from the constraint set. The result applies for an arbitrary feasible one.

\paragraph{Example 1 (Cont')} Equation \eqref{eq:bias-bound-MD} operates in our example in the following manner. Consider a case where a researcher uses a constant, i.e., $r_1(X)=1$, and the control $X$ itself, i.e., $r_2(X)=X$, as her summaries; hence, $\mathcal J=\{1,2\}$. Set $L=\infty$.

Since $x\mapsto f(x,0)-l^A(x,0)=-2p + x$ is a linear combination of the two summary functions, $\xi_{\sigma,A}^d=0$'s are feasible. Furthermore, since $L=\infty$, $\xi_{\sigma,A}^{d\star}=0$. It follows that $\zeta_{1,A}^\star=\zeta_{-1,A}^\star=(-2p,1)'$ and therefore $m_{1,A}^{\MD}=m_{-1,A}^{\MD}=(2p,1)'$. Similarly, the linearity of  $x\mapsto f(x,0)-l^B(x,0)=p-(1/2)x$ implies $\xi_{\sigma,B}^{d\star}=0$ and $m_{1,B}^{\MD}=m_{-1,B}^{\MD}=(p,1/2)'$.

Given
\begin{align*}
	c^{\MD}=(|1-1|,|(1/2-p)/(1/2)-p/(1/2)|)'=(0,|1-4p|)',
\end{align*}
equation \eqref{eq:bias-bound-MD} takes the forms of
\begin{align*}
	&\beta_A-\tau = 1-4p\\
	&\begin{cases}
		\leq m_{1,A}^{\MD\prime}c^{\MD} +\sum_{d\in\{0,1\}}\Exp_{G^d}[\xi_{1,A}^{d\star}(\cdot)]=(2p,1)'(0,|1-4p|) + 0 = |1-4p|\\
		\geq - m_{-1,A}^{\MD\prime}c^{\MD} - \sum_{d\in\{0,1\}}\Exp_{G^d}[\xi_{-1,A}^{d\star}(\cdot)]=-(2p,1)'(0,|1-4p|)-0=-|1-4p| 
	\end{cases}
	\intertext{and}
	&\beta_B-\tau=-1/2+2p\\
	&\begin{cases}
		\leq m_{1,B}^{\MD\prime}c^{\MD} + \sum_{d\in\{0,1\}}\Exp_{G^d}[\xi_{1,B}^{d\star}(\cdot)]=(p,1/2)'(0,|1-4p|) + 0 = (1/2)|1-4p|\\
		\geq - m_{-1,B}^{\MD\prime}c^{\MD} - \sum_{d\in\{0,1\}}\Exp_{G^d}[\xi_{-1,B}^{d\star}(\cdot)]=-(p,1/2)'(0,|1-4p|)-0=-(1/2)|1-4p|.
	\end{cases}\tag*{\qed}
\end{align*}

\section{Role of design phase} \label{sec:role-design-phase}

The bias bounds in Section \ref{sec:bias-bounds} clarify the role of \emph{population} covariate balance in linear regression: Small $c$ can keep the bias small even under severe misspecification of the regression model. This section clarifies that \emph{sample} covariate balance serves the same role as population balance, except that it controls the bias of the regression estimand to which the least squares estimator \emph{conditionally} converges to.

The outline of this section is as follows: Section \ref{sec:adjustment} formalizes the design phase as a process of estimand adjustment via selecting a subsample. Section \ref{sec:assessment} justifies the design phase by showing that the covariate balance of a selected subsample can serve as a useful tool for assessing the bias of the adjusted estimand. Section \ref{sec:inference} concerns statistical inference for the adjusted estimand. We describe how one can robustify an existing confidence interval against model misspecification using imbalance measures.

\subsection{Estimand adjustment} \label{sec:adjustment}
Let $\mathcal S^*\subseteq \mathcal S$ be the subsample constructed in the design phase. Let $\hat\theta^*=(\hat\alpha^*,\hat\beta^*,\hat\gamma^{*\prime})'$ denote the least squares estimator obtained by regressing $Y_i$ on $Z_i$ within $\mathcal S^*$, i.e., 
\begin{align}
	\hat\theta^*\equiv \biggl(\frac{1}{|\mathcal S^*|}\sum_{i\in \mathcal S^*}Z_iZ_i'\biggr)^{-1}\biggl(\frac{1}{|\mathcal S^*|}\sum_{i\in \mathcal S^*}Z_iY_i\biggr),
\end{align}
where we recall that $Z_i\equiv (1,D_i,s(X_i)')'$.

Define $\eG^*(\mathrm dx\mathrm dd)\equiv |\mathcal S^*|^{-1}\sum_{i\in\mathcal S^*}\delta_{(X_i',D_i)'}(\mathrm dx\mathrm dd)$ to be the empirical distribution of $(X_i',D_i)'$ within $\mathcal S^*$. Also, for each $d\in\{0,1\}$, define $\eG^{*d}(\mathrm dx)\equiv |\mathcal S^{*d}|^{-1}\sum_{i\in\mathcal S^{*d}}\delta_{X_i}(\mathrm dx)$, where $\mathcal S^{*d}\equiv \{i\in\mathcal S^*:D_i=d\}$, to be the empirical conditional distribution of $X$ given $D=d$ within $\mathcal S^*$. Let $\theta_{\eG^*,F}=(\alpha_{\eG^*,F},\beta_{\eG^*,F},\gamma_{\eG^*,F}')'$ denote the least squares estimator obtained from the hypothetical regression of $\Exp_F[Y|X=X_i,D=D_i]$ on $Z_i$, i.e., 
\begin{align}
	\theta_{\eG^*,F}\equiv \biggl(\frac{1}{|\mathcal S^*|}\sum_{i\in\mathcal S^*}Z_iZ_i'\biggr)^{-1}\biggl(\frac{1}{|\mathcal S^*|}\sum_{i\in\mathcal S^*}Z_i\Exp_F[Y|X=X_i,D=D_i]\biggr),
\end{align}
whose subscript is appropriate, since, recalling the notation introduced in Section \ref{sec:framework},
\begin{align*}
	\theta_{\eG^*,F}=(\Exp_{\eG^*}[ZZ'])^{-1}\Exp_{\eG^*}[Z\Exp_F[Y|X,D]].
\end{align*}

We make the following assumptions.
\begin{assume} \label{ass:function-subsample}
	$\mathcal S^*$ is a function of $\{(X_i',D_i)'\}_{i\in\mathcal S}$, i.e., $\mathcal S^*=\mathcal S^*(\{(X_i',D_i)'\}_{i\in\mathcal S})$.
\end{assume}
In the design phase, a subsample must be constructed solely based on the controls and treatment status. A researcher may use estimated propensity scores, for instance, as they do not depend on the outcomes. By Assumption \ref{ass:function-subsample}, $\eG^*$ and $\eG^{*d}$ are functions of $\{(X_i',D_i)'\}_{i\in\mathcal S}$.

\begin{assume}
	$s$ is a function of $\{(X_i',D_i)\}_{i\in\mathcal S}$, i.e., $s=s(\{(X_i',D_i)'\}_{i\in\mathcal S})$.	
\end{assume}

The covariate function $s$ may depend on $\{(X_i',D_i)'\}_{i\in\mathcal S}$. This allows the estimated propensity scores to be included as covariates.

\begin{assume} \label{ass:function-F}
	$F$ is a function of $\{(X_i',D_i)'\}_{i\in\mathcal S}$, i.e., $F=F(\{(X_i',D_i)'\}_{i\in\mathcal S})$.
\end{assume}

We allow the conditional distribution $\mathcal L_{Y|X,D}$ itself to vary with $\{(X_i',D_i)'\}_{i\in\mathcal S}$. By Assumption \ref{ass:function-F}, $Y_i,i\in\mathcal S$ can be considered as a random triangular array. 

Fix vectors $(x_i',d_i)',i\in\mathcal S$ in $\sqcup_{d\in\{0,1\}}(\mathcal X^d\times \{d\})$, which models the realized values of $(X_i',D_i)',i\in\mathcal S$ encountered by the researcher. Let $\mathbb X^{*d}$ denote the support of $\eG^{*d}$, i.e., $\mathbb X^{*d}=\mathbb X^{*d}(\{(X_i',D_i)'\}_{i\in\mathcal S})\equiv\{X_i:i\in\mathcal S^{*d}\}$. Let $U\equiv Y-\Exp_F[Y|X,D]$.

\begin{assume} \label{ass:subsample-large}
	$\lim_{|\mathcal S|\rightarrow \infty}|\mathcal S^*|=\infty$ given $\{(X_i',D_i)'\}_{i\in\mathcal S}=\{(x_i',d_i)'\}_{i\in\mathcal S}$.
\end{assume}
\begin{assume} \label{ass:bounded}
	For each $d\in\{0,1\}$, $\Exp_F[\|ZU\|^2|X=\cdot,D=d]$ is bounded on $\mathbb X^{*d}$ by a fixed constant given $\{(X_i',D_i)'\}_{i\in\mathcal S}=\{(x_i',d_i)'\}_{i\in\mathcal S}$.
\end{assume}

Define $\Gamma^*=\Gamma^*(\{(X_i',D_i)'\}_{i\in\mathcal S})\equiv|\mathcal S^*|^{-1}\sum_{i\in\mathcal S^*}Z_iZ_i$.
\begin{assume} \label{ass:design-matrix-inverse}
	For a symmetric positive semi-definite matrix $A$, let $\lambda_{\min}[A]$ denote its smallest eigenvalue. For some fixed positive constant $\underline\lambda>0$, 
\begin{align}
	\liminf_{|\mathcal S|\rightarrow\infty}\lambda_{\min}[\Gamma^*]\geq \underline\lambda
\end{align}
given $\{(X_i',D_i)'\}_{i\in\mathcal S}=\{(x_i',d_i)'\}_{i\in\mathcal S}$.
\end{assume}

\begin{prop} \label{thm:consistency}
Under Assumptions \ref{ass:function-subsample}--\ref{ass:design-matrix-inverse}, $\hat\theta^*$ converges in conditional probability to $\theta_{\eG^*,F}$ given $\{(X_i',D_i)'\}_{i\in\mathcal S}=\{(x_i',d_i)'\}_{i\in\mathcal S}$, i.e.,
\begin{align}
	\Pr_F[\|\hat\theta^* - \theta_{\eG^*,F}\|>\eta|\{(X_i',D_i)'\}_{i\in\mathcal S}]\rightarrow 0,\forall \eta>0\label{eq:consistency}
\end{align}
given $\{(X_i',D_i)'\}_{i\in\mathcal S}=\{(x_i',d_i)'\}_{i\in\mathcal S}$, as $|\mathcal S|$ tends to infinity.
\end{prop}

The type of convergence used in Proposition \ref{thm:consistency} is conditional. It asserts that, when $|\mathcal S|$ is large, if a researcher confronts multiple samples with the same realizations of $(X_i',D_i)'$'s, then $\hat\beta^*$---in particular---must exhibit little variation around $\beta_{\eG^*,F}$.

It relates to Theorem 1 of \citet{AbadieImbensZheng2014}, who show the unconditional version of equation \eqref{eq:consistency} in the case where $\mathcal S^*=\mathcal S$.\footnote{However, it should be noted that their result is not confined to linear regression.} Accordingly, this result can be viewed as an extension of their full-sample statement to arbitrary subsamples.

However, Proposition \ref{thm:consistency} provides a useful formalization of subsample construction in the design phase: It can be viewed as adjusting the regression estimand from $\beta_{\eG,F}$ to $\beta_{\eG^*,F}$, possibly under the expectation that the latter is closer to the corresponding parameter of interest, here assumed to be
\begin{align}
	\tau_{\eG^{*1},F}\equiv\frac{1}{|\mathcal S^{*1}|}\sum_{i\in\mathcal S^{*1}}(\Exp_F[Y|X=X_i,D=1]-\Exp_F[Y|X=X_i,D=0]).
\end{align}
Note that $\tau_{\eG^{*1},F}$ identifies the average treatment effect on the treated units within $\mathcal S^*$, under the conditional independence assumptions on $Y_i(0)$'s, and also that
\begin{align*}
	\tau_{\eG^{*1},F}= \Exp_{\eG^{*1}}[\Exp_F[Y|X,D=1]-\Exp_F[Y|X,D=0]|D=1],
\end{align*}
which explains the subscript $(\eG^{*1},F)$.


\subsection{Estimand assessment via covariate balance}\label{sec:assessment}
Pursuing different subsample construction strategies $\mathcal S\mapsto\mathcal S^*$---which result in different $\eG^*$'s---leads to different estimands $\beta_{\eG^*,F}$. In the design phase, a researcher pursues a subsample $\mathcal S^*$ with small covariate imbalance. Thus, the design phase can be justified if covariate imbalance can reveal a desirable attribute of $\beta_{\eG^*,F}$. Here, we elaborate on how imbalance measurements provide information about the robustness of $\beta_{\eG^*,F}$---or the sensitivity of its bias---to the misspecification of the researcher's model 
\begin{align*}
	l_{\theta_{\eG^*,F}}^s(\cdot,0)=(1,0,s(\cdot)')'\theta_{\eG^*,F}
\end{align*}
with respect to the true conditional expectation function $f(\cdot,0)=\Exp_F[Y|X=\cdot,D=0]$.

Assume that the conditions of either Corollary 
\ref{thm:bias-bound-KS}, 
\ref{thm:bias-bound-MKW}, or \ref{thm:bias-bound-MD} hold under $(\eG^*,F)$; they do, for instance, given $\{(X_i',D_i)'\}_{i\in\mathcal S}=\{(x_i',d_i)'\}_{i\in\mathcal S}$ under Assumption \ref{ass:design-matrix-inverse}. Then, 
\begin{align}
	|\beta_{\eG^*,F} - \tau_{\eG^{*1},F}|\leq m_{\eG^*,F}c_{\eG^*},\label{eq:bias-bound-subsample}
\end{align}
whose right-hand side is either 
Kolmogorov-Smirnov,
Monge-Kantorovich/Wasserstein, or mean difference bound,\footnote{
Clearly, for mean difference bound to take the form of equation \eqref{eq:bias-bound-subsample}, one requires extra conditions that ensure separabilities. Here, we limit ourselves to cases where the slacks are minimal.}
with the relevant superscripts suppressed. Although the only difference between inequalities \eqref{eq:bias-bound} and \eqref{eq:bias-bound-subsample} is their dependence on $G$ versus $\eG^*$, the practical significance of the latter is substantially higher, as $c_{\eG^*}$---unlike $c_G$---can be \emph{computed} from $\mathcal S^*$. One can now use a \emph{statistic} to constructively discuss the bias of the estimand $\beta_{\eG^*,F}$ in \emph{finite} samples. For instance, a researcher can argue, if $|\mathcal S^*|$ is large enough, that the bias of her estimate $\hat\beta^*$ is at most $\epsilon$ by appealing to the small Kolmogorov-Smirnov distance $c_{\eG^*}^{\KS}$ between $\eG^{*d}$'s. Unless the total variation $m_{\eG^*,F}^{\KS}$ of the misspecification function $f(\cdot,0)-l_{\theta_{\eG^*,F}}^s(\cdot,0)$ on $\cup_{d\in\{0,1\}}\mathbb X^{*d}$ exceeds $\epsilon/c_{\eG^*}^{\KS}$---which is large---the bias of $\beta_{\eG^*,F}$ must be no greater than $\epsilon$. In other words, bias bound \eqref{eq:bias-bound-subsample} provides a framework to \emph{translate} various evaluations $c_{\eG^*}$ of the quality of a ``design'' $\mathbb G^*$ into the performance of the least squares estimator $\hat\beta^*$.

As another example, mean difference bound enables a researcher to associate the numbers displayed in her balance tables---which compare the means $\Exp_{\eG^{*d}}[r_j(X)|D=d]$ of summaries $r_j(X),j\in\mathcal J$---directly to her main regression tables. To be more specific, she can refer to Corollary \ref{thm:bias-bound-MD} to articulate an argument that, if the true state $F$ is such that the graphs
\begin{align*}
	\Gr_{f(\cdot,0)-l_{\theta_{\eG^*,F}}^s(\cdot,0)}^{\mathbb X^{*d}}\equiv \{(x',f(x,0)-l_{\theta_{\eG^*,F}}^s(x,0))\}_{x\in\mathbb X^{*d}}
\end{align*}
of $f(\cdot,0)-l_{\theta_{\eG^*,F}}^s(\cdot,0)$ are well-separated by $r(\cdot)'\zeta_{\sigma,\eG^*,F}^{\star}$'s,\footnote{$\zeta_{\sigma,\eG^*,F}^\star$ refers to $\zeta_\sigma^\star$ defined in Section \ref{sec:bias-bound-MD}, except that the reference population is $(\eG^*,F)$ instead of $(G,F)$. The same applies to $m_{\eG^*,F}^{\MD}$.}  
	the small values $c_{\eG^*}^{\MD}$ observed from her balance tables ensure a small $\epsilon$-bias of $\beta_{\eG^*,F}$ even when $|f(\cdot,0)-l_{\theta_{\eG^*,F}}^s(\cdot,0)|$ can be covered only by $|r(\cdot)|'m_{\eG^*,F}^{\MD}$ with a large $m_{\eG^*,F}^{\MD}$.\footnote{as long as $m_{\eG^*,F}^{\MD\prime}c_{\eG^*}^{\MD}$ is smaller than $\epsilon$.} We detail in Section \ref{sec:application} how one could form a concrete statement about the bias of $\beta_{\eG^*,F}$ using summaries.

Note that $c_{\eG^*}$ decides the set $\mathcal F^\epsilon$ of states $F$ at which $\beta_{\eG^*,F}$ secures a bias less than $\epsilon$, and that the smaller value of $c_{\eG^*}$ implies larger $\mathcal F^\epsilon$. However, we highlight that the empirical situations we consider neither require a researcher's ability to fully process nor enumerate the members of $\mathcal F^{\epsilon}$. $\mathcal F^\epsilon$ can serve purely as a diagnostic tool to check whether a state $f(\cdot,0)$---posed against the researcher's model $l_{\theta_{\eG^*,F}}^s(\cdot,0)$---may carry a risk of bias larger than $\epsilon$. In Section \ref{sec:application}, we illustrate how the readers of a researcher's report may utilize $\mathcal F^\epsilon$ to gauge the robustness of her conclusion.

\paragraph{Example 2} We illustrate our point with a simple setup. Let $Y_i=X_i+U_i$, where $X_i$ is conditionally normal given $D_i=d$ such that $\mathcal L_{X_i|D_i=d}=\mathcal N(d,1)$,\footnote{$\mathcal N(\theta,\Sigma)$ denotes the normal distribution with mean $\theta$ and covariance matrix $\Sigma$.} and $U_i$ is an arbitrary random variable with a finite first moment, $\Exp[U_i]=0$, and is independent of $(X_i,D_i)$. A sample $\mathcal S$ with 24 units is given, i.e., $|\mathcal S|=24$, where the half of them, T1, $\dots$, T12, are treated, while the other half, U1, $\dots$, U12, are not. Table \ref{tab:empirical-conditional-distributions} shows the empirical conditional distributions $\mathbb G^d$ of $X_i$ given $D_i=d$ within $\mathcal S$.

\begin{table}[ht]
\bigskip\centering\small
\begin{tabular}{lcccccccccccc}
  \toprule 
  & U1 & U2 & U3 & T1 & U4 & T2 & U5 & U6 & T3 & T4 & U7 & T5\\
  \midrule 
 $\mathbb G^1$ &  &  &  & -1.35 &  & -0.91 &  &  & 0.16 & 0.33 &  & 0.44 \\ 
  $\mathbb G^0$ & -2.32 & -1.96 & -1.36 &  & -0.91 &  & 0.11 & 0.14 &  &  & 0.42 &  \\ 
   \bottomrule \\
 \toprule & U8 & T6 & T7 & U9 & U10 & T8 & U11 & U12 & T9 & T10 & T11 & T12\\
\midrule 
$\mathbb G^1$ &  & 0.70 & 0.75 &  &  & 0.92 &  &  & 1.95 & 2.16 & 2.22 & 3.19 \\ 
  $\mathbb G^0$ & 0.55 &  &  & 0.76 & 0.83 &  & 1.02 & 1.52 &  &  &  &  \\ 
   \bottomrule 
\end{tabular}
\caption{Empirical conditional distributions, $\eG^d$'s}
\label{tab:empirical-conditional-distributions}
\end{table}

We set up the regression model as $Y = \alpha + \beta D + E$, which is misspecified since it omits the interaction term between $X$ and $D$. Regressing $\Exp_F[Y|X=X_i,D=D_i]$ on $D_i$ within $\mathcal S$ produces the mean difference in $\Exp_F[Y|X=X_i,D=D_i]$ between treated and untreated units, i.e., $\beta_{\eG,F}=\Exp_{\eG^1,F}[Y|D=1]-\Exp_{\eG^0,F}[Y|D=0]$.  

The Kolmogorov-Smirnov distance $c_{\eG}^{\KS}$ between $\mathbb G^d$'s is $1/3$. This indicates that, if we wish to be confident of $\tau_{\eG^1,F}$ using $\beta_{\eG,F}$ at the precision of $\epsilon=0.5$, for example, the degree of model misspecification, i.e., $m_{\eG,F}^{\KS}$, should be less than 3/2.\footnote{Here, $\alpha_{\eG,F}$ is the mean of $\eG^0$, which is $-0.100$. Thus, $f(x,0)-l_{\theta_{\mathbb G,F}}^0(x,0)=x-\alpha_{\mathbb G,F}=x+0.100$, yielding $m_{\eG,F}^{\KS}=|(3.19+0.100)-(-2.32+0.100)|=5.51$, which is larger than $3/2$.}

Now, consider a case where we instead construct a subsample $\mathcal S^*$ consisting of U3 and T1, U4 and T2, U6 and T3, T5 and U8, and T7 and U9. Then, the Kolmogorov-Smirnov distance improves to $c_{\eG^*}^{\KS}=1/6$, and the upper limit of model misspecification $m_{\eG^*,F}^{\KS}$ is relaxed to 3.\footnote{$m_{\eG^*,F}^{\KS}=|(0.76+0.100)-(-1.36+0.100)|=2.12$; it is smaller than 3.}

We can perform similar exercises for the other bounds: 
The Monge-Kantorovich or Wasserstein distance $c_{\eG}^{\MKW}$ based on $\mathcal S$ is 0.981, while that based on $\mathcal S^*$, i.e., $c_{\eG^*}^{\MKW}$, is 0.039.
If we use as summaries a constant $1$ and $X$, the mean differences $c_{\eG}^{\MD}$ based on $\mathcal S$ are $(|1-1|,|0.881-(-0.100)|)'=(0,0.981)'$, while those based on $\mathcal S^*$, i.e., $c_{\eG^*}^{\MD}$, are $(|1-1|,|(-0.096)-(-0.067)|)'=(0,0.028)'$.

The shared implication of these computations is that the adjusted estimand $\beta_{\eG^*,F}$ is likely of a smaller bias. The small $c_{\eG^*}$'s suggest that $\beta_{\eG^*,F}$ is more robust than $\beta_{\eG,F}$ against various adversarial scenarios. Indeed, one can check that the bias $|\beta_{\eG,F}-\tau_{\eG^1,F}|$ of the latter is $|0.981-0|=0.981$, whereas that of the former, i.e., $|\beta_{\eG^*,F}-\tau_{\eG^{*1},F}|$, is $|(-0.028)-0|=0.028$. \qed

\subsection{Inference for adjusted estimand}\label{sec:inference}

We establish an asymptotic normality result for the deviation $\hat\beta^*-\beta_{\eG^*,F}$, and propose a corresponding standard error for $\hat\beta^*$. Combining these results with bias bound \eqref{eq:bias-bound-subsample}, we then establish a confidence procedure for $\tau_{\eG^{*1},F}$, which is robust to misspecification of the regression model.

Define $\Delta^*=\Delta^*(\{(X_i',D_i)'\}_{i\in\mathcal S})\equiv |\mathcal S^*|^{-1}\sum_{i\in\mathcal S^*}Z_iZ_i\Exp_F[U|X=X_i,D=D_i]$.
\begin{assume} \label{ass:variance-consistency}
For some fixed positive constant $\underline\lambda>0$,
\begin{align*}
	\liminf_{|\mathcal S|\rightarrow\infty}\lambda_{\min}[\Delta^*]\geq \underline\lambda
\end{align*}
given $\{(X_i',D_i)'\}_{i\in\mathcal S}=\{(x_i',d_i)'\}_{i\in\mathcal S}$.
\end{assume}

Let $\delta>0$ be a fixed positive constant.
\begin{assume} \label{ass:bounded-lyapunov}
	For each $d\in\{0,1\}$, $\Exp_F[\|ZU\|^{2+\delta}|X=\cdot,D=d]$ is bounded on $\mathbb X^{*d}$ by a fixed positive constant given $\{(X_i',D_i)'\}_{i\in\mathcal S}=\{(x_i',d_i)'\}_{i\in\mathcal S}$. 
\end{assume}

Let $\Sigma^*\equiv (\Gamma^*)^{-1}\Delta^*(\Gamma^*)^{-1}$. Let $\rho$ denote the L\'evy-Prokhorov metric on the space of probability measures, which metrizes the convergence in law. A formal definition is provided in equation \eqref{eq:LP}.
\begin{prop} \label{thm:normality} Under Assumptions \ref{ass:function-subsample}--\ref{ass:subsample-large} and \ref{ass:design-matrix-inverse}--\ref{ass:bounded-lyapunov}, $\Sigma^{*-1/2}\sqrt{|\mathcal S^*|}(\hat\theta^*-\theta_{\eG^*,F})$ converges in conditional distribution to $\mathcal N(\mathbf 0,I)$ given $\{(X_i',D_i)'\}_{i\in\mathcal S}=\{(x_i',d_i)'\}_{i\in\mathcal S}$, i.e., 
\begin{align}
	\rho(\mathcal L_{\Sigma^{*-1/2}\sqrt{|\mathcal S^*|}(\hat\theta^*
	-\theta_{\eG^*,F})|\{(X_i',D_i)'\}_{i\in\mathcal S}},\mathcal N(\mathbf 0,I))\rightarrow 0
\end{align}
given $\{(X_i',D_i)'\}_{i\in\mathcal S}=\{(x_i',d_i)'\}_{i\in\mathcal S}$, as $|\mathcal S|$ tends to infinity.
\end{prop}

Proposition \ref{thm:normality} shows that, under suitable normalization, the variation of $\hat\theta^*$ around $\theta_{\eG^*,F}$ given $\{(X_i',D_i)'\}_{i\in\mathcal S}$ is asymptotically distributed as standard normal. It  extends to arbitrary subsamples the full-sample result by \citet[Theorem 2]{AbadieImbensZheng2014}.

Note that, both here and in Proposition \ref{thm:consistency}, the asymptotic behavior of $\theta_{\eG^*,F}$ is left unrestricted: We do not require that $\theta_{\eG^*,F}$ converge to $\theta_{G^*,F}$ for some $G^*$ under which $X$ and $D$ become independent, the case explored in \citet{AbadieSpiess2022}.\footnote{They consider a case where $\mathcal S^*$ is constructed through matching, and show that, if $\eG^{*d}$'s converge to each other at a sufficient rate, match-level clustering yields valid standard errors.}

As addressed in \citet{AbadieImbensZheng2014}, the conditional variance $\Exp_F[U^2|X,D]$ in $\Sigma^*$ complicates its estimation, especially when $X$ is continuous. They propose alternative estimators, which do not involve a separate estimation of $\Exp_F[U^2|X,D]$. We show that one of their estimators is also valid when using subsamples. 
\begin{assume} \label{ass:metric} There exists a metric $\psi$ on $\mathbb R^p$, dominated by the Euclidean distance up to a fixed constant, such that, for each $d\in\{0,1\}$, 
	\begin{enumerate}[label=(\roman*)]
		\item $\Exp_F[Z_{j,i}^{r_j}Z_{k,i}^{r_k}(Y_i-Z_i'\theta_{\eG^*,F})^{r_{j,k}}|X_i=\cdot,D_i=d]$ is bounded on $\mathbb X^{*d}$ by a fixed constant, where $Z_{j,i}$ and $Z_{k,i}$ indicate the $j$th and $k$th components of $Z_i$, respectively, and $r_j$, $r_k$, and $r_{j,k}$ are natural numbers such that $r_j,r_k\leq 2$ and $r_{j,k}\leq r_j+r_k$; 
		\item $\Exp_F[Z_{j,i}^{r_j}Z_{k,i}^{r_k}(Y_i-Z_i'\theta_{\eG^*,F})^{r_{j,k}}|X_i=\cdot,D_i=d]$ is $\psi$-Lipschitz with a fixed Lipschitz constants on $\mathbb X^{*d}$; and 
		\item $\diam(\mathbb X^{*d})\equiv \sup_{i,j\in\mathcal S^{*d}}\|X_i-X_j\|$ are bounded by a fixed constant
	\end{enumerate}
	given $\{(X_i',D_i)'\}_{i\in\mathcal S}=\{(x_i',d_i)'\}_{i\in\mathcal S}$.
\end{assume}

For each unit $i\in\mathcal S^*$, let $\hat E_i^*\equiv Y_i-Z_i'\hat\theta^*$, and denote by
\begin{align}
	l_{(X',D)'}^*(i) \equiv \arg\min_{j\in\mathcal S^*/\{i\}} \psi(x_i,x_j) + \kappa |d_i-d_j|
\end{align}
the index of the closest unit in $\mathcal S^*$ to $i$, where $\kappa>0$ is a large, fixed positive constant. Our proposed estimator for $\Delta^*$ is then
\begin{align}
	\hat\Delta^* \equiv \frac{1}{2|\mathcal S^*|}\sum_{i\in\mathcal S^*}(Z_i\hat E_i^* - Z_{l_{(X',D)'}^*(i)}\hat E_{l_{(X',D)'}^*(i)}^*)(Z_i\hat E_i^* - Z_{l_{(X',D)'}^*(i)}\hat E_{l_{(X',D)'}^*(i)}^*)'.
\end{align}

Let $\hat\Sigma^*\equiv (\Gamma^*)^{-1}\hat\Delta^*(\Gamma^*)^{-1}$.
\begin{prop} \label{thm:variance-estimator-consistency} Provided the conditions of Proposition \ref{thm:consistency} and Assumption \ref{ass:metric} hold, $\hat\Sigma^*$ converges in conditional probability to $\Sigma^*$ given $\{(X_i',D_i)'\}_{i\in\mathcal S}=\{(x_i',d_i)'\}_{i\in\mathcal S}$, i.e.,
\begin{align}
	\Pr_F[\|\hat\Sigma^* - \Sigma^*\|>\eta|\{(X_i',D_i)'\}_{i\in\mathcal S}]\rightarrow 0, \forall \eta>0 \label{eq:variance-estimator-consistency}
\end{align}
given $\{(X_i',D_i)'\}_{i\in\mathcal S}=\{(x_i',d_i)'\}_{i\in\mathcal S}$, as $|\mathcal S|$ tends to infinity.
\end{prop}

Denote by $se_{\beta_{\eG^*,F}}(\hat\beta^*)$ the $|\mathcal S^*|^{-1/2}$-scaled square root of the second diagonal element of $\hat\Sigma^*$. 
For a function $\tau_0^*$ of $\{(X_i',D_i)'\}_{i\in\mathcal S}$, define a $t$-statistic
\begin{align}
	t^*\equiv \frac{\hat\beta^*-\tau_0^*}{se_{\beta_{\eG^*,F}}(\hat\beta^*)}.
\end{align}
\begin{assume}\label{ass:local-asymptotics}
	For a positive function $v$ of $\{(X_i',D_i)'\}_{i\in\mathcal S}$, 
	\begin{align}
	m_{\eG^*,F}c_{\eG^*}=\frac{v}{\sqrt{|\mathcal S^*|}},
	\end{align}
	which restricts $F$ to depend on $\{(X_i',D_i)'\}_{i\in\mathcal S}$ through $\eG^*$ and $v$.
\end{assume}

Fix $\{(X_i',D_i)'\}_{i\in\mathcal S}=\{(x_i',d_i)'\}_{i\in\mathcal S}$. The positive number $v$ locally parametrizes the states $F$ in the vicinity of the ``linear'' ones $F^\star$ at which $m_{\eG^*,F^\star}$ is zero.\footnote{It serves as a global parametrization if $c_{\eG^*}=O(1/\sqrt{|\mathcal S^*|})$.} It quantifies the misspecification of the regression model in units of $|\mathcal S^*|^{-1/2}c_{\eG^*}^{-1}$.

Recall that $\Phi$ denotes the distribution function of the standard normal distribution.
\begin{cor} \label{thm:size-distortion-subsample} Suppose that 
the conditions of
Propositions \ref{thm:normality}--\ref{thm:variance-estimator-consistency} are satisfied, and also that Assumption \ref{ass:local-asymptotics} hold. Then, under the null hypothesis $H_0:\tau_{\eG^{*1},F}=\tau_0^*$,
\begin{equation}
\begin{aligned}
	&\Phi(z)\leq \liminf_{|\mathcal S|\rightarrow\infty}\Pr_F\biggl[t^*-\frac{v}{\sqrt{|\mathcal S^*|}se_{\beta_{\eG^*,F}}(\hat\beta^*)}\leq z\biggm|\{(X_i',D_i)'\}_{i\in\mathcal S}\biggr]\text{ and}\\
	&\limsup_{|\mathcal S|\rightarrow\infty}\Pr_F\biggl[t^* + \frac{v}{\sqrt{|\mathcal S^*|}se_{\beta_{\eG^*,F}}(\hat\beta^*)}\leq z\biggm|\{(X_i',D_i)'\}_{i\in\mathcal S}\biggr]\leq \Phi(z), \forall z\in\mathbb R
\end{aligned}
\end{equation}
given $\{(X_i',D_i)'\}_{i\in\mathcal S}=\{(x_i',d_i)'\}_{i\in\mathcal S}$. 
\end{cor}

For a positive number $m$ and $\alpha\in(0,1)$, consider a random interval
\begin{align}
	C_\alpha(m)\equiv \bigl\{\tau\in\mathbb R: \hat\beta^*-z_{1-\alpha/2}se^* - mc_{\eG^*}< \tau\leq \hat\beta^*-z_{\alpha/2}se^* + mc_{\eG^*} \bigr\}, \label{eq:random-set}
\end{align}
where we abbreviate $se_{\beta_{\eG^*,F}}(\hat\beta^*)$ by $se^*$, and $z_\alpha\equiv \Phi^{-1}(\alpha)$. By Corollary \ref{thm:size-distortion-subsample}, whenever
\begin{align}
	m	\geq m_{\eG^*,F}, \label{eq:robust}
\end{align}
$C_\alpha(m)$ achieves asymptotic conditional coverage no smaller than $1-\alpha$. Reporting the graph of the interval-valued function $C_\alpha$ can be an effective communication practice, since it enables the computation of the minimum degree $\mathfrak m_\tau$ of model misspecification required to \emph{un}-reject a rejected null hypothesis $H_0:\tau_{\eG^{*1},F}=\tau$ by the usual confidence interval $C_\alpha(0)$. We refer to $\mathfrak m_\tau$ as the ``m-value'' of the null hypothesis.

We add the superscript 
$\KS$, 
$\MKW$, or $\MD$ to the notation $C_\alpha$ if the corresponding imbalance metric is used in its definition. For example, if $c_{\eG^*}$ of equation \eqref{eq:random-set} is the Kolmogorov-Smirnov distance, this choice is conveyed by denoting $C_{\alpha}^{\KS}$.

\section{Application}\label{sec:application}
Our justification of the design phase not only underscores the critical role of covariate balance in making a regression estimate robust and credible, but also suggests how a researcher and the readers of her report may communicate using imbalance statistics. Our bound results---compactly summarized in equation \eqref{eq:bias-bound-subsample}---indicate that, by using the available imbalance measures $c_{\eG}$, the readers can themselves explore various true states $f(\cdot,0)$ in the vicinity of the researcher's model $l_{\theta_{\eG,F}}^s(\cdot,0)$, compute the distances $m_{\eG,F}$ between $f(\cdot,0)$'s and $l_{\theta_{\eG,F}}^s(\cdot,0)$, and derive the maximum possible biases at those states by multiplying $m_{\eG,F}$'s with $c_{\eG}$'s. They can determine whether the researcher's estimate $\hat\beta$ is large enough to ensure that, even after accounting for potential bias due to model misspecification, her conclusion can be sustained. We use the Boston House Mortgage Disclosure Act (HMDA) dataset to illustrate how this communication can be implemented.

This expanded version of the prior 1990 HMDA reports, enriched by the Federal Reserve Bank of Boston, includes 38 additional variables tied to minority status and, at the same time, critical to banks' mortgage lending decisions, such as credit histories or loan-to-value ratios. It has been employed in the literature to examine the presence of racial discrimination in the mortgage market.\footnote{For example, \citet{MunnellTootellBrowneMcEneaney1996}.
}

The sample is limited to male applicants purchasing single-family residences, who are either black or white, not self-employed, and were approved for private mortgage insurance; furthermore, they must have no public record of default. This results in a sample $\mathcal S$ of 148 black and 1336 white applicants.\footnote{This is the sample restriction used by \citet[Section 6.1]{AbadieImbens2012}.}

For each applicant $i\in \mathcal S$, let $Y_i$ be the indicator which takes value 1 if $i$'s mortgage application is denied, and 0 if approved, and $D_i$ an indicator that takes value 1 if $i$ is black, and 0 if white. Six variables---the housing expense to income ratio, the total debt payments to income ratio, consumer credit history, mortgage credit history, the probability of unemployment, and the loan to appraised value ratio---are used as the controls, i.e., $p=6$. Let $X_i$ denote a six-dimensional vector that collects their values for unit $i$.\footnote{The detailed definitions of the variables are provided in \citet[Appendix]{MunnellTootellBrowneMcEneaney1996}.}

Given these defined variables, consider a researcher who runs the linear regression \eqref{eq:reg-model} on the sample $\mathcal S$ using the identity $x\mapsto Ix$ as the covariate function $s$, and obtains the regression estimates $(\hat\alpha,\hat\beta,\hat\gamma')'$ for the coefficients of the regressors $Z_i=(1,D_i,X_i')'$. To convey the robustness of her estimate $\hat\beta$ to her readers, she computes the distances $c_{\eG}$ between the empirical conditional distributions $\eG^d$ of the controls $X_i$ given minority status $D_i=d$. The readers then use the reported imbalance measurements to compute the maximum possible biases of $\beta_{\eG,F}$ at various adversarial states $f(\cdot,0)$ that deviate from the researcher's model $l_{\theta_{\eG,F}}^I(\cdot,0)$: They compute their distances $m_{\eG,F}$, and then multiply them by the corresponding $c_{\eG}$'s. However, since $X_i$ is multi-dimensional, i.e., $p>1$, the computations of $c_{\eG}$'s and $m_{\eG,F}$'s may be involved. We present an approach that reduces them into one-dimensional problems.

It is based on the following key observation: One can show that $\hat\beta$ is \emph{also} the least squares estimate
	for the coefficient $\tilde\beta$ of $D_i$ of the regression model
\begin{align}
	Y_i = \tilde\alpha + \tilde\beta D_i + X_i'\hat\gamma\cdot \tilde\gamma + \tilde E_i, \label{eq:reg-model-index} 
\end{align}
where we treat the estimated index $X_i'\hat\gamma$ as the only \emph{covariate};\footnote{The estimated linear propensity score also shares this property.} that is, we take the covariate function $s$ to be $x\mapsto x'\hat\gamma$. Thus, considering $\hat\gamma$ \emph{fixed},\footnote{This treatment is legitimate essentially because, as $|\mathcal S|$ tends to infinity, $\hat\gamma$ converges conditionally to $\gamma_{\eG,F}$, which is a function of $\{(X_i',D_i)'\}_{i\in\mathcal S}$.} a researcher may use the bias of the regression estimand $\tilde\beta_{\eG,F}$ of this induced model as that of her estimate $\hat\beta$. Now, assessing the bias of $\tilde\beta_{\eG,F}$ is a one-dimensional problem; it involves computing the discrepancy between the push-forwards $\hat\gamma\sharp\eG^d$ of $\eG^d$'s through $x\mapsto x'\hat\gamma$, which are distributions on the real line $\mathbb R$.

To be specific, let $\tilde Z_i\equiv (1,D_i,X_i'\hat\gamma)'$ collect the regressors of regression model \eqref{eq:reg-model-index}. $\tilde\beta_{\eG,F}$ is then the second component of 
\begin{align}
	\tilde\theta_{\eG,F}\equiv (\Exp_{\eG}[\tilde Z\tilde Z'])^{-1}\Exp_{\eG}[\tilde Z\Exp_F[Y|X,D]]. \label{eq:theta-index}
\end{align}
Define $\Exp_{F_{\hat\gamma}^{}}[Y|X'\hat\gamma,D]\equiv \Exp_{\eG}[\Exp_F[Y|X,D]|X'\hat\gamma,D]$. Abbreviate $\hat\gamma\sharp \eG^d$ by $\eG_{\hat\gamma}^d$. Consider a parameter
\begin{align}
	\tau_{\eG_{\hat\gamma}^1,F_{\hat\gamma}^{}}\equiv \Exp_{\eG_{\hat\gamma}^1}[\Exp_{F_{\hat\gamma}^{}}[Y|X'\hat\gamma,D=1] - \Exp_{F_{\hat\gamma}^{}}[Y|X'\hat\gamma,D=0]|D=1],	
\end{align}
which may be given a causal interpretation---for example, in cases where $X_i'\hat\gamma$ predicts $Y_i(0)$ sufficiently well that conditioning on it ensures the independence between $Y_i(0)$ and $D_i$. In this example, the supports of $\hat\gamma\sharp\eG^d$'s do not intersect. This makes $X_i'\hat\gamma$ a balancing score, which implies $\tau_{\eG_{\hat\gamma}^1,F_{\hat\gamma}^{}}=\tau_{\eG^1,F}$. Thus, we can assess the bias of $\tilde\beta_{\eG,F}$ for the latter via its bias for the former, which admits a nice representation involving the discrepancy between one-dimensional $\eG_{\hat\gamma}^d$'s.

Abbreviate $\Exp_{F_{\hat\gamma}}[Y|X'\hat\gamma=\cdot,D=0]$ by $f^{\hat\gamma}(\cdot,0)$. Let $l_{\tilde\theta_{\eG,F}}(\cdot,0)\equiv(1,0,\cdot)'\tilde\theta_{\eG,F}$ be the regression function defined at the level of the estimated index $X_i'\hat\gamma$. By equation \eqref{eq:representation-covariates},
\begin{align}
	\tilde\beta_{\eG,F} - \tau_{\eG_{\hat\gamma}^1,F_{\hat\gamma}^{}} = \int (\underbrace{f^{\hat\gamma}(t,0) - l_{\tilde\theta_{\eG,F}}(t,0)}_{\smash{\text{model misspecification}}})(\overbrace{\eG_{\hat\gamma}^1 - \eG_{\hat\gamma}^0}^{\smash{\mathclap{\text{covariate imbalance}}}})(\mathrm dt). \label{thm:representation-index}
\end{align}
Let $\eG_{\hat\gamma}$ be shorthand for the push-forward $\hat\gamma\sharp \eG$ of $\eG$ through $(x',d)'\mapsto (x'\hat\gamma,d)'$. Our bound results in Section \ref{sec:bias-bounds} then imply that the distances $c_{\eG_{\hat\gamma}}$ between the push-forwards $\eG_{\hat\gamma}^d$  are informative about the maximum possible bias of $\tilde\beta_{\eG,F}$---and thus about that of $\hat\beta$---for $\tau_{\eG_{\hat\gamma}^1,F_{\hat\gamma}^{}}$ under the corresponding forms of model misspecifications.

Note, however, that a perturbation of the true state $f^{\hat\gamma}(\cdot,0)$ from the researcher's model $l_{\tilde\theta_{\eG,F}}(\cdot,0)$ must now be articulated at the level $t$ of the estimated index. At the cost of reducing the problem to a one-dimensional one, one can no longer perturb the true state $f(\cdot,0)$ at the level $x$ of the original control.

Now, we begin our empirical analysis based on the induced model \eqref{eq:reg-model-index}. We assume that the sample size $|\mathcal S|$ is sufficiently large so that the population regression function $l_{\tilde \theta_{\eG,F}}(\cdot,0)$ is well-approximated by its empirical analogue, i.e.,
\begin{align*}
	l_{\tilde\theta_{\eG,F}}(t,0)\approx \hat\alpha	+ t,
\end{align*}
where we recall that $(\hat\alpha,\hat\beta,1)'$ is the least squares estimate for the coefficients $(\tilde\alpha,\tilde\beta,\tilde\gamma)'$ of the regressors $\tilde Z_i$ in the induced model; the regression of $Y_i$ on $(1,D_i,X_i'\hat\gamma)'$ within $\mathcal S$ results in them. The realized values of $\hat\alpha$ and $\hat\beta$ are $-0.249$ and 0.099, respectively.

The black circles of Figure \ref{fig:model} represent the joint distribution of the denial indicator $Y_i$ and the estimated index $X_i'\hat\gamma$. The red and blue diamonds display the distributions of the values of the population regression function, i.e., $l_{\tilde\theta_{\eG,F}}(X_i'\hat\gamma,0)\approx-0.249 + X_i'\hat\gamma$, for white and black applicants, respectively.

\begin{figure}
	\centering
	\includegraphics[width=0.45\textwidth]{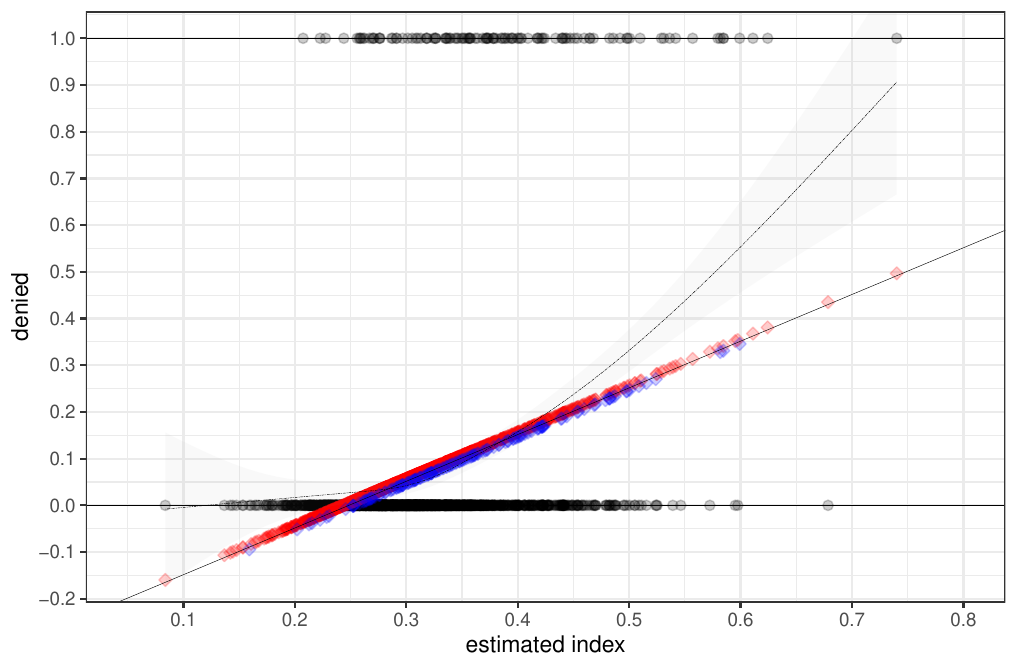}	
	\caption{Researcher's model, $l_{\tilde\theta_{\eG,F}}(\cdot,0)$}
	\label{fig:model}
\end{figure}

The readers initiate the perturbation exercise from the researcher's model $l_{\tilde\theta_{\eG,F}}(\cdot,0)$. If it were the true conditional expectation function $f^{\hat\gamma}(\cdot,0)$, then representation \eqref{thm:representation-index} implies that the bias of her estimand $\tilde\beta_{\eG,F}$ is zero, and thus, as long as $\hat\beta\approx \tilde\beta_{\eG,F}$, she may safely interpret her regression estimate $\hat\beta$ as $\tau_{\eG_{\hat\gamma}^1,F_{\hat\gamma}^{}}$. In most practical situations, however, this level of exactness is unlikely. The readers may wish to experiment with various adversarial scenarios in which the true state $f^{\hat\gamma}(\cdot,0)$ deviates from her model $l_{\tilde\theta_{\eG,F}}(\cdot,0)$ to some degree, accordingly. 

Figure \ref{fig:pKS} illustrates a case where the readers challenge that the true state $f^{\hat\gamma}(\cdot,0)$ deviates from the researcher's model $l_{\tilde\theta_{\eG,F}}(\cdot,0)$ in a periodic manner. If the researcher reports the Kolmogorov-Smirnov distance $c_{\eG_{\hat\gamma}}^{\KS}$ between the push-forwards $\eG_{\hat\gamma}^d$ as well as a finite set that contains their supports $\mathbb T^d$, the readers themselves can compute the total variation $m_{\eG_{\hat\gamma},F_{\hat\gamma}}^{\KS}$ of the resulting misspecification function, i.e., $f^{\hat\gamma}(\cdot,0)-l_{\tilde\theta_{\eG,F}}(\cdot,0)$, and multiply it with $c_{\eG_{\hat\gamma}}^{\KS}$ to obtain the maximum possible bias of her estimand $\tilde \beta_{\eG,F}$ for $\tau_{\eG_{\hat\gamma}^1,F_{\hat\gamma}^{}}$. If, even after reflecting such maximum possible bias, her regression estimate $\hat\beta$ remains distant from zero, for instance, if
\begin{align*}
	\hat\beta - m_{\eG_{\hat\gamma},F_{\hat\gamma}}^{\KS}c_{\eG_{\hat\gamma}}^{\KS}\gg 0,
\end{align*}
the readers may agree with the researcher's conclusion that the minority status affects banks' lending decision.

In our sample $\mathcal S$, $c_{\eG_{\hat\gamma}}^{\KS}=0.233$, which is not particularly small; this indicates that for perturbations similar to the previous one, which involves multiple fluctuations, the maximum possible bias computed based on $c_{\eG_{\hat\gamma}}^{\KS}$ could be large. Nevertheless, assuming that the Monge-Kantorovich/Wasserstein distance $c_{\eG_{\hat\gamma}}^{\MKW}$ between $\eG_{\hat\gamma}^d$'s, which takes the value 0.044, is also reported by the researcher, since $f^{\hat\gamma}(\cdot,0)-l_{\tilde\theta_{\eG,F}}(\cdot,0)$ has slopes of moderate magnitudes, the readers may---using the Monge-Kantorovich/Wasserstein bound---conclude that the bias of $\tilde\beta_{\eG,F}$ would remain small. However, for perturbations like those in Figure \ref{fig:pMKW}, even a small $c_{\eG_{\hat\gamma}}^{\MKW}$ is insufficient to ensure a small bias bound. It exhibits little fluctuations, but since $c_{\eG_{\hat\gamma}}^{\KS}$ is not very small, the Kolmogorov-Smirnov bound would not also yield a small number either.

\begin{figure}
	\centering
	\subfloat[Periodic fluctuations\label{fig:pKS}]{\includegraphics[width=0.45\textwidth]{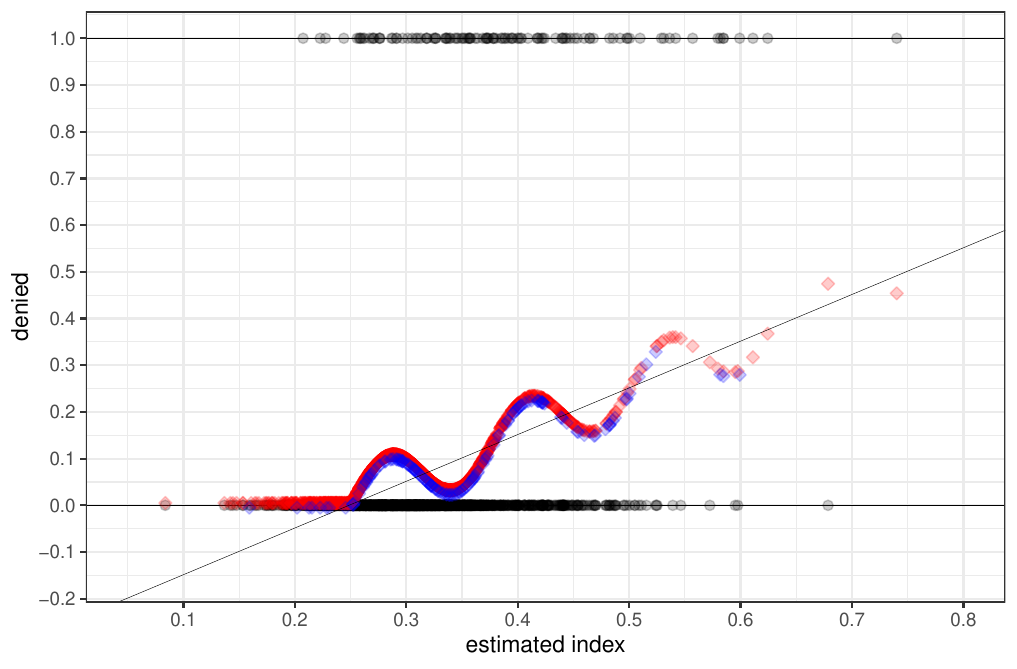}}	
	\subfloat[Single big spike\label{fig:pMKW}]{\includegraphics[width=0.45\textwidth]{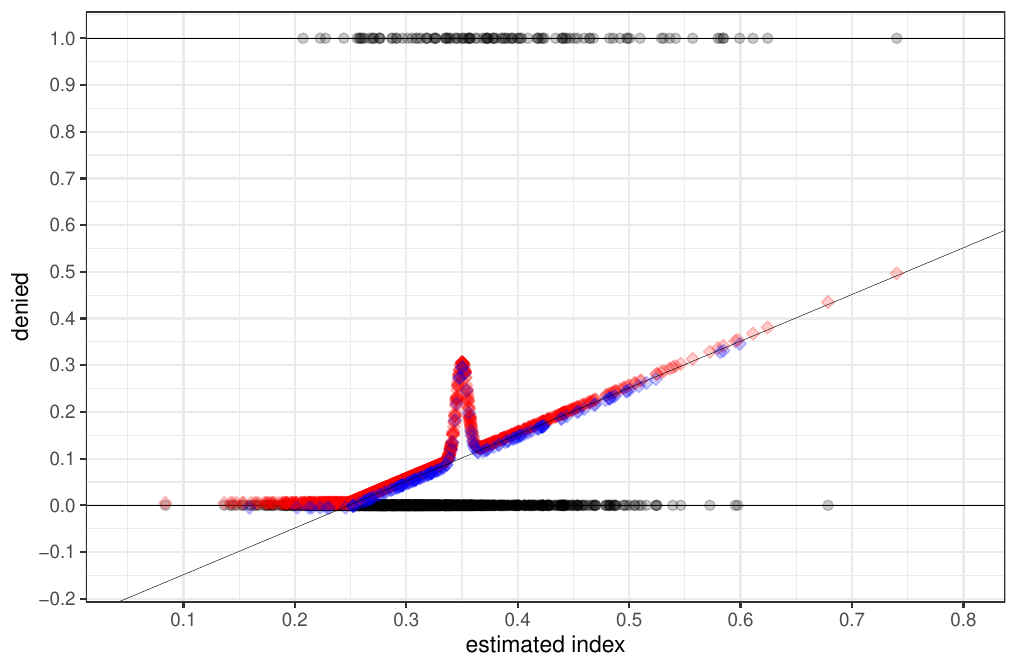}}
	\caption{Readers' perturbations}
\end{figure}

Nevertheless, note that, except for a few (mostly) red diamonds along the spike,
\begin{align*}
	f^{\hat\gamma}(X_i'\hat\gamma,0)-l_{\tilde\theta_{\eG,F}}(X_i'\hat\gamma,0)=-l_{\tilde\theta_{\eG,F}}(X_i'\hat\gamma,0)\wedge 0.
\end{align*}
Because the spike is local and of small mass, the function $-l_{\tilde\theta_{\eG,F}}(\cdot,0)\wedge 0$ can separate
\begin{align*}
	\Gr_{f^{\hat\gamma}(\cdot,0)-l_{\tilde\theta_{\eG,F}}(\cdot,0)}^{\mathbb T^d}\equiv \{(t,\tilde f^{\hat\gamma}(t,0)-l_{\tilde\theta_{\eG,F}}(t,0))'\}_{t\in \mathbb T^d}\text{'s}
\end{align*}
without consuming much slack, on average.\footnote{That is, $\sum_{d\in\{0,1\}}\Exp_{\eG^d}[\tilde \xi_\sigma^d(s(X_i)'\hat\gamma)|D=d]$ is presumed to be small.} Thus, if the researcher also reports the mean difference $c_{\eG_{\hat\gamma}}^{\MD}$ in the summary
\begin{align*}
	-l_{\tilde\theta_{\eG,F}}(X_i'\hat\gamma,0)\wedge 0\approx -(\hat\alpha+X_i'\hat\gamma)\wedge 0
\end{align*}
between black and white applicants, which is 0.002, the readers may refer to the mean difference bound to conclude that the bias of her estimand will remain small.

Assuming that the researcher reports the standard error $se_{\tilde\beta_{\eG,F}}(\hat\beta)$ of her regression estimate $\hat\beta$, the readers may themselves construct the robustified confidence intervals $C_{0.95}^{\KS}$, $C_{0.95}^{\MKW}$, and $C_{0.95}^{\MD}$, as in equation \eqref{eq:random-set}, using the reported distances $c_{\eG_{\hat\gamma}}^{\KS}$, $c_{\eG_{\hat\gamma}}^{\MKW}$, and $c_{\eG_{\hat\gamma}}^{\MD}$ between the push-forwards $\eG_{\hat\gamma}^d$. They are shown in the left panels of Figures \ref{fig:KS}--\ref{fig:MD}, where each vertical segment $C_{0.95}(m)$ of the trapezoid is a robustification of the traditional confidence interval, i.e., $C_{0.95}(0)$, assuming that the degree $m_{\eG_{\hat\gamma},F_{\hat\gamma}}$ of model misspecification is $m$. Note that the imbalance $c_{\eG_{\hat\gamma}}$ between $\eG_{\hat\gamma}^d$'s determines the angle of the trapezoid's rightward expansion. Referring to the graph, the readers can identify the minimum degrees $\mathfrak m_0$ of misspecification, i.e., the m-values, that render her initial rejection of the null $H_0:\tau_{\eG_{\hat\gamma}^1,F_{\hat\gamma}^{}}=0$ based on $C_{0.95}(0)$ inconclusive.

\begin{figure}
	\centering
	\subfloat[Pre-matching, $c_{\eG_{\hat\gamma}}^{\KS}=0.233$]{\includegraphics[width=0.45\textwidth]{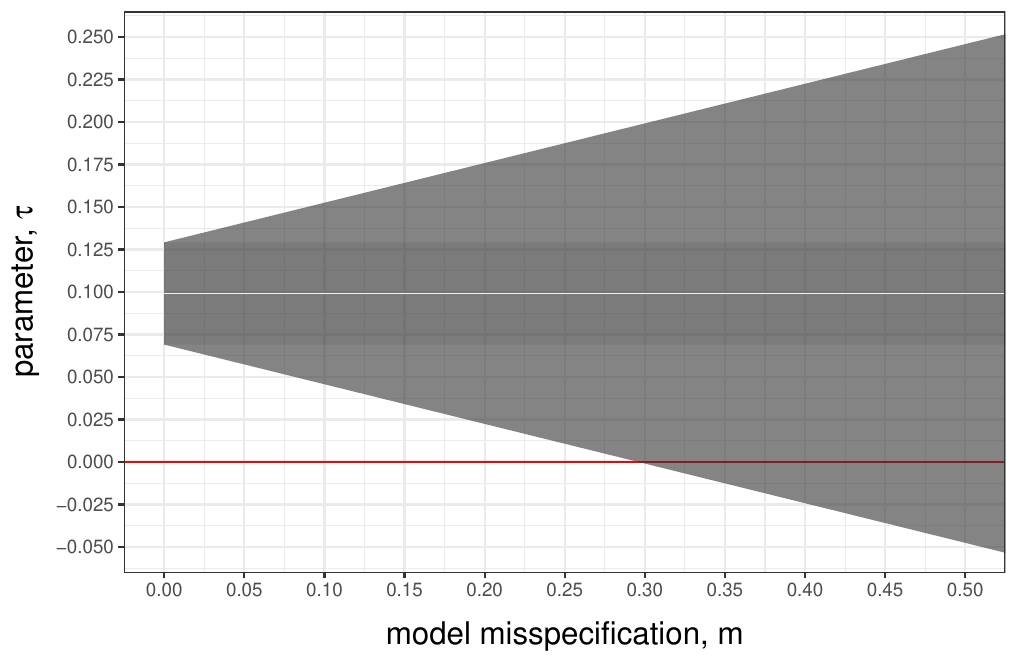}}
	\quad
	\subfloat[Post-matching, $c_{\eG_{\hat\gamma}^*}^{\KS}=0.047$]{\includegraphics[width=0.45\textwidth]{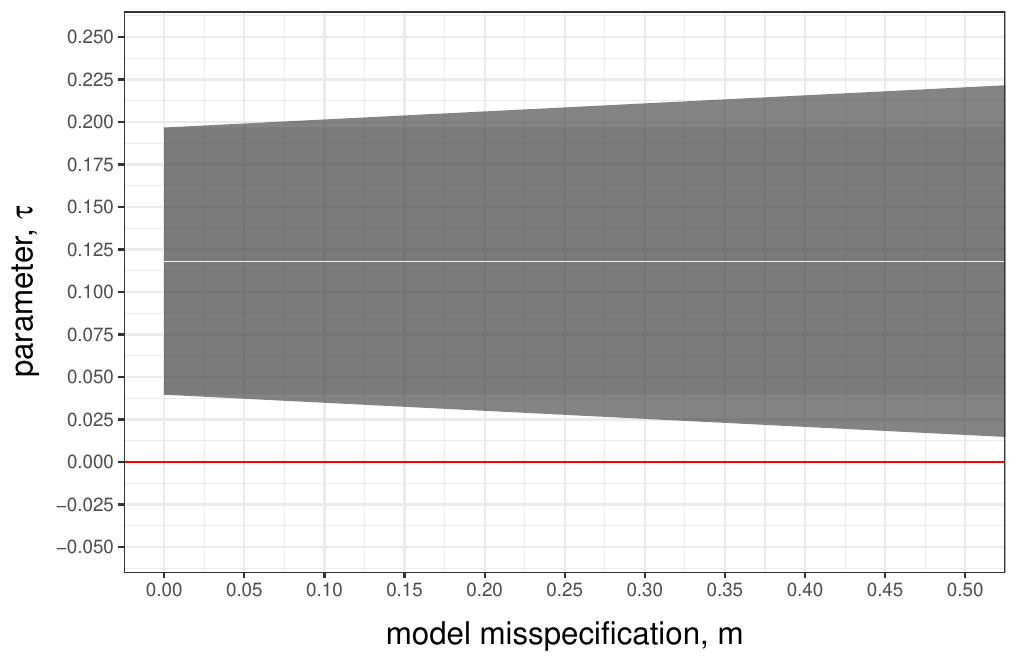}}
	\caption	{The graphs of the robustified confidence intervals, $C_{0.95}^{\KS}(m)$}
	\label{fig:KS}
	\vspace{\baselineskip}
	\subfloat[Pre-matching, $c_{\eG_{\hat\gamma}}^{\MKW}=0.044$]{\includegraphics[width=0.45\textwidth]{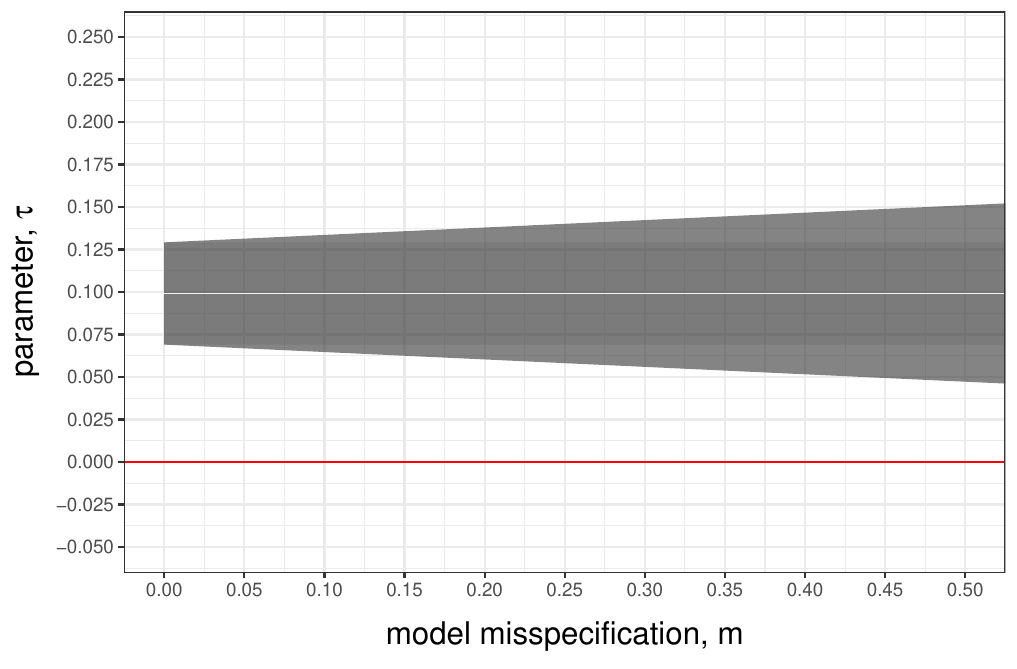}}
	\quad
	\subfloat[Post-matching, $c_{\eG_{\hat\gamma}^*}^{\MKW}=0.007$]{\includegraphics[width=0.45\textwidth]{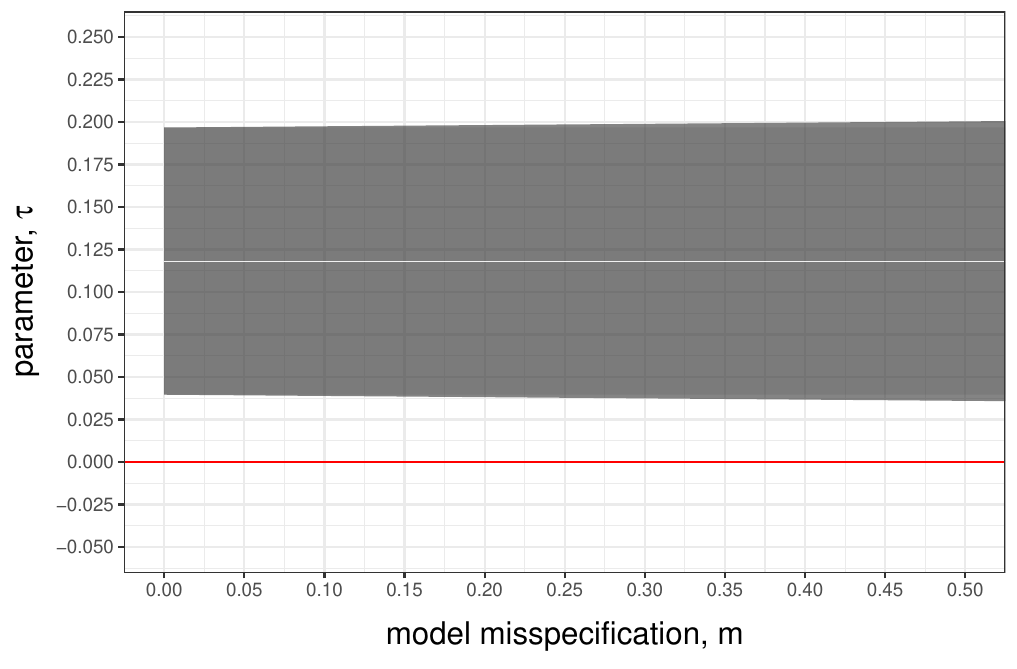}}
	\caption	{The graphs of the robustified confidence intervals, $C_{0.95}^{\MKW}(m)$}
	\label{fig:MKW}
	\vspace{\baselineskip}
	\subfloat[Pre-matching, $c_{\eG_{\hat\gamma}}^{\MD}=0.040$]{\includegraphics[width=0.45\textwidth]{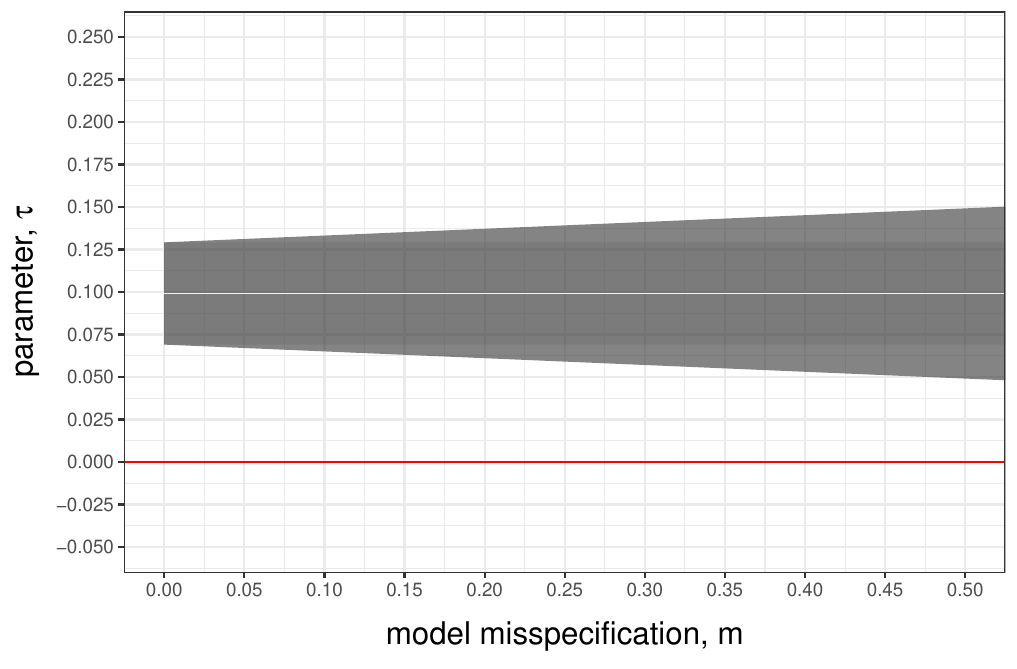}}
	\quad
	\subfloat[Post-matching, $c_{\eG_{\hat\gamma}^*}^{\MD}=0.003$]{\includegraphics[width=0.45\textwidth]{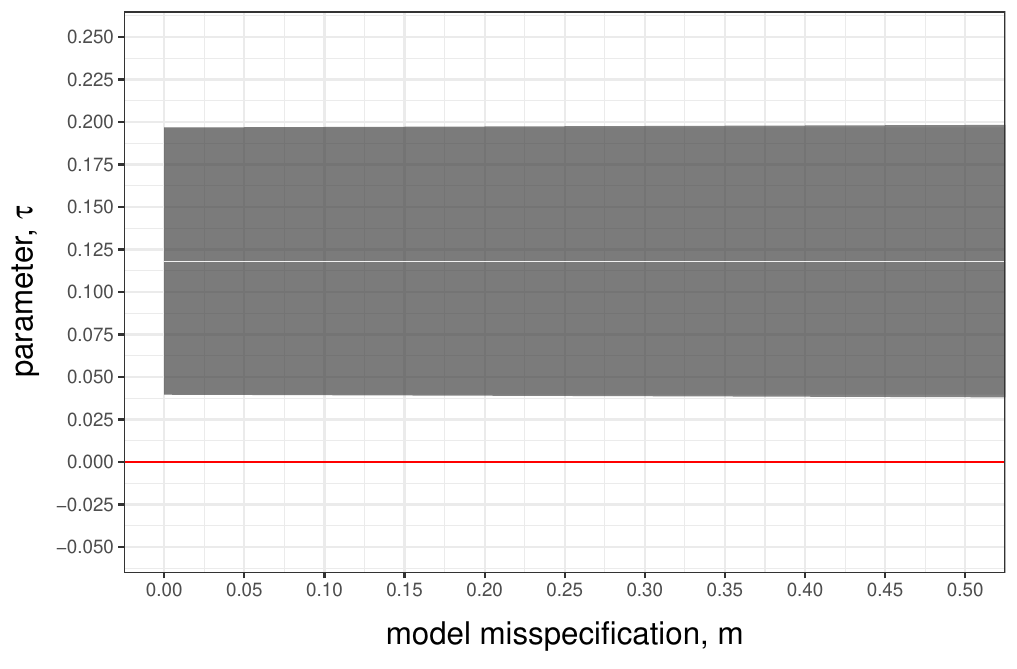}}
	\caption	{The graphs of the robustified confidence intervals, $C_{0.95}^{\MD}(m)$}
	\label{fig:MD}
\end{figure}

Also, a researcher may offer a rough estimate for the true conditional expectation function $f^{\hat\gamma}(\cdot,0)$, as indicated by the black dashed curve in Figure \ref{fig:model}, to guide readers on where to initiate their perturbation exercise. This can be useful particularly when the readers are not expected to hold a reliable prior about the true state $f^{\hat\gamma}(\cdot,0)$, and thus about the actual degree of misspecification $m_{\eG_{\hat\gamma},F_{\hat\gamma}}$.

The right panels of Figures \ref{fig:KS}--\ref{fig:MD} are the counterparts of the left ones, assuming the researcher has completed her design phase. It is assumed that she has applied nearest neighbor matching on the controls $X_i$. The imbalances $c_{\eG_{\hat\gamma}^*}$ between the push-forwards $\eG_{\hat\gamma}^{*d}$ within the resulting subsample $\mathcal S^*$ are minimal: $c_{\eG_{\hat\gamma}^*}^{\KS}=0.047$, $c_{\eG_{\hat\gamma}^*}^{\MKW}=0.007$, and $c_{\eG_{\hat\gamma}^*}^{\MD}=0.003$. This indicates that, unless the readers have a strong justification for an extreme perturbation---at least within this subsample $\mathcal S^*$---the bias of the researcher's estimand $\tilde\beta_{\eG^*,F}$ appears very small. 

This improved robustness comes at the cost of increased statistical uncertainty due to the smaller sample size. Let $\hat{\tilde\beta}^*$ denote the regression estimate for $\tilde\beta$ within $\mathcal S^*$. While its value, 0.118, is similar to that of $\hat\beta=0.099$, its standard error $se_{\tilde\beta_{\eG^*,F}}(\hat{\tilde\beta}^*)=0.040$ has substantially increased from $se_{\tilde\beta_{\eG,F}}(\hat\beta)=0.015$.\footnote{$se_{\tilde\beta_{\eG^*,F}}(\hat{\tilde\beta}^*)$ is the $|\mathcal S^*|^{-1/2}$-scaled square root of the second diagonal element of $\hat\Sigma^*$ in equation \eqref{eq:variance-estimator-consistency}, but computed using $\tilde Z_i=(1,D_i,X_i'\hat\gamma)'$ instead of $Z_i=(1,D_i,X_i')'$.} Therefore, a researcher may wish to navigate a trade-off between efficiency and robustness during the design phase. For instance, given her target level $\epsilon$ of precision, she may set the minimum robustness $\mathfrak m$ that she aims to uphold,   and search for a subsample $\mathcal S^*$ with the greatest cardinality, subject to its $c_{\eG_{\hat\gamma}^*}$ being less than $\epsilon/\mathfrak m$. We leave for future research the formal modeling of a researcher's decision when faced with these two competing concerns. 

Similar exercises can be performed using either the total variation or density ratio distance between the push-forwards, the results of which we present in the appendix.\footnote{As the L\'evy-Prokhorov distance requires non-trivial computation, we omit the corresponding results.}

\section{Conclusion}\label{sec:conclusion}
The design phase is formalized and justified in the context of linear regression: It is a process of adjusting the estimand via subsample selection, and the covariate balance of a subsample---used in the design phase as a criterion for the selection---informs on the maximum degree of misspecification that can be compromised when the subsample is used. The design phase can be viewed as a means to identify a \emph{better} estimand that is more robust to bias in the face of potential misspecification.

\bibliographystyle{chicago}
\bibliography{semi-te.bib}

\begin{thebibliography}{}

\bibitem[\protect\citeauthoryear{Abadie, Athey, Imbens, and Wooldridge}{Abadie
  et~al.}{2020}]{AbadieAtheyImbensWooldridge2020}
Abadie, A., S.~Athey, G.~W. Imbens, and J.~M. Wooldridge (2020).
\newblock Sampling-based versus design-based uncertainty in regression
  analysis.
\newblock {\em Econometrica\/}~{\em 88\/}(1), 265--296.

\bibitem[\protect\citeauthoryear{Abadie and Imbens}{Abadie and
  Imbens}{2008}]{AbadieImbens2008}
Abadie, A. and G.~W. Imbens (2008).
\newblock Estimation of the {C}onditional {V}ariance in {P}aired {E}xperiments.
\newblock {\em Annales d'{\'e}conomie et de statistique\/}~(91/92), 175--187.

\bibitem[\protect\citeauthoryear{Abadie and Imbens}{Abadie and
  Imbens}{2012}]{AbadieImbens2012}
Abadie, A. and G.~W. Imbens (2012).
\newblock Martingale {R}epresentation for {M}atching {E}stimators.
\newblock {\em Journal of the American Statistical Association\/}~{\em
  107\/}(498), 833--843.

\bibitem[\protect\citeauthoryear{Abadie, Imbens, and Zheng}{Abadie
  et~al.}{2014}]{AbadieImbensZheng2014}
Abadie, A., G.~W. Imbens, and F.~Zheng (2014).
\newblock Inference for {M}isspecified {M}odels {W}ith {F}ixed {R}egressors.
\newblock {\em Journal of the American Statistical Association\/}~{\em
  109\/}(508), 1601--1614.

\bibitem[\protect\citeauthoryear{Abadie and Spiess}{Abadie and
  Spiess}{2022}]{AbadieSpiess2022}
Abadie, A. and J.~Spiess (2022).
\newblock Robust {P}ost-{M}atching {I}nference.
\newblock {\em Journal of the American Statistical Association\/}~{\em
  117\/}(538), 983--995.

\bibitem[\protect\citeauthoryear{Andrews, Gentzkow, and Shapiro}{Andrews
  et~al.}{2017}]{AndrewsGentzkowShapiro2017}
Andrews, I., M.~Gentzkow, and J.~M. Shapiro (2017, 06).
\newblock Measuring the {S}ensitivity of {P}arameter {E}stimates to
  {E}stimation {M}oments.
\newblock {\em The Quarterly Journal of Economics\/}~{\em 132\/}(4),
  1553--1592.

\bibitem[\protect\citeauthoryear{Armstrong and Koles{\'a}r}{Armstrong and
  Koles{\'a}r}{2021}]{ArmstrongKolesar2021}
Armstrong, T.~B. and M.~Koles{\'a}r (2021).
\newblock Finite-{S}ample {O}ptimal {E}stimation and {I}nference on {A}verage
  {T}reatment {E}ffects {U}nder {U}nconfoundedness.
\newblock {\em Econometrica\/}~{\em 89\/}(3), 1141--1177.

\bibitem[\protect\citeauthoryear{Bonhomme and Weidner}{Bonhomme and
  Weidner}{2022}]{BonhommeWeidner}
Bonhomme, S. and M.~Weidner (2022).
\newblock Minimizing sensitivity to model misspecification.
\newblock {\em Quantitative Economics\/}~{\em 13\/}(3), 907--954.

\bibitem[\protect\citeauthoryear{Crump, Hotz, Imbens, and Mitnik}{Crump
  et~al.}{2009}]{CrumpHotzImbensMitnik2009}
Crump, R.~K., V.~J. Hotz, G.~W. Imbens, and O.~A. Mitnik (2009).
\newblock Dealing with limited overlap in estimation of average treatment
  effects.
\newblock {\em Biometrika\/}~{\em 96\/}(1), 187--199.

\bibitem[\protect\citeauthoryear{Dudley}{Dudley}{2002}]{Dudley2002}
Dudley, R. (2002).
\newblock {\em Real {A}nalysis and {P}robability}.
\newblock Cambridge studies in advanced mathematics. Cambridge University
  Press.

\bibitem[\protect\citeauthoryear{Dudley}{Dudley}{1968}]{Dudley1968}
Dudley, R.~M. (1968).
\newblock Distances of {P}robability {M}easures and {R}andom {V}ariables.
\newblock {\em The Annals of Mathematical Statistics\/}~{\em 39\/}(5), 1563 --
  1572.

\bibitem[\protect\citeauthoryear{Durrett}{Durrett}{2019}]{Durrett2019}
Durrett, R. (2019).
\newblock {\em Probability: {T}heory and {E}xamples\/} (5 ed.).
\newblock Cambridge Series in Statistical and Probabilistic Mathematics.
  Cambridge University Press.

\bibitem[\protect\citeauthoryear{Freedman}{Freedman}{2008a}]{Freedman2008a}
Freedman, D.~A. (2008a).
\newblock On {R}egression {A}djustments in {E}xperiments with {S}everal
  {T}reatments.
\newblock {\em The Annals of Applied Statistics\/}~{\em 2\/}(1), 176--196.

\bibitem[\protect\citeauthoryear{Freedman}{Freedman}{2008b}]{Freedman2008b}
Freedman, D.~A. (2008b).
\newblock On regression adjustments to experimental data.
\newblock {\em Advances in Applied Mathematics\/}~{\em 40\/}(2), 180--193.

\bibitem[\protect\citeauthoryear{Ho, Imai, King, and Stuart}{Ho
  et~al.}{2007}]{HoImaiKingStuart2007}
Ho, D.~E., K.~Imai, G.~King, and E.~A. Stuart (2007).
\newblock Matching as {N}onparametric {P}reprocessing for {R}educing {M}odel
  {D}ependence in {P}arametric {C}ausal {I}nference.
\newblock {\em Political Analysis\/}~{\em 15\/}(3), 199--236.

\bibitem[\protect\citeauthoryear{Huber and Ronchetti}{Huber and
  Ronchetti}{2009}]{HuberRonchetti2009}
Huber, P.~J. and E.~M. Ronchetti (2009).
\newblock {\em Robust Statistics}.
\newblock John Wiley \& Sons.

\bibitem[\protect\citeauthoryear{Imbens}{Imbens}{2015}]{Imbens2015}
Imbens, G.~W. (2015).
\newblock Matching {M}ethods in {P}ractice: {T}hree {E}xamples.
\newblock {\em The Journal of Human Resources\/}~{\em 50\/}(2), 373--419.

\bibitem[\protect\citeauthoryear{Imbens and Rubin}{Imbens and
  Rubin}{2015}]{ImbensRubin2015}
Imbens, G.~W. and D.~B. Rubin (2015).
\newblock {\em Causal {I}nference for {S}tatistics, {S}ocial, and {B}iomedical
  {S}ciences: {A}n {I}ntroduction}.
\newblock Cambridge University Press.

\bibitem[\protect\citeauthoryear{Kallenberg}{Kallenberg}{2021}]{Kallenberg2021}
Kallenberg, O. (2021).
\newblock {\em Foundations of {M}odern {P}robability}.
\newblock Springer Cham.

\bibitem[\protect\citeauthoryear{Kolmogorov and Fomin}{Kolmogorov and
  Fomin}{1975}]{KolmogorovFomin1975}
Kolmogorov, A.~N. and S.~V. Fomin (1975).
\newblock {\em Introductory {R}eal {A}nalysis}.
\newblock Dover Publications.

\bibitem[\protect\citeauthoryear{Lin}{Lin}{2013}]{Lin2013}
Lin, W. (2013).
\newblock Agonistic {N}otes on {R}egression {A}djustments to {E}xperimental
  {D}ata: {R}eexamining {F}reedman's {C}ritique.
\newblock {\em The Annals of Applied Statistics\/}~{\em 7\/}(1), 295--318.

\bibitem[\protect\citeauthoryear{Munnell, Tootell, Browne, and
  McEneaney}{Munnell et~al.}{1996}]{MunnellTootellBrowneMcEneaney1996}
Munnell, A.~H., G.~M.~B. Tootell, L.~E. Browne, and J.~McEneaney (1996).
\newblock Mortgage {L}ending in {B}oston: {I}nterpreting {HMDA} {D}ata.
\newblock {\em The American Economic Review\/}~{\em 86\/}(1), 25--53.

\bibitem[\protect\citeauthoryear{Negi and Wooldridge}{Negi and
  Wooldridge}{2021}]{NegiWooldridge2021}
Negi, A. and J.~M. Wooldridge (2021).
\newblock Revisiting {R}egression {A}djustment in {E}xperiments with
  {H}eterogeneous {T}reatment {E}ffects.
\newblock {\em Econometric Reviews\/}~{\em 40\/}(5), 504--534.

\bibitem[\protect\citeauthoryear{{Online Causal Inference Seminar}}{{Online
  Causal Inference Seminar}}{2022}]{Imbens2022}
{Online Causal Inference Seminar} (2022, 1).
\newblock Interview with {G}uido {I}mbens.
\newblock \url{https://youtu.be/DuVVy1WM-qM}.

\bibitem[\protect\citeauthoryear{Protter and Morrey}{Protter and
  Morrey}{1991}]{ProtterMorrey1991}
Protter, M.~H. and C.~B. Morrey (1991).
\newblock {\em A {F}irst {C}ourse in {R}eal {A}nalysis}.
\newblock Springer-Verlag New York.

\bibitem[\protect\citeauthoryear{Stuart}{Stuart}{2010}]{Stuart2010}
Stuart, E.~A. (2010).
\newblock {Matching Methods for Causal Inference: A Review and a Look Forward}.
\newblock {\em Statistical Science\/}~{\em 25\/}(1), 1 -- 21.

\bibitem[\protect\citeauthoryear{Vallender}{Vallender}{1974}]{Vallender1974}
Vallender, S.~S. (1974).
\newblock Calculation of the wasserstein distance between probability
  distributions on the line.
\newblock {\em Theory of Probability \& Its Applications\/}~{\em 18\/}(4),
  784--786.

\end{thebibliography}

\appendix
\section{Proofs}
\subsection{Proof of Proposition \ref{thm:representation}}
We use the following observation for our result.
\begin{lem} \label{thm:Y-moments} Suppose that Assumption \ref{ass:regularity} holds. Then, 
\begin{align}
	\Exp[Y|D=d] = \alpha + \beta d + \Exp[s(X)'\gamma|D=d].	
\end{align}
\end{lem}

\begin{proof}
Since the right-hand side of
\begin{align*}
	Y- s(X)'\gamma = \alpha + \beta D + E
\end{align*}
is saturated-in-$D$, $\Exp[E]=\Exp[DE]=0$, and $\Exp[ZZ']$ is positive definite,
\begin{align*}
	\Exp[Y- s(X)'\gamma|D] = \alpha + \beta D,
\end{align*}
where the left-hand side is well-defined, as $Y$ and $s(X)$ have their first moments. Then, expanding the left-hand side yields the desired result.
\end{proof}
Now, we have
\begin{align*}
	\tau &=\Exp[Y|D=1] - \Exp[\Exp[Y|X,D=0]|D=1]\\
	&= \Exp[Y|D=1] - \biggl(\int \Exp[Y|X=x,D=0]G^0(\mathrm dx) + \int \Exp[Y|X=x,D=0](G^1-G^0)(\mathrm dx)\biggr)\\
	&= \beta + (\Exp[s(X)'\gamma|D=1] - \Exp[s(X)'\gamma|D=0]) - \int \Exp[Y|X=x,D=0](G^1-G^0)(\mathrm dx)\\
	&= \beta - \int(\Exp[Y|X=x,D=0] - s(x)'\gamma)(G^1-G^0)(\mathrm dx)\\
	&= \beta - \int(\Exp[Y|X=x,D=0] - (\alpha + s(x)'\gamma))(G^1-G^0)(\mathrm dx),
\end{align*}
where Lemma \ref{thm:Y-moments} is used in the third equality.

\subsection{Proof of Corollary \ref{thm:bias-bound-KS}}
By Assumption \ref{ass:continuity}, $f(\cdot,0)-l(\cdot,0)$ is continuous and bounded on $\cup_{d\in\{0,1\}}\mathcal X^d$, which is closed by definition. By the Tietze extension theorem, there exists a continuous and bounded extension on $\mathbb R$. Thus, $\mathcal H$ is non-empty. Take any $h\in\mathcal H$. Then,
\begin{align*}
	\int_{[a,b]} (f(x,0)- l(x,0))G^d(\mathrm dx) &= \int_{[a,b]\cap \mathcal X^d} (f(x,0)- l(x,0))G^d(\mathrm dx)\\
	&=\int_{[a,b]\cap \mathcal X^d} h(x)G^d(\mathrm dx) = \int_{[a,b]} h(x)G^d(\mathrm dx).
\end{align*}
As $h$ is continuous and $G^d$ is nondecreasing, the Riemann-Stieltjes integral $\int_a^b h(x)\mathrm dG^d(x)$ exists and coincides with the Lebesgue-Stieltjes integral $\int_{[a,b]} h(x)G^d(\mathrm dx)$. \citep[p.368]{KolmogorovFomin1975} Integrating by parts,
\begin{align*}
	\int_{[a,b]} h(x)G^d(\mathrm dx) = \int_a^b h(x)\mathrm dG^d(x) &= h(b)G^d(b) - h(a)G^d(a) - \int_a^b G^d(x)\mathrm dh(x).
\end{align*}
\citep[p.320]{ProtterMorrey1991} Combining the previous two observations, we have
\begin{align*}
	&\int_{[a,b]} (f(x,0)- l(x,0))(G^1-G^0)(\mathrm dx)\\
	&= h(b)(G^1-G^0)(b) - h(a)(G^1-G^0)(a) - \int_a^b (G^1-G^0)(x)\mathrm dh(x).
\end{align*}
The third term of the right-hand side is bounded by
\begin{align*}
	\biggl|\int_a^b (G^1-G^0)(x)\mathrm dh(x)\biggr| &\leq \sup_{x\in[a,b]}|G^1(x)-G^0(x)|V_a^b[h] \leq \sup_{x\in[-\infty,\infty]}|G^1(x)-G^0(x)|V_{-\infty}^\infty[h].
\end{align*}
Taking $a,b\rightarrow\pm\infty$, by the dominated convergence theorem, and since $\lim_{b\rightarrow\infty}h(b)(G^1-G^0)(b)=\lim_{a\rightarrow-\infty}h(a)(G^1-G^0)(a)=0$, where we use the boundedness of $h$,
\begin{align*}
	\lim_{a,b\rightarrow\pm \infty}\biggl|\int_a^b (G^1-G^0)(x)\mathrm dh(x)\biggr| = \biggl|\int_{[-\infty,\infty]} (f(x,0)- l(x,0))(G^1-G^0)(\mathrm dx)\biggr|.
\end{align*}
This term does not depend on $h$, and thus we have
\begin{align*}
	\biggl|\int (f(x,0)- l(x,0))(G^1-G^0)(\mathrm dx)\biggr|	 \leq \sup_{x\in[-\infty,\infty]}|G^1(x)-G^0(x)| \inf_{h\in\mathcal H}V_{-\infty}^\infty[h].
\end{align*}
Now, the desired result follows from Proposition \ref{thm:representation}.

\subsection{Proof of Corollary \ref{thm:bias-bound-MKW}}
Suppose that $\|f(\cdot,0)-l(\cdot,0)\|_{\Lip}=0$. Then, $f(\cdot,0)-l(\cdot,0)$ is constant on $\cup_{d\in\{0,1\}}\mathcal X^d$. Then, by Proposition \ref{thm:representation}, 
\begin{align*}
	\beta - \tau	 = (f(\cdot,0)-l(\cdot,0))\int (G^1-G^0)(\mathrm dx) = (f(\cdot,0)-l(\cdot,0))(1-1)=0.
\end{align*}
That is, equation \eqref{eq:bias-bound-MKW} holds in the form of $0=0$.

Now, assume that the Lipschitz seminorm is positive. By Theorem 11.8.2 of \citet{Dudley2002},
\begin{align*}
	\inf_{\pi\in\Pi(G^1,G^0)}\int \|x_1-x_2\|\pi(\mathrm dx_1\mathrm dx_2)=\sup_{\|h\|_{\Lip}\leq 1}\biggl|\int h(x) (G^1-G^0)(\mathrm dx) \biggr|,
\end{align*}
where $h(\cdot)$ is a real-valued function defined on $\cup_{d\in\{0,1\}}\mathcal X^d$. The desired result follows from the observation
\begin{align*}
	|\beta-\tau|	=\biggl|\int \frac{f(x,0)-l(x,0)}{\|f(\cdot,0)-l(\cdot,0)\|_{\Lip}}(G^1-G^0)(\mathrm dx)\biggr|\|f(\cdot,0)-l(\cdot,0)\|_{\Lip}.
\end{align*}

\subsection{Proof of Corollary \ref{thm:bias-bound-MD}}
Note that
\begin{align*}
	\int (f(x,0)-l(x,0))G^1(\mathrm dx)&\leq \int (r(x)'\zeta_1^\star + \xi_1^{1\star}(x))G^1(\mathrm dx)\\
	&=\zeta_1^{\star\prime}\int r(x)G^1(\mathrm dx) + \int \xi_1^{1\star}(x)G^1(\mathrm dx)\text{ and}\\
	-\int (f(x,0)-l(x,0))G^0(\mathrm dx)&\leq \int (-r(x)'\zeta_1^\star + \xi_1^{0\star}(x))G^0(\mathrm dx)\\
	&=\zeta_1^{\star\prime}\int -r(x)G^0(\mathrm dx) + \int \xi_1^{0\star}(x)G^0(\mathrm dx).
\end{align*}
By Proposition \ref{thm:representation},
\begin{align*}
	\beta - \tau \leq \underbrace{\zeta_1^{\star\prime}\int r(x)(G^1-G^0)(\mathrm dx)}_{\displaystyle\mathclap{\leq |\zeta_1^{\star}|'\biggl|\int r(x)(G^1-G^0)(\mathrm dx)\biggr|}} + \sum_{d\in\{0,1\}}\int \xi_1^{d\star}(x)G^d(\mathrm dx).
\end{align*}
Similarly, the inequalities
\begin{align*}
	\int (f(x,0)-l(x,0))G^1(\mathrm dx)&\geq \int (r(x)'\zeta_{-1}^\star - \xi_{-1}^{1\star}(x))G^1(\mathrm dx)\\
	&= \zeta_{-1}^{\star\prime}\int r(x)G^1(\mathrm dx) - \int \xi_{-1}^{1\star}(x)G^1(\mathrm dx)\text{ and}\\
	-\int (f(x,0)-l(x,0))G^0(\mathrm dx)&\geq \int (-r(x)'\zeta_{-1}^\star-\xi_{-1}^{0\star}(x))G^0(\mathrm dx)\\
	&=\zeta_{-1}^{\star\prime}\int-r(x)G^0(\mathrm dx) - \int \xi_{-1}^{0\star}(x)G^0(\mathrm dx)
\end{align*}
combined with Proposition \ref{thm:representation} implies
\begin{align*}
	\beta - \tau \geq \underbrace{\zeta_{-1}^{\star\prime}\int r(x)(G^1-G^0)(\mathrm dx)}_{\displaystyle\mathclap{\geq -|\zeta_{-1}^\star|'\biggl|\int r(x)(G^1-G^0)(\mathrm dx)\biggr|}}-\sum_{d\in\{0,1\}}\int \xi_{-1}^{d\star}(x)G^d(\mathrm dx).
\end{align*}

\subsection{Proof of Proposition \ref{thm:consistency}}
Fix $\{(X_i',D_i)'\}_{i\in\mathcal S}=\{(x_i',d_i)'\}_{i\in\mathcal S}$. Since
\begin{align*}
	\hat\theta^* - \theta_{\eG^*,F} = \biggl(\frac{1}{|\mathcal S^*|}\sum_{i\in\mathcal S^*}Z_iZ_i'\biggr)^{-1}\frac{1}{|\mathcal S^*|}	\sum_{i\in\mathcal S^*}Z_iU_i, 
\end{align*}
and by Assumption \ref{ass:design-matrix-inverse}, it is enough to show that, as $|\mathcal S|$ tends to infinity, 
\begin{align*}
	\Var_F\biggl[\frac{1}{|\mathcal S^*|}\sum_{i\in\mathcal S^*} Z_iU_i\biggm|\{(X_i',D_i)'\}_{i\in\mathcal S}\biggr]\rightarrow 0.	
\end{align*}

Note that
\begin{align*}
	&\Var_F\biggl[\sum_{i\in\mathcal S}\frac{1}{|\mathcal S^*|}\mathbf 1\{i\in\mathcal S^*\}Z_iU_i\biggm|\{(X_i',D_i)'\}_{i\in\mathcal S}\biggr]\\
	&= \sum_{i\in\mathcal S}\sum_{j\in\mathcal S}\Cov_F\biggl[\frac{1}{|\mathcal S^*|}\mathbf 1\{i\in\mathcal S^*\}Z_iU_i,\frac{1}{|\mathcal S^*|}\mathbf 1\{j\in\mathcal S^*\}Z_jU_j\biggm|\{(X_i',D_i)'\}_{i\in\mathcal S}\biggr]\\
	&= \frac{1}{|\mathcal S^*|^2}\sum_{i\in\mathcal S}\sum_{j\in\mathcal S}\mathbf 1\{i\in\mathcal S^*\}\mathbf 1\{j\in\mathcal S^*\}\Cov_F[Z_iU_i,Z_jU_j|\{(X_i',D_i)'\}_{i\in\mathcal S}]\\
	&= \frac{1}{|\mathcal S^*|^2}\sum_{i\in\mathcal S^*}\Var_F[Z_iU_i|(X_i',D_i)']= \frac{1}{|\mathcal S^*|^2}\sum_{i\in\mathcal S^*}Z_iZ_i'\Exp_F[U_i^2|(X_i',D_i)'],
\end{align*}
where the second and third equations hold by Assumption \ref{ass:function-subsample} and
\begin{align*}
	&\Cov_F[Z_iU_i,Z_jU_j|\{(X_i',D_i)'\}_{i\in\mathcal S}]\\
	&= Z_iZ_j'\Exp_F[U_iU_j|\{(X_i',D_i)'\}_{i\in\mathcal S}]
	- Z_iZ_j'\Exp_F[U_i|(X_i',D_i)']\Exp_F[U_j|(X_j',D_j)']=0,
\end{align*}
respectively. By Assumption \ref{ass:bounded},
\begin{align*}
	\biggl\|\Var_F\biggl[\frac{1}{|\mathcal S^*|}\sum_{i\in\mathcal S^*} Z_iU_i\biggm|\{(X_i',D_i)'\}_{i\in\mathcal S}\biggr]\biggr\|
	\leq \frac{1}{|\mathcal S^*|^2}\sum_{i\in\mathcal S^*}\Exp_F[\|Z_iU_i\|^2|(X_i',D_i)']\leq \frac{M}{|\mathcal S^*|}.
\end{align*}
for some fixed positive constant $M>0$. This yields the desired result.

\subsection{Proof of Proposition \ref{thm:normality}}
Redefine the index set as $\mathcal S=\{1,\dots,|\mathcal S|\}\times \{\mathcal S\}$. To simplify notation, however, we suppress the dependence of the relabeled indices on $\mathcal S$.

For each $i=1,\dots,|\mathcal S|$, let $\xi_i\equiv\Sigma^{*-\frac{1}{2}}(\Gamma^*)^{-1}|\mathcal S^*|^{-1/2}\mathbf 1\{i\in\mathcal S^*\}Z_iU_i$, so that 
\begin{align*}
	\Sigma^{*-\frac{1}{2}}\sqrt{|\mathcal S^*|}(\hat\theta^*
	-\theta_{\eG^*,F}) = \sum_{i=1}^{|\mathcal S|}\xi_i.
\end{align*}
Fix $\{(X_i',D_i)'\}_{i=1}^{|\mathcal S|}=\{(x_i',d_i)'\}_{i=1}^{|\mathcal S|}$. Let $t\in\mathbb R^{\dim\theta}$. Note that, by Assumptions \ref{ass:design-matrix-inverse}--\ref{ass:variance-consistency},
\begin{align*}
	&\sum_{i=1}^{|\mathcal S|}\Exp_F[(t'\xi_i)^2|\{(X_i',D_i)'\}_{i\in\mathcal S}]\\
	&= t'\Sigma^{*-\frac{1}{2}}(\Gamma^*)^{-1}\biggl(\frac{1}{\sqrt{|\mathcal S^*|}}\sum_{i\in\mathcal S^*}Z_iZ_i'\Exp_F[U_i^2|(X_i',D_i)']\biggr)(\Gamma^*)^{-1}\Sigma^{*-\frac{1}{2}}t\rightarrow t't,
\end{align*}
and that, by Assumption \ref{ass:bounded-lyapunov}, for an arbitrary fixed positive constant $\varepsilon>0$,
\begin{align*}
	&\sum_{i=1}^{|\mathcal S|}\Exp_F[|t'\xi_i|^2\mathbf 1\{|t'\xi_i|>\varepsilon\}|\{(X_i',D_i)'\}_{i=1}^{|\mathcal S|}]\\
	&\leq \sum_{i=1}^{|\mathcal S|}\Exp_F\biggl[|t'\xi_i|^{2+\delta}\frac{1}{\varepsilon^\delta}\biggm|\{(X_i',D_i)'\}_{i=1}^{|\mathcal S|}\biggr]\\
	&\leq \frac{\|t\|^{2+\delta}}{\varepsilon^\delta}\underline\lambda^{-\frac{3}{2}(2+\delta)}\frac{1}{|\mathcal S^*|^{\frac{1}{2}(2+\delta)}}\sum_{i=1}^{|\mathcal S|}\Exp_F[\|Z_iU_i\|^{2+\delta}|(X_i',D_i)']\rightarrow 0,
\end{align*}
By the Lindeberg-Feller theorem \citep[Theorem 3.4.10]{Durrett2019}, 
\begin{align*}
	\sum_{i=1}^{|\mathcal S|}t'\xi_i\rightarrow_d \mathcal N(0,t't).
\end{align*}
The desired result follows from the Cram\'er-Wold device.

\subsection{Proof of Proposition \ref{thm:variance-estimator-consistency}}
Our result requires Lemma \ref{thm:closest-variance}, and for that, we establish Lemma \ref{thm:closest-moments}. Lemma \ref{thm:closest} provides lower-level conditions for the key condition used in Lemma \ref{thm:closest-moments}. The proofs closely follow those of \citet[Lemma A.2, Lemma A.3]{AbadieImbensZheng2014} and \citet[Lemma 1]{AbadieImbens2008}.

\begin{lem} \label{thm:closest}
	Let $\|\cdot\|$ denote the Euclidean norm on $\mathbb R^q$. Let $\mathcal S^*$ be a finite index set, and $\mathbb W_{\mathcal S^*}$ a subset of $\mathbb R^q$ such that $\diam(\mathbb W_{\mathcal S^*})\equiv \sup_{w,w'\in\mathbb W_{\mathcal S^*}}d(w,w')$ is bounded by a fixed constant. Let $\{w_i\}_{i\in\mathcal S^*}$ be a finite subset of $\mathbb W_{\mathcal S^*}$. Let $\nu$ be a metric on $\mathbb R^q$, which is dominated by the Euclidean metric. For each $i\in\mathcal S$, let
	\begin{align*}
		l_W^*(i)=l_W^*(i,w_i)\equiv \arg\min_{j\in\mathcal S^*/\{i\}}\nu(w_i,w_j). 
	\end{align*}
	Then, as $|\mathcal S^*|$ tends to infinity, 
	\begin{align*}
		\frac{1}{|\mathcal S^*|}\sum_{i\in\mathcal S^*}\nu(w_i,w_{l_W^*(i)})\rightarrow 0.
	\end{align*}
\end{lem}

\begin{proof}
By assumption, there exists some positive constant $K>0$ such that $\nu(w,w')\leq K\|w-w'\|$. Let $\varepsilon>0$ be arbitrary. Let $B_\nu^\varepsilon(w)\equiv\{w'\in\mathbb W_{\mathcal S^*}:\nu(w,w')<\varepsilon\}$. Suppose that there are $N^\varepsilon$ $i$'s such that $\nu(w_i,w_{l_W^*(i)})>2\varepsilon K$. Then, by definition, for such $i$'s, $B_\nu^{\varepsilon K}(w_i)\cap B_\nu^{\varepsilon K}(w_j)=\emptyset$ for all $j\in\mathcal S^*/\{i\}$, which implies
\begin{align}
	B^{\varepsilon}(w_i)\cap B^\varepsilon(w_j)=\emptyset,\text{ for all }j\in\mathcal S^*/\{i\}.\label{eq:exclusive}
\end{align}

Let $\mathcal O_{\mathcal S^*}$ be a closed Euclidean ball with radius $\diam(\mathbb W_{\mathcal S^*})$ such that $\mathbb W_{\mathcal S^*}\subseteq \mathcal O_{\mathcal S^*}$. By equation \eqref{eq:exclusive}, 
\begin{align*}
	N^\epsilon \frac{\pi^{q/2}\varepsilon^q}{\Gamma(\frac{q}{2}+1)}
	&= \sum_{i\in\mathcal S^*:\nu(w_i,w_{l_W^*(i)})>2\varepsilon K}\Vol[B^\varepsilon(w_i)]\\
	&=\Vol\biggl[\bigsqcup_{i\in\mathcal S^*:\nu(w_i,w_{l_W^*(i)})>2\varepsilon K} B^\varepsilon (w_i)\biggr]
	\leq\Vol[\mathcal O_{\mathcal S^*}^\varepsilon]= \frac{\pi^{q/2}(\diam(\mathbb W_{\mathcal S^*})(1+\varepsilon))^q}{\Gamma(\frac{q}{2}+1)},
\end{align*}
where $\Vol[\cdot]$ computes the argument's volume, $\Gamma(\cdot)$ denotes the Gamma function, and $\mathcal O_{\mathcal S^*}^\varepsilon$ the $\varepsilon$-enlargement of $\mathcal O_{\mathcal S^*}$. Thus, $N^\varepsilon\leq \diam(\mathbb W_{\mathcal S^*})^q((1+\varepsilon)/\varepsilon)^q$. 

Now, it follows from
\begin{align*}
	\frac{1}{|\mathcal S^*|}\sum_{i\in\mathcal S^*}\nu(w_i,w_{l_W(i)})
	&\leq \frac{1}{|\mathcal S^*|}\biggl(N^\varepsilon \diam_\nu(\mathbb W_{\mathcal S^*}) + 2\varepsilon K(|\mathcal S^*|-N^\varepsilon)\biggr)\\
	&\leq K\diam(\mathbb W_{\mathcal S^*})\frac{N^\varepsilon}{|\mathcal S^*|} + 2\varepsilon K
	\leq K\diam(\mathbb W_{\mathcal S^*})^{q+1}\frac{(\frac{1+\varepsilon}{\varepsilon})^q}{|\mathcal S^*|} + 2\varepsilon K
\end{align*}
that $\limsup_{|\mathcal S^*|\rightarrow\infty}|\mathcal S^*|^{-1}\sum_{i\in\mathcal S^*}\nu(w_i,w_{l_W(i)})\leq 2\varepsilon$. 
\end{proof}

\begin{lem} \label{thm:closest-moments}
Let $\{(V_i,W_i')'\}_{i\in\mathcal S}$ be a finite random sample in $\mathbb R\times \mathbb R^{\dim W}$, where $\mathcal L_{V|W}$ may be a function of $\{W_i\}_{i\in\mathcal S}$, i.e., $\mathcal L_{V|W}=\mathcal L_{V|W}(\{W_i\}_{i\in\mathcal S})$. Let $\mathcal S^*\subseteq \mathcal S$ be a function of $\{W_i\}_{i\in\mathcal S}$, i.e., $\mathcal S^*=\mathcal S^*(\{W_i\}_{i\in\mathcal S})$, and $i\in\mathcal S^*\mapsto l_W^*(i)\neq i\in\mathcal S^*$ a derangement that is a function of $\{W_i\}_{i\in\mathcal S}$, i.e., $l_W^*=l_W^*(\{W_i\}_{i\in\mathcal S})$. Let $R$ be a natural number, and $\nu$ a metric on $\mathbb R^{\dim W}$. Let $\mathbb W_{\mathcal S^*}\equiv \{W_i:i\in\mathcal S^*\}$. Now, fix vectors $\{w_i\}_{i\in\mathcal S}\subseteq\mathbb R^{\dim W}$, and assume that, given $\{W_i\}_{i\in\mathcal S}=\{w_i\}_{i\in\mathcal S}$, the following conditions hold for every natural number $r$ no greater than $R$:
	\begin{enumerate}[label=(\roman*)]
		\item $\Exp_{\mathcal L_{V|W}}[V^r|W=\cdot]$ is  $\nu$-Lipschitz on $\mathbb W_{\mathcal S^*}$ with a fixed Lipschitz constant $L_r$;
		\item $\Exp_{\mathcal L_{V|W}}[V^r|W=\cdot]$ is bounded on $\mathbb W_{\mathcal S^*}$ by a fixed constant $M_r$; and
		\item $\lim_{|\mathcal S|\rightarrow\infty}|\mathcal S^*|=\infty$ and $\lim_{|\mathcal S^*|\rightarrow\infty} |\mathcal S^*|^{-1}\sum_{i\in\mathcal S^*}\nu(W_i,W_{l_W^*(i)})=0$.
	\end{enumerate}
	Then, for all natural numbers $k$ and $m$ such that $k\vee m\leq R/2$,
	\begin{align*}
		&\Pr_{\mathcal L_{V|W}}\biggl[\biggl|\frac{1}{|\mathcal S^*|}\sum_{i\in\mathcal S^*}V_i^kV_{l_W^*(i)}^m\\
		&\hphantom{\Pr_{\mathcal L_{V|W}}\biggl[\biggl|} 
		- \frac{1}{|\mathcal S^*|}\sum_{i\in\mathcal S^*}\Exp_{\mathcal L_{V|W}}[V_i^k|W_i]\Exp_{\mathcal L_{V|W}}[V_i^m|W_i]\biggr|>\eta\biggm|\{W_i\}_{i\in\mathcal S}\biggr]\rightarrow 0, \forall \eta>0
	\end{align*}
	given $\{W_i\}_{i\in\mathcal S}=\{w_i\}_{i\in\mathcal S}$.
\end{lem}

\begin{proof} Fix $\{W_i\}_{i\in\mathcal S}=\{w_i\}_{i\in\mathcal S}$. It is enough to show that, as $|\mathcal S|$ tends to infinity, 
\begin{equation}
\begin{aligned}
	&\Exp_{\mathcal L_{V|W}}\biggl[\biggl(\frac{1}{|\mathcal S^*|}\sum_{i\in\mathcal S^{*[n]}}V_i^kV_{l_W^*(i)}^m\\
	&\hphantom{\Exp_{\mathcal L_{V|W}}\biggl[\biggl(\frac{1}{|\mathcal S^*|}}
	- \frac{1}{|\mathcal S^*|}\sum_{i\in\mathcal S^*}\Exp_{\mathcal L_{V|W}}[V_i^k|W_i]\Exp_{\mathcal L_{V|W}}[V_i^m|W_i]\biggr)^2\biggm|\{W_i\}_{i\in\mathcal S}\biggr]
	\label{eq:L2}
\end{aligned}
\end{equation}
is $o(1)$, since, by the Markov inequality, it then implies
\begin{align*}
	\frac{1}{|\mathcal S^*|}\sum_{i\in\mathcal S^*}V_i^kV_{l_W^*(i)}^m - \frac{1}{|\mathcal S^*|}\sum_{i\in\mathcal S^*}\Exp_{\mathcal L_{V|W}}[V_i^k|W_i]\Exp_{\mathcal L_{V|W}}[V_i^m|W_i]\rightarrow_p 0.
\end{align*}
Let $\mu_r(\cdot)\equiv \Exp_{\mathcal L_{V|W}}[V_i^r|W_i=\cdot]$. We henceforth omit the subscript $\mathcal L_{V|W}$ for notational simplicity.  

Before proceeding, note that
\begin{align*}
	\biggl|\Exp\biggl[\frac{1}{|\mathcal S^*|}\sum_{i\in\mathcal S^*}V_i^kV_{l_W^*(i)}^m
	- \frac{1}{|\mathcal S^*|}\sum_{i\in\mathcal S^*}\Exp[V_i^k|W_i]\Exp[V_i^m|W_i]\biggm|\{W_i\}_{i\in\mathcal S}\biggr]\biggr|=o(1).
\end{align*}
This follows from  
\begin{align*}
	&\Exp\biggl[\frac{1}{|\mathcal S^*|}\sum_{i\in\mathcal S^*}V_i^kV_{l_W^*(i)}^m\biggm|\{W_i\}_{i\in\mathcal S}\biggr]\\
	&=\frac{1}{|\mathcal S^*|}\sum_{i\in\mathcal S^*}\Exp[V_i^k|W_i]\Exp[V_{l_W^*(i)}^m|W_{l_W^*(i)}]\\
	&=\frac{1}{|\mathcal S^*|}\sum_{i\in\mathcal S^*}\mu_k(W_i)\mu_m(W_i)
	+ \frac{1}{|\mathcal S^*|}\sum_{i\in\mathcal S^*}\mu_k(W_i)(\mu_m(W_{l_W^*(i)})-\mu_m(W_i)),
\end{align*}
where we note that $V_i$ and $V_{l_W^*(i)}$ are independent conditional on $\{W_i\}_{i\in\mathcal S}$,\footnote{\citet[Theorem 8.12]{Kallenberg2021}} and
\begin{align*}
	\biggl|\frac{1}{|\mathcal S^*|}\sum_{i\in\mathcal S^*}\mu_k(W_i)(\mu_m(W_{l_W^*(i)})-\mu_m(W_i))\biggr|
	&\leq \frac{1}{|\mathcal S^*|}\sum_{i\in\mathcal S^*}|\mu_k(W_i)||\mu_m(W_{l_W^*(i)})-\mu_m(W_i)|\\
	&\leq M_kL_m\frac{1}{|\mathcal S^*|}\sum_{i\in\mathcal S^*}\nu(W_{l_W^*(i)},W_i)= o(1).
\end{align*}

We expand equation \eqref{eq:L2} as
\begin{align*}
	&\underbrace{\Exp\biggl[\biggl(\frac{1}{|\mathcal S^*|}\sum_{i\in\mathcal S^*}V_i^kV_{l_W^*(i)}^m\biggr)^2\biggm|\{W_i\}_{i\in\mathcal S}\biggr]}_{\equiv A} + \underbrace{\biggl(\frac{1}{|\mathcal S^*|}\sum_{i\in\mathcal S^*}\Exp[V_i^k|W_i]\Exp[V_i^m|W_i]\biggr)^2}_{\equiv B}\\
	&-	2\underbrace{\Exp\biggl[\frac{1}{|\mathcal S^*|}\sum_{i\in\mathcal S^*}V_i^kV_{l_W^*(i)}^m\biggm|\{W_i\}_{i\in\mathcal S}\biggr]\biggl(\frac{1}{|\mathcal S^*|}\sum_{i\in\mathcal S^*}\Exp[V_i^k|W_i]\Exp[V_i^m|W_i]\biggr)}_{\equiv C}.
\end{align*}
Here, $C=B+o(1)$, since
\begin{align*}
	C &= \biggl(\frac{1}{|\mathcal S^*|}\sum_{i\in\mathcal S^*}\Exp[V_i^k|W_i]\Exp[V_i^m|W_i]\biggr)^2\\
	&\phantom{=}+ \Exp\biggl[\frac{1}{|\mathcal S^*|}\sum_{i\in\mathcal S^*}V_i^kV_{l_W^*(i)}^m-\frac{1}{|\mathcal S^*|}\sum_{i\in\mathcal S^*}\Exp[V_i^k|W_i]\Exp[V_i^m|W_i]\biggm|\{W_i\}_{i\in\mathcal S}\biggr]\biggl(\frac{1}{|\mathcal S^*|}\sum_{i\in\mathcal S^*}\Exp[V_i^k|W_i]\Exp[V_i^m|W_i]\biggr),
\end{align*}
where the second term is bounded by
\begin{align*}
	&\biggl|\Exp\biggl[\frac{1}{|\mathcal S^*|}\sum_{i\in\mathcal S^*}V_i^kV_{l_W^*(i)}^m-\frac{1}{|\mathcal S^*|}\sum_{i\in\mathcal S^*}\Exp[V_i^k|W_i]\Exp[V_i^m|W_i]\biggm|\{W_i\}_{i\in\mathcal S}\biggr]\biggr|M_kM_m = o(1).
\end{align*}
It is thus enough to establish that $A=B+ o(1)$.\footnote{$A+B-2C = (B+o(1))+B-2(B+o(1))=o(1)$.}

The first term on the right-hand side of
\begin{align*}
	A
	= \Exp\biggl[\frac{1}{|\mathcal S^*|^2}\sum_{i\in\mathcal S^*}V_i^{2k}V_{l_W^*(i)}^{2m}\biggm|\{W_i\}_{i\in\mathcal S}\biggr]
	+ \Exp\biggl[\frac{1}{|\mathcal S^*|^2}\sum_{i\in\mathcal S^*}\sum_{j\neq i}V_i^kV_{l_W^*(i)}^mV_j^kV_{l_W^*(j)}^m\biggm|\{W_i\}_{i\in\mathcal S}\biggr]
\end{align*}
is $o(1)$, because
\begin{align*}
	\Exp\biggl[\frac{1}{|\mathcal S^*|^2}\sum_{i\in\mathcal S^*}V_i^{2k}V_{l_W^*(i)}^{2m}\biggm|\{W_i\}_{i\in\mathcal S}\biggr]
	&= \frac{1}{|\mathcal S^*|^2}\sum_{i\in\mathcal S^*}\Exp[V_i^{2k}|W_i]\Exp[V_{l_W^*(i)}^{2m}|W_{l_W^*(i)}]\leq \frac{1}{|\mathcal S^*|}M_{2k}M_{2m}=o(1),
\end{align*}
where we note that $2k, 2m \leq R$. The second term is further decomposed into
\begin{equation}
\begin{aligned}
	&\frac{1}{|\mathcal S^*|^2}\sum_{i\in\mathcal S^*}\sum_{j\neq i: l_W^*(j)\neq i, l_W^*(i)\neq j, l_W^*(j)\neq l_W^*(i)}\Exp[V_i^kV_{l_W^*(i)}^mV_j^kV_{l_W^*(j)}^m|\{W_i\}_{i\in\mathcal S}]\\
	&+ \frac{1}{|\mathcal S^*|^2}\sum_{i\in\mathcal S^*}\sum_{j\neq i: l_W^*(j)= i\text{ or }l_W^*(i)= j\text{ or }l_W^*(j)= l_W^*(i)}\Exp[V_i^kV_{l_W^*(i)}^mV_j^kV_{l_W^*(j)}^m|\{W_i\}_{i\in\mathcal S}].\label{eq:A22}
\end{aligned}
\end{equation}
Note that, for each $i\in\mathcal S^*$, $|\{j\in\mathcal S^*:  l_W^*(j)=i\}|\leq \overline K(\dim W)$, where $\overline K(q)$ indicates the ``kissing number'' i.e., the maximum number of times that each unit can be \emph{used} as a match in the  $q$-dimensional Euclidean space. Thus, as $|\{j\in\mathcal S^*:l_W^*(i)=j\}|=1$, for each $i\in\mathcal S^*$,
\begin{align*}
	|\{j\neq i:l_W^*(j)=i\text{ or }l_W^*(i)=j\text{ or }l_W^*(j)=l_W^*(i)\}|\leq \overline K(\dim W) + 1 + \overline K(\dim W).
\end{align*}
From this,
\begin{align*}
	&\sum_{j\neq i: l_W^*(j)= i\text{ or }l_W^*(i)= j\text{ or }l_W^*(j)= l_W^*(i)}|\Exp[V_i^kV_{l_W^*(i)}^mV_j^kV_{l_W^*(j)}^m|\{W_i\}_{i\in\mathcal S}]|\\
	&\leq (2\overline K(\dim W) + 1)\cdot (M_{k+m}^2\vee M_{k+m}M_kM_m \vee M_k^2M_{2m}\vee M_k^2M_m^2),
\end{align*}
and the second term of equation \eqref{eq:A22} is $O(|\mathcal S^*|^{-1})=o(1)$. Now, we show that the first term converges to $B$.

We proceed as follows: Note that
\begin{align*}
	&\frac{1}{|\mathcal S^*|^2}\sum_{i\in\mathcal S^*}\sum_{j\neq i: l_W^*(j)\neq i, l_W^*(i)\neq j, l_W^*(j)\neq l_W^*(i)}\Exp[V_i^kV_{l_W^*(i)}^mV_j^kV_{l_W^*(j)}^m|\{W_i\}_{i\in\mathcal S}]\\
	&\phantom{\Exp\biggl[}-\frac{1}{|\mathcal S^*|^2}\sum_{i\in\mathcal S^*}\sum_{j\neq i: l_W^*(j)\neq i, l_W^*(i)\neq j, l_W^*(j)\neq l_W^*(i)}\Exp[V_i^k|W_i]\Exp[V_i^m|W_i]\Exp[V_j^k|W_j]\Exp[V_j^m|W_j]\\
	&=\frac{1}{|\mathcal S^*|^2}\sum_{i\in\mathcal S^*}\sum_{j\neq i: l_W^*(j)\neq i, l_W^*(i)\neq j, l_W^*(j)\neq l_W^*(i)}\Exp[V_i^k|W_i]\Exp[V_j^k|W_j]\\
	&\phantom{=\frac{1}{|\mathcal S^*|^2}}\times (\Exp[V_{l_W^*(i)}^m|W_{l_W^*(i)}]\Exp[V_{l_W^*(j)}^m|W_{l_W^*(j)}]-\Exp[V_i^m|W_i]\Exp[V_j^m|W_j])\\
	&\leq \frac{1}{|\mathcal S^*|^2}\sum_{i\in\mathcal S^*}\sum_{j\neq i: l_W^*(j)\neq i, l_W^*(i)\neq j, l_W^*(j)\neq l_W^*(i)}\Exp[V_i^k|W_i]\Exp[V_j^k|W_j]\\
	&\phantom{\leq \frac{1}{|\mathcal S^*|^2}}
	\times \frac{\Exp[V_{l_W^*(i)}^m|W_{l_W^*(i)}]+\Exp[V_i^m|W_i]}{2}(\Exp[V_{l_W^*(j)}^m|W_{l_W^*(j)}]-\Exp[V_j^m|W_j])\\
	&\phantom{\leq}
	+ \frac{1}{|\mathcal S^*|^2}\sum_{i\in\mathcal S^*}\sum_{j\neq i: l_W^*(j)\neq i, l_W^*(i)\neq j, l_W^*(j)\neq l_W^*(i)}\Exp[V_i^k|W_i]\Exp[V_j^k|W_j]\\
	&\phantom{\leq \frac{1}{|\mathcal S^*|^2}}
	\times \frac{\Exp[V_{l_W^*(j)}^m|W_{l_W^*(j)}]+\Exp[V_j^m|W_j]}{2}(\Exp[V_{l_W^*(i)}^m|W_{l_W^*(i)}]-\Exp[V_i^m|W_i]),
\end{align*}
where the first term in the last equation is bounded by
\begin{align*}
	&\biggl|\frac{1}{|\mathcal S^*|^2}\sum_{i\in\mathcal S^*}\sum_{j\neq i: l_W^*(j)\neq i, l_W^*(i)\neq j, l_W^*(j)\neq l_W^*(i)}\Exp[V_i^k|W_i]\Exp[V_j^k|W_j]\\
	&\phantom{=\Exp\biggl[}\times \frac{\Exp[V_{l_W^*(i)}^m|W_{l_W^*(i)}]+\Exp[V_i^m|W_i]}{2}(\Exp[V_{l_W^*(j)}^m|W_{l_W^*(j)}]-\Exp[V_j^m|W_j])\biggr|\\
	&\leq \frac{1}{|\mathcal S^*|^2}\sum_{i\in\mathcal S^*}\sum_{j\neq i: l_W^*(j)\neq i, l_W^*(i)\neq j, l_W^*(j)\neq l_W^*(i)}|\Exp[V_i^k|W_i]||\Exp[V_j^k|W_j]|\\
	&\phantom{=\Exp\biggl[}\times \frac{|\Exp[V_{l_W^*(i)}^m|W_{l_W^*(i)}]|+|\Exp[V_i^m|W_i]|}{2}|\Exp[V_{l_W^*(j)}^m|W_{l_W^*(j)}]-\Exp[V_j^m|W_j]|\biggr|\\
	&= M_k^2M_mL_m\frac{1}{|\mathcal S^*|}\sum_{j\in\mathcal S^*}\nu(W_{l_W^*(j)},W_j)=o(1),
\end{align*}
Similarly, the second term can be bounded by $o(1)$. Now, from
\begin{align*}
	&\biggl|\frac{1}{|\mathcal S^*|^2}\sum_{i\in\mathcal S^*}\sum_{j\neq i: l_W^*(j)\neq i, l_W^*(i)\neq j, l_W^*(j)\neq l_W^*(i)}\Exp[V_i^k|W_i]\Exp[V_i^m|W_i]\Exp[V_j^k|W_j]\Exp[V_j^m|W_j]\\
	&\phantom{\Exp\biggl[}-\frac{1}{|\mathcal S^*|^2}\sum_{i\in\mathcal S^*}\sum_{j\neq i}\Exp[V_i^k|W_i]\Exp[V_i^m|W_i]\Exp[V_j^k|W_j]\Exp[V_j^m|W_j]\biggr|\\
	&=\frac{1}{|\mathcal S^*|^2}\sum_{i\in\mathcal S^*}\sum_{j\neq i:l_W^*(j)=i\text{ or }l_W^*(i)=j\text{ or }l_W^*(j)=l_W^*(i)}\underbrace{|\Exp[V_i^k|W_i]||\Exp[V_i^m|W_i]||\Exp[V_j^k|W_j]||\Exp[V_j^m|W_j]|}_{\leq M_k^2M_m^2}\\
	&\leq M_k^2M_m^2(2\overline K(\dim W)+1)\frac{1}{|\mathcal S^*|}=o(1),
\end{align*}
it follows that
\begin{align*}
	&\biggl|B - \frac{1}{|\mathcal S^*|^2}\sum_{i\in\mathcal S^*}\sum_{j\neq i: l_W^*(j)\neq i, l_W^*(i)\neq j, l_W^*(j)\neq l_W^*(i)}\Exp[V_i^kV_{l_W^*(i)}^mV_j^kV_{l_W^*(j)}^m|\{W_i\}_{i=1}^n]\biggr|\\
	&\leq \biggl|B - \frac{1}{|\mathcal S^*|^2}\sum_{i\in\mathcal S^*}\sum_{j\neq i}\Exp[V_i^k|W_i]\Exp[V_i^m|W_i]\Exp[V_j^k|W_j]\Exp[V_j^m|W_j]\biggr| + o(1)\\
	&=\frac{1}{|\mathcal S^*|^2}\sum_{i\in\mathcal S^*}(\Exp[V_i^k|W_i]\Exp[V_i^m|W_i])^2\leq M_k^2M_m^2\frac{1}{|\mathcal S^*|}=o(1).
\end{align*}
This completes the proof.
\end{proof}

\begin{lem} \label{thm:closest-variance}
Now, consider a finite random sample $\{(V_i',W_i')'\}_{i\in\mathcal S}$ in $\mathbb R^{\dim V}\times \mathbb R^{\dim W}$. We adopt the same setup for $\mathcal S^*$, $\mathcal L_{V|W}$, $\{w_i\}_{i\in\mathcal S}$, $\mathbb W_{\mathcal S^*}$, and $\nu$ as in Lemma \ref{thm:closest-moments}, with the exception that $\nu$ is now further assumed to be dominated by the Euclidean metric. For $l_W^*$, we follow the setup of Lemma \ref{thm:closest}, i.e., $l_W^*(i)=\arg\min_{j\in\mathcal S^*/\{i\}}\nu(W_j,W_i)$. Suppose that, given $\{W_i\}_{i\in\mathcal S}=\{w_i\}_{i\in\mathcal S}$, the following conditions hold for all natural numbers $r_j$ and $r_k$ no greater than 2:
\begin{enumerate}[label=(\roman*)]
	\item $\Exp_{\mathcal L_{V|W}}[V_{j,i}^{r_j}V_{k,i}^{r_k}|W_i=\cdot]$ is Lipschitz on $\mathbb W_{\mathcal S^*}$ with a fixed Lipschitz constant $L_{r_j,r_k}$,	where $j$ and $k$ are natural numbers no larger than $\dim V$. 
	\item $\Exp_{\mathcal L_{V|W}}[V_{j,i}^{r_j}V_{k,i}^{r_k}|W_i=\cdot]$ is bounded on $\mathbb W_{\mathcal S^*}$ by a fixed constant $M_{r_j,r_k}$, and
	\item $\lim_{|\mathcal S|\rightarrow\infty}|\mathcal S^*|=\infty$, and $\dim(\mathbb W_{\mathcal S^*})$ is bounded by a fixed constant. 
\end{enumerate}
Then, as $|\mathcal S|$ tends to infinity,
\begin{align*}
	&\Pr_{\mathcal L_{V|W}}\biggl[\biggl\|\frac{1}{2|\mathcal S^*|}\sum_{i\in\mathcal S^*}(V_i - V_{l_W^*(i)})(V_i - V_{l_W^*(i)})'\\
	&\hphantom{\Pr_{\mathcal L_{V|W}}\biggl[\biggl\|\frac{1}{2|\mathcal S^*|}} - \frac{1}{|\mathcal S^*|}\sum_{i\in\mathcal S^*}\Var_{\mathcal L_{V|W}}[V_i|W_i]\biggr\|>\eta\biggm|\{W_i\}_{i\in\mathcal S}\biggr]\rightarrow 0, \forall\eta>0
\end{align*}
given $\{W_i\}_{i\in\mathcal S}=\{w_i\}_{i\in\mathcal S}$. 
\end{lem}

\begin{proof} Fix $\{W_i\}_{i\in\mathcal S}=\{w_i\}_{i\in\mathcal S}$.  First, note that the condition (iii) and Lemma \ref{thm:closest} imply $\lim_{|\mathcal S^*|\rightarrow\infty}|\mathcal \mathcal S^*|^{-1}\sum_{i\in\mathcal S^*}\nu(W_i,W_{l_W^*(i)})=0$, thereby making Lemma \ref{thm:closest-moments} applicable. Again, for notational simplicity, we omit the subscript $\mathcal L_{V|W}$.

Let $\hat{\mathbb V}_{j,k}^*$ be the $(j,k)$th element of $\hat{\mathbb V}^*$. Then,
\begin{align*}
	\hat{\mathbb V}_{j,k}^*
	&=\frac{1}{2|\mathcal S^*|}\sum_{i\in\mathcal S^*}(V_{j,i}-V_{j,l_W^*(i)})(V_{k,i}-V_{k,l_W^*(i)})\\
	&=\frac{1}{2|\mathcal S^*|}\sum_{i\in\mathcal S^*}V_{j,i}V_{k,i} + \frac{1}{2|\mathcal S^*|}\sum_{i\in\mathcal S^*}V_{j,l_W^*(i)}V_{k,l_W^*(i)}
	- \biggl(\frac{1}{2|\mathcal S^*|}\sum_{i\in\mathcal S^*}V_{j,i}V_{k,l_W^*(i)} + \frac{1}{2|\mathcal S^*|}\sum_{i\in\mathcal S^*}V_{j,l_W^*(i)}V_{k,i}\biggr).
\end{align*}
Applying Lemma \ref{thm:closest-moments} to $V_{j,i}V_{k,i}$ with $(k,m)=(1,0)$ and $(k,m)=(0,1)$ yields
\begin{align*}
	&\frac{1}{|\mathcal S^*|}\sum_{i\in\mathcal S^*}V_{j,i}V_{k,i} -\frac{1}{|\mathcal S^*|}\sum_{i\in\mathcal S^*}\Exp[V_{j,i}V_{k,i}|W_i]\rightarrow_p 0,\text{ and}\\
	&\frac{1}{|\mathcal S^*|}\sum_{i\in\mathcal S^*}V_{j,l_W^*(i)}V_{k,l_W^*(i)} - \frac{1}{|\mathcal S^*|}\sum_{i\in\mathcal S^*}\Exp[V_{j,i}V_{k,i}|W_i]\rightarrow_p 0.
\end{align*}
Note that
\begin{align*}
	&\frac{1}{|\mathcal S^*|}\sum_{i\in\mathcal S^*}V_{j,i}V_{k,l_W^*(i)} + \frac{1}{|\mathcal S^*|}\sum_{i\in\mathcal S^*}V_{j,l_W^*(i)}V_{k,i}\\
	&= \frac{1}{|\mathcal S^*|}\sum_{i\in\mathcal S^*}(V_{j,i}+V_{k,i})(V_{j,l_W^*(i)} + V_{k,l_W^*(i)}) - \frac{1}{|\mathcal S^*|}\sum_{i\in\mathcal S^*} V_{j,i}V_{j,l_W^*(i)} - \frac{1}{|\mathcal S^*|}\sum_{i\in\mathcal S^*} V_{k,i}V_{k,l_W^*(i)}.
\end{align*}
Applying Lemma \ref{thm:closest-moments} to $V_{j,i} + V_{k,i}$, $V_{j,i}$, and $V_{k,i}$ with $(k,m)=(1,1)$ yields
\begin{align*}
	&\frac{1}{|\mathcal S^*|}\sum_{i\in\mathcal S^*}(V_{j,i}+V_{k,i})(V_{j,l_W^*(i)} + V_{k,l_W^*(i)}) - \frac{1}{|\mathcal S^*|}\sum_{i\in\mathcal S^*}(\Exp[V_{j,i}+V_{k,i}|W_i])^2\rightarrow_p 0,\\
	&\frac{1}{|\mathcal S^*|}\sum_{i\in\mathcal S^*} V_{j,i}V_{j,l_W^*(i)} - \frac{1}{|\mathcal S^*|}\sum_{i\in\mathcal S^*}(\Exp[V_{j,i}|W_i])^2 \rightarrow_p 0, \text{ and}\\
	&\frac{1}{|\mathcal S^*|}\sum_{i\in\mathcal S^*} V_{k,i}V_{k,l_W^*(i)} - \frac{1}{|\mathcal S^*|}\sum_{i\in\mathcal S^*}(\Exp[V_{k,i}|W_i])^2 \rightarrow_p 0.
\end{align*}
Combining results,
\begin{align*}
	&\hat{\mathbb V}_{j,k}^* - \biggl(\frac{1}{2|\mathcal S^*|}\sum_{i\in\mathcal S^*}\Exp[V_{j,i}V_{k,i}|W_i] + \frac{1}{2|\mathcal S^*|}\sum_{i\in\mathcal S^*}\Exp[V_{j,i}V_{k,i}|W_i]\\
	&\phantom{\hat{\mathbb V}_{j,k} - \biggl(}- \biggl(\frac{1}{2|\mathcal S^*|}\sum_{i\in\mathcal S^*}(\Exp[V_{j,i}+V_{k,i}|W_i])^2-\frac{1}{2|\mathcal S^*|}\sum_{i\in\mathcal S^*}(\Exp[V_{j,i}|W_i])^2-\frac{1}{2|\mathcal S^*|}\sum_{i\in\mathcal S^*}(\Exp[V_{k,i}|W_i])^2\biggr)\biggr)\\
	&=\hat{\mathbb V}_{j,k}^* - \frac{1}{|\mathcal S^*|}\sum_{i\in\mathcal S^*}(\Exp[V_{j,i}V_{k,i}|W_i]-\Exp[V_{j,i}|W_i]\Exp[V_{k,i}|W_i])\rightarrow_p 0.
\end{align*}
This completes the proof.
\end{proof}

We now begin the proof of Proposition \ref{thm:variance-estimator-consistency}. For each $i\in\mathcal S$, let $W_i\equiv (X_i',D_i)'$. Fix $\{W_i\}_{i\in\mathcal S}=\{(x_i',d_i)'\}_{i\in\mathcal S}$. Let $\mathbb W_{\mathcal S^*}\equiv (\mathbb X^{*1} \times \{1\}) \sqcup (\mathbb X^{*0} \times \{0\})$. For $w,w'\in \mathbb R^p\times \{0,1\}$, define $\nu(w,w')\equiv \psi(x,x') + \kappa|d-d'|$. Since $\psi$ is dominated by the Euclidean distance up to a constant, so is $\nu$. Lemma \ref{thm:closest}---together with the condition (iii) of Assumption \ref{ass:metric}---implies that $\lim_{|\mathcal S^*|\rightarrow\infty}|\mathcal S^*|^{-1}\sum_{i\in\mathcal S^*}\nu(W_i,W_{l_W^*(i)})=0$.\footnote{by choosing $\kappa$ sufficiently large.}

Define $V_i\equiv Z_i(Y_i-Z_i'\theta_{\eG^*,F})$ and $\mu_{r_j,r_k}(\cdot)\equiv\Exp_F[V_{j,i}^{r_j}V_{k,i}^{r_k}|W_i=\cdot]$. From the condition (ii) of Assumption \ref{ass:metric}, it follows that $\mu_{r_j,r_k}$ is $\nu$-Lipschitz on $\mathbb W_{\mathcal S^*}$; for $w,w'\in\mathbb W_{\mathcal S^*}$, 
\begin{align*}
	|\mu_{r_j,r_k}(w)-\mu_{r_j,r_k}(w')|
	&\leq |\mu_{r_j,r_k}(x,d)-\mu_{r_j,r_k}(x',d)| + |\mu_{r_j,r_k}(x',d)-\mu_{r_j,r_k}(x',d')|\\
	&\leq L_{r_j,r_k}^d\psi(x,x') + (M_{r_j,r_k}^d + M_{r_j,r_k}^{d'})|d-d'|\\
	&\leq (L_{r_j,r_k}^d \vee (M_{r_j,r_k}^d + M_{r_j,r_k}^{d'})/\kappa) \nu(w,w'),
\end{align*}
where $L_{r_j,r_k}^d$ is a fixed Lipschitz constant of $\mu_{r_j,r_k}(\cdot,d)$, and $M_{r_j,r_k}^d$ a fixed maximum of $\mu_{r_j,r_k}(\cdot,d)$ on $\mathbb X^{*d}$.\footnote{By the conditions (i)--(ii) of Assumption \ref{ass:metric}, such $L_{r_j,r_k}^d$'s and $M_{r_j,r_k}^d$'s exist.} The condition (i) of Assumption \ref{ass:metric} implies that $\mu_{r_j,r_k}$ is bounded on $\mathbb W_{\mathcal S^*}$ by a fixed constant. 

Thus, Lemma \ref{thm:closest-variance} yields that, as $|\mathcal S|$ tends to infinity, 
\begin{align}
	\frac{1}{2|\mathcal S^*|}\sum_{i\in\mathcal S^*}(V_i - V_{l_W^*(i)})(V_i - V_{l_W^*(i)})' - \frac{1}{|\mathcal S^*|}	\sum_{i\in\mathcal S^*}\Var_F[V_i|W_i] \rightarrow_p 0, \label{eq:closest-variance}
\end{align}
where we note that
\begin{align*}
	\frac{1}{|\mathcal S^*|}\sum_{i\in\mathcal S^*}\Var_F[V_i|W_i] = \Delta^*,
\end{align*}
since $\Var_F[V_i|W_i]
=Z_iZ_i'\Exp_F[U_i^2|X_i,D_i]$. Let $V_i(\theta)\equiv Z_i(Y_i-Z_i'\theta)$. Using Proposition \ref{thm:consistency}, we can show that
\begin{equation}
\begin{aligned}
	&\frac{1}{|\mathcal S^*|}\sum_{i\in\mathcal S^*}(V_i(\hat\theta^*) - V_{l_W^*(i)}(\hat\theta^*))(V_i(\hat\theta^*) - V_{l_W^*(i)}(\hat\theta^*))'\\
	&- \frac{1}{|\mathcal S^*|}\sum_{i\in\mathcal S^*}(V_i - V_{l_W^*(i)})(V_i - V_{l_W^*(i)})'\rightarrow_p 0. \label{eq:V}
\end{aligned}
\end{equation}
Now, from
\begin{align*}
	&\biggl\|\biggl(\frac{1}{|\mathcal S^*|}\sum_{i\in\mathcal S}Z_iZ_i'\biggr)^{-1}\frac{1}{2|\mathcal S^*|}\sum_{i\in\mathcal S^*}(V_i(\hat\theta^*) - V_{l_W^*(i)}(\hat\theta^*))(V_i(\hat\theta^*) - V_{l_W^*(i)}(\hat\theta^*))'\biggl(\frac{1}{|\mathcal S^*|}\sum_{i\in\mathcal S^*} Z_iZ_i'\biggr)^{-1}\\
	 &\hphantom{\biggl\|}- \biggl(\frac{1}{|\mathcal S^*|}\sum_{i\in\mathcal S^*} Z_iZ_i'\biggr)^{-1}\frac{1}{|\mathcal S^*|}\sum_{i\in\mathcal S^*}\Var_F[V_i|W_i]\biggl(\frac{1}{|\mathcal S^*|}\sum_{i\in\mathcal S^*} Z_iZ_i'\biggr)^{-1}\biggr\|\\
	 &\leq \lambda_{\min}[\Gamma^*]^{-2}\biggl\|\frac{1}{2|\mathcal S^*|}\sum_{i\in\mathcal S^*}(V_i(\hat\theta^*) - V_{l_W^*(i)}(\hat\theta^*))(V_i(\hat\theta^*) - V_{l_W^*(i)}(\hat\theta^*))' - \frac{1}{|\mathcal S^*|}\sum_{i\in\mathcal S^*}\Var_F[V_i|W_i]\biggr\|,
\end{align*}
the desired result follows.

\subsection{Proof of Corollary \ref{thm:size-distortion-subsample}}
Under the null $H_0$,
\begin{align*}
	- \frac{m_{\eG^*,F}c_{\eG^*}}{se_{\beta_{\eG^*,F}}(\hat\beta^*)}\leq \frac{\beta_{\eG^*,F}-\tau_0^*}{se_{\beta_{\eG^*,F}}(\hat\beta^*)}\leq \frac{m_{\eG^*,F}c_{\eG^*}}{se_{\beta_{\eG^*,F}}(\hat\beta^*)}
\end{align*}
By Assumption \ref{ass:local-asymptotics},
\begin{align*}
	t^* - \frac{v}{\sqrt{|\mathcal S^*|}se_{\beta_{\eG^*,F}}(\hat\beta^*)} \leq \frac{\hat\beta^*-\beta_{\eG^*,F}}{se_{\beta_{\eG^*,F}}(\hat\beta^*)} = \frac{\hat\beta^*-\tau_0^*}{se_{\beta_{\eG^*,F}}(\hat\beta^*)} - \frac{\beta_{\eG^*,F}-\tau_0^*}{se_{\beta_{\eG^*,F}}(\hat\beta^*)}	\leq t^* + \frac{v}{\sqrt{|\mathcal S^*|}se_{\beta_{\eG^*,F}}(\hat\beta^*)},
\end{align*}
Now, the desired result follows from Propositions \ref{thm:normality}--\ref{thm:variance-estimator-consistency} and Assumptions \ref{ass:design-matrix-inverse}--\ref{ass:variance-consistency}.

\newpage
\section{Three additional bounds} \label{sec:bias-bounds-TVDRLP}

\subsection{Total variation bound}
The bound of the first form is called the ``total variation bound,'' as the total variation distance is used as a measure of covariate imbalance. Define
\begin{align}
	m^{\TV}	&\equiv \|f(\cdot,0)-l_\theta^s(\cdot,0)\|_{L^\infty(\mu)}\text{ and}\\
	c^{\TV} &\equiv \int |g^1(x)-g^0(x)|\mu(\mathrm dx),
\end{align}
where $m^{\TV}$ is the essential supremum of $f(\cdot,0)-l_\theta^s(\cdot,0)$ with respect to $\mu$, and $c^{\TV}$ is the total variation distance between $G^1$ and $G^0$. $m^{\TV}=0$ if and only if $f(\cdot,0)=l_\theta^s(\cdot,0)$ $\mu$-almost surely, and $c^{\TV}=0$ if and only if $D$ and $X$ are independent.

\begin{cor} \label{thm:bias-bound-TV}
Suppose that the conditions of Proposition \ref{thm:representation} are satisfied. Then,
\begin{align}
	|\beta - \tau| \leq m^{\TV}c^{\TV}. \label{eq:bias-bound-TV}
\end{align}
\end{cor}

\paragraph{Example 1 (Cont')} We illustrate the result of Corollary \ref{thm:bias-bound-TV} using our toy example. We calculate $m^{\TV}$'s for each specification, and report the specific forms of equation \eqref{eq:bias-bound-TV}. First, note that 
\begin{align*}
	c^{\TV}=\int|g^1(x)-g^0(x)|\mu(\mathrm dx)=|g^1(1)-g^0(1)|+|g^1(0)-g^0(0)|=2|1-4p|.
\end{align*}
For Specification A, as $m_A^{\TV}=\|f(\cdot,0)-l^A(\cdot,0)\|_{L^\infty(\mu)}=\|\cdot-2p\|_{L^\infty(\mu)}=2p\vee(1-2p)$, equation \eqref{eq:bias-bound-TV} holds in the form of
\begin{align*}
	|1-4p|=|\beta_A-\tau|\leq m_A^{\TV}c^{\TV}=\underbrace{(2p\vee (1-2p))}_{\geq 1/2}\times 2|1-4p|.
\end{align*}
For Specification B, since $m_B^{\TV}=\|f(\cdot,0)-l^B(\cdot,0)\|_{L^\infty(\mu)}=\|\cdot - (-p+(3/2)\cdot)\|_{L^\infty(\mu)}=p\vee(1/2-p)$, equation \eqref{eq:bias-bound-TV} takes the form of
\begin{align*}
	|-1/2+2p|=|\beta_B-\tau|\leq m_B^{\TV}c^{\TV}=\underbrace{(p\vee(1/2-p))}_{\geq 1/4}\times 2|1-4p|.\tag*{\qed}
\end{align*}

\subsection{Density ratio bound}

$m^{\TV}$ and $m^{\KS}$ could be overly stringent measures for misspecification, as they basically pick up the most extreme discrepancy between $f(\cdot,0)$ and $l_\theta^s(\cdot,0)$ irrespective of how likely it occurs. To elaborate, suppose that $p$ is close to zero. As $G^0[\{0\}]$ is then close to one, one might argue that $l^A(\cdot,0)=2p$ is close enough to $f(\cdot,0)=\cdot$, and that the regression model $l^A$ is \emph{almost} correctly specified. Nevertheless, this aspect cannot be reflected in $m_A^{\TV}$ or $m_A^{\KS}$; both of them are no smaller than $1/4$. The bound we now provide uses a mean squared metric to measure misspecification, which circumvents this issue.

The bound of the second form is called the ``density ratio bound,'' as it uses the moments of the density ratio $\mathrm dG^1/\mathrm dG^0$ to assess the overlap in covariate distributions. 
Define
\begin{align}
	m^{\DR} &\equiv \|f(\cdot,0)-l_\theta^s(\cdot,0)\|_{L^2(G^0)} \text{ and}\\
	c^{\DR} &\equiv \|(\mathrm dG^1/\mathrm dG^0)(\cdot)-1\|_{L^2(G^0)}.
\end{align}
$m^{\DR}=0$ if the regression model is correctly specified, but the reverse is not necessarily true. $m^{\DR}$ is a relaxed measure of model misspecification compared to $m^{\TV}$ or $m^{\KS}$; in the previous scenario where $p\approx 0$, $m_A^{\DR}=\Exp[(X-2p)^2|D=0]^{1/2}=(2p(1-2p))^{1/2}\approx 0$.
When $G^0$ dominates $G^1$, $c^{\DR}=0$ if and only if $X$ and $D$ are independent; also, it is finite if the propensity scores $\Pr[D=1|X]$ are bounded away from one. $m^{\DR}$ is finite when $Y$ has a second moment.
\begin{cor} \label{thm:bias-bound-DR} Suppose that the conditions of Proposition \ref{thm:representation} are satisfied. Then,
\begin{align}
	|\beta - \tau| \leq m^{\DR}c^{\DR} + |\Exp_{(G^1)^{\perp}}[f(\cdot,0)-l_\theta^s(\cdot,0)]|, \label{eq:bias-bound-DR}
\end{align}
where $(G^1)^{\perp}$ denotes the singular part of $G^1$ with respect to $G^0$.
\end{cor}

\paragraph{Example 1 (Cont')} We revisit the preceding example to illustrate the result of Corollary \ref{thm:bias-bound-DR}. First, since $p\in(0,1/2)$, $G^0$ dominates $G^1$. $c^{\DR}=|1-4p|/(2p(1-2p))^{1/2}$, which equals zero if and only if $p=1/4$. For Specifications A and B, equation \eqref{eq:bias-bound-DR} holds as an equality in the forms of
\begin{align*}
	|1-4p|=|\beta_A-\tau|&= m_A^{\DR}c^{\DR} = (2p(1-2p))^{1/2}\times |1-4p|/(2p(1-2p))^{1/2} \text{ and}\\
	|-1/2+2p|=|\beta_B-\tau|&= m_B^{\DR}c^{\DR} = (p(1/2-p))^{1/2}\times |1-4p|/(2p(1-2p))^{1/2}.\tag*{\qed}
\end{align*}

\subsection{L\'evy-Prokhorov bound}
The bound of the third form is referred to as the ``L\'evy-Prokhorov bound,'' because it uses the L\'evy-Prokhorov metric to gauge the discrepancy in covariate distributions between the treated and untreated. Define
\begin{align}
	m^{\LP}&\equiv \|f(\cdot,0)-l_\theta^s(\cdot,0)\|_{\Lip} + \|f(\cdot,0)-l_\theta^s(\cdot,0)\|_\infty\text{ and}\\
	c^{\LP}&\equiv 2\rho(G^1,G^0),
\end{align}
where $m^{\LP}$ is the sum of the Lipschitz seminorm and the supremum norm of $f(\cdot,0)-l_\theta^s(\cdot,0)$, where the latter is defined as
\begin{align*}
	\|f(\cdot,0)-l_\theta^s(\cdot,0)\|_\infty \equiv \sup_{x\in\cup_{d\in\{0,1\}}\mathcal X^d} |f(x,0)-l_\theta^s(x,0)|, 	
\end{align*}
and we recall that $\rho$ 
denotes the L\'evy-Prokhorov metric, which is defined as 
\begin{align}
	\rho(G^1,G^0) \equiv \inf\{\varepsilon>0:G^1[A]\leq G^0[A^\varepsilon]+\varepsilon, \forall  A\in\mathcal B(\mathbb R^p)\}, \label{eq:LP}
\end{align}
where $A^\varepsilon$ is the $\varepsilon$-enlargement of $A$, and $\mathcal B(\mathbb R^p)$ is the class of all Borel sets in $\mathbb R^p$.  

\begin{cor} \label{thm:bias-bound-LP} Suppose that the conditions of Proposition \ref{thm:representation} hold. Then,
\begin{align}
	|\beta - \tau|\leq m^{\LP}c^{\LP}. \label{eq:bias-bound-LP}
\end{align}
\end{cor}

A distinctive feature of this bound is, the informativeness of $c^{\LP}$ is both for the Lipschitz-smoothness of $f(\cdot,0)-l_\theta^s(\cdot,0)$ and the maximum degree of misspecification that can be allowed; for a given level $\epsilon$ of target precision, the closer the distributions $G^1$ and $G^0$ become, the less smooth and the farther from zero $f(\cdot,0)-l_\theta^s(\cdot,0)$ can be.

\paragraph{Example 1 (Cont')} We illustrate how equation \eqref{eq:bias-bound-LP} concretizes in our toy example. $\rho(G^1,G^0)=\inf\{\varepsilon>0:G^1[\{0\}]\leq G^0[\{0\}^\varepsilon]+\varepsilon, G^1[\{1\}]\leq G^0[\{1\}^\varepsilon]+\varepsilon, G^1[\{0,1\}]\leq G^0[\{0,1\}^\varepsilon]+\varepsilon\}=|1-4p|$, whereby $c^{\LP}=2|1-4p|$. Also, $m_A^{\LP}=m_A^{\MKW}+m_A^{\TV}=1+2p\vee(1-2p)$ and $m_B^{\LP}=m_B^{\MKW}+m_B^{\TV}=1/2+p\vee(1/2-p)$. As a result, equation \eqref{eq:bias-bound-LP} takes the forms of
\begin{align*}
	|1-4p|=|\beta_A - \tau|&\leq m^{\LP}c^{\LP} = (1+2p\vee(1-2p))\times 2|1-4p|\text{ and}\\
	|-1/2+2p|=|\beta_B - \tau|&\leq m^{\LP}c^{\LP} = (1/2+p\vee(1/2-p))\times 2|1-4p|.\tag*{\qed}
\end{align*}

As with the Monge-Kantorovich/Wasserstein bound, $m^{\LP}$ could have been defined in a slightly weaker manner: The Lipschitz seminorm can be replaced by a no larger
\begin{align}
	\Exp_{G^0}\biggl[\sup_{x\in\mathcal X^1}\frac{|f(x,0)-l_\theta^s(x,0)-(f(\cdot,0)-l_\theta^s(\cdot,0))|}{\|x-\cdot\|}\mathbf 1\{0<\|x-\cdot\|\leq c^{\LP}/2\}\biggr], \label{eq:bias-bound-LP-m}
\end{align}
which averages the local unsmoothness of the function $f(\cdot,0)-l_\theta^s(\cdot,0)$; $c^{\LP}$ parametrizes the locality. 

The exact computation of the L\'evy-Prokhorov distance $c^{\LP}$ is non-trivial. One may utilize its relationship $(1/2)c^{\text{MK/W}}\leq (1/2)c^{\LP}\leq \sqrt{c^{\text{MK/W}}}$ with the Monge-Kantorovich or Wasserstein distance $c^{\MKW}$ \citep[Corollary 2.18]{HuberRonchetti2009}, whose computation is less involved.

\section{Extensions of Proposition \ref{thm:representation}}

\subsection{Interaction terms}

\begin{assume}
	$Y=l(X,D)+E$ such that $l(X,D)$ and $E$ have finite first moments, and $\Exp[DE]=0$. Also, $\Pr[D=1]\in(0,1)$. \label{ass:regularity-extension}
\end{assume}

For notational convenience, let $\Delta l(X,\cdot)\equiv l(X,1)-l(X,0)$ and $\Delta \Exp[l(X,0)|D=\cdot]\equiv \Exp[l(X,0)|D=1]-\Exp[l(X,0)|D=0]$.

\begin{prop} \label{thm:representation-extension} Suppose that Assumption \ref{ass:regularity-extension} holds. Then,
\begin{align}
\Exp[\Delta l(X,\cdot)|D=1] - \tau = \int (f(x,0) - l(x,0))(g^1(x)-g^0(x))\mu(\mathrm dx). \label{eq:representation-extension}
\end{align}
\end{prop}

As an example, consider an extension of \eqref{eq:reg-model}
\begin{align}
	Y = \underbrace{\alpha + \beta D + s(X)'\gamma + D(s(X)-\Exp[s(X)|D=1])'\delta}_{l(X,D)} + E,
\end{align}
which includes the interaction terms between $D$ and $s(X)$. Note that 
\begin{align*}
	\Delta l(X,\cdot)=\beta + (s(X)-\Exp[s(X)|D=1]), 
\end{align*}
and thus $\Exp[\Delta l(X,\cdot)|D=1]=\beta$.

\subsubsection{Proof of Proposition \ref{thm:representation-extension}}

\begin{lem} \label{thm:Y-moments-extension}Suppose that Assumption \ref{ass:regularity-extension} holds. Then,
\begin{align}
	\Exp[Y|D] =\Exp[l(X,0)|D=0]+(\Exp[\Delta l(X,\cdot)|D=1] + \Delta \Exp[l(X,0)|D=\cdot])D.
\end{align}
\end{lem}

\begin{proof} Let $a\equiv \Exp[l(X,D)|D=0]$ and $b\equiv\Exp[l(X,D)|D=1]-\Exp[l(X,D)|D=0]$; both of which exist and finite. Define $\tilde E \equiv l(X,D) - (a+ bD)$. Then, 
\begin{align*}
	\Exp[D\tilde E]=\Exp[Dl(X,D)]-(a+b)\Pr[D=1]=0.	
\end{align*}
Thus, $\Exp[l(X,D)|D]=a + bD$. In addition, since $\Exp[DE]=0$ and $\Pr[D=1]\in(0,1)$, $\Exp[E|D]=0$. Now, combining results,
\begin{align*}
	\Exp[Y|D]&=\Exp[l(X,D)|D] + \Exp[E|D]\\
	&= \Exp[l(X,D)|D=0] + (\Exp[l(X,D)|D=1]-\Exp[l(X,D)|D=0])D.
\end{align*}
Telescoping $\Exp[l(X,0)|D=1]$ in the coefficient of $D$ yields the desired result.
\end{proof}

Now,
\begin{align*}
	\tau &= \Exp[Y|D=1] - \int \Exp[Y|X=x,D=0]G^1(\mathrm dx)\\
	&= \Exp[Y|D=1] - \Exp[Y|D=0] - \int \Exp[Y|X=x,D=0](G^1-G^0)(\mathrm dx)\\
	&= \Exp[\Delta l(X,\cdot)|D=1] + \Delta \Exp[l(X,0)|D=\cdot] - \int \Exp[Y|X=x,D=0](G^1-G^0)(\mathrm dx)\\
	&= \Exp[\Delta l(X,\cdot)|D=1] - \int (\Exp[Y|X=x,D=0]-l(x,0))(G^1-G^0)(\mathrm dx),
\end{align*}
where we use Lemma \ref{thm:Y-moments-extension} in the third equality.

\subsection{Other parameters}
Here, we explore other parameters that may simplify into the average treatment effect on the untreated (ATEU) and the average treatment effect (ATE), respectively. 

Define
\begin{align}
	\tau^0 &\equiv \Exp[\Exp[Y|X,D=1]|D=0] - \Exp[Y|D=0]\text{, and}\\
	\tau^{10} &\equiv \Exp[Y|D=1]\Pr[D=1] + \Exp[\Exp[Y|X,D=1]|D=0]\Pr[D=0]\nonumber\\
	&\hphantom{\equiv \Exp} - (\Exp[\Exp[Y|X,D=0]|D=1]\Pr[D=1] + \Exp[Y|D=0]\Pr[D=0]).
\end{align}
$\tau^0$ is interpreted as the ATEU under the conditional independence between $Y(1)$ and $D$ given $X$, where $Y(1)$ denotes the potential outcome when treated, and $\tau^{10}$ as the ATE under that between $(Y(1),Y(0))$ and $D$ given $X$.
\begin{prop} \label{thm:representation-other-parameters} Suppose that Assumption \ref{ass:regularity} holds. Then,
\begin{align}
	\beta - \tau^0 &= \int (f(x,1) - l(x,1))(g^1(x)-g^0(x))\mu(\mathrm dx)\text{, and}\\
	\beta - \tau^{10} &= \Pr[D=1]\int (f(x,0)-l(x,0))(g^1(x)-g^0(x))\mu(\mathrm dx)\nonumber\\
	&\hphantom{= \Pr}+\Pr[D=0]\int (f(x,1)-l(x,1))(g^1(x)-g^0(x))\mu(\mathrm dx).	
\end{align}
\end{prop}

The proof is given in the online appendix.

Also, we note that our discussion extends readily to situations where the parameter of a researcher's interest is a weighted version of $\tau$, i.e.,
\begin{align*}
	\tau^w\equiv \Exp[w^1(X)(\Exp[Y|X,D=1]-\Exp[Y|X,D=0])|D=1],
\end{align*}
where $w^d(\cdot)$ is her chosen weight function that satisfies $\Exp[w^d(X)|D=d]=1$.\footnote{Since $G$ is unknown, this choice is infeasible. However, given $\{(X_i',D_i)\}_{i\in\mathcal S}=\{(x_i',d_i)'\}_{i\in\mathcal S}$, it is feasible.} This is due to the observation that $\tau^w=\tau_{G_w,F}$, where $G_w$ is a modification of $G$ by replacing $G^d$'s with their dominated probability measures $w^d\cdot G^d$'s.\footnote{$(w^d\cdot G^d)[B]\equiv \int_Bw^d(x)G^d(\mathrm dx)$.} In fact, this observation yields a useful characterization of subsample selection: It is a transition $\eG\mapsto \eG_w$ such that $w^d\cdot\mathbb G^d$'s are again empirical distributions.

\end{mtchideinmaintoc}

\newpage
\title{\Large Online appendix to\\
``On the role of the design phase in a linear regression''}
\author{Junho Choi\\
University of Wisconsin-Madison\\
junho.choi@wisc.edu}
\hypersetup{bookmarksdepth=1}
\setcounter{page}{1}
\maketitle

\newpage
\tableofcontents
\newpage\appendix

\setcounter{section}{2}
\setcounter{assume}{14}
\setcounter{figure}{5}
\setcounter{prop}{6}
\setcounter{equation}{66}
\section{Simulation}\label{sec:simulation}
Using a simulation exercise, we illustrate that the design phase adjusts the estimand of a regression and makes its identification of the parameter more robust to potential model misspecification.

We adopt the simulation setup from \citet[Section 6.1]{AbadieImbens2012}, who use the restricted sample $\mathcal S$ of the Boston Home Mortgage Disclosure Act (HMDA) dataset, defined in Section \ref{sec:application}. 

Simulation proceeds in the following steps. First, we construct a simulation sample $\mathfrak S$, consisting of $\mathsf N_1$ black applicants and $\mathsf N_0(\geq \mathsf N_1)$ white applicants, randomly drawn from $\mathcal S$ without replacement. Let $(\check\alpha,\check\gamma^\prime)'$ denote the least squares estimates obtained from the regression of $D_i$ on $(1,X_i')'$ within $\mathfrak S$. The estimated linear propensity scores $\hat e(X_i)\equiv \check\alpha + X_i'\check\gamma$ and their exponentials $\hat e(X_i)^r$ are considered as only \emph{covariates}; the estimated linear propensity score function $\hat e$ generates the covariate function. By our notational convention, $\eG_{\hat e}$ denotes the empirical distribution of $(\hat e(X_i),D_i)'$ within $\mathfrak S$. Second, we run a logistic regression of $Y_i$ on $(1,D_i,\hat e(X_i),D_i\hat e(X_i))'$ within $\mathfrak S$, and use the estimated model as the population conditional distribution of $Y_i$ given $(\hat e(X_i),D_i)'$, denoted by $F_{\hat e}$. This models a situation in which the linear regression on the estimated propensity scores may offer a close approximation to the true conditional expectation. 

Third, we compute the full-sample objects $\beta_{\eG,F}$, $\tau_{\eG_{\hat e}^1,F_{\hat e}^{}}$, $m_{\eG_{\hat e},F_{\hat e}}$'s, and $c_{\eG_{\hat e}}$'s, where we  recall that $c_{\eG_{\hat e}}$ is the distance between the push-forwards $\hat e\sharp \eG^d$ of $\eG^d$ via $\hat e$ and
	\begin{align*}
		\tau_{\eG_{\hat e}^1,F_{\hat e}^{}}=\Exp_{\eG_{\hat e}^1}[\Exp_{F_{\hat e}^{}}[Y|\hat e(X),D=1] - \Exp_{F_{\hat e}^{}}[Y|\hat e(X_i),D=0]|D=1].	
	\end{align*}
	Fourth, we apply nearest-neighbor matching to the constructed simulation sample $\mathfrak S$, pairing each of the $\mathsf N_1$ black applicants to a unique white applicant having a similar $\hat e(X_i)$.\footnote{The remaining $N_0-N_1$ white applicants are discarded.} This step corresponds to the design phase in observational studies, though, in practice, one may choose alternative balancing techniques, such as matching on the full set of controls using the Mahalanobis metric, or trimming based on the estimated propensity scores. Recall that, by our notational convention, $\eG_{\hat e}^*$ denotes the empirical distribution of $(\hat e(X),D)'$ within the matched subsample $\mathfrak S^*\subseteq \mathfrak S$. Fifth, we compute the subsample objects $\beta_{\eG^*,F}$, $\tau_{\eG_{\hat e}^{*1},F_{\hat e}^{}}$, $m_{\eG_{\hat e}^*,F_{\hat e}^{}}$'s, and $c_{\eG_{\hat e}^*}$'s. Finally, we repeat steps 2 to 5 for $\mathsf R$ times.

We assume that a researcher considers three specifications:\begin{description}[itemsep=0em,topsep=3pt]
	\item[Specification A:] $Y_i=\alpha_A + \beta_A D_i + E_i$,
	\item[Specification B:] $Y_i=\alpha_B + \beta_B D_i + \gamma_B \hat e(X_i) + E_i$, and 
	\item[Specification C:] $Y_i=\alpha_C + \beta_C D_i + \sum_{r=1}^3\gamma_{C,r} \hat e(X_i)^r + E_i$,
\end{description}
These are relevant to a situation where a researcher weighs increasing the strength of the linear approximation against the potential statistical imprecision it entails.\footnote{Note, however, that including higher order terms does not necessarily improve the strength of approximation to $f(\cdot,0)$.}
  
We explore three scenarios $(\mathsf N_1,\mathsf N_0)\in\{(50,75),(50,100),(50,125)\}$, in all of which the number $\mathsf N_0$ of white applicants is kept small to underscore the utility of the design phase; for conciseness, we present the results only for the second one. Set $\mathsf R=500$.

\begin{figure}[ht]
\centering
\subfloat[Specification A]{\includegraphics[width=0.33\textwidth]{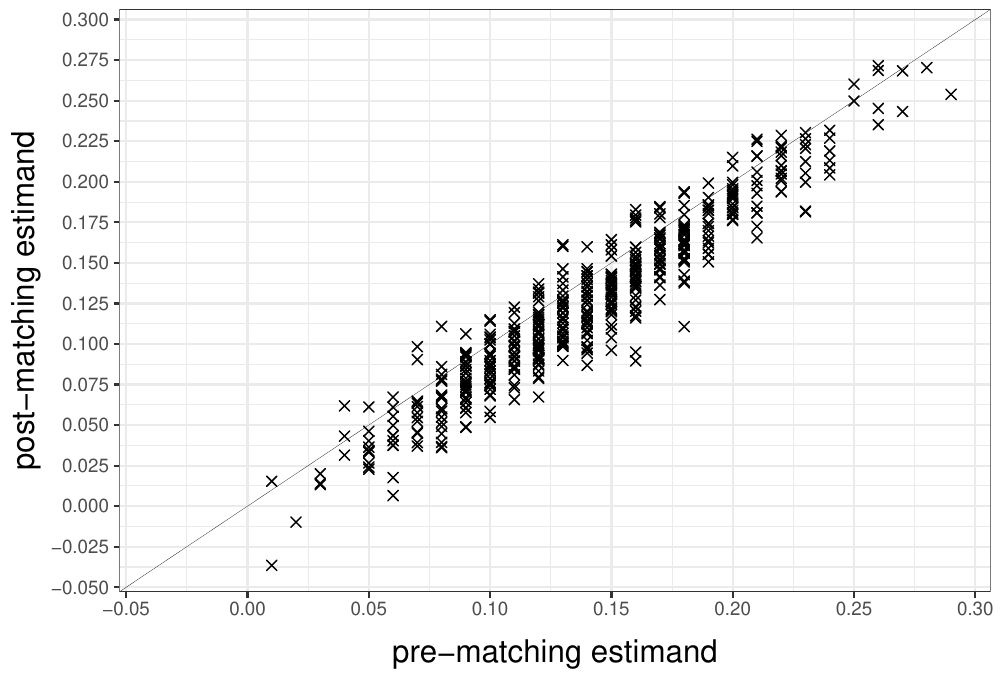}}\hfil
\subfloat[Specification B]{\includegraphics[width=0.33\textwidth]{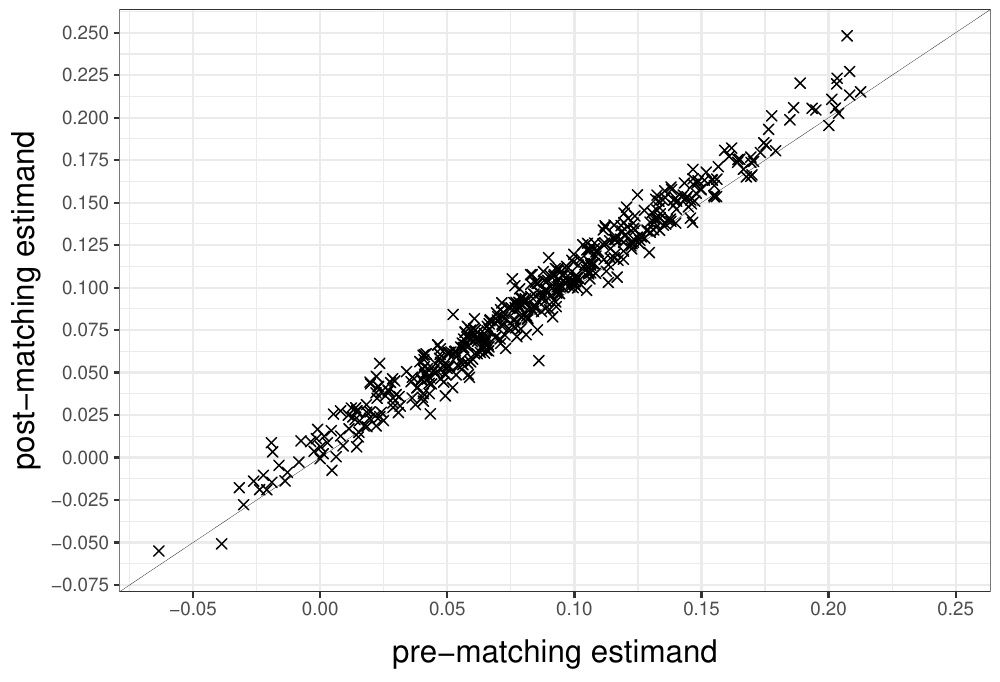}}\hfil
\subfloat[Specification C]{\includegraphics[width=0.33\textwidth]{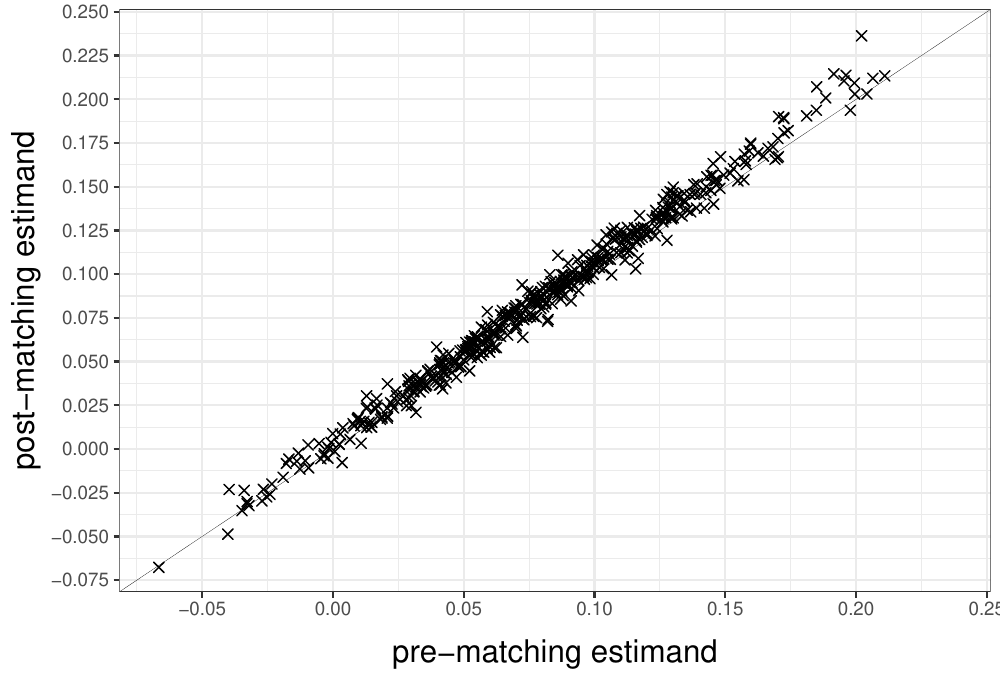}}
\caption{Pre- and Post-estimands ($\mathsf N_1=50$, $\mathsf N_0=100$)}
\label{fig:estimand-100}
\end{figure}

Figure \ref{fig:estimand-100} displays the pairs of the pre- and post-matching estimands, i.e., $\beta_{\eG,F}$ and $\beta_{\eG^*,F}$, showing how nearest-neighbor matching adjusts the estimands of the regressions. In general, it reduces the estimands for Specification A, while slightly increasing those for Specifications B--C.

Figure \ref{fig:bias-100} depicts the pairs $(|\beta_{\eG,F}-\tau_{\eG_{\hat e}^1,F_{\hat e}^{}}|,|\beta_{\eG^*,F}-\tau_{\eG_{\hat e}^{*1},F_{\hat e}^{}}|)$ of the values of the pre- and post-matching biases. It shows that the estimand adjustment in the design phase in general reduces the bias of a regression. In light of our justification, this reduction is presumably driven by the extended range of the degree of misspecification that can be allowed for the regression model. The next set of figures confirms this.
\begin{figure}[ht]
\subfloat[Specification A]{\includegraphics[width=0.33\textwidth]{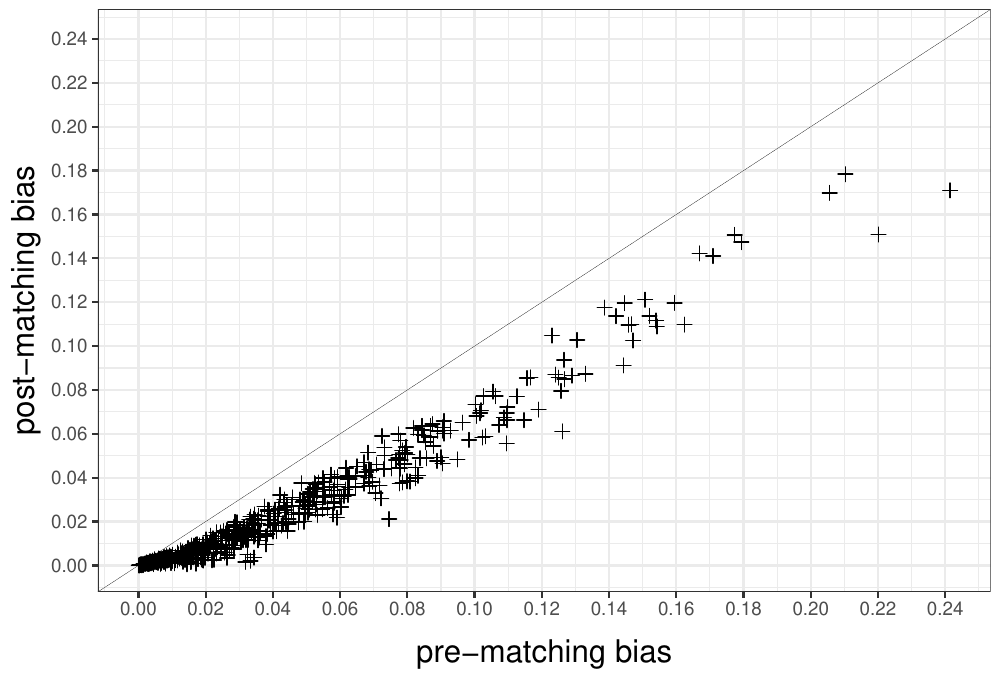}}\hfil
\subfloat[Specification B]{\includegraphics[width=0.33\textwidth]{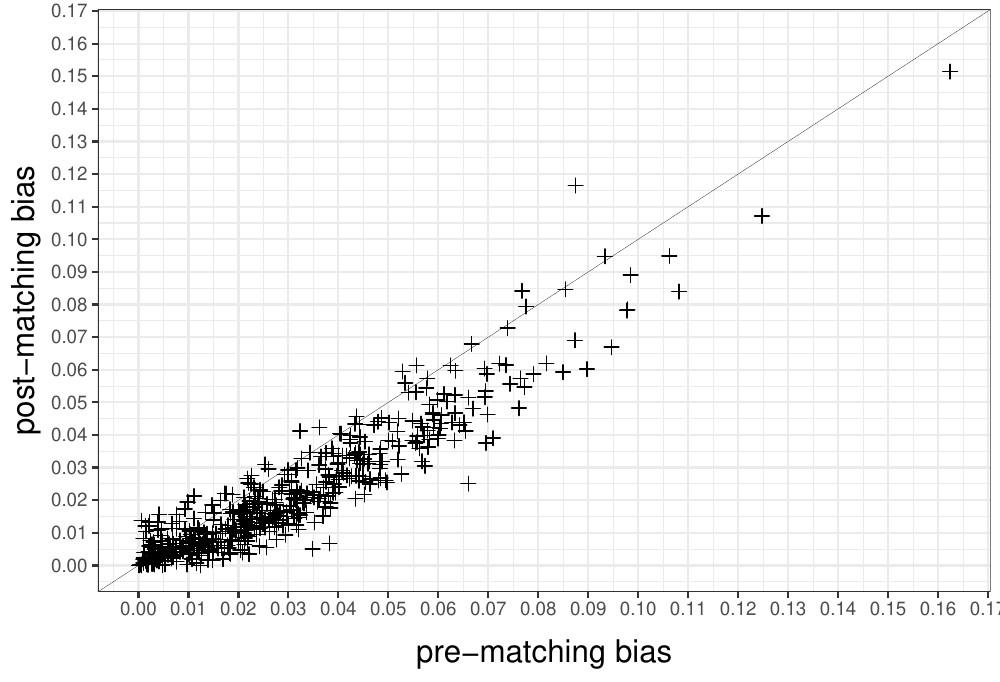}}\hfil	
\subfloat[Specification C]{\includegraphics[width=0.33\textwidth]{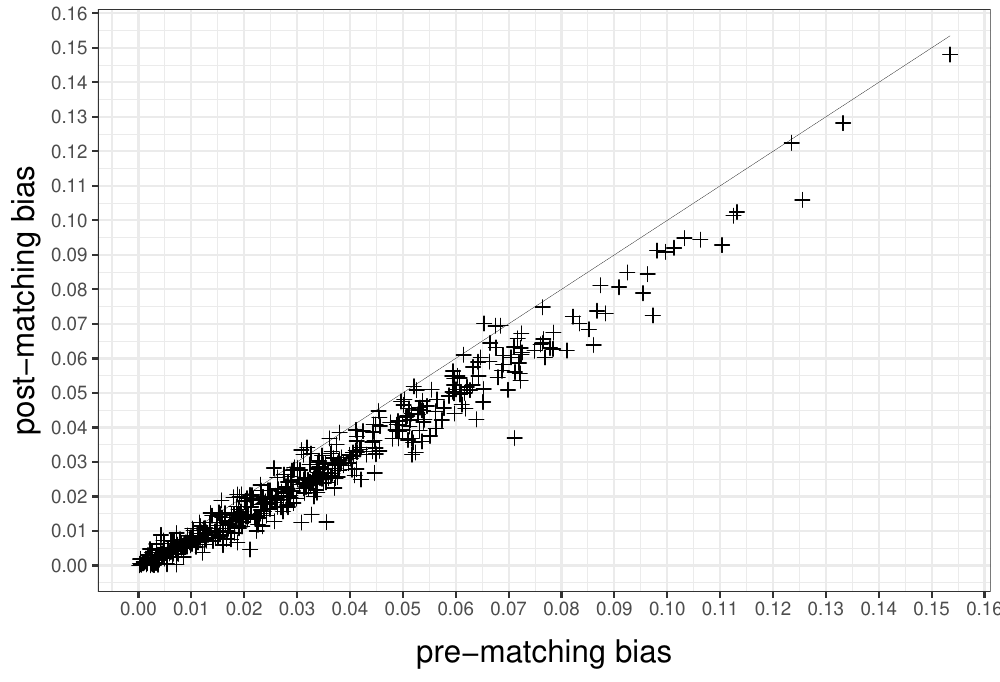}}
\caption{Pre- and Post-biases ($\mathsf N_1=50$, $\mathsf N_0=100$)}
\label{fig:bias-100}
\end{figure}

The points in Figure \ref{fig:total-variation}--\ref{fig:mean-difference} denote the pairs of the values of model misspecification $m_{\eG_{\hat e}^*,F_{\hat e}^{}}$ and covariate imbalance $c_{\eG_{\hat e}^*}$ for the bounds derived in Section \ref{sec:bias-bounds} and \ref{sec:bias-bounds-TVDRLP}.\footnote{The figure for the L\'evy-Prokhorov bound is omitted due to the non-trivial computation of the L\'evy-Prokhorov distance.} The red ones are from the full samples $\mathfrak S$, i.e., $\eG_{\hat e}^*=\eG_{\hat e}^{}$, while the blue ones are from the subsamples $\mathfrak S^*$ constructed using the nearest-neighbor matching. The curves $c\mapsto \epsilon/c$ illustrate the target levels of precision: If a point lies within the curve with $\epsilon=0.005$, it indicates that the product of $m_{\eG_{\hat e}^*,F_{\hat e}^{}}$ and $c_{\eG_{\hat e}^*}$, and thus the bias, is of less than 0.005.

For total variation and density ratio bounds, to guarantee that imbalance measure takes non-trivial values, the estimated linear propensity scores $\hat e(X_i)$ are stratified so that the supports of the resulting conditional distributions can overlap. To be specific, we use the indicators of whether $\hat e(X_i)$ falls into each interquartile of $\eG_{\hat e}^0$, as covariates. Accordingly, we substitute Specifications B--C with a single saturated regression. 

For mean difference bound, we use a constant and the estimated linear propensity score as the summaries, i.e., $r(X_i)=(1,\hat e(X_i))'$. Thus, $c_{\eG_{\hat e}^*}^{\MD}$ is the pair of zero and the absolute mean difference in $\hat e(X_i)$'s between treated and untreated units.

A key observation from the figures is that the design phase yields the subsamples $\mathfrak S^*$ that demonstrate better covariate balance when compared to the full samples $\mathfrak S$. The distribution of the blue points is shifted towards the left relative to the red ones. As a result, more blue points are located within each reciprocal curve, which explains the overall reduction in bias after the design phase.

\begin{figure}[htp]
\centering
\subfloat[Without covariates]{\includegraphics[width=0.33\textwidth]{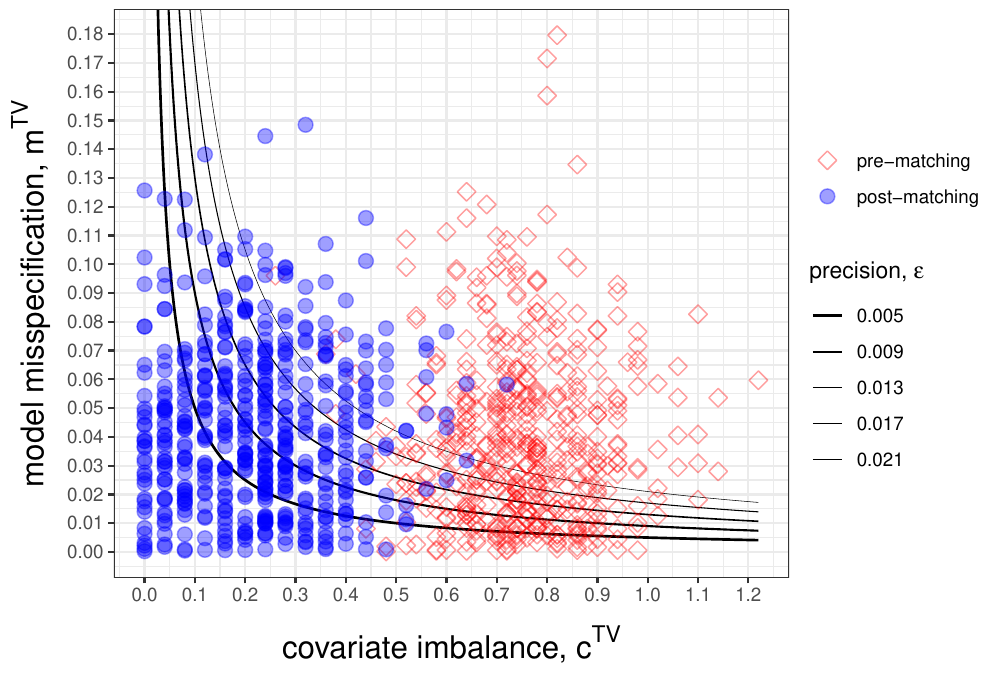}}
\subfloat[Saturated]{\includegraphics[width=0.33\textwidth]{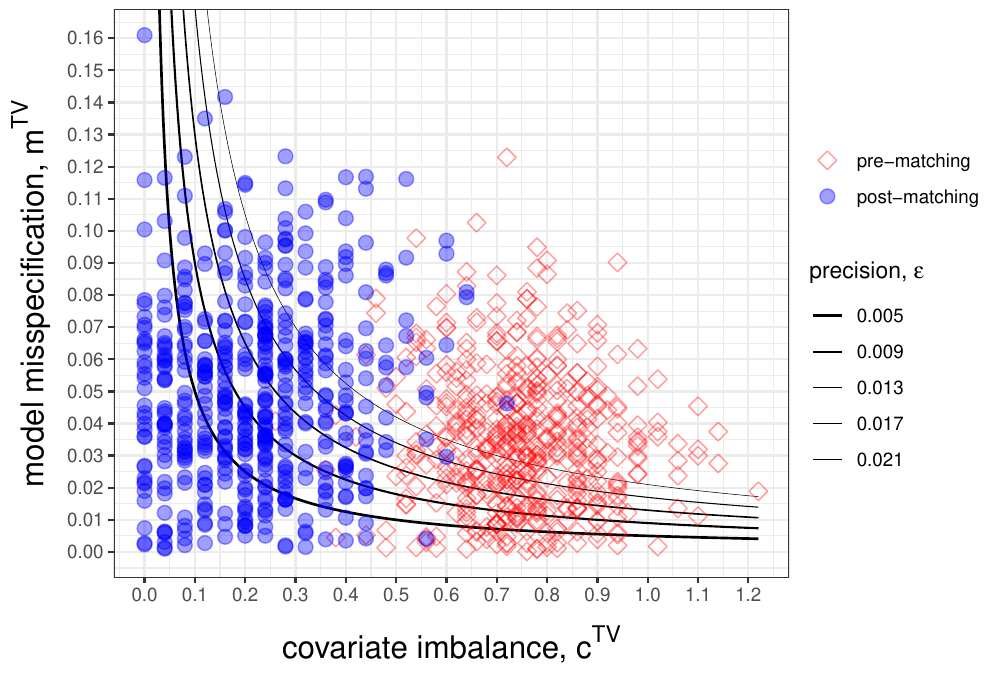}}	
\caption{Total variation bound ($\mathsf N_1=50$, $\mathsf N_0=100$)}
\label{fig:total-variation}\bigskip
\subfloat[Specification A]{\includegraphics[width=0.33\textwidth]{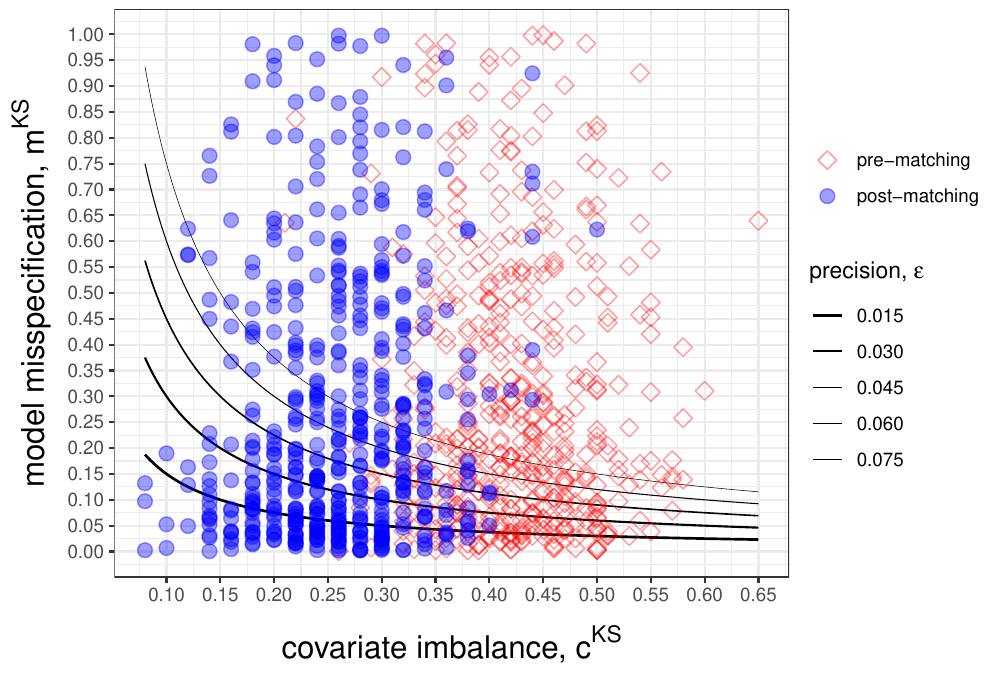}}\hfil
\subfloat[Specification B]{\includegraphics[width=0.33\textwidth]{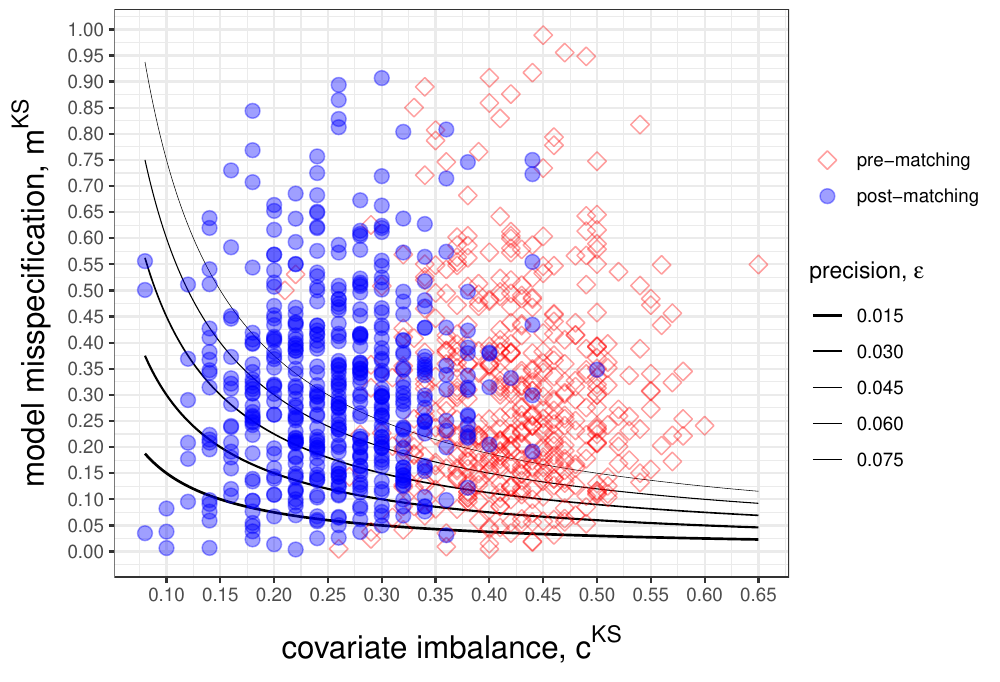}}\hfil
\subfloat[Specification C]{\includegraphics[width=0.33\textwidth]{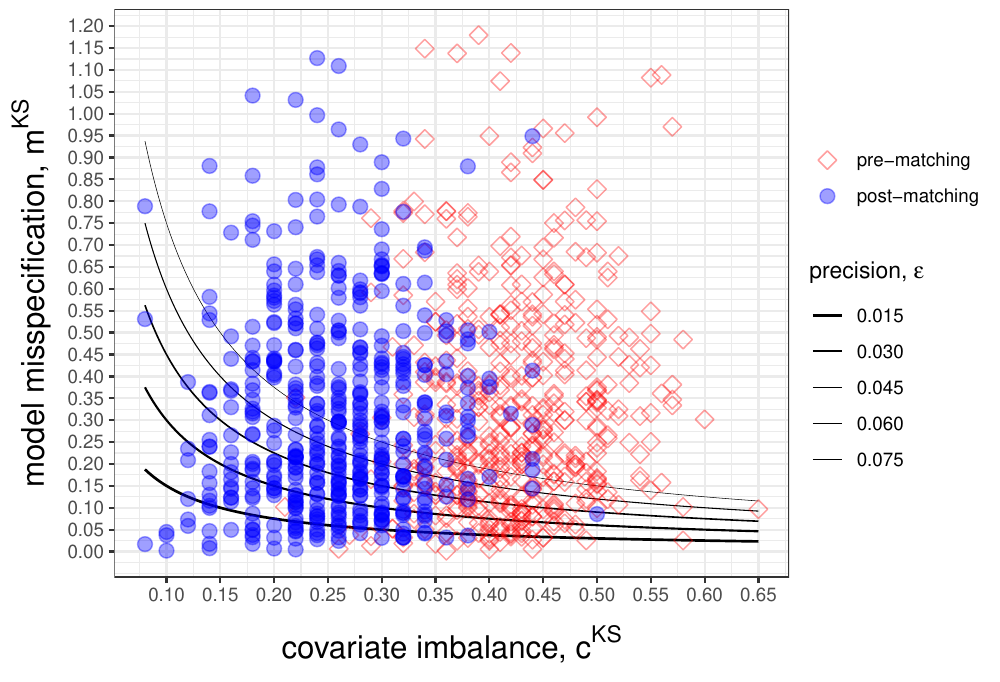}}
\caption{Kolmogorov-Smirnov bound ($\mathsf N_1=50$, $\mathsf N_0=100$)}
\label{fig:kolmogorov-smirnov}\bigskip
\subfloat[Without covariates]{\includegraphics[width=0.33\textwidth]{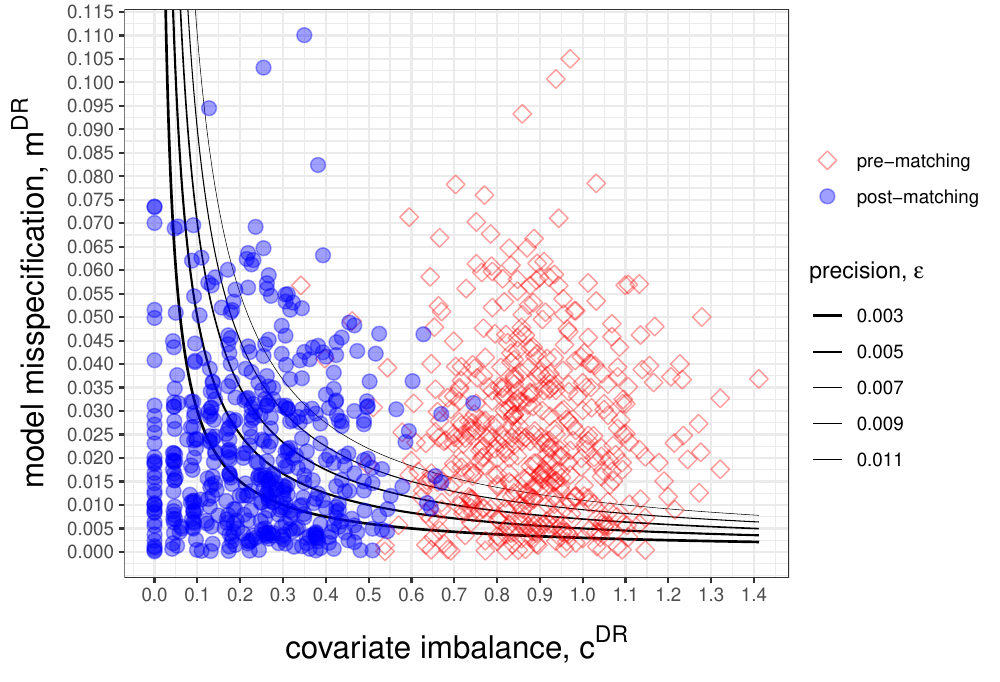}}
\subfloat[Saturated]{\includegraphics[width=0.33\textwidth]{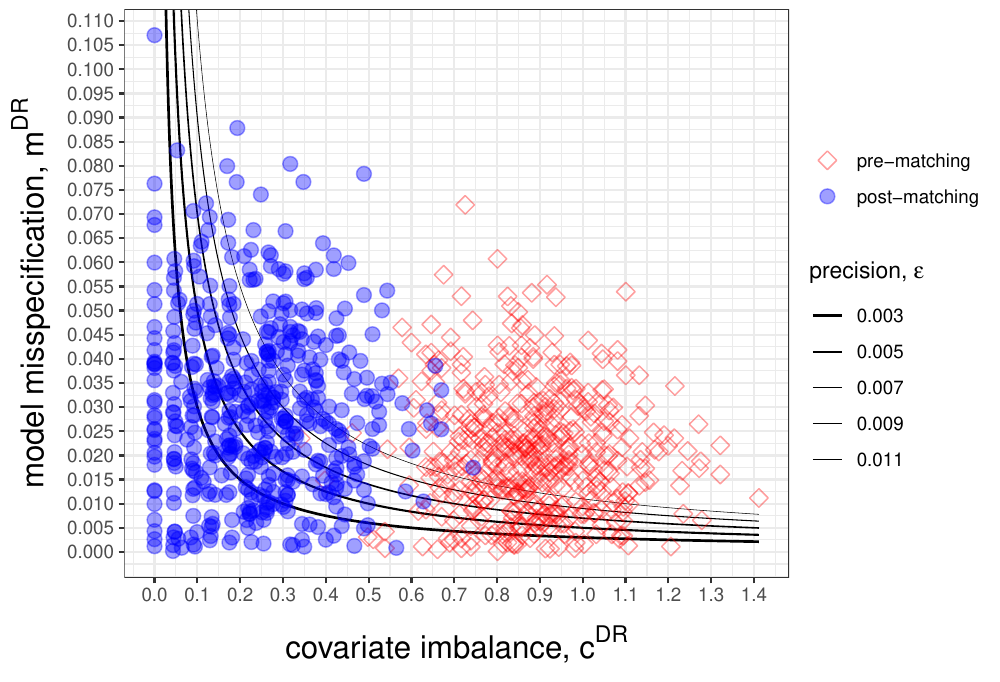}}
\caption{Density ratio bound ($\mathsf N_1=50$, $\mathsf N_0=100$)}
\label{fig:density-ratio}
\end{figure}

\begin{figure}
\centering
\subfloat[Specification A]{\includegraphics[width=0.33\textwidth]{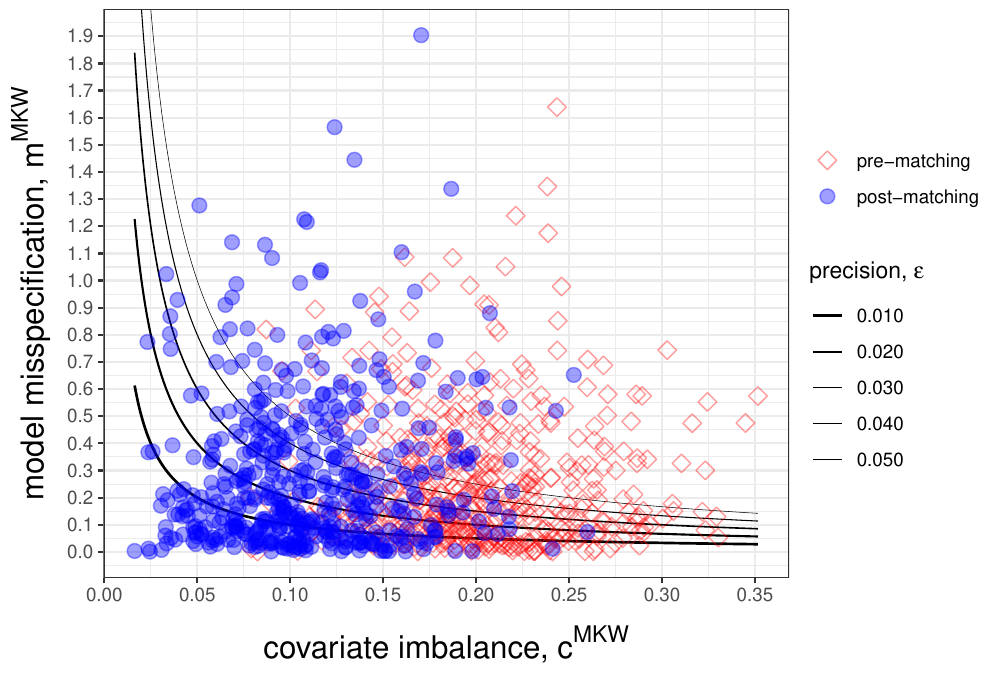}}\hfil
\subfloat[Specification B]{\includegraphics[width=0.33\textwidth]{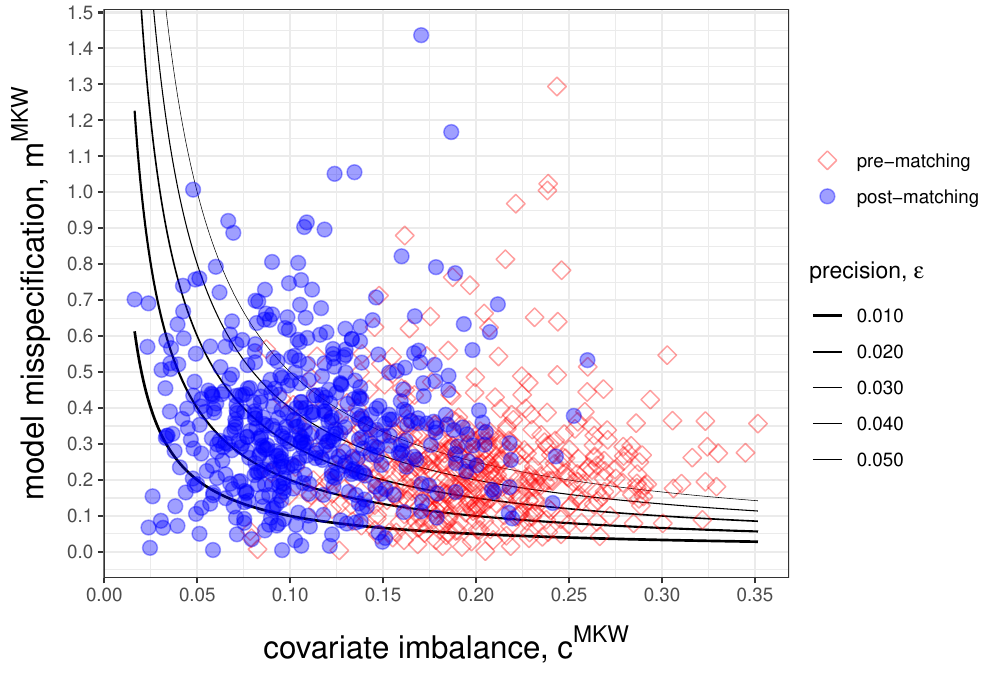}}\hfil
\subfloat[Specification C]{\includegraphics[width=0.33\textwidth]{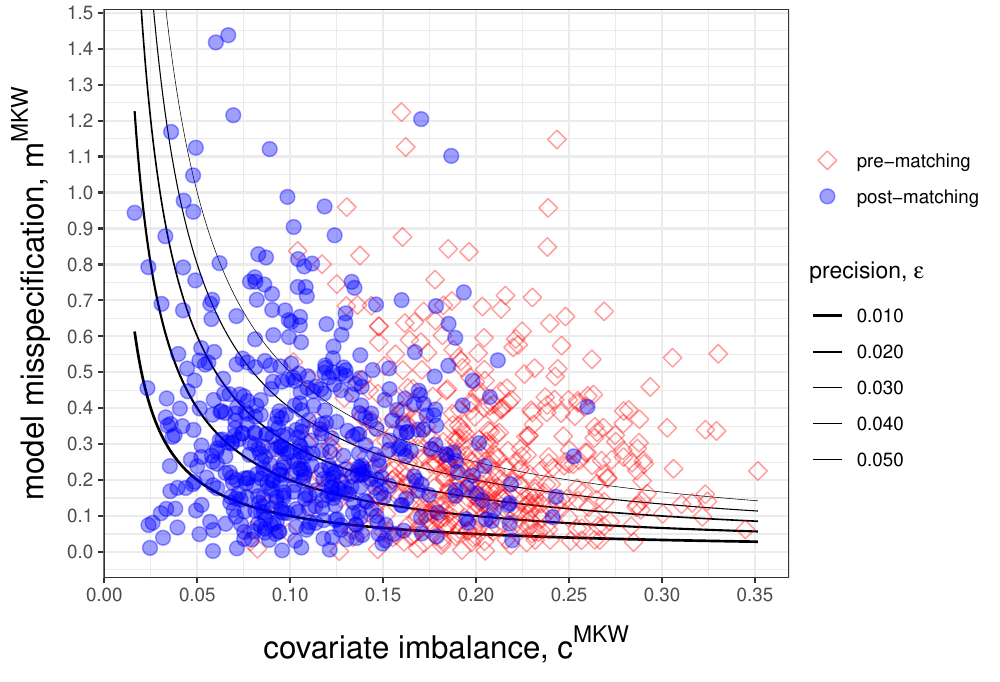}}
\caption{Monge-Kantorovich/Wasserstein bound ($\mathsf N_1=50$, $\mathsf N_0=100$)}
\label{fig:monge-kantorovich-wasserstein}\bigskip
\subfloat[Specification A]{\includegraphics[width=0.33\textwidth]{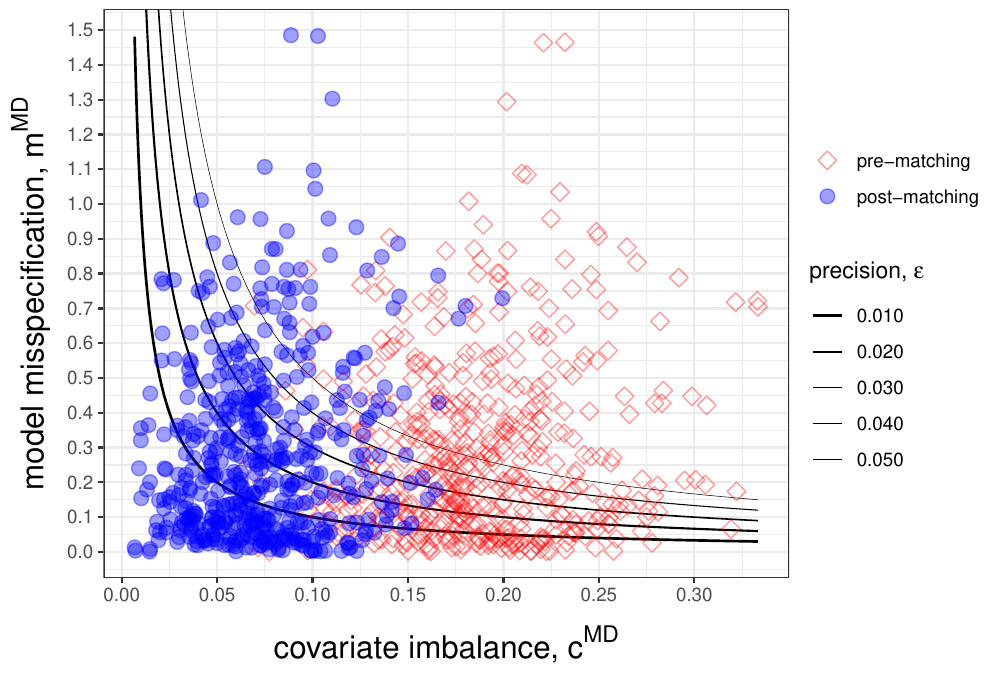}}\hfil
\subfloat[Specification B]{\includegraphics[width=0.33\textwidth]{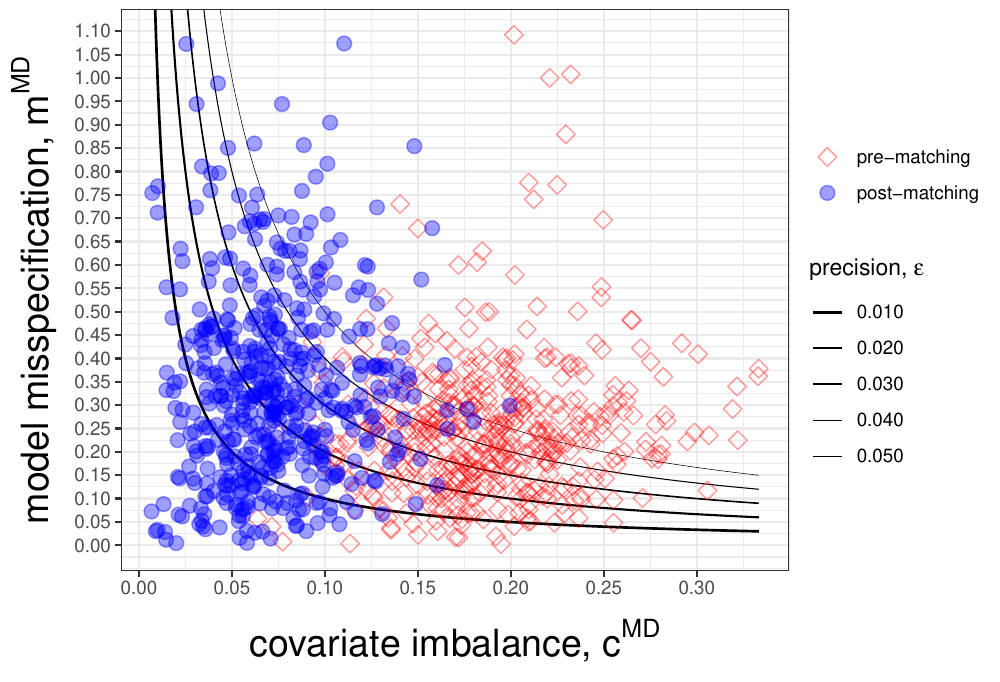}}\hfil
\subfloat[Specification C]{\includegraphics[width=0.33\textwidth]{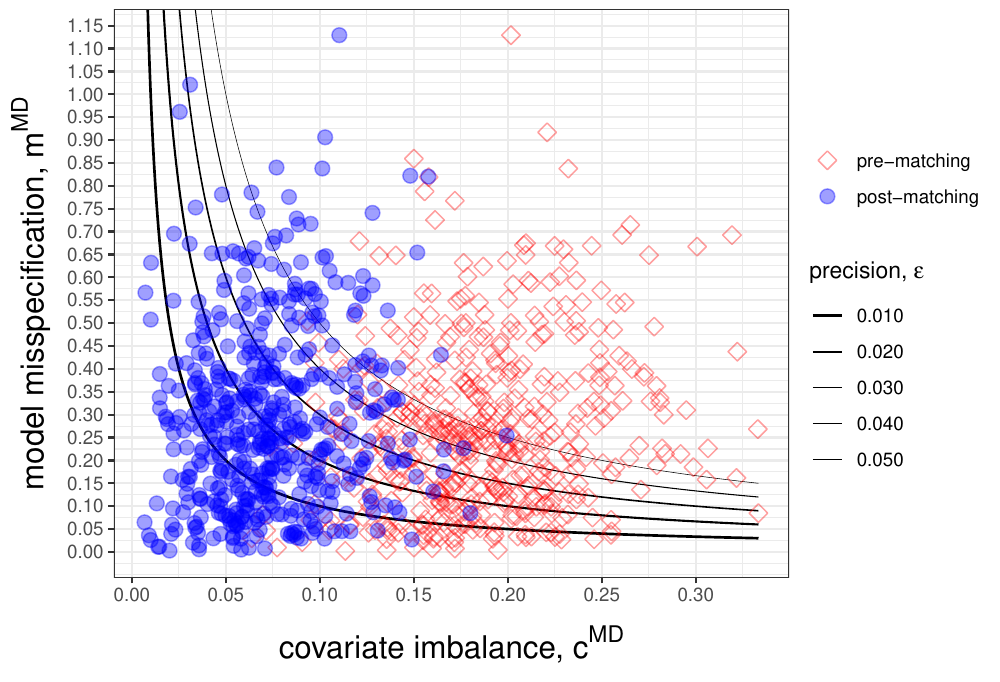}}
\caption{Mean difference bound ($\mathsf N_1=50$, $\mathsf N_0=100$)}
\label{fig:mean-difference}
\end{figure}

\clearpage
\section{More on application}
We illustrate how a researcher may utilize the total variation and density ratio distances to convey the robustness of her regression estimate $\hat\beta$ to the readers of her report. Here we assume that, prior to running the regression, she coarsened her control variables $X_i$ with a covariate function $s$ defined in the following manner: First, she divides the support of the conditional distribution $\hat e\sharp \eG^0$ of the estimated propensity scores $\hat e(X_i)$ given $D_i=0$ into $\mathsf T$ consecutive strata, say, $1,\dots, \mathsf T$. Then, define
\begin{align}
	s(\cdot) \equiv (\mathbf 1\{\hat e(\cdot)\in 1\},\dots,\mathbf 1\{\hat e(\cdot)\in \mathsf T\})'.
\end{align}
The parameter of our interest is $\tau_{\eG_s^1,F_s^{}}$, where we recall that
\begin{align}
\tau_{\eG_s^1,F_s^{}}=\Exp_{\eG_s^1}[\Exp_{F_s}[Y|s(X),D=1]-\Exp_{F_s}[Y|s(X),D=0]|D=1],
\end{align}
which reduces to the original $\tau_{\eG^1,F}$ if $s$ has the balancing property.

We consider a case where the researcher uses the estimated linear propensity score and runs a saturated regression. To be specific, if $(\check\alpha,\check\gamma')'$ denotes the estimates from the linear regression of $D_i$ on $(1,X_i')'$, then she sets $\hat e(X_i)\equiv \check\alpha  + X_i'\check\gamma$. Her regression model is
\begin{align}
	Y_i = \alpha + \beta D_i + (\mathbf 1\{\hat e(X_i)\in 1\},\dots,\mathbf 1\{\hat e(X_i)\in \mathsf T\})'\gamma + E_i.
\end{align}
The strata are partitioned based on the quartiles 0.157, 0.196, and 0.229 of $\hat e\sharp \eG^0$; in other words, $\mathsf T=4$.

The black circles in Figure \ref{fig:model-index} depicts the joint distribution of the estimated linear propensity score $\hat e(X_i)$ and the denial indicator $Y_i$. The red and blue diamonds show the distribution of the values $l_{\theta_{\eG,F}}(s(X_i),0)\approx\hat\alpha + s(X_i)'\hat\gamma$ of the population regression function for white and black applicants. Figure \ref{fig:ps} illustrates the imbalance between the conditional distributions $s\sharp \eG^d$ of $s(X_i)$'s. The black applicants tend to have higher estimated linear propensity scores $\hat e(X_i)$, placing them in the upper quartile regions.

\begin{figure}[htp]
	\centering
	\includegraphics[width=0.45\textwidth]{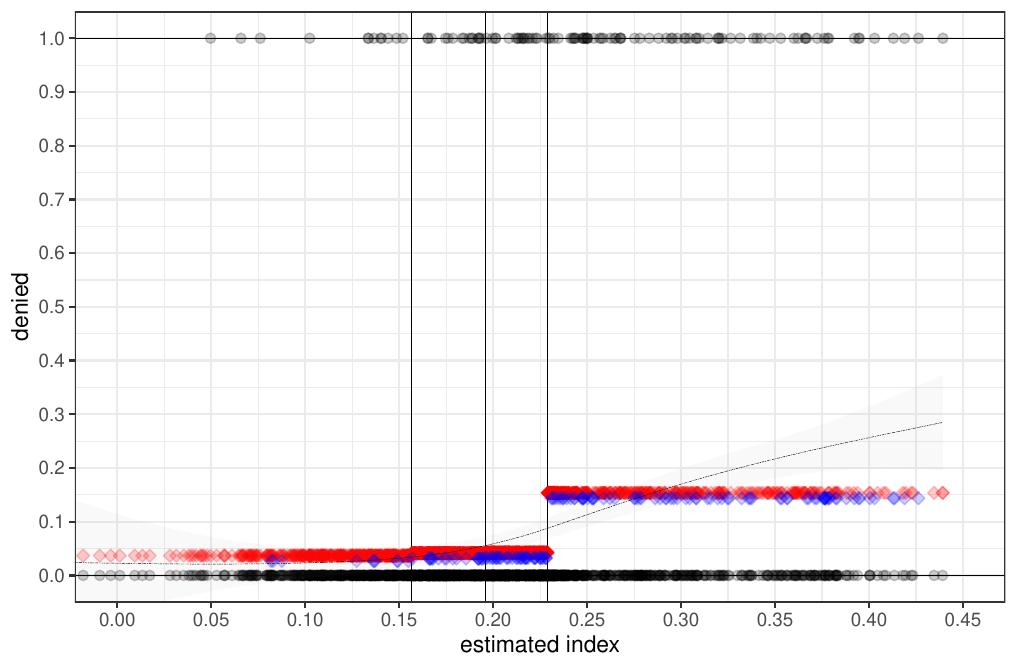}
	\caption{Researcher's model, $l_{\theta_{\eG,F}}(\cdot,0)$}
	\label{fig:model-index}
	\vspace{\baselineskip}
	\includegraphics[width=0.45\textwidth]{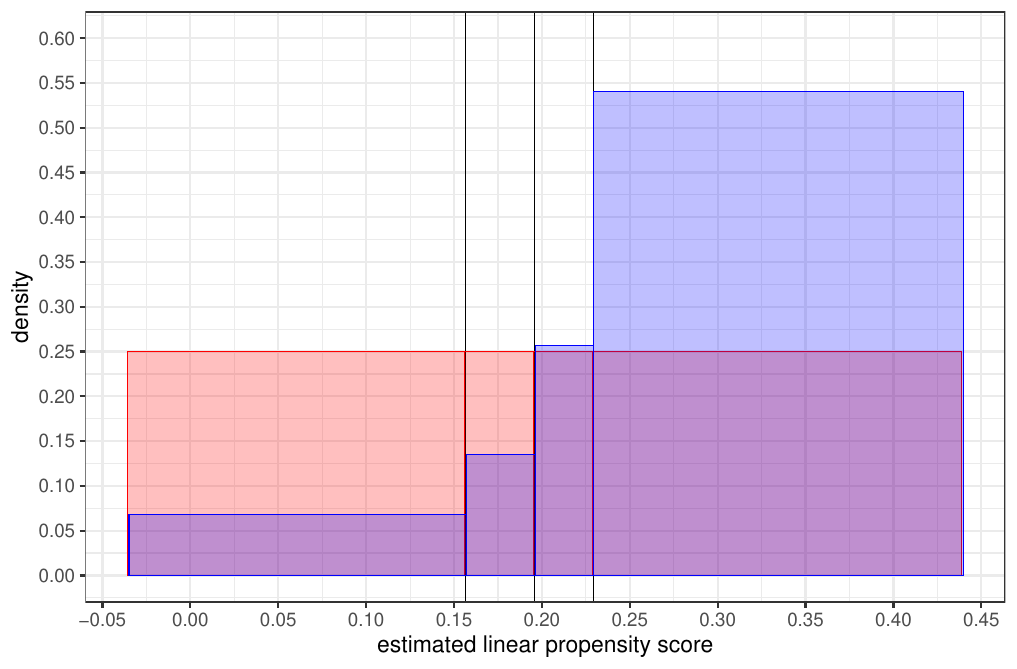}
	\caption{Imbalance in $s\sharp\eG^d$'s}
	\label{fig:ps}
	\vspace{\baselineskip}
	\includegraphics[width=0.45\textwidth]{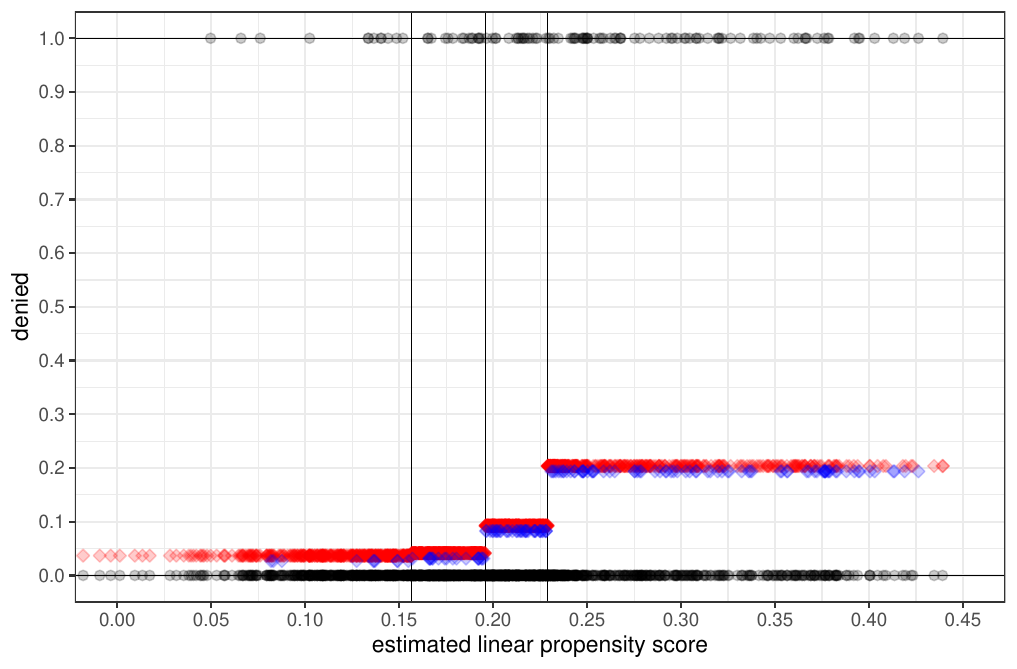}	
	\caption{Readers' perturbation}
	\label{fig:pTVDR}
\end{figure}

Some readers---based on the rough estimate for the conditional expectation function $\Exp_{F_{\hat e}}[Y|\hat e(X)=\cdot,D=0]$ indicated by the black dashed curve in Figure \ref{fig:model-index}---may argue that the denial probabilities of the applicants with higher $\hat e(X_i)$'s should be greater than her model predictions $l_{\theta_{\eG,F}}(3,0)\approx 0.038$ and $l_{\theta_{\eG,F}}(4,0)\approx 0.149$. 

Accordingly, they examine an adversarial state 
\begin{align*}
	f^s(\cdot,0)
	=l_{\theta_{\eG,F}}(\cdot,0) + 0.05\cdot \mathbf 1\{\cdot\in \{3,4\}\},
\end{align*}
allocating values 0.05 greater to the third and fourth quartiles, as illustrated in Figure \ref{fig:pTVDR}. The readers multiply the total variation and density ratio metrics $c_{\eG_s}^{\TV}=0.594$ and $c_{\eG_s}^{\DR}=0.723$ respectively by $m_{\eG_s,F_s}^{\TV}=0.05$ and $m_{\eG_s,F_s}^{\DR}=\sqrt{0.05^2\times (0.25+0.25)}=0.035$ to obtain the maximum possible biases $0.030$ and $0.026$.\footnote{Note that the latter calculation presumes partial access to the conditional distribution $\hat e(\cdot)\sharp \eG^0$ by the readers.} Given her regression estimate $\hat\beta=0.106$ of moderate magnitude, the readers may accept her conclusion that minority status affects banks' lending decision: Even if the denial probabilities for some white applicants with high propensity scores were 0.05 higher, the researcher's argument remains essentially unaffected.

Using the reported standard error $se_{\beta_{\eG,F}}(\hat\beta)=0.016$, the readers may themselves construct the robustified confidence intervals $C_{0.95}^{\TV}(m)$ and $C_{0.95}^{\DR}(m)$---illustrated in the left panels of Figures \ref{fig:TV}--\ref{fig:DR}---and obtain the minimum degrees $\mathfrak m_0$ of misspecification that render the rejections based on $C_{0.95}^{\TV}(0)$ and $C_{0.95}^{\DR}(0)$ 
	inconclusive. The right panels are their subsample counterparts, where the researcher has applied in the design phase the nearest neighbor matching on $\hat e(X_i)$'s. The regression estimate $\hat\beta^*=0.101$ remains similar, but the standard error $se_{\beta_{\eG^*,F}}(\hat\beta^*)$ has increased to 0.043. 	The figures indicate that the readers may no longer concern, at least, about potential misspecification.

\begin{figure}
	\subfloat[Pre-matching, $c_{\eG_s}^{\TV}=0.594$]{\includegraphics[width=0.45\textwidth]{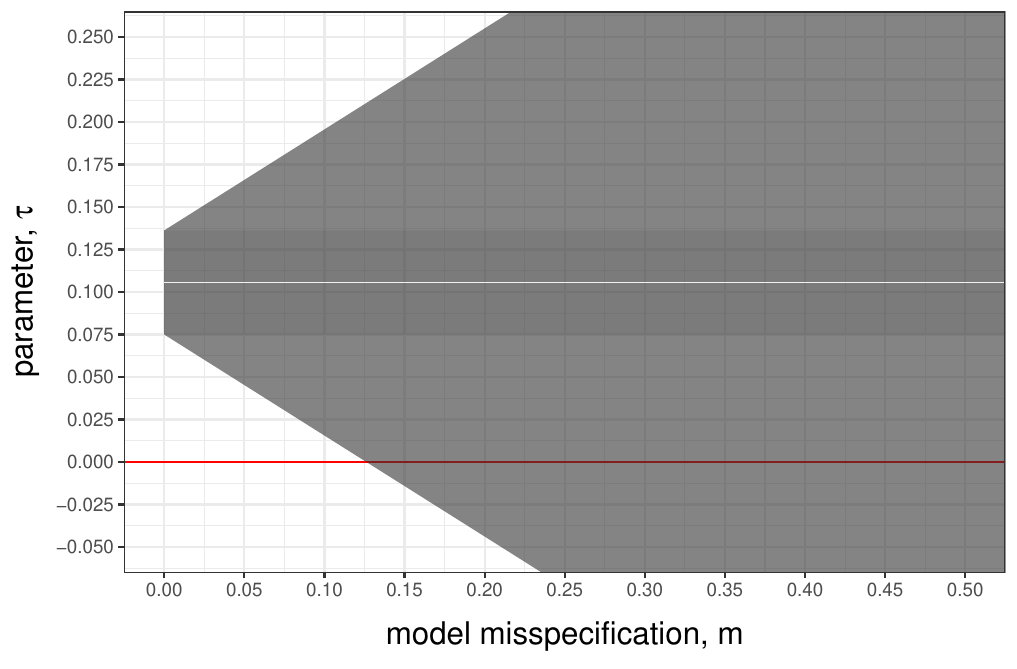}}\quad
	\subfloat[Post-matching, $c_{\eG_s^*}^{\TV}=0$]{\includegraphics[width=0.45\textwidth]{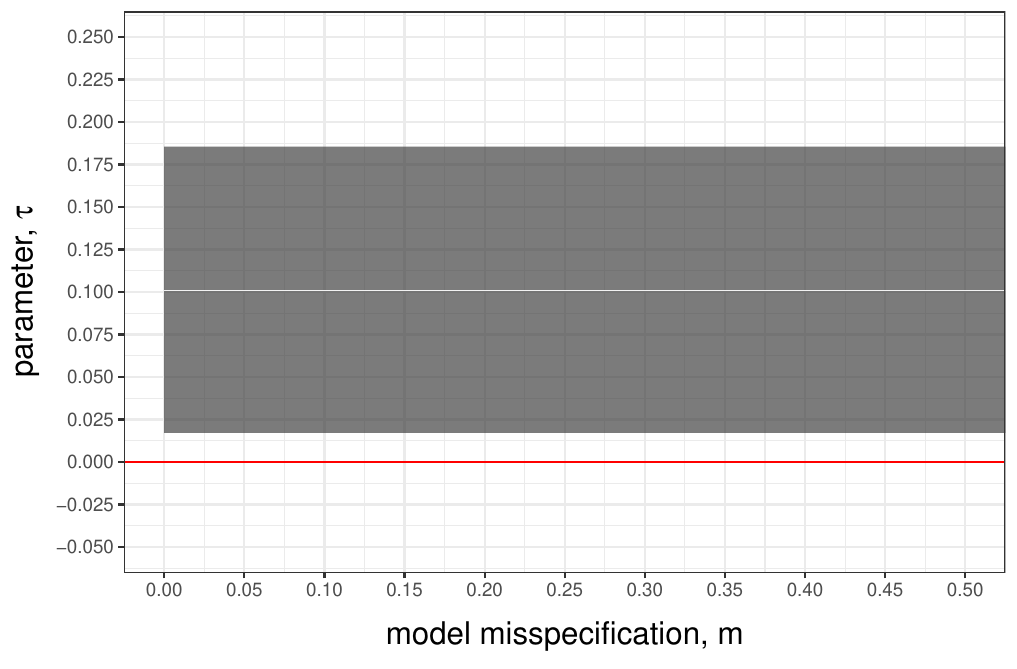}}
	\caption{The graphs of the robustified confidence intervals $C_{0.95}^{\TV}(m)$}
	\label{fig:TV}
	\vspace{\baselineskip}
	\subfloat[Pre-matching, $c_{\eG_s}^{\DR}=0.723$]{\includegraphics[width=0.45\textwidth]{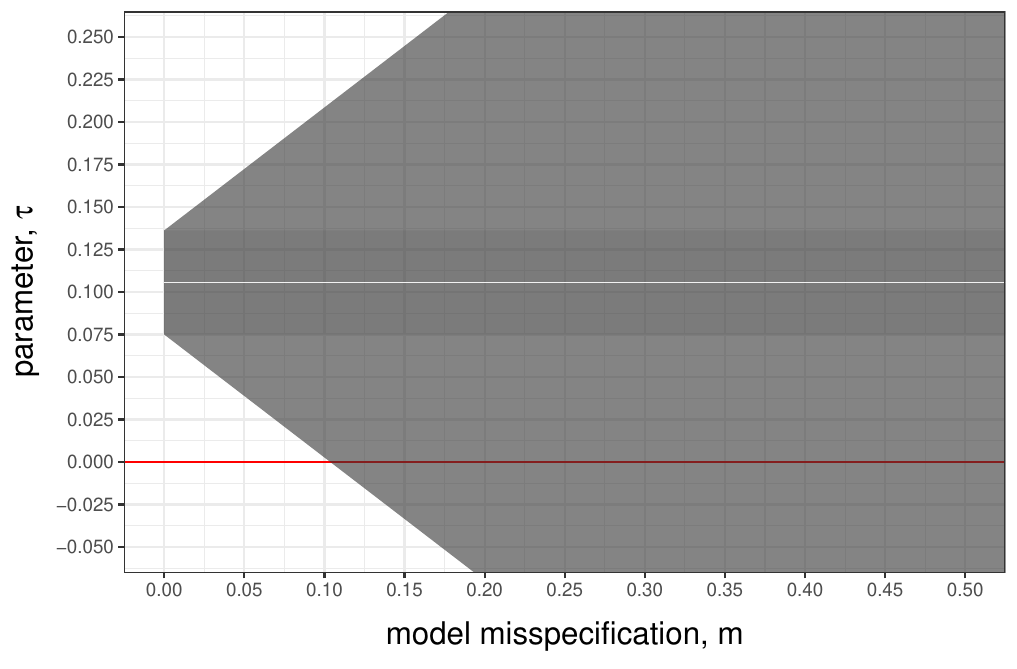}}\quad
	\subfloat[Pre-matching, $c_{\eG_s^*}^{\DR}=0$]{\includegraphics[width=0.45\textwidth]{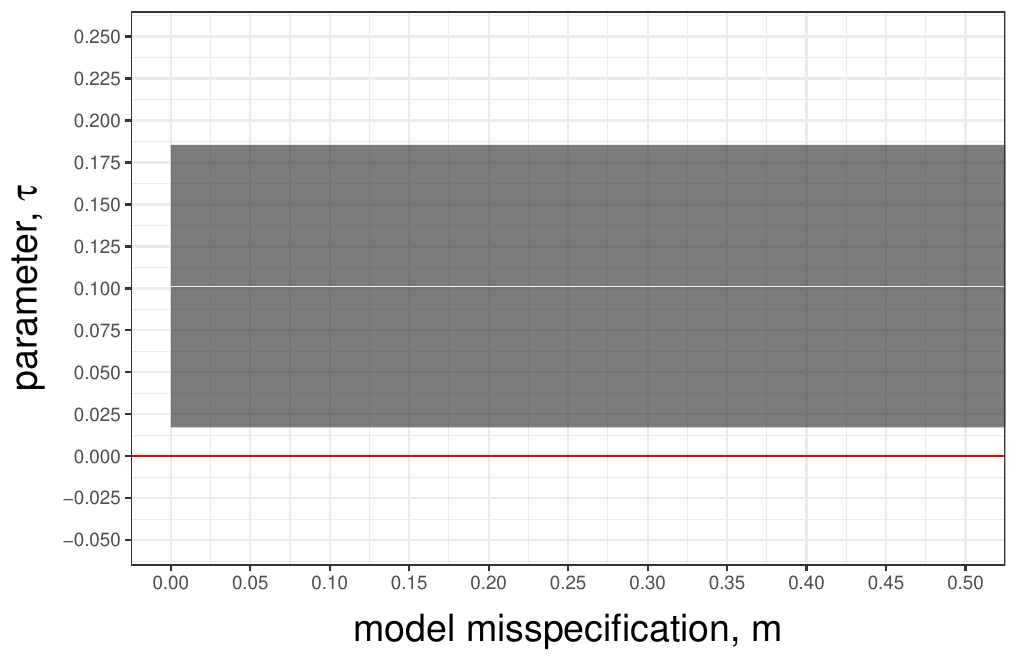}}
	\caption{The graphs of the robustified confidence intervals $C_{0.95}^{\DR}(m)$}
	\label{fig:DR}
\end{figure}

\clearpage
\section{Estimated index as a covariate}
Consider a regression model
\begin{align}
	Y_i = \tilde\alpha + \tilde\beta D_i + s(X_i)'\phi \cdot \tilde\gamma + \tilde E_i,  \label{eq:reg-model-gamma}
\end{align}
which specializes to the induced regression model \eqref{eq:reg-model-index} in Section \ref{sec:application} when $\phi=\hat\gamma$.

Let $\tilde Z_i(\phi)\equiv (1,D_i,s(X_i)'\phi)'$. Define $\hat{\tilde\theta}^*(\phi)\equiv (\tilde\alpha^*(\phi),\tilde\beta^*(\phi),\tilde\gamma^*(\phi))'$ to be the least squares estimator obtained by regressing $Y_i$ on $\tilde Z_i(\phi)$ within $\mathcal S^*$, i.e., 
\begin{align}
	\hat{\tilde\theta}^*(\phi)=\biggl(\frac{1}{|\mathcal S^*|}\sum_{i\in\mathcal S^*}\tilde Z_i(\phi)\tilde Z_i(\phi)'\biggr)^{-1}\biggl(\frac{1}{|\mathcal S^*|}\sum_{i\in\mathcal S^*}\tilde Z_i(\phi)Y_i\biggr).
\end{align}
Also, define $\tilde\theta_{\eG^*,F}(\phi)\equiv (\tilde\alpha_{\eG^*,F}(\phi),\tilde\beta_{\eG^*,F}(\phi),\tilde\gamma_{\eG^*,F}(\phi)')'$ to be the least squares estimator obtained from the hypothetical regression of $\Exp_F[Y|X=X_i,D=D_i]$ on $\tilde Z_i(\phi)$, i.e., 
\begin{align}
	\tilde\theta_{\eG^*,F}(\phi) = \biggl(\frac{1}{|\mathcal S^*|}\sum_{i\in\mathcal S^*}\tilde Z_i(\phi)\tilde Z_i(\phi)'\biggr)^{-1}\biggl(\frac{1}{|\mathcal S^*|}\sum_{i\in\mathcal S^*}\tilde Z_i(\phi)\Exp_F[Y|X=X_i,D=D_i]\biggr).
\end{align}
Abbreviate $\tilde Z_i(\hat\gamma)$, $\hat{\tilde\theta}^*(\hat\gamma)$, and $\tilde\theta_{\eG^*,F}(\hat\gamma)$ by $\tilde Z_i$, $\hat{\tilde\theta}^*$, and $\tilde\theta_{\eG^*,F}$, respectively, which accords with our notation in Section \ref{sec:application}. 

We justify treating $\hat\gamma$ as fixed in Section \ref{sec:application}: The next assumptions are analogues of Assumptions \ref{ass:design-matrix-inverse}--\ref{ass:variance-consistency}. Define
\begin{align}
	\tilde\Gamma^*(\phi)\equiv \Exp_{\eG^*}[\tilde Z(\phi)\tilde Z(\phi)']=\frac{1}{|\mathcal S^*|}\sum_{i\in\mathcal S^*}\tilde Z_i(\phi)\tilde Z_i(\phi)'.
\end{align}

\begin{assume}\label{ass:design-matrix-inverse-index} 
	For some fixed positive constant $\underline\lambda>0$, 
	\begin{align*}
		\liminf_{|\mathcal S|\rightarrow\infty}	\lambda_{\min}[\tilde\Gamma^*(\gamma_{\eG,F})]\geq \underline\lambda
	\end{align*}
	given $\{(X_i',D_i)'\}_{i\in\mathcal S}=\{(x_i',d_i)'\}_{i\in\mathcal S}$. 
\end{assume}

Denote by $\mathbb T^{*d}(\phi)$ the support of the push-forward $s(\cdot)'\phi\sharp\eG^{*d}$ of $\eG^{*d}$ via $x\mapsto s(x)'\phi$.  
\begin{assume}\label{ass:bounded-lyapunov-index}
For each $d\in\{0,1\}$, $\Exp_F[\|\tilde Z(\gamma_{\eG^*,F})U\|^{2+\delta}|X=\cdot,D=d]$ is bounded on $\mathbb T^d(\gamma_{\eG^*,F})$ by a fixed positive constant given $\{(X_i',D_i)'\}_{i\in\mathcal S}=\{(x_i',d_i)'\}_{i\in\mathcal S}$. 
\end{assume}

Define
\begin{align}
	\tilde\Delta^*(\phi)
	\equiv \frac{1}{|\mathcal S^*|}\sum_{i\in\mathcal S^*}\tilde Z_i(\phi) \tilde Z_i(\phi)'\Exp_F[U_i^2|X_i,D_i].
\end{align}

\begin{assume}\label{ass:variance-consistency-index} 
For some positive constant $\underline\lambda>0$, 
\begin{align*}
	\liminf_{|\mathcal S|\rightarrow\infty}\lambda_{\min}[\tilde\Delta^*(\gamma_{\eG,F})]\geq \underline\lambda
\end{align*}
given $\{(X_i',D_i)'\}_{i\in\mathcal S}=\{(x_i',d_i)'\}_{i\in\mathcal S}$. 
\end{assume}

Let $\tilde\Sigma^*(\phi)\equiv (\tilde\Gamma^*(\phi))^{-1}\tilde\Delta^*(\phi)(\tilde\Gamma^*(\phi))^{-1}$.
\begin{prop}\label{thm:normality-index}
Under Assumptions \ref{ass:function-subsample}--\ref{ass:subsample-large} and \ref{ass:design-matrix-inverse-index}--\ref{ass:variance-consistency-index}, as $|\mathcal S|$ tends to infinity,
\begin{align}
	\rho(\mathcal L_{\tilde\Sigma^*(\gamma_{\eG,F})^{-\frac{1}{2}}\sqrt{|\mathcal S^*|}(\hat{\tilde\theta}^*(\gamma_{\eG,F}) - \tilde\theta_{\eG^*,F}(\gamma_{\eG,F}))|\{(X_i',D_i)'\}_{i\in\mathcal S}},\mathcal N(\mathbf 0,I))\rightarrow 0
\end{align}
given $\{(X_i',D_i)'\}_{i\in\mathcal S}=\{(x_i',d_i)'\}_{i\in\mathcal S}$. 
\end{prop}

\begin{assume}\label{ass:metric-index}
There exists a metric $\psi$ on $\mathbb R^p$, dominated by the Euclidean distance up to a fixed constant, such that, for each $d\in\{0,1\}$, 
 	\begin{enumerate}[label=(\alph*)]
 		\item $\Exp_F[\tilde Z_{j,i}(\gamma_{\eG,F})^{r_j}\tilde Z_{k,i}(\gamma_{\eG,F})^{r_k}(Y_i-\tilde Z_i(\gamma_{\eG,F})'\tilde\theta_{\eG^*,F}^*(\gamma_{\eG,F}))^{r_{j,k}}|X_i=\cdot,D_i=d]$ is bounded on $\mathbb X^{*d}$ by a fixed constant;
 		\item $\Exp_F[\tilde Z_{j,i}(\gamma_{\eG,F})^{r_j}\tilde Z_{k,i}(\gamma_{\eG,F})^{r_k}(Y_i-\tilde Z_i(\gamma_{\eG,F})'\tilde\theta_{\eG^*,F}^*(\gamma_{\eG,F}))^{r_{j,k}}|X_i=\cdot,D_i=d]$ is Lipschitz on $\mathbb X^{*d}$ with respect to $\psi$; and 
	 	\item $\diam(\mathbb X^{*d})\equiv \sup_{i,j\in\mathcal S^{*d}}\|X_i-X_j\|$ is bounded by a fixed constant
 	\end{enumerate}
 	given $\{(X_i',D_i)\}_{i\in\mathcal S}=\{(x_i',d_i)'\}_{i\in\mathcal S}$.
\end{assume}

Define
\begin{equation}
\begin{aligned}
	&\hat{\tilde\Delta}(\phi) \equiv \frac{1}{2|\mathcal S^*|}\sum_{i\in\mathcal S^*}(\tilde Z_i(\phi)\hat{\tilde E}_i^*(\phi)-\tilde Z_{l_{(X',D)'}^*(i)}(\phi)\hat{\tilde E}_{l_{(X',D)'}^*(i)}^*(\phi))\\
	&\hphantom{ \frac{1}{2|\mathcal S^*|}\sum_{i\in\mathcal S^*}(\tilde Z_i(\phi)\hat{\tilde E}_i^*(\phi)}
	\times (\tilde Z_i(\phi)\hat{\tilde E}_i^*(\phi)-\tilde Z_{l_{(X',D)'}^*(i)}(\phi)\hat{\tilde E}_{l_{(X',D)'}^*(i)}^*(\phi))',
\end{aligned}
\end{equation}
where $\hat{\tilde E}_i^*(\phi) \equiv Y_i - \tilde Z_i(\phi)'\hat{\tilde\theta}^*(\phi)$. Let $\hat{\tilde\Sigma}^*(\phi)\equiv (\tilde\Gamma^*(\phi))^{-1}\hat{\tilde\Delta}^*(\phi)(\tilde\Gamma^*(\phi))^{-1}$.
\begin{prop}\label{thm:variance-estimator-consistency-index} Under the conditions of Proposition \ref{thm:normality-index} and Assumption \ref{ass:metric-index},
\begin{align}
	 \Pr_F[\|\hat{\tilde\Sigma}^*(\gamma_{\eG,F}) - \tilde\Sigma^*(\gamma_{\eG,F})\|>\eta|\{(X_i',D_i)\}_{i\in\mathcal S}]\rightarrow 0, \forall \eta>0
\end{align}
given $\{(X_i',D_i)'\}_{i\in\mathcal S}=\{(x_i',d_i)'\}_{i\in\mathcal S}$, as $|\mathcal S|$ tends to infinity.
\end{prop}

Abbreviate $\tilde\Sigma^*(\hat\gamma)$ and $\hat{\tilde\Sigma}^*(\hat\gamma)$ by $\tilde\Sigma^*$ and $\hat{\tilde\Sigma}^*$. Then, since $\hat\gamma$ converges in conditional probability to $\gamma_{\eG,F}$ given $\{(X_i',D_i)\}_{i\in\mathcal S}=\{(x_i',d_i)'\}_{i\in\mathcal S}$,
\begin{align}
	\rho(\mathcal L_{\tilde\Sigma^{*-\frac{1}{2}}\sqrt{|\mathcal S^*|}(\hat{\tilde\theta}^* - \tilde\theta_{\eG^*,F})|\{(X_i',D_i)'\}_{i\in\mathcal S}},\mathcal N(\mathbf 0,I))&\rightarrow 0\label{eq:normality-index}  \text{ and}\\
	\hat{\tilde\Sigma}^* - \tilde\Sigma^* &\rightarrow_p 0 \label{eq:variance-estimator-consistency-index}
\end{align}
given $\{(X_i',D_i)\}_{i\in\mathcal S}=\{(x_i',d_i)\}_{i\in\mathcal S}$, using the following assumption:
\begin{assume}
	$\Gamma^*$, $\theta_{\eG^*,F}$, and $\Delta^*$ are bounded by a fixed constant, i.e., $O(1)$, given $\{(X_i',D_i)'\}_{i\in\mathcal S}=\{(x_i',d_i)'\}_{i\in\mathcal S}$. 	
\end{assume}

Denote by $se_{\tilde\beta_{\eG^*,F}}(\hat{\tilde\beta}^*)$ the $|\mathcal S^*|^{-1/2}$-scaled square root of the second diagonal element of $\hat{\tilde\Sigma}^*$. 

\begin{assume}\label{ass:local-asymptotics-index}
	For a positive function $v$ of $\{(X_i',D_i)'\}_{i\in\mathcal S}$, 
	\begin{align}
		m_{\eG_{\phi}^*,F_\phi^{}}c_{\eG_{\phi}^*}=\frac{v}{\sqrt{|\mathcal S^*|}}
	\end{align}
\end{assume}

Under Assumption \ref{ass:local-asymptotics-index}, whenever $m \geq m_{\eG_{\phi}^*,F_\phi^{}}$,
\begin{align}
	C_\alpha(m)\equiv \biggl\{\tau\in\mathbb R:\hat{\tilde\beta}^*-z_{1-\alpha/2}\tilde{se}^*-mc_{\eG_\phi^*}\leq \tau\leq  \hat{\tilde\beta}^*-z_{\alpha/2}\tilde{se}^* + m c_{\eG_\phi^*}\biggr\},	
\end{align}
where $\tilde{se}^*$ is shorthand for $se_{\tilde\beta_{\eG^*,F}}(\hat{\tilde\beta}^*)$, has asymptotic conditional coverage no smaller than $1-\alpha$.

\newpage 
\section{Omitted proofs}
This section includes omitted proofs for several equations, arguments, and Corollaries \ref{thm:bias-bound-TV}--\ref{thm:bias-bound-LP} and Proposition \ref{thm:representation-other-parameters} in the main text.

\subsection{Proof of equation \eqref{eq:beta}}
Let $(a,b')'\equiv \argmin_{(\tilde a,\tilde b')'\in\mathbb R\times\mathbb R^p}\Exp[(D-(\tilde a+s(X)'\tilde b))^2]$ be the linear projection coefficient of $(1,s(X)')$ obtained from the regression $D$ on $(1,s(X)')'$ in the population. Define $\dot D\equiv D-(a+s(X)'b)$. We reformulate the linear regression model as 
\begin{align*}
Y &= \alpha + \beta D + s(X)'\gamma + E\\
&= \alpha + \beta (a + s(X)'b + \dot D) + s(X)'\gamma + E\\
&= \beta \dot D + \underbrace{(\alpha + \beta a) + s(X)'(\beta b + \gamma) + E}_{\equiv \dot E}.
\end{align*}

Suppose $\Exp[\dot D^2]=0$. Then, $D=a+s(X)'b$ holds almost surely, and thus, so does
\begin{align*}
	Z=\begin{pmatrix}
		1\\ D\\ s(X)
	\end{pmatrix}
	= \underbrace{\begin{pmatrix}
		1 & 0\\
		a & b'\\
		\mathbf 0 & I
 	\end{pmatrix}}_{M}
	\begin{pmatrix}
		1 \\ s(X)
	\end{pmatrix}.
\end{align*}
The rank of $M$ is $\dim\theta -1$. Thus, the rank of $\Exp[ZZ']$ cannot be larger than $\dim\theta -1$, which violates Assumption \ref{ass:regularity}. From contraposition, $\Exp[\dot D^2]>0$ then follows.

Furthermore,
\begin{align*}
	\Exp[\dot D\dot E] = \underbrace{\Exp[\dot D(
		1\ s(X)')]}_{\mathclap{=0\text{ by the definition of  a linear projection}\phantom{\quad\quad\quad\quad\quad\quad\quad}}}\begin{pmatrix}
		\alpha + \beta a \\ \beta b + \gamma
	\end{pmatrix} + (-a,1,-b')\underbrace{\Exp[ZE]}_{\mathclap{=0 \text{ by Assumption 1}}}= 0.
\end{align*}

Combining results, $(\Exp[\dot D^2])^{-1}\Exp[\dot D Y] = \beta + (\Exp[\dot D^2])^{-1}\Exp[\dot D \dot E] = \beta$.

\subsection{Proof of equations \eqref{eq:bias-bound-MKW-m}--\eqref{eq:bias-bound-MKW-c}}
Fix $\pi\in\Pi(G^1,G^0)$. Using the Cauchy-Schwarz inequality, 
\begin{align*}
	&\int |f(x_1,0)-l(x_1,0)-(f(x_2,0)-l(x_2,0))|\pi(\mathrm dx_1\mathrm dx_2)\\
	&=\int \frac{|f(x_1,0)-l(x_1,0)-(f(x_2,0)-l(x_2,0))|}{\|x_1-x_2\|}\mathbf 1\{\|x_1-x_2\|>0\}\|x_1-x_2\|\pi(\mathrm dx_1\mathrm dx_2)\\
	&\leq \biggl(\int \|x_1-x_2\|^2\pi(\mathrm dx_1\mathrm dx_2)\biggr)^{\frac{1}{2}}\\
	&\hphantom{\leq \biggl(}\times \sup_{\pi\in\Pi(G^1,G^0)}\biggl(\int \biggl( \frac{|f(x_1,0)-l(x_1,0)-(f(x_2,0)-l(x_2,0))|}{\|x_1-x_2\|}\mathbf 1\{\|x_1-x_2\|>0\}\biggr)^2\pi(\mathrm dx_1\mathrm dx_2)\biggr)^{\frac{1}{2}},
\end{align*}
which bounds from above $|\beta-\tau|$ by Proposition \ref{thm:representation}. Taking $\inf_{\pi\in\Pi(G^1,G^0)}$ then yields the result.

\subsection{$|f(\cdot,0)-l(\cdot,0)|\leq (|\zeta_1^\star|\vee |\zeta_{-1}^\star|)'|r(\cdot)|$ under separabilities.}
If $r(\cdot)'\zeta_1^\star$ attains separation, i.e., $\xi_1^{d\star}(\cdot)=0$, by definition,
\begin{align*}
	f(x,0)-l(x,0) - r(x)'\zeta_1^\star &\leq 0\text{ for every }x\in\mathcal X^1\text{ and}\\
	f(x,0)-l(x,0) - r(x)'\zeta_1^\star &\geq 0\text{ for every }x\in\mathcal X^0.
\end{align*}
Similarly, if $r(\cdot)'\zeta_{-1}^\star$ attains separation, i.e., $\xi_{-1}^{d\star}(\cdot)=0$, again by definition,
\begin{align*}
	-(f(x,0)-l(x,0)) + r(x)'\zeta_{-1}^\star&\leq 0\text{ for every }x\in\mathcal X^1\text{ and}\\
	-(f(x,0)-l(x,0)) + r(x)'\zeta_{-1}^\star&\geq 0\text{ for every }x\in\mathcal X^0.
\end{align*}
Combining the inequalities,
\begin{align*}
	-|r(x)|'|\xi_{-1}^\star|\leq r(x)'\zeta_{-1}^\star\leq f(x,0)-l(x,0)&\leq r(x)'\zeta_1^\star\leq |r(x)|'|\zeta_1^\star|\text{ for every }x\in\mathcal X^1,\text{ and}\\
	-|r(x)|'|\zeta_1^\star|\leq r(x)'\zeta_1^\star\leq f(x,0)-l(x,0)&\leq r(x)'\zeta_{-1}^\star\leq |r(x)|'|\zeta_{-1}^\star|\text{ for every }x\in\mathcal X^0.
\end{align*}
Thus, for every $x\in\cup_{d\in\{0,1\}}\mathcal X^d$,
\begin{align*}
	|f(x,0)-l(x,0)|\leq \sum_{\sigma=\pm 1}|r(x)|'|\zeta_\sigma^\star|=|r(x)|'\sum_{\sigma=\pm 1}|\zeta_\sigma^\star|.
\end{align*} 

\subsection{Asymptotic conditional consistency of $C_\alpha(m)$}
The eligibility for inclusion in the set $C_\alpha(m)$, i.e., $\tau\in C_\alpha(m)$, can be stated as
\begin{align*}
	z_{\alpha/2}-\frac{mc_{\eG^*}}{se_{\beta_{\eG^*,F}}^*}< \frac{\hat\beta^*-\tau}{se_{\beta_{\eG^*,F}}^*}\leq z_{1-\alpha/2} + \frac{mc_{\eG^*}}{se_{\beta_{\eG^*,F}}^*}.
\end{align*}
Below, we show that, whenever equation \eqref{eq:robust} holds, $\tau_{\eG^*,F}$ satisfies this condition with conditional probability at least $1-\alpha$ in large samples. 

Fix $\{(X_i',D_i)\}_{i\in\mathcal S}=\{(x_i',d_i)'\}_{i\in\mathcal S}$. Since $mc_{\eG^*}\geq m_{\eG^*,F}c_{\eG^*}=v/\sqrt{|\mathcal S^*|}$, 
\begin{align*}
	&\Pr_F\biggl[z_{\alpha/2}\geq \frac{\hat\beta^*-\tau_{\eG^*,F}}{se_{\beta_{\eG^*,F}}^*} + \frac{v}{\sqrt{|\mathcal S^*|}se_{\beta_{\eG^*,F}}^*}\biggm|\{(X_i',D_i)'\}_{i\in\mathcal S}\biggr]\\
	&\geq \Pr_F\biggl[z_{\alpha/2}\geq \frac{\hat\beta^*-\tau_{\eG^*,F}}{se_{\beta_{\eG^*,F}}^*} + \frac{mc_{\eG^*}}{se_{\beta_{\eG^*,F}}^*}\biggm|\{(X_i',D_i)'\}_{i\in\mathcal S}\biggr]\\
	&= \Pr_F\biggl[z_{\alpha/2}-\frac{mc_{\eG^*}}{se_{\beta_{\eG^*,F}}^*} \geq  \frac{\hat\beta^*-\tau_{\eG^*,F}}{se_{\beta_{\eG^*,F}}^*}\biggm|\{(X_i',D_i)'\}_{i\in\mathcal S}\biggr],
\end{align*}
and by Corollary \ref{thm:size-distortion-subsample}, 
\begin{align*}
	\limsup_{|\mathcal S|\rightarrow\infty}\Pr_F\biggl[z_{\alpha/2}-\frac{mc_{\eG^*}}{se_{\beta_{\eG^*,F}}^*} \geq  \frac{\hat\beta^*-\tau_{\eG^*,F}}{se_{\beta_{\eG^*,F}}^*}\biggm|\{(X_i',D_i)'\}_{i\in\mathcal S}\biggr]\leq \Phi(z_{\alpha/2})=\alpha/2
\end{align*}
Similarly, since
\begin{align*}
	&1 - \Pr_F\biggl[\frac{\hat\beta^*-\tau_{\eG^*,F}}{se_{\beta_{\eG^*,F}}^*} - \frac{v}{\sqrt{|\mathcal S^*|}se_{\beta_{\eG^*,F}}^*}\leq z_{1-\alpha/2}\biggm|\{(X_i',D_i)'\}_{i\in\mathcal S}\biggr]\\
	&= \Pr_F\biggl[z_{1-\alpha/2} < \frac{\hat\beta^*-\tau_{\eG^*,F}}{se_{\beta_{\eG^*,F}}^*} - \frac{v}{\sqrt{|\mathcal S^*|}se_{\beta_{\eG^*,F}}^*}\biggm|\{(X_i',D_i)'\}_{i\in\mathcal S}\biggr]\\
	&\geq \Pr_F\biggl[z_{1-\alpha/2} < \frac{\hat\beta^*-\tau_{\eG^*,F}}{se_{\beta_{\eG^*,F}}^*} - \frac{mc_{\eG^*}}{se_{\beta_{\eG^*,F}}^*}\biggm|\{(X_i',D_i)'\}_{i\in\mathcal S}\biggr]\\
	&= \Pr_F\biggl[z_{1-\alpha/2} + \frac{mc_{\eG^*}}{se_{\beta_{\eG^*,F}}^*}< \frac{\hat\beta^*-\tau_{\eG^*,F}}{se_{\beta_{\eG^*,F}}^*}\biggm|\{(X_i',D_i)'\}_{i\in\mathcal S}\biggr],
\end{align*}
Corollary \ref{thm:size-distortion-subsample} implies
\begin{align*}
	&\limsup_{|\mathcal S|\rightarrow \infty}\Pr_F\biggl[z_{1-\alpha/2} + \frac{mc_{\eG^*}}{se_{\beta_{\eG^*,F}}^*}< \frac{\hat\beta^*-\tau_{\eG^*,F}}{se_{\beta_{\eG^*,F}}^*}\biggm|\{(X_i',D_i)'\}_{i\in\mathcal S}\biggr]\\
	&\leq 1 - \liminf_{|\mathcal S|\rightarrow\infty}\Pr_F\biggl[\frac{\hat\beta^*-\tau_{\eG^*,F}}{se_{\beta_{\eG^*,F}}^*} - \frac{v}{\sqrt{|\mathcal S^*|}se_{\beta_{\eG^*,F}}^*}\leq z_{1-\alpha/2}\biggm|\{(X_i',D_i)'\}_{i\in\mathcal S}\biggr]\\
	&\leq 1-\Phi(1-\alpha/2) = 1- (1-\alpha/2)  = \alpha/2.
\end{align*}

\subsection{Proof of equation \eqref{eq:V}}
Fix $\{W_i\}_{i\in\mathcal S}=\{(x_i',d_i)'\}_{i\in\mathcal S}$. Since $V_i(\theta) - V_i = Z_iZ_i'(\theta_{\eG^*,F}-\theta)$,
\begin{align*}
	&\frac{1}{|\mathcal S^*|}\sum_{i\in\mathcal S^*}(V_i(\theta) - V_{l_W^*(i)}(\theta))(V_i(\theta) - V_{l_W^*(i)}(\theta))' - \frac{1}{|\mathcal S^*|}\sum_{i\in\mathcal S^*}(V_i - V_{l_W^*(i)})(V_i - V_{l_W^*(i)})'\\
	&= \frac{1}{|\mathcal S^*|}\sum_{i\in\mathcal S^*}(Z_iZ_i'-Z_{l_W^*(i)}Z_{l_W^*(i)}')(\theta_{\eG^*,F}-\theta)(V_i - V_{l_W^*(i)})'\\
	&\hphantom{=} + \frac{1}{|\mathcal S^*|}\sum_{i\in\mathcal S^*}(V_i - V_{l_W^*(i)})(\theta_{\eG^*,F}-\theta)'(Z_iZ_i'-Z_{l_W^*(i)}Z_{l_W^*(i)}')'\\
	&\hphantom{=} + \frac{1}{|\mathcal S^*|}\sum_{i\in\mathcal S^*}(Z_iZ_i'-Z_{l_W^*(i)}Z_{l_W^*(i)}')(\theta_{\eG^*,F}-\theta)(\theta_{\eG^*,F}-\theta)'(Z_iZ_i'-Z_{l_W^*(i)}Z_{l_W^*(i)}')'.
\end{align*}
The Frobenius norms of the first and second terms are bounded by
\begin{align*}
	A(\theta)\equiv \|\theta_{\eG^*,F}-\theta\|\biggl(\frac{1}{|\mathcal S^*|}\sum_{i\in\mathcal S^*}\|Z_iZ_i'-Z_{l_W^*(i)}Z_{l_W^*(i)}'\|^2\biggr)^{\frac{1}{2}}\lambda_{\max}\biggl[\frac{1}{2|\mathcal S^*|}\sum_{i\in\mathcal S^*}(V_i - V_{l_W^*(i)})(V_i - V_{l_W^*(i)})'\biggr]^{\frac{1}{2}}
\end{align*}
up to a fixed constant, and the norm of the third term is bounded by the square of
\begin{align*}
	B(\theta)\equiv \|\theta_{\eG^*,F}-\theta\|\biggl(\frac{1}{|\mathcal S^*|}\sum_{i\in\mathcal S^*}\|Z_iZ_i'-Z_{l_W^*(i)}Z_{l_W^*(i)}'\|\biggr)^{\frac{1}{2}}.	
\end{align*}

First, we show that, as $|\mathcal S|$ tends to infinity, 
\begin{align*}
	\frac{1}{|\mathcal S^*|}\sum_{i\in\mathcal S^*}\|Z_iZ_i'-Z_{l_W^*(i)}Z_{l_W^*(i)}'\|^2 \rightarrow 0.
\end{align*}
By definition,
\begin{align*}
	\|Z_iZ_i'-Z_{l_W^*(i)}Z_{l_W^*(i)}'\|^2 = \sum_{j=1}^{\dim \theta}\sum_{k=1}^{\dim \theta}(Z_{j,i}Z_{k,i}-Z_{j,l_W^*(i)}Z_{k,l_W^*(i)})^2.
\end{align*}
As $\dim \theta$ is fixed, it is enough to show that
\begin{align*}
	\frac{1}{|\mathcal S^*|}\sum_{i\in\mathcal S^*}(Z_{j,i}Z_{k,i}-Z_{j,l_W^*(i)}Z_{k,l_W^*(i)})^2 \rightarrow 0. 
\end{align*}
Define $\tilde V_i\equiv Z_{j,i}Z_{k,i}$, and recall that $W_i= (X_i',D_i)'$. Note that Assumption \ref{ass:metric} implies $\Exp_{\mathcal L_{\tilde V|W}}[\tilde V_i^2|W_i=\cdot]$ is $\nu$-Lipschitz and bounded. Lemma \ref{thm:closest-variance} implies
\begin{align*}
	\frac{1}{2|\mathcal S^*|}\sum_{i\in\mathcal S^*}(Z_{j,i}Z_{k,i}-Z_{j,l_W^*(i)}Z_{k,l_W^*(i)})^2	 - \frac{1}{|\mathcal S^*|}\sum_{i\in\mathcal S^*}\Var_{\mathcal L_{\tilde V|W}}[\tilde V_i|W_i] \rightarrow_p 0,
\end{align*}
where 
\begin{align*}
	\Var_{\mathcal L_{\tilde V|W}}[\tilde V_i|W_i]=\Exp_{\mathcal L_{\tilde V|W}}[Z_{j,i}^2Z_{k,i}^2|W_i]-\Exp_{\mathcal L_{\tilde V|W}}[Z_{j,i}Z_{k,i}|W_i]^2 = Z_{j,i}^2Z_{k,i}^2 - (Z_{j,i}Z_{k,i})^2 = 0. 
\end{align*}
This yields the desired result.

By Proposition \ref{thm:consistency} and the preceding result, 
\begin{align*}
	B(\hat\theta^*) = \|\theta_{\eG^*,F}-\hat\theta^*\|\biggl(\frac{1}{|\mathcal S^*|}\sum_{i\in\mathcal S^*}\|Z_iZ_i'-Z_{l_W^*(i)}Z_{l_W^*(i)}'\|^2\biggr)^{\frac{1}{2}}\rightarrow_p 0.
\end{align*}
It therefore remains to show that $A(\hat\theta^*)\rightarrow_p 0$. Let
\begin{align*}
	A_1(\theta)&\equiv \|\theta_{\eG^*,F}-\theta\|\biggl(\frac{1}{|\mathcal S^*|}\sum_{i\in\mathcal S^*}\|Z_iZ_i'-Z_{l_W^*(i)}Z_{l_W^*(i)}'\|^2\biggr)^{\frac{1}{2}}\\
	&\hphantom{\equiv} \times \biggl(\lambda_{\max}\biggl[\frac{1}{2|\mathcal S^*|}\sum_{i\in\mathcal S^*}(V_i - V_{l_W^*(i)})(V_i - V_{l_W^*(i)})'\biggr]^{\frac{1}{2}} - \lambda_{\max}[\Delta^*]^{\frac{1}{2}}\biggr)
	\intertext{and}
	A_2(\theta)&\equiv \|\theta_{\eG^*,F}-\theta\|\biggl(\frac{1}{|\mathcal S^*|}\sum_{i\in\mathcal S^*}\|Z_iZ_i'-Z_{l_W^*(i)}Z_{l_W^*(i)}'\|^2\biggr)^{\frac{1}{2}}\lambda_{\max}[\Delta^*]^{\frac{1}{2}},
\end{align*}
so that $A(\theta)=A_1(\theta)+A_2(\theta)$. Note that $\lambda_{\max}$ is Lipschitz. Thus, for a fixed positive constant $\eta>0$,  
\begin{align*}
&\Pr_F\biggl[\biggl\|\frac{1}{2|\mathcal S^*|}\sum_{i\in\mathcal S^*}(V_i-V_{l_W^*(i)})(V_i-V_{l_W^*(i)})' - \Delta^*\biggl\|>\sqrt{\underline\lambda}\eta\biggm|\{W_i\}_{i\in\mathcal S}\biggr]\rightarrow 0
\end{align*}
---i.e., equation \eqref{eq:closest-variance}---implies
\begin{align*}
	\Pr_F\biggl[\biggl|\lambda_{\max}\biggl(\frac{1}{2|\mathcal S^*|}\sum_{i\in\mathcal S^*}(V_i-V_{l_W^*(i)})(V_i-V_{l_W^*(i)})'\biggr)^{\frac{1}{2}} - \lambda_{\max}(\Delta^*)^{\frac{1}{2}}\biggr|>\eta\biggm|\{W_i\}_{i\in\mathcal S}\biggr]\rightarrow 0.
\end{align*}
From $B(\hat\theta^*)\rightarrow_p 0$, it follows that $A_1(\hat\theta^*)\rightarrow_p 0$. By Assumptions \ref{ass:bounded} and \ref{ass:metric},
\begin{align*}
	\biggl(\frac{1}{|\mathcal S^*|}\sum_{i\in\mathcal S^*}\|Z_iZ_i'-Z_{l_W^*(i)}Z_{l_W^*(i)}'\|^2\biggr)^{\frac{1}{2}}\lambda_{\max}(\Delta^*)^{\frac{1}{2}}=o(1)O(1) = o(1).
\end{align*}
Combined with Proposition \ref{thm:consistency}, this yields $A_2(\hat\theta^*)\rightarrow_p 0$.

\subsection{Proof of Corollary \ref{thm:bias-bound-TV}}
Using the Jensen's inequality, 
\begin{align*}
&\biggl|\int (f(x,0)-l(x,0))(g^1(x)-g^0(x))\mu(\mathrm dx)\biggr|\\
&\leq \int |f(x,0)-l(x,0)||g^1(x)-g^0(x)|\mu(\mathrm dx)\\
&\leq \|f(\cdot,0)-l(\cdot,0)\|_{L^\infty(\mu)}\int |g^1(x)-g^0(x)|\mu(\mathrm dx).
\end{align*}

\subsection{If $G^0\gg G^1$, $c^{\DR}=0$ if and only if $X$ and $D$ are independent.}
$c^{\DR}=0$ if and only if $(\mathrm dG^1/\mathrm dG^0)(\cdot)=1$, $G^0$-almost surely. Thus,
\begin{align*}
	(\mathrm dG^1/\mathrm dG^0)\cdot G^0 = 1\cdot G^0,
\end{align*}
where the domination implies that the left-hand side is $G^1$.

\subsection{$c^{\DR}$ is finite if $\Pr[D=1|X=\cdot]$ is bounded from one.}
By the Lebesgue decomposition theorem, it holds $G^0$-almost surely that 
\begin{align*}
	\frac{\mathrm dG^1}{\mathrm dG^0}= \frac{g^1}{g^0},
\end{align*}
where we recall that $G^d=g^d\cdot \mu$.

Now, assuming that $\Pr[D=1]\in (0,1)$, choose the marginal distribution $\mathcal L_X$ of $X$ as the dominating measure $\mu$. Then, by the Bayes' theorem,
\begin{align*}
	\Pr[D=d]g^d(\cdot)=\Pr[D=d|X=\cdot]
\end{align*}
holds $\mathcal L_X$-almost surely.

As a result, if the propensity score function $\Pr[D=1|X=\cdot]$ is bounded from one, it holds $\mathcal L_X$-almost surely that
\begin{align*}
	\frac{g^1(\cdot)}{g^0(\cdot)}=\frac{\frac{1}{\Pr[D=1]}\Pr[D=1|X=\cdot]}{\frac{1}{\Pr[D=0]}\Pr[D=0|X=\cdot]}&\leq \frac{\Pr[D=0]}{\Pr[D=1]\Pr[D=0|X=\cdot]}\\
	&\leq \frac{\Pr[D=0]}{\Pr[D=1](1-\|\Pr[D=1|X=\cdot]\|_{L_{\infty}(\mathcal L_X)})}<\infty;
\end{align*}
in particular, the density ratio is bounded $G^0$-almost surely.

\subsection{Proof of Corollary \ref{thm:bias-bound-DR}}
In cases where either $m^{\DR}$ or $c^{\DR}$ is infinity, the inequality trivially holds; we assume that both are finite. By the Cauchy-Schwarz inequality and the Jensen's inequality, 
\begin{align*}
	m^{\DR}c^{\DR}&\geq \int \biggl|(f(x,0)-l(x,0))\biggl(\frac{\mathrm dG^1}{\mathrm dG^0}(x)-1\biggr)\biggr|G^0(\mathrm dx)\\
	&\geq \biggl|\int (f(x,0)-l(x,0))\biggl(\frac{\mathrm dG^1}{\mathrm dG^0}(x)-1\biggr) G^0(\mathrm dx)\biggr|\\
	&=\biggl|\int (f(x,0)-l(x,0))((G^1)^a-G^0)(\mathrm dx)\biggr|,
\end{align*}
where $(G^1)^a$ denotes the absolutely continuous part of $G^1$ with respect to $G^0$. Add
\begin{align*}
	\biggl|\int (f(x,0)-l(x,0))(G^1)^{\perp}(\mathrm dx)\biggr|	
\end{align*}
to the both sides of the preceding inequality. Proposition \ref{thm:representation} then yields the result.

\subsection{Proof of Corollary \ref{thm:bias-bound-LP}}
When $m^{\LP}$ is infinity, the inequality trivially holds. When it is zero, then, by definition, $\|f(\cdot,0)-l(\cdot,0)\|_{L^\infty(\mu)}=0$, and Corollary \ref{thm:bias-bound-TV} implies $|\beta-\tau|=0$. That is, equation \eqref{eq:bias-bound-LP} holds in the form of $0=0$. Thus, we assume $m^{\LP}\in(0,\infty)$. Then, by Corollary 2 of \citet{Dudley1968},
\begin{align*}
	\bigg|\int \frac{f(x,0)-l(x,0)}{\|f(\cdot,0)-l(\cdot,0)\|_{\Lip}+\|f(\cdot,0)-l(\cdot,0)\|_\infty}(G^1-G^0)(\mathrm dx)\biggr|
	\leq 2\rho(G^1,G^0).
\end{align*}

\subsection{Proof of equation \eqref{eq:bias-bound-LP-m}}
By Corollary 1 of \citet{Dudley1968}, for any $\varepsilon>0$, there exists a probability measure $\pi$ on $\mathbb R^p\times\mathbb R^p$ such that its push-forward onto the first coordinate is $G^1$, onto the second coordinate is $G^0$, and  
\begin{align*}
	\pi[\{(x_1',x_2')'\in\mathbb R^p\times \mathbb R^p:\|x_1-x_2\|>\rho(G^1,G^0)+\varepsilon\}]\leq \rho(G^1,G^0) + \varepsilon.
\end{align*}

The first two properties imply
\begin{align*}
	&\int (f(\cdot,0)-l(\cdot,0))(G^1-G^0)(\mathrm dx)\\
	&= \int (f(x_1,0)-l(x_1,0)-(f(x_2,0)-l(x_2,0)))\pi(\mathrm dx_1\mathrm dx_2)\\
	&= \int (f(x_1,0)-l(x_1,0)-(f(x_2,0)-l(x_2,0)))\mathbf 1\{\|x_1-x_2\|>\rho(G^1,G^0)+\varepsilon\}\pi(\mathrm dx_1\mathrm dx_2)\\
	&\hphantom{= \int} + \int (f(x_1,0)-l(x_1,0)-(f(x_2,0)-l(x_2,0)))\mathbf 1\{\|x_1-x_2\|\leq \rho(G^1,G^0) +\varepsilon\}\pi(\mathrm dx_1\mathrm dx_2). 
\end{align*}
Combined with the third property, the absolute value of the first term in the preceding equation is bounded by
\begin{align*}
	2\sup_{x\in\cup_{d\in\{0,1\}}\mathcal X^d}|f(x,0)-l(x,0)|(\rho(G^1,G^0)+\varepsilon),
\end{align*}
and that of the second term by
\begin{align*}
	&\int \frac{|f(x_1,0)-l(x_1,0)-(f(x_2,0)-l(x_2,0))|}{\|x_1-x_2\|}\mathbf 1\{0<\|x_1-x_2\|\leq \rho(G^1,G^0)+\varepsilon\}\pi(\mathrm dx_1\mathrm dx_2)\\
	&\phantom{\int \frac{||}{||}}\times (\rho(G^1,G^0)+\varepsilon).
\end{align*}
Now, taking $\varepsilon\downarrow 0$ and then $\sup_{x_1\in\mathcal X^1}$ inside the integral yields the desired result.

\subsection{Parameter is preserved if $s(X)$ is a balancing score.} 
Suppose that $s(X)$ is a balancing score, i.e., $X$ and $D$ are independent conditional on $s(X)$. For each $d\in\{0,1\}$,
\begin{align*}
	\Exp_G[\Exp_F[Y|X,D=d]|s(X)=\cdot,D=1]
	=\Exp_G[\Exp_F[Y|X,D=d]|s(X)=\cdot,D=0].
\end{align*}
Thus, for each $(d,d')'\in\{0,1\}^2$,
\begin{align*}
	\Exp_{F_s}[Y|s(X)=\cdot,D=d] = \Exp_G[\Exp_F[Y|X,D=d]|s(X)=\cdot,D=d'] .
\end{align*}
Using this,
\begin{align*}
	&\Exp_G[\Exp_{F_s}[Y|s(X),D=1]-\Exp_{F_s}[Y|s(X),D=0]|D=1]\\
 	&=\Exp_G[\Exp_G[\Exp_F[Y|X,D=1]|s(X),D=1]|D=1]\\
 	&\hphantom{=\Exp_G} - \Exp_G[\Exp_G[\Exp_F[Y|X,D=0]|s(X),D=1]|D=1]\\ 
	&=\Exp_G[\Exp_G[\Exp_F[Y|X,D=1]-\Exp_F[Y|X,D=0]|s(X),D=1]|D=1]\\
	&=\Exp_G[\Exp_F[Y|X,D=1]-\Exp_F[Y|X,D=0]|D=1]=\tau_{G,F}.
\end{align*}

\subsection{Proof of Proposition \ref{thm:representation-other-parameters}}
We start with $\tau^0$.
\begin{align*}
	\tau^0 &= \Exp[\Exp[Y|X,D=1]|D=0]- \Exp[Y|D=0]\\
	&= \biggl(\int \Exp[Y|X=x,D=1](G^0-G^1)(\mathrm dx) + \int \Exp[Y|X=x,D=1]G^1(\mathrm dx)\biggr) - \Exp[Y|D=0]\\
	&= \int \Exp[Y|X=x,D=1](G^0-G^1)(\mathrm dx) + \beta + \Exp[s(X)'\gamma|D=1] - \Exp[s(X)'\gamma|D=0]\\
	&= \beta - \int (\Exp[Y|X=x,D=1] - s(x)'\gamma)(G^1-G^0)(\mathrm dx)\\
	&= \beta - \int (\Exp[Y|X=x,D=1] - (\alpha + \beta + s(x)'\gamma))(G^1-G^0)(\mathrm dx).
\end{align*}
Combining this with Proposition \ref{thm:representation},
\begin{align*}
	\tau^{10} &= \Pr[D=1]\tau + \Pr[D=0]\tau^0\\
	&= \Pr[D=1]\biggl(\beta - \int (f(x,0)-l(x,0))(G^1-G^0)(\mathrm dx)\biggr)\\
	&\hphantom{= \Pr} + \Pr[D=0]\biggl(\beta - \int(f(x,1)-l(x,1))(G^1-G^0)(\mathrm dx)\biggr).
\end{align*}

\subsection{$\hat{\tilde\theta}^*=(\hat\alpha,\hat\beta,1)$ if $\mathcal S^*=\mathcal S$.}
Let $s$ be the dimension of $s(\cdot)$. By definition, for every $(\alpha,\beta,\gamma')'\in\mathbb R^{1+1+s}$, 
\begin{align*}
	\frac{1}{|\mathcal S|}\sum_{i\in\mathcal S}(Y_i-(\hat\alpha+\hat\beta D_i+s(X_i)'\hat\gamma))^2\leq \frac{1}{|\mathcal S|}\sum_{i\in\mathcal S}(Y_i-(\alpha+\beta D_i+s(X_i)'\gamma))^2.
\end{align*}
Take $\gamma=\hat\gamma q$, where $q\in\mathbb R$. Then,
\begin{align*}
	\frac{1}{|\mathcal S|}\sum_{i\in\mathcal S}(Y_i-(\hat\alpha+\hat\beta D_i+(s(X_i)'\hat\gamma)\cdot 1))^2\leq 
	\frac{1}{|\mathcal S|}\sum_{i\in\mathcal S}(Y_i-(\alpha+\beta D_i+(s(X_i)'\hat\gamma)\cdot q))^2.
\end{align*}

Now, note that
\begin{align*}
\tilde Z_i'\equiv \begin{pmatrix}
	1 & D_i & s(X_i)'\hat\gamma	
\end{pmatrix} = 
Z_i'
\underbrace{\begin{pmatrix}
	1 & 0 & 0\\
	0 & 1 & 0\\
	\mathbf 0 & \mathbf 0 & \hat\gamma 
\end{pmatrix}}_{\equiv B},
\end{align*}
where $B$ is of full column rank if $\hat \gamma \neq \mathbf 0$. Thus,
\begin{align*}
	\frac{1}{|\mathcal S|}\sum_{i\in\mathcal S}\tilde Z_i\tilde Z_i'= 
	B'\frac{1}{|\mathcal S|}\sum_{i\in\mathcal S}Z_iZ_i'B
\end{align*}
is of full rank if $(1/|\mathcal S|)\sum_{i\in\mathcal S}Z_iZ_i'$ is of full rank and $\hat\gamma\neq \mathbf 0$.

\subsection{More on equation \eqref{eq:theta-index}}
We show that 
\begin{align*}
	\frac{1}{|\mathcal S^*|}\sum_{i\in\mathcal S^*}\tilde Z_i\Exp_{F_{\hat\gamma}}[Y_i|s(X_i)'\hat\gamma,D_i] = \frac{1}{|\mathcal S^*|}\sum_{i\in\mathcal S^*}\tilde Z_i\Exp_F[Y_i|X_i,D_i],
\end{align*}
which implies that $\tilde\theta_{\eG^*,F}=\theta_{\eG_{\hat\gamma}^*,F_{\hat\gamma}^{}}$.

Let $\mathcal S^*(d,t)\equiv \{i\in\mathcal S^*: D_i=d,s(X_i)'\hat\gamma =t\}$.
Note that $\sum_{i\in\mathcal S^*}=\sum_{(d,t)}\sum_{i\in\mathcal S^*(d,t)}$.
\begin{align*}
	&\frac{1}{|\mathcal S^*|}\sum_{i\in\mathcal S^*}\tilde Z_i\Exp_{F_{\hat\gamma}}[Y_i|s(X_i)'\hat\gamma,D_i]\\
	&= \frac{1}{|\mathcal S^*|}\sum_{i\in\mathcal S^*}\begin{pmatrix}
		1 \\ D_i \\ s(X_i)'\hat\gamma	
	\end{pmatrix}
	\frac{1}{|\mathcal S^*(D_i,s(X_i)'\hat\gamma)|}\sum_{j\in \mathcal S^*(D_i,s(X_i)'\hat\gamma)}\Exp_F[Y|X=X_j,D=D_j]\\
	&=\frac{1}{|\mathcal S^*|}\sum_{(d,t)}\underbrace{\sum_{i\in\mathcal S^*(d,t)}\frac{1}{|\mathcal S^*(d,t)|}}_{=1}\begin{pmatrix}
		1 \\ d \\ t
	\end{pmatrix}\sum_{j\in \mathcal S^*(d,t)}\Exp_F[Y|X=X_j,D=D_j]\\
	&=\frac{1}{|\mathcal S^*|}\sum_{(d,t)}\sum_{j\in \mathcal S^*(d,t)}\begin{pmatrix}
		1 \\ D_j \\ s(X_j)'\hat\gamma
	\end{pmatrix}\Exp_F[Y|X=X_j,D=D_j] = \frac{1}{|\mathcal S^*|}\sum_{j\in\mathcal S^*}\tilde Z_j\Exp_F[Y|X=X_j,D=D_j].
\end{align*}

\subsection{Proofs of Propositions \ref{thm:normality-index}--\ref{thm:variance-estimator-consistency-index}}

Note that $\gamma_{\eG,F}$ is a function of $\{(X_i',D_i)\}_{i\in\mathcal S}$. This reduces Propositions \ref{thm:normality-index}--\ref{thm:variance-estimator-consistency-index} to special cases of Propositions \ref{thm:normality}--\ref{thm:variance-estimator-consistency}. 

\subsection{Proofs of equations \eqref{eq:normality-index}--\eqref{eq:variance-estimator-consistency-index}}
Fix $\{(X_i',D_i)'\}_{i\in\mathcal S}=\{(x_i',d_i)'\}_{i\in\mathcal S}$. It is enough to verify that
\begin{align*}
	\frac{1}{\sqrt{|\mathcal S^*|}}\sum_{i\in\mathcal S^*}\tilde Z_i(\gamma_{\eG,F})U_i  - 	\frac{1}{\sqrt{|\mathcal S^*|}}\sum_{i\in\mathcal S^*}\tilde Z_i(\hat\gamma)U_i&\rightarrow_p 0, \\
	\tilde\Gamma^*(\gamma_{\eG,F}) - \tilde\Gamma^*(\hat\gamma) &\rightarrow_p 0,\\
	\tilde\Delta^*(\gamma_{\eG,F}) - \tilde\Delta^*(\hat\gamma) &\rightarrow_p 0,\text{ and}\\
	\hat{\tilde\Delta}^*(\gamma_{\eG,F}) - \hat {\tilde\Delta}^*(\hat\gamma)&\rightarrow_p 0. 
\end{align*}

From
\begin{align*}
	\tilde Z_i(\gamma_{\eG,F})-\tilde Z_i(\hat\gamma) =
	\begin{pmatrix}
		1 & 0 & \mathbf 0'\\
		0 & 1 & 	\mathbf 0'\\
		0 & 0 & \gamma_{\eG,F}' 
	\end{pmatrix}Z_i
	- \begin{pmatrix}
		1 & 0 & \mathbf 0'\\
		0 & 1 & 	\mathbf 0'\\
		0 & 0 & \hat\gamma' 
	\end{pmatrix}Z_i
	= 	\begin{pmatrix}
		0 & 0 & \mathbf 0'\\
		0 & 0 & 	\mathbf 0'\\
		0 & 0 & \gamma_{\eG,F}'-\hat\gamma' 
	\end{pmatrix}Z_i, 
\end{align*}
it follows that
\begin{align*}
	&\|\tilde\Gamma^*(\gamma_{\eG,F}) - \tilde\Gamma^*(\hat\gamma)\|\\
	&\leq \biggl\|\frac{1}{|\mathcal S^*|}\sum_{i\in\mathcal S^*}(\tilde Z_i(\gamma_{\eG,F})-\tilde Z_i(\hat\gamma))\tilde Z_i(\gamma_{\eG,F})'\biggr\|
	+ \biggl\|\frac{1}{|\mathcal S^*|}\sum_{i\in\mathcal S^*}\tilde Z_i(\gamma_{\eG,F})(\tilde Z_i(\gamma_{\eG,F})-\tilde Z_i(\hat\gamma))'\biggr\|\\
	&\hphantom{=\biggl\|} +  \biggl\|\frac{1}{|\mathcal S^*|}\sum_{i\in\mathcal S^*}(\tilde Z_i(\gamma_{\eG,F})-\tilde Z_i(\hat\gamma))(\tilde Z_i(\gamma_{\eG,F})-\tilde Z_i(\hat\gamma))'\biggr\|\\
	&\leq 2\|\gamma_{\eG,F}-\hat\gamma\|\biggl\|\frac{1}{|\mathcal S^*|}\sum_{i\in\mathcal S^*}Z_iZ_i'\biggr\|\sqrt{\|\gamma_{\eG,F}\|^2+2} + \|\gamma_{\eG,F}-\hat\gamma\|^2\biggl\|\frac{1}{|\mathcal S^*|}\sum_{i\in\mathcal S^*}Z_iZ_i'\biggr\| = o_p(1),
\end{align*}
that
\begin{align*}
	&\|\tilde\Delta^*(\gamma_{\eG,F})-\tilde\Delta^*(\hat\gamma)\|\\
	&\leq \|\gamma-\hat\gamma\|^2\biggl\|\frac{1}{|\mathcal S^*|}\sum_{i\in\mathcal S^*}Z_iZ_i'\Exp_F[U_i^2|X_i,D_i]\biggr\|\lesssim \|\gamma-\hat\gamma\|^2\biggl\|\frac{1}{|\mathcal S^*|}\sum_{i\in\mathcal S^*}Z_iZ_i'\biggr\| = o_p(1), 
\end{align*}
and that
\begin{align*}
	&\biggl\|\frac{1}{\sqrt{|\mathcal S^*|}}\sum_{i\in\mathcal S^*}\tilde Z_i(\gamma_{\eG,F})U_i - \frac{1}{\sqrt{|\mathcal S^*|}}\sum_{i\in\mathcal S^*}\tilde Z_i(\hat\gamma)U_i\biggr\|\\
	&\leq \|\gamma_{\eG,F}-\hat\gamma\|\biggl\|\frac{1}{\sqrt{|\mathcal S^*|}}\sum_{i\in\mathcal S^*}Z_iU_i\biggr\|= o_p(1). 
\end{align*}

We proceed as follows to establish the final convergence.
\begin{itemize}[listparindent=\parindent]
	\item Note that
	\begin{align*}
		&\hat{\tilde\theta}(\hat\gamma)-\hat{\tilde\theta}(\gamma_{\eG,F})\\
		&=(\tilde\Gamma^*(\hat\gamma))^{-1}\frac{1}{|\mathcal S^*|}\sum_{i\in\mathcal S^*}\tilde Z_i(\hat\gamma)Y_i 
		- (\tilde\Gamma^*(\gamma_{\eG,F}))^{-1}\frac{1}{|\mathcal S^*|}\sum_{i\in\mathcal S^*}\tilde Z_i(\gamma_{\eG,F})Y_i\\
		&=((\tilde\Gamma^*(\hat\gamma))^{-1}-(\tilde\Gamma^*(\gamma_{\eG,F}))^{-1})\underbrace{\frac{1}{|\mathcal S^*|}\sum_{i\in\mathcal S^*}\tilde Z_i(\hat\gamma)Y_i}_{\equiv Q}\\ 
		&\hphantom{=(} +(\tilde\Gamma^*(\gamma_{\eG,F}))^{-1}\biggl(\underbrace{\frac{1}{|\mathcal S^*|}\sum_{i\in\mathcal S^*}\tilde Z_i(\hat\gamma)Y_i-\frac{1}{|\mathcal S^*|}\sum_{i\in\mathcal S^*}\tilde Z_i(\gamma_{\eG,F})Y_i}_{\equiv R}\biggr).
	\end{align*}
	Since
	\begin{align*}
		R&=\frac{1}{|\mathcal S^*|}\sum_{i\in\mathcal S^*}\begin{pmatrix}
			0 & 0 & \mathbf 0'\\
			0 & 0 & 	\mathbf 0'\\
			0 & 0 & \hat\gamma'-\gamma_{\eG,F}'
		\end{pmatrix}Z_iY_i\\
		&= \begin{pmatrix}
			0 & 0 & \mathbf 0'\\
			0 & 0 & 	\mathbf 0'\\
			0 & 0 & \hat\gamma'-\gamma_{\eG,F}'
		\end{pmatrix}\biggl(\underbrace{\frac{1}{|\mathcal S^*|}\sum_{i\in\mathcal S^*}Z_i\Exp_F[Y_i|X_i,D_i]}_{\displaystyle=\biggl(\frac{1}{|\mathcal S^*|}\sum_{i\in\mathcal S^*}Z_iZ_i\biggr)\theta_{\eG^*,F}} + \frac{1}{|\mathcal S^*|}\sum_{i\in\mathcal S^*}Z_iU_i\biggr)=o_p(1)
	\end{align*}
	and
	\begin{align*}
		Q&=R + \begin{pmatrix}
			1 & 0 & \mathbf 0'\\
			0 & 1 & 	\mathbf 0'\\
			0 & 0 & \gamma_{\eG,F}' 
		\end{pmatrix}\frac{1}{|\mathcal S^*|}\sum_{i\in\mathcal S^*}Z_iY_i\\
		&=R + \begin{pmatrix}
			1 & 0 & \mathbf 0'\\
			0 & 1 & 	\mathbf 0'\\
			0 & 0 & \gamma_{\eG,F}'
		\end{pmatrix}\biggl(\biggl(\frac{1}{|\mathcal S^*|}\sum_{i\in\mathcal S^*}Z_iZ_i\biggr)\theta_{\eG^*,F} + \frac{1}{|\mathcal S^*|}\sum_{i\in\mathcal S^*}Z_iU_i\biggr)=O_p(1),
	\end{align*}
	$\hat{\tilde\theta}(\hat\gamma)-\hat{\tilde\theta}(\gamma_{\eG,F})=o_p(1)$.
	\item Let
	\begin{align*}
		A_i&\equiv \tilde Z_i(\gamma_{\eG,F})\hat{\tilde E}_i^*(\gamma_{\eG,F})-\tilde Z_i(\hat\gamma)\hat{\tilde E}_i^*(\hat\gamma)\text{ and}\\
		B_{l_{(X',D)'}^*(i)}&\equiv \tilde Z_{l_{(X',D)'}^*(i)}(\gamma_{\eG,F})\hat{\tilde E}_{l_{(X',D)'}^*(i)}^*(\gamma_{\eG,F})-\tilde Z_{l_{(X',D)'}^*(i)}(\hat\gamma)\hat{\tilde E}_{l_{(X',D)'}^*(i)}^*(\hat\gamma),	
	\end{align*}
	so that
	\begin{align*}
		&(\tilde Z_i(\gamma_{\eG,F})\hat{\tilde E}_i^*(\gamma_{\eG,F})-\tilde Z_{l_{(X',D)'}^*(i)}(\gamma_{\eG,F})\hat{\tilde E}_{l_{(X',D)'}^*(i)}^*(\gamma_{\eG,F}))\\
		&- (\tilde Z_i(\hat\gamma)\hat{\tilde E}_i^*(\hat\gamma)-\tilde Z_{l_{(X',D)'}^*(i)}(\hat\gamma)\hat{\tilde E}_{l_{(X',D)'}^*(i)}^*(\hat\gamma))\\
		&= A_i - B_{l_{(X',D)'}^*(i)}.
	\end{align*}
	Note that the Frobenius norm of 
	\begin{align*}
		\frac{1}{|\mathcal S^*|}\sum_{i\in\mathcal S^*}(A_i-B_{l_{(X',D)'}^*(i)})(A_i-B_{l_{(X',D)'}^*(i)})'
	\end{align*}
	is bounded by (the products of) the norms of 
	\begin{align*}
	\frac{1}{|\mathcal S^*|}\sum_{i\in\mathcal S^*}A_iA_i'\text{ and }\frac{1}{|\mathcal S^*|}\sum_{i\in\mathcal S^*}B_{l_{(X',D)'}^*(i)}B_{l_{(X',D)'}^*(i)}'.
	\end{align*}
	
	Since
	\begin{align*}
		&\hat{\tilde E}_i^*(\gamma_{\eG,F})-\hat{\tilde E}_i^*(\hat\gamma)\\
		&=(Y_i-\tilde Z_i(\gamma_{\eG,F})'\hat{\tilde\theta}(\gamma_{\eG,F}))-(Y_i-\tilde Z_i(\hat\gamma)'\hat{\tilde\theta}(\hat\gamma))\\
		&=-\tilde Z_i(\gamma_{\eG,F})'\hat{\tilde\theta}(\gamma_{\eG,F})+\tilde Z_i(\hat\gamma)'\hat{\tilde\theta}(\hat\gamma)\\
		&=(\tilde Z_i(\hat\gamma)-\tilde Z_i(\gamma_{\eG,F}))'\hat{\tilde\theta}(\gamma_{\eG,F}) + \tilde Z_i(\gamma_{\eG,F})(\hat{\tilde\theta}(\hat\gamma)-\hat{\tilde\theta}(\gamma_{\eG,F}))\\
		&\hphantom{=(} + (\tilde Z_i(\hat\gamma)-\tilde Z_i(\gamma_{\eG,F}))'(\hat{\tilde\theta}(\hat\gamma)-\hat{\tilde\theta}(\gamma_{\eG,F})),
	\end{align*}
	when combined with the preceding result, 
	\begin{align*}
		\frac{1}{|\mathcal S^*|}\sum_{i\in\mathcal S^*}(\hat{\tilde E}_i^*(\gamma_{\eG,F})-\hat{\tilde E}_i^*(\hat\gamma))(\hat{\tilde E}_i^*(\gamma_{\eG,F})-\hat{\tilde E}_i^*(\hat\gamma))' = o_p(1).
	\end{align*}
	This yields $|\mathcal S^*|^{-1}\sum_{i\in\mathcal S^*}A_iA_i'=o_p(1)$

	Similarly, we can show that $|\mathcal S^*|^{-1}\sum_{i\in\mathcal S^*}B_{l_{(X',D)'}^*(i)}B_{l_{(X',D)'}^*(i)}'=o_p(1)$, noting that
	\begin{align*}
		&\frac{1}{|\mathcal S^*|}\sum_{i\in\mathcal S^*}Z_{l_{(X',D)'}^*(i)}Z_{l_{(X',D)'}^*(i)}'\\
		&= \biggl(\frac{1}{|\mathcal S^*|}\sum_{i\in\mathcal S^*}Z_{l_{(X',D)'}^*(i)}Z_{l_{(X',D)'}^*(i)}'-\frac{1}{|\mathcal S^*|}\sum_{i\in\mathcal S^*}Z_iZ_i\biggr) + \frac{1}{|\mathcal S^*|}\sum_{i\in\mathcal S^*}Z_iZ_i = O(1).
	\end{align*}
	\item Since the Frobenius norm of $\hat{\tilde\Delta}^*(\gamma_{\eG,F}) - \hat {\tilde\Delta}^*(\hat\gamma)$ is bounded by the product of the norms of $|\mathcal S^*|^{-1}\sum_{i\in\mathcal S^*}(A_i-B_{l_{(X',D)'}^*(i)})(A_i-B_{l_{(X',D)'}^*(i)})'$ and $\hat{\tilde\Delta}^*(\gamma_{\eG,F})$, where 
	\begin{align*}
		\hat{\tilde\Delta}^*(\gamma_{\eG,F})=(\hat{\tilde\Delta}^*(\gamma_{\eG,F})-\tilde\Delta^*(\gamma_{\eG,F}))+\tilde\Delta^*(\gamma_{\eG,F})=O(1),
	\end{align*}
	we obtain the desired result.
\end{itemize}


\subsection{Short versus long regressions}
Let $s_A$ and $s_B$ denote two different covariate functions. For each $j\in \{A,B\}$, consider a linear regression model 
\begin{align*}
	Y &= \alpha_j + \beta_j D + s_j(X)'\gamma_j + E_j,
\end{align*}
where $\Exp[E_j]=\Exp[DE_j]=0$ and $\Exp[s_j(X)E_j]=\mathbf 0$; let
\begin{align*}
	l^j(x,d) &\equiv \alpha_j + \beta_j d + s_j(x)'\gamma_j
\end{align*}
be the corresponding regression function. Then, by Proposition \ref{thm:representation},
\begin{align}
	\beta_A - \beta_B = \int (l^A(x,0) - l^B(x,0))(G^1 - G^0)(\mathrm dx).
\end{align}

\newpage
\section{Omitted simulation results}
We provide additional simulation results for cases with $(\mathsf N_1,\mathsf N_0)\in\{(50,75),(50,125)\}$. We observe similar patterns to the results in Section \ref{sec:simulation}, even when $\mathsf N_0$ is small.

Figures \ref{fig:estimand-75}--\ref{fig:mean-difference-75} display the results for $(\mathsf N_1,\mathsf N_0)=(50,75)$. As in the preceding case, the design phase identifies the estimand with less bias (Figures \ref{fig:estimand-75}--\ref{fig:bias-75}), and the source of this reduction is the robustification against misspecification (Figures \ref{fig:total-variation-75}--\ref{fig:mean-difference-75}).

The results for $(\mathsf N_1,\mathsf N_0)=(50,125)$ are shown in Figures \ref{fig:estimand-125}--\ref{fig:mean-difference-125}. The design phase has become more powerful with the aid of a larger ``donor pool.'' In Figure \ref{fig:bias-125}, relative to Figure \ref{fig:bias-75}, more points now lie farther below the identity line. In Figures \ref{fig:total-variation-125}--\ref{fig:mean-difference-125}, the blue and red points are more separated than Figures \ref{fig:total-variation-75}--\ref{fig:mean-difference-75}. Specifically, the blue points have shifted further towards the left, thus benefiting from the extended degree of model misspecification that can be allowed.

\begin{figure}[ht]
\centering
\subfloat[Specification A]{\includegraphics[width=0.33\textwidth]{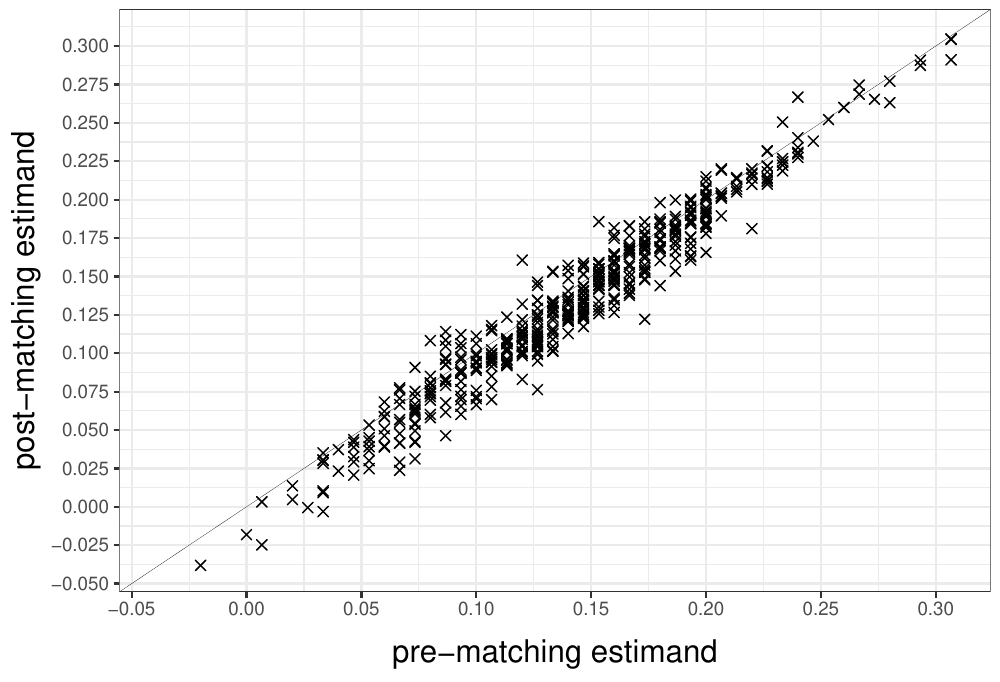}}\hfil
\subfloat[Specification B]{\includegraphics[width=0.33\textwidth]{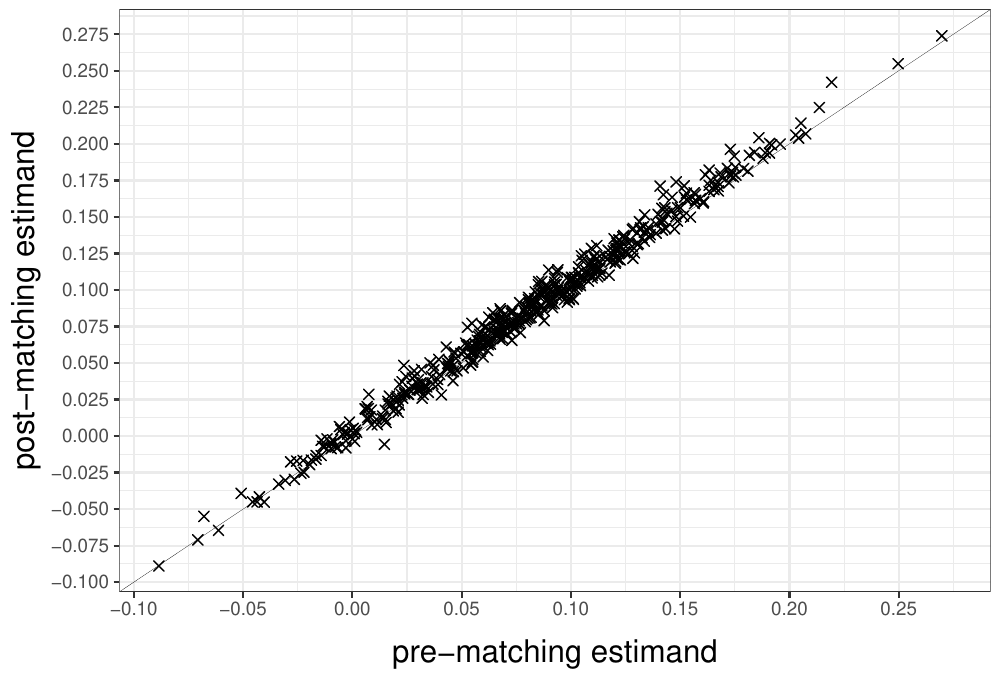}}\hfil
\subfloat[Specification C]{\includegraphics[width=0.33\textwidth]{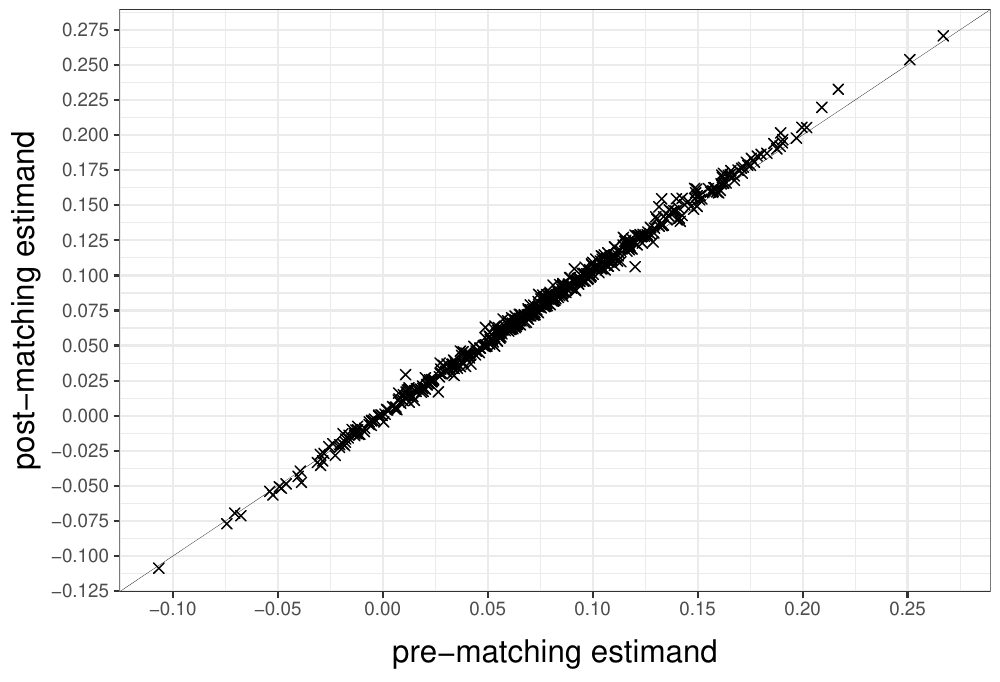}}
\caption{Pre- and Post-estimands ($\mathsf N_1=50$, $\mathsf N_0=75$)}
\label{fig:estimand-75}\bigskip
\subfloat[Specification A]{\includegraphics[width=0.33\textwidth]{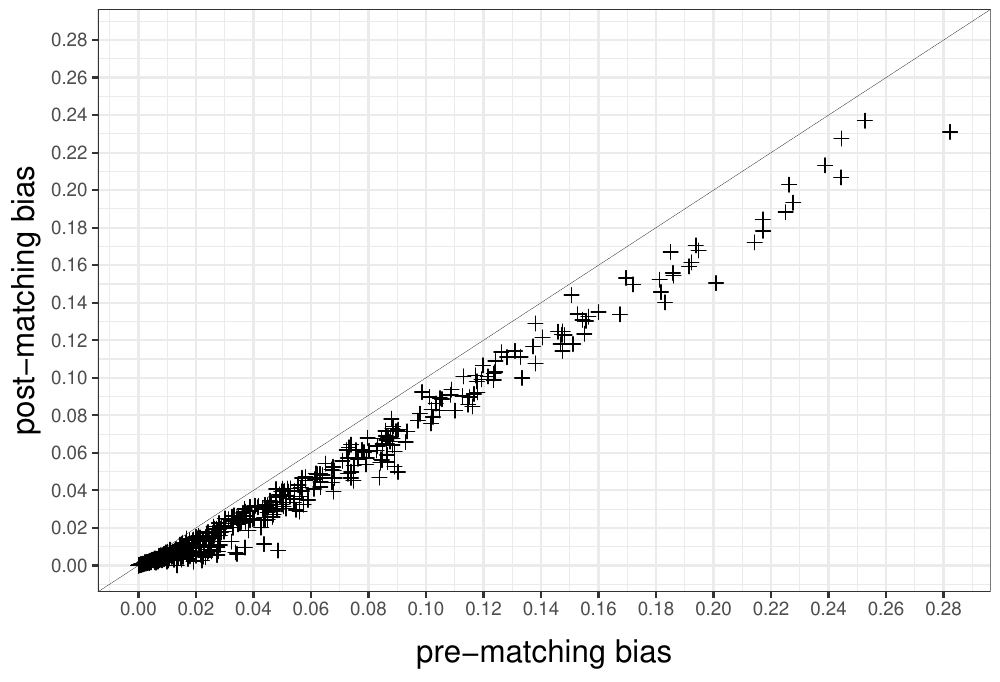}}\hfil
\subfloat[Specification B]{\includegraphics[width=0.33\textwidth]{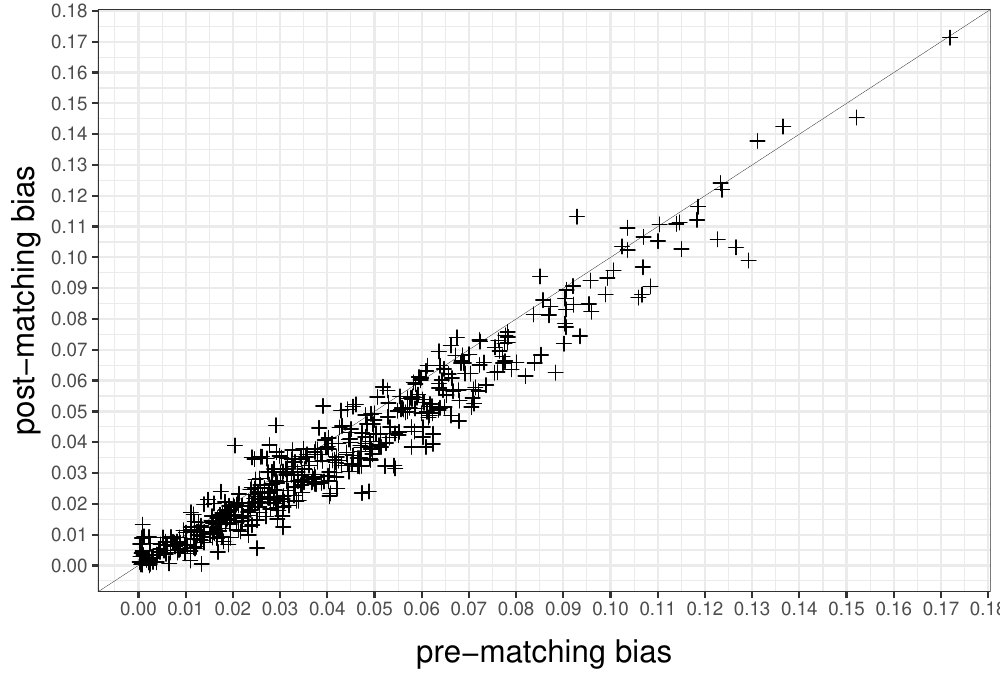}}\hfil	
\subfloat[Specification C]{\includegraphics[width=0.33\textwidth]{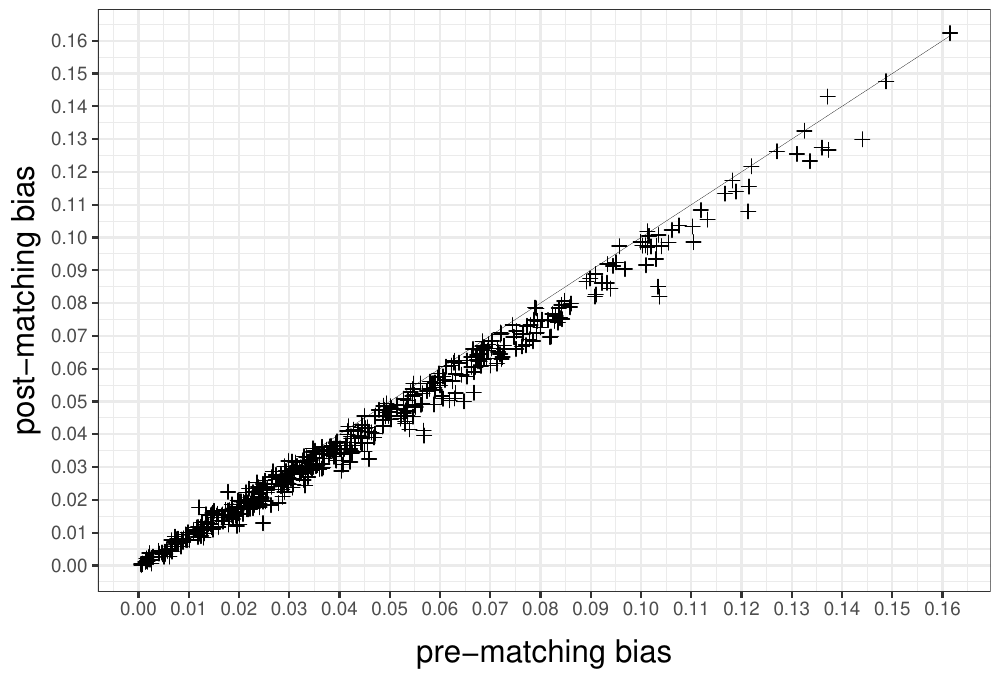}}
\caption{Pre- and Post-biases ($\mathsf N_1=50$, $\mathsf N_0=75$)}
\label{fig:bias-75}\bigskip
\end{figure}

\begin{figure}[htp]
\centering
\subfloat[Without covariates]{\includegraphics[width=0.33\textwidth]{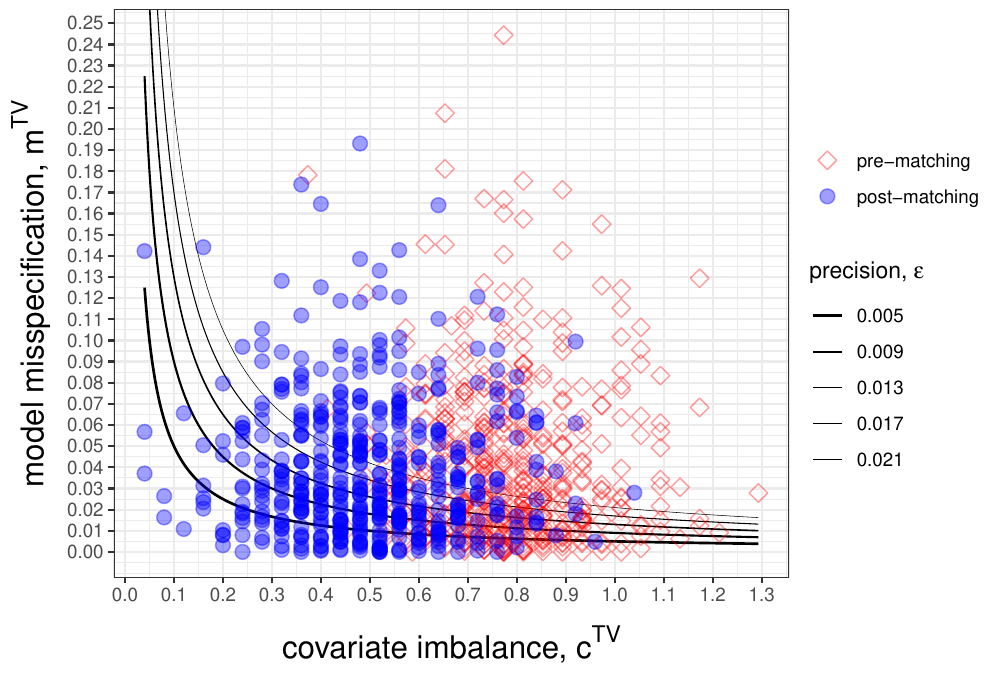}}
\subfloat[Saturated]{\includegraphics[width=0.33\textwidth]{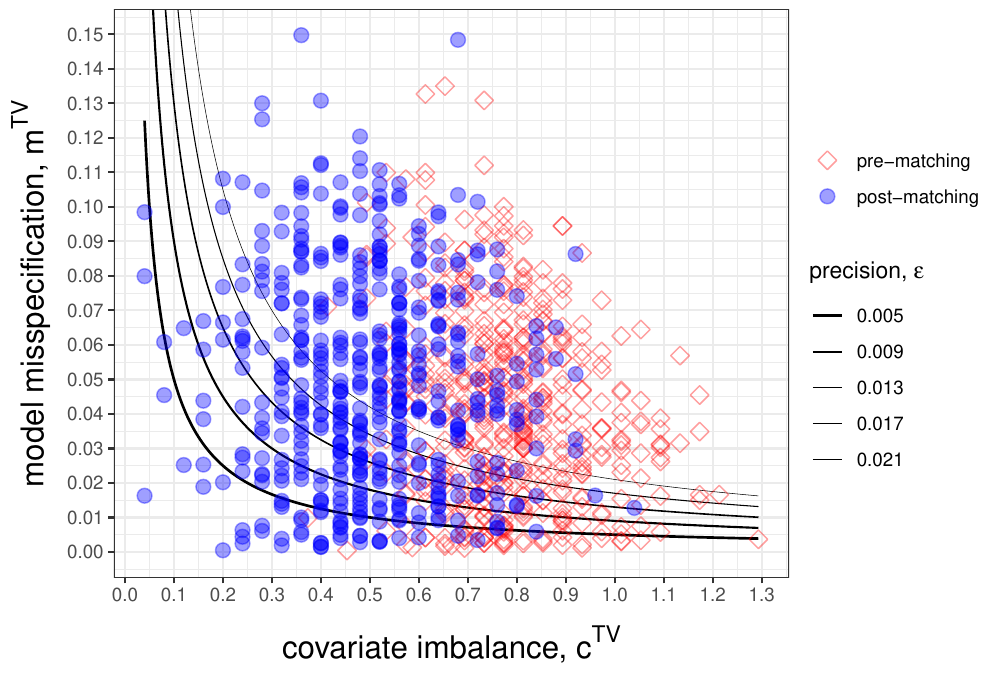}}
\caption{Total variation bound ($\mathsf N_1=50$, $\mathsf N_0=75$)}
\label{fig:total-variation-75}\bigskip
\subfloat[Specification A]{\includegraphics[width=0.33\textwidth]{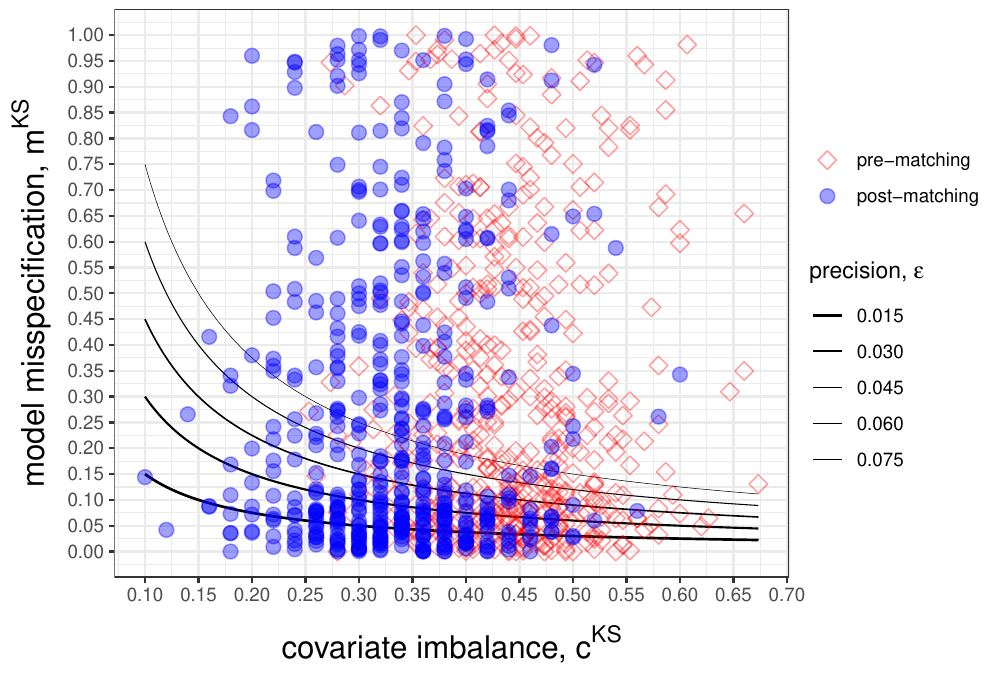}}\hfil
\subfloat[Specification B]{\includegraphics[width=0.33\textwidth]{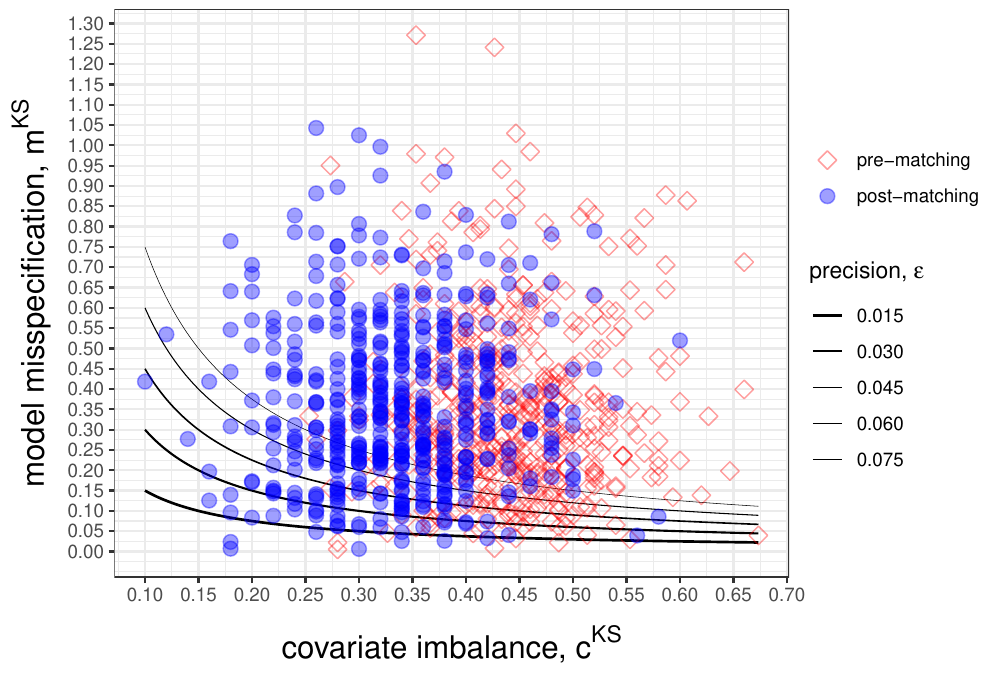}}\hfil
\subfloat[Specification C]{\includegraphics[width=0.33\textwidth]{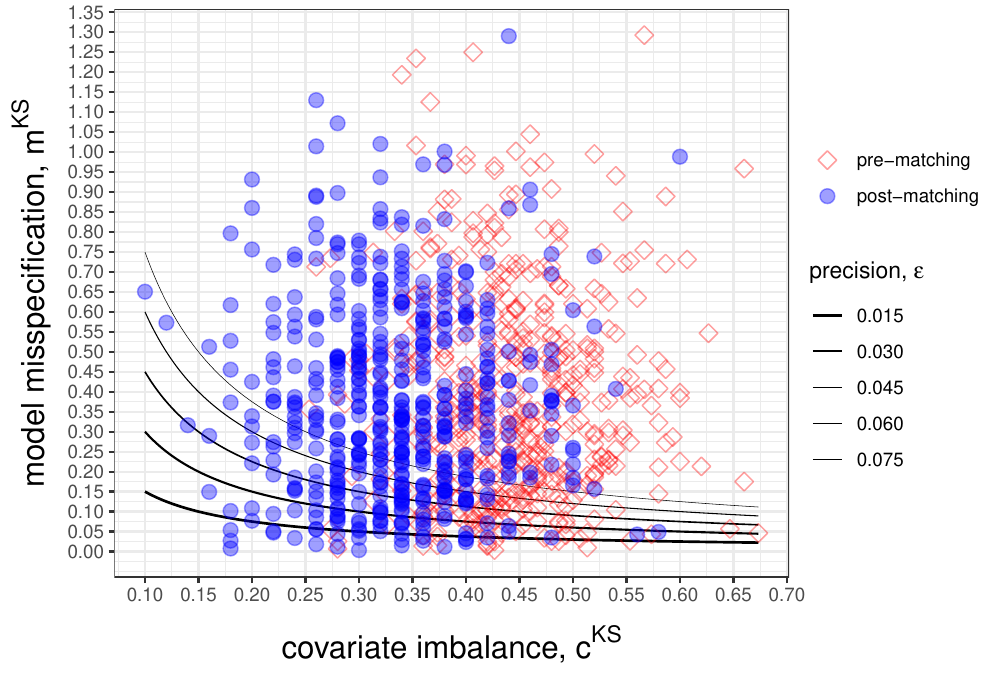}}
\caption{Kolmogorov-Smirnov bound ($\mathsf N_1=50$, $\mathsf N_0=75$)}
\label{fig:kolmogorov-smirnov-75}\bigskip
\subfloat[Without covariates]{\includegraphics[width=0.33\textwidth]{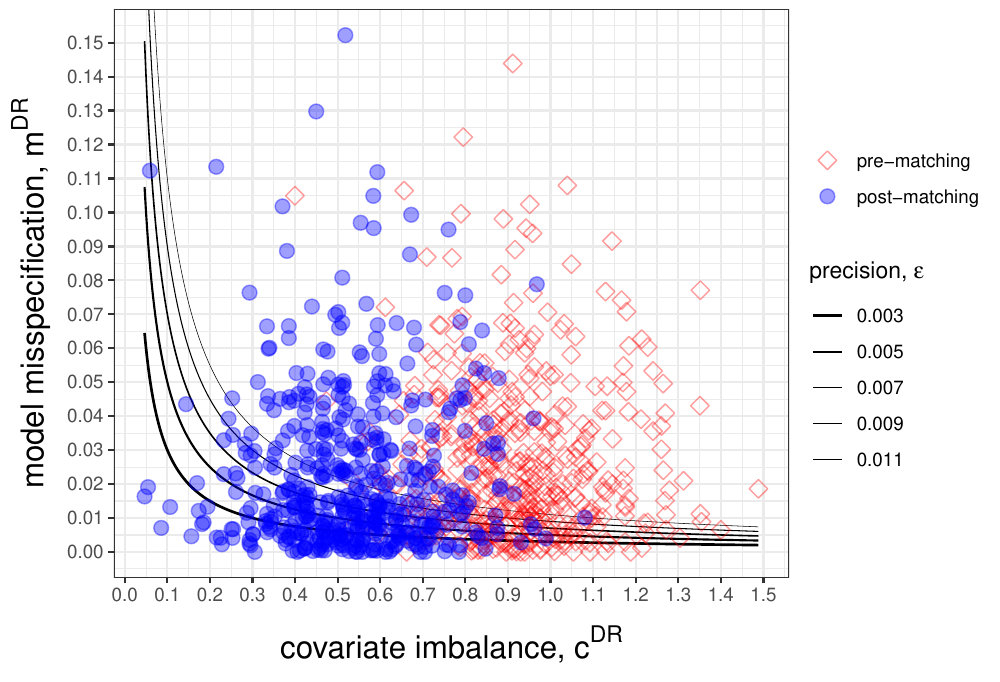}}
\subfloat[Saturated]{\includegraphics[width=0.33\textwidth]{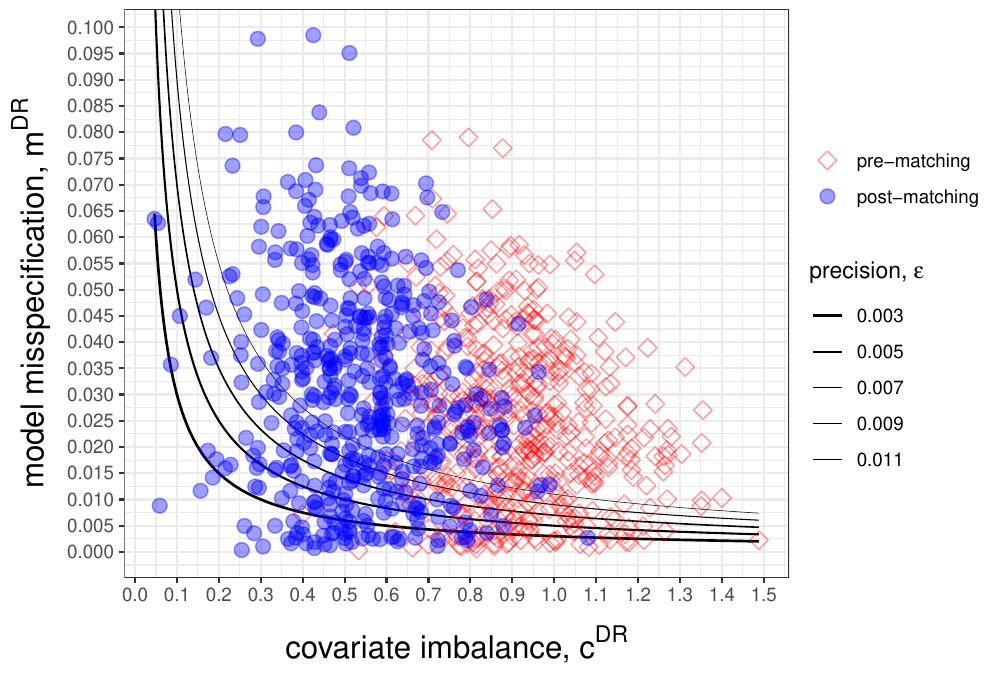}}
\caption{Density ratio bound ($\mathsf N_1=50$, $\mathsf N_0=75$)}
\label{fig:density-ratio-75}
\end{figure}

\begin{figure}
\centering
\subfloat[Specification A]{\includegraphics[width=0.33\textwidth]{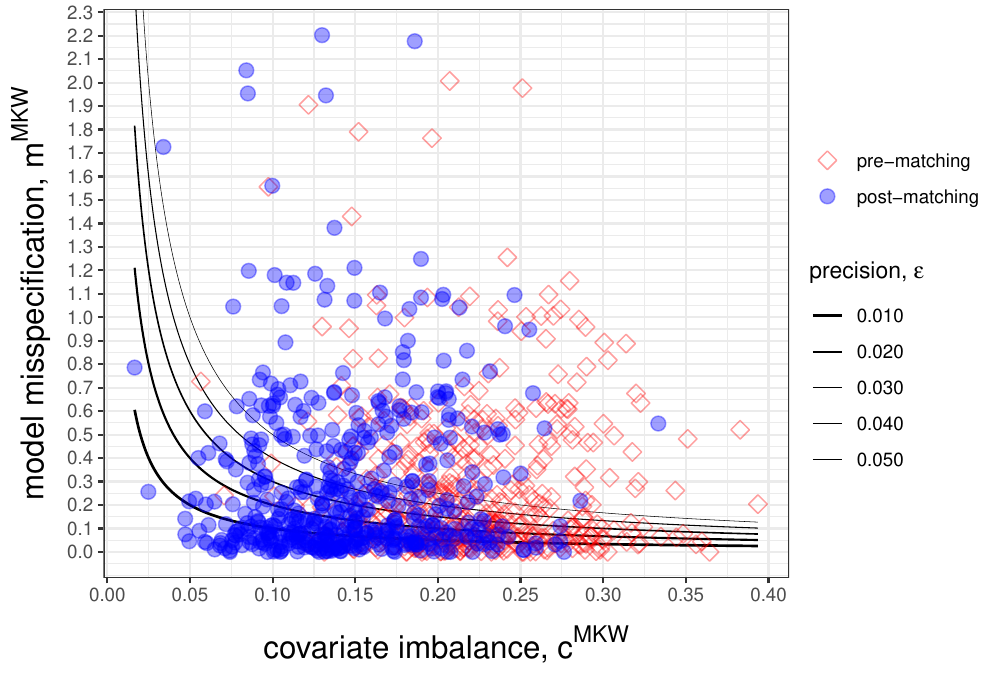}}
\subfloat[Specification B]{\includegraphics[width=0.33\textwidth]{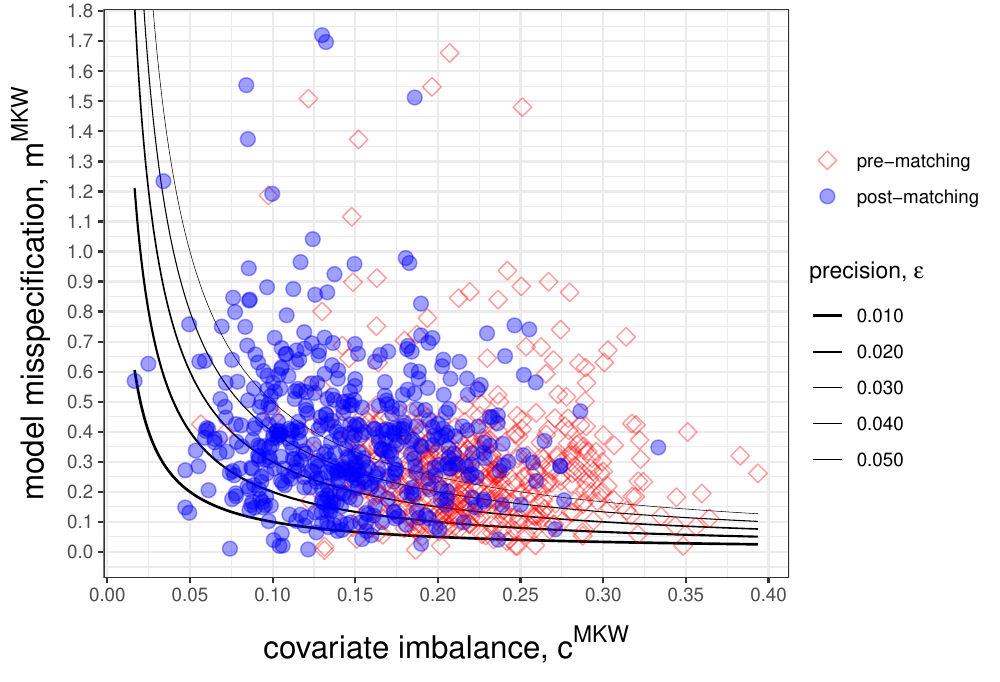}}
\subfloat[Specification C]{\includegraphics[width=0.33\textwidth]{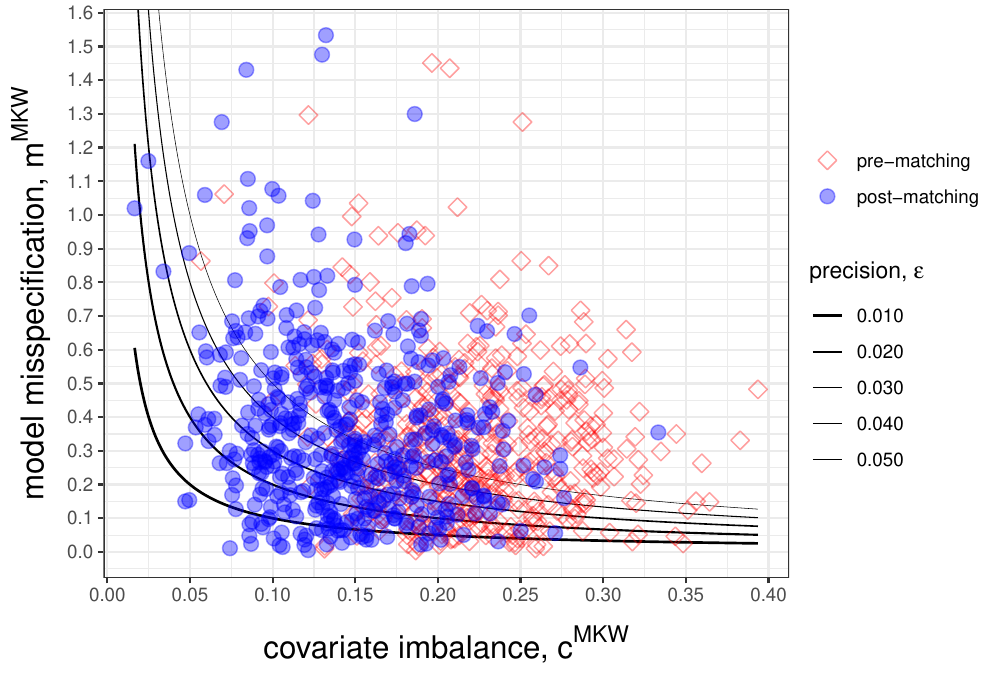}}
\caption{Monge-Kantorovich/Wasserstein Bound ($\mathsf N_1=50$, $\mathsf N_0=75$)}
\label{fig:monge-kantorovich/wasserstein-75}\bigskip
\centering
\subfloat[Specification A]{\includegraphics[width=0.33\textwidth]{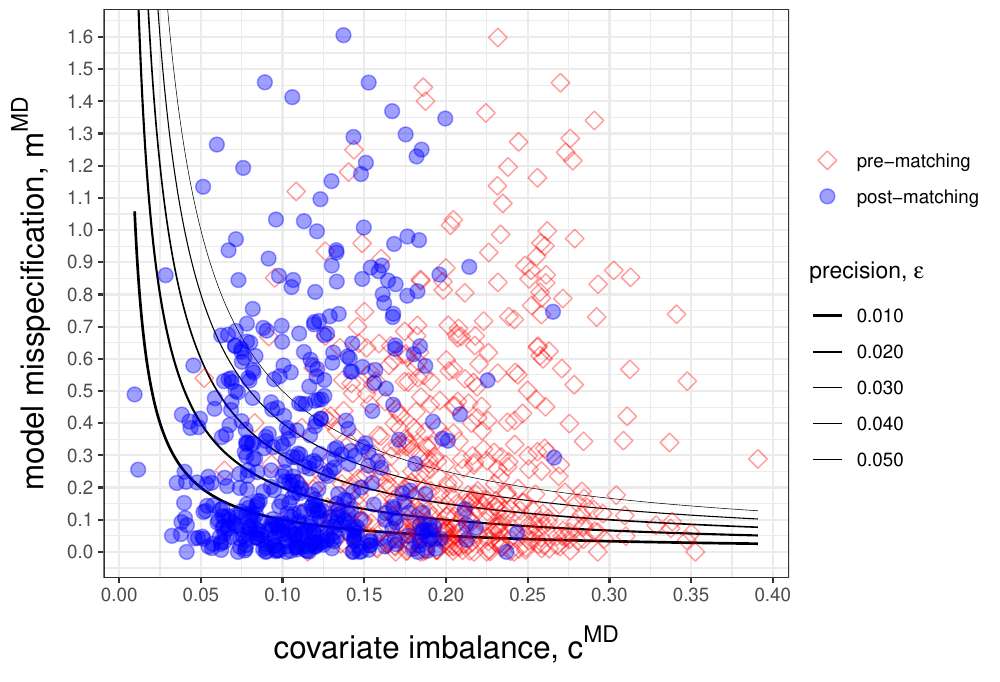}}
\subfloat[Specification B]{\includegraphics[width=0.33\textwidth]{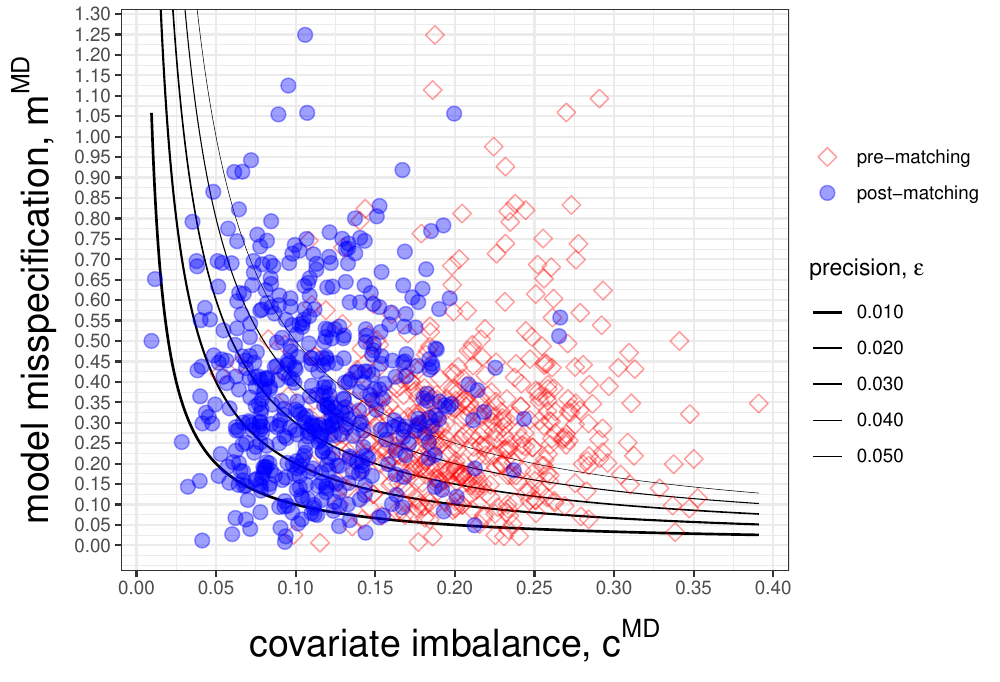}}
\subfloat[Specification C]{\includegraphics[width=0.33\textwidth]{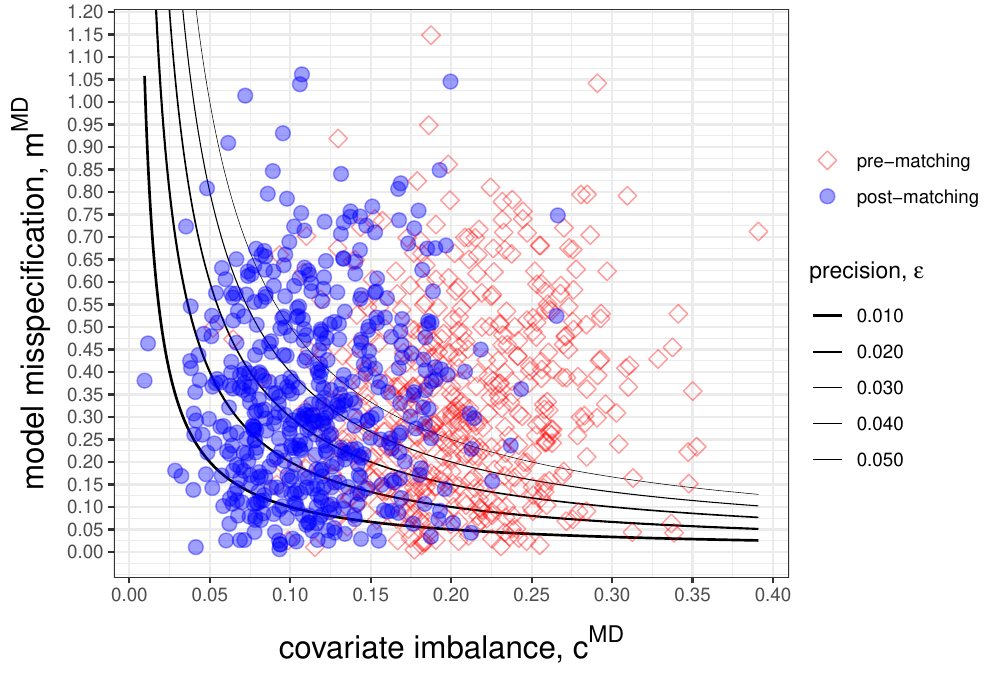}}
\caption{Mean difference Bound ($\mathsf N_1=50$, $\mathsf N_0=75$)}
\label{fig:mean-difference-75}
\end{figure}

\begin{figure}[ht]
\centering
\subfloat[Specification A]{\includegraphics[width=0.33\textwidth]{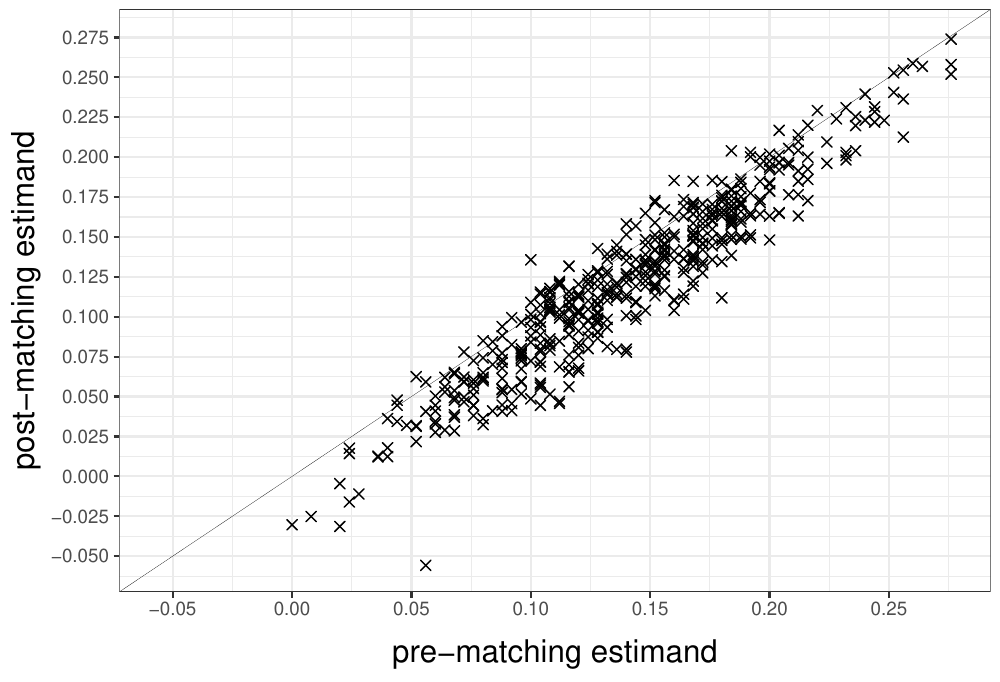}}
\subfloat[Specification B]{\includegraphics[width=0.33\textwidth]{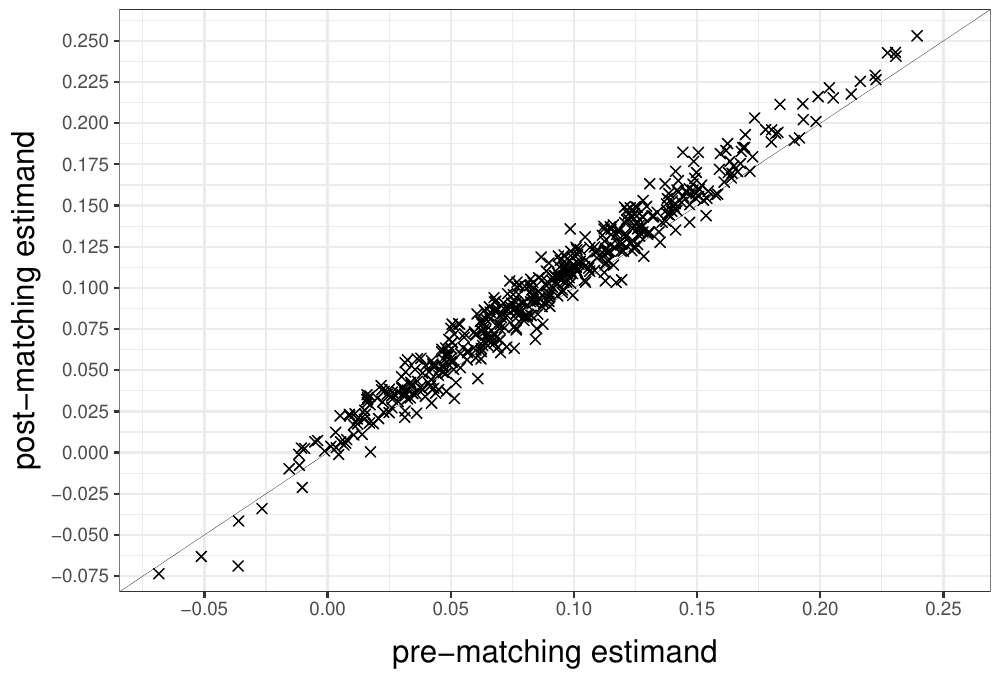}}
\subfloat[Specification C]{\includegraphics[width=0.33\textwidth]{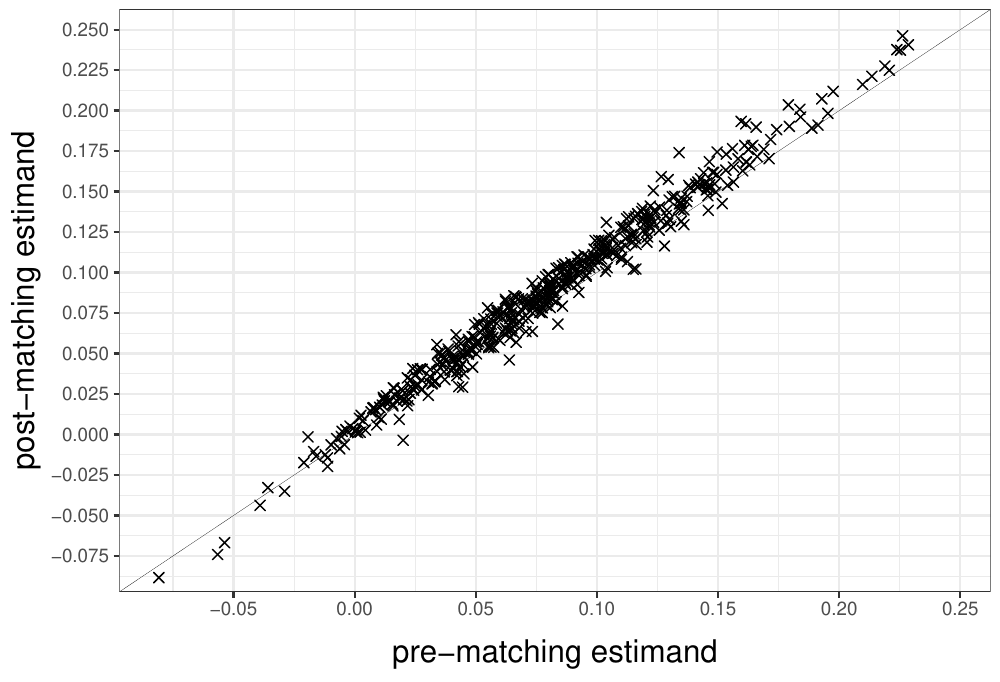}}
\caption{Pre- and Post-estimands ($\mathsf N_1=50$, $\mathsf N_0=125$)}
\label{fig:estimand-125}\bigskip
\subfloat[Specification A]{\includegraphics[width=0.33\textwidth]{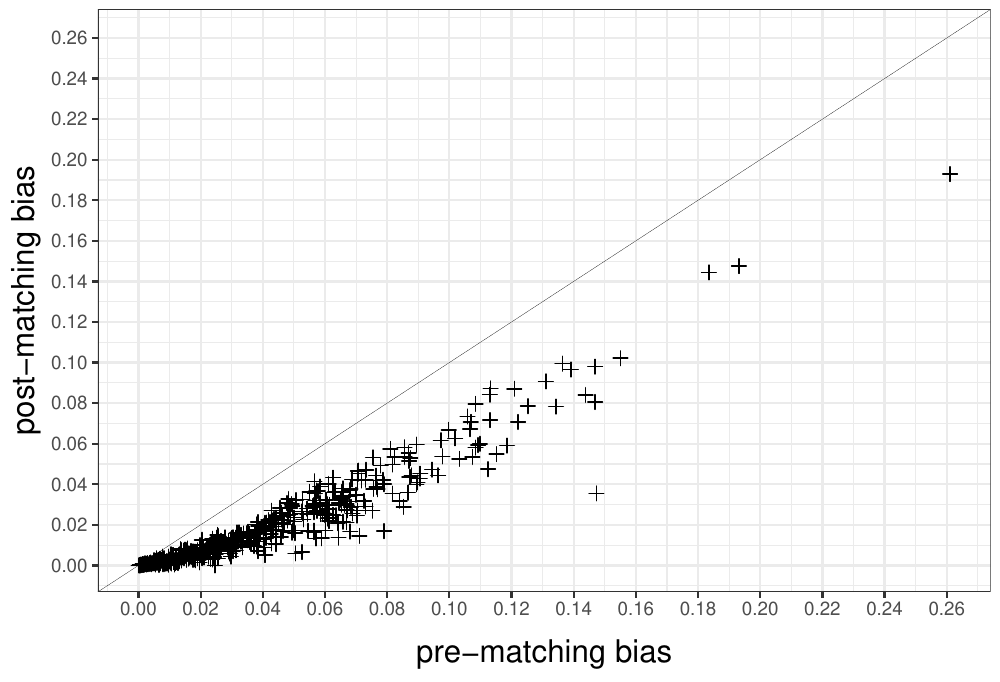}}
\subfloat[Specification B]{\includegraphics[width=0.33\textwidth]{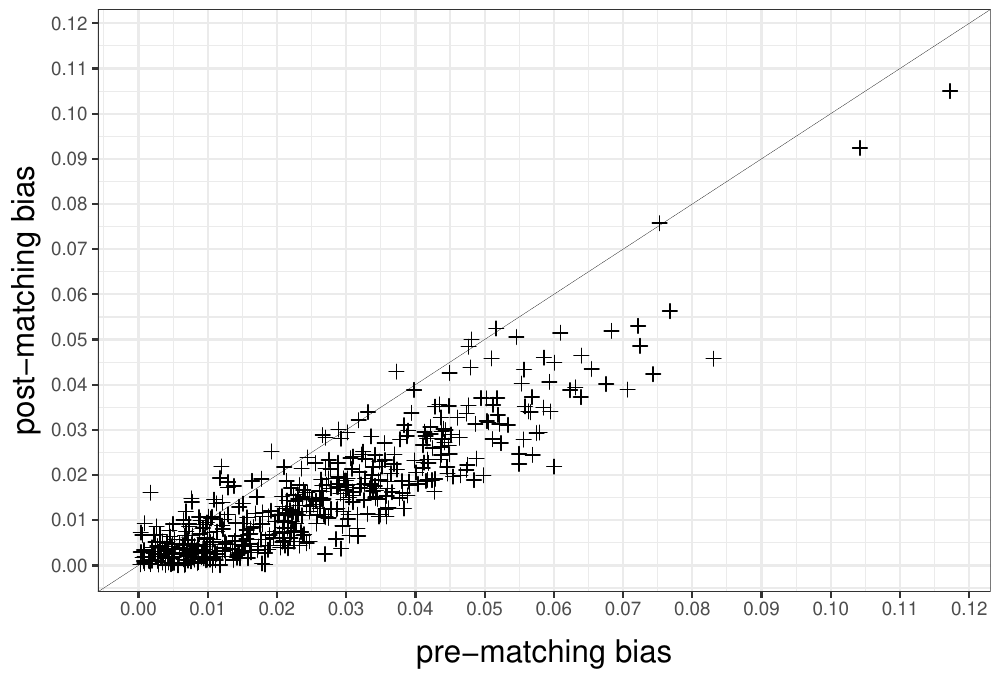}}
\subfloat[Specification C]{\includegraphics[width=0.33\textwidth]{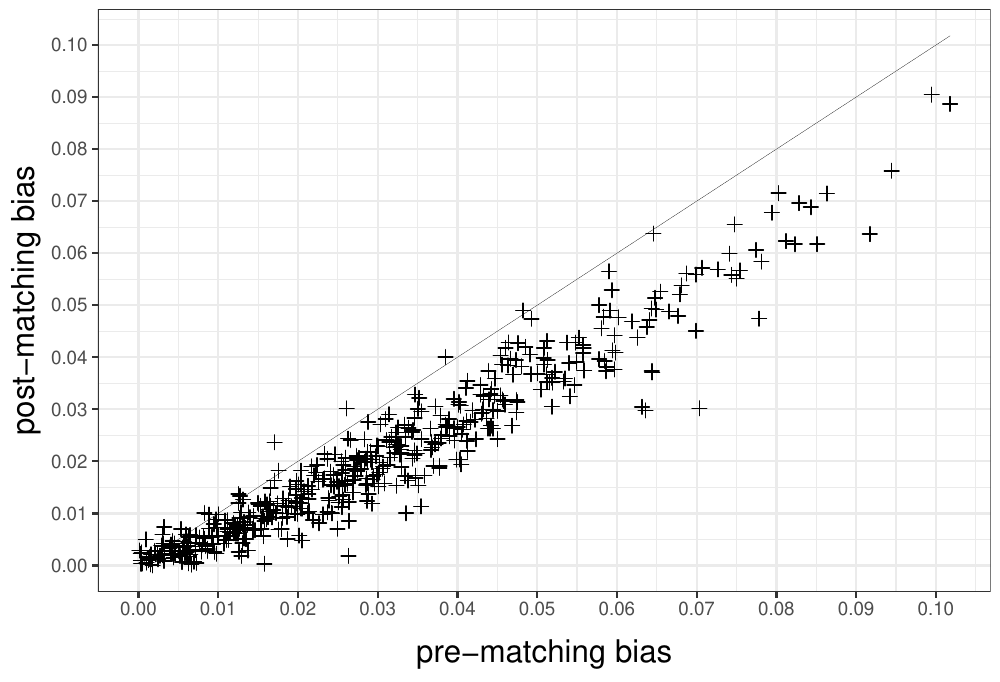}}
\caption{Pre- and Post-biases ($\mathsf N_1=50$, $\mathsf N_0=125$)}
\label{fig:bias-125}\bigskip
\end{figure}

\begin{figure}[htp]
\centering
\subfloat[Without covariates]{\includegraphics[width=0.33\textwidth]{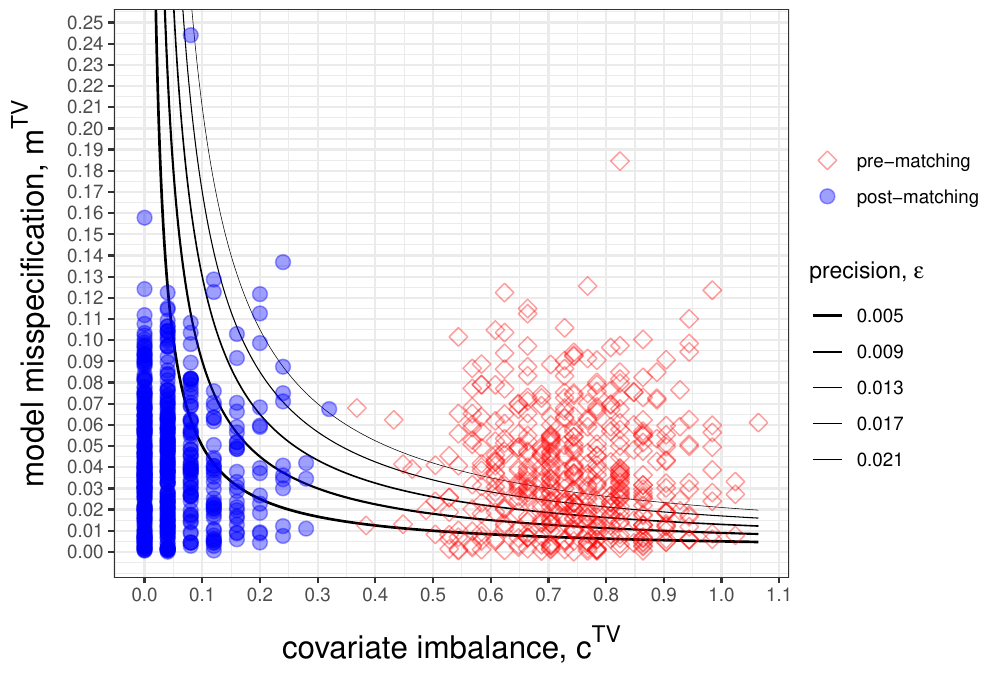}}
\subfloat[Saturated]{\includegraphics[width=0.33\textwidth]{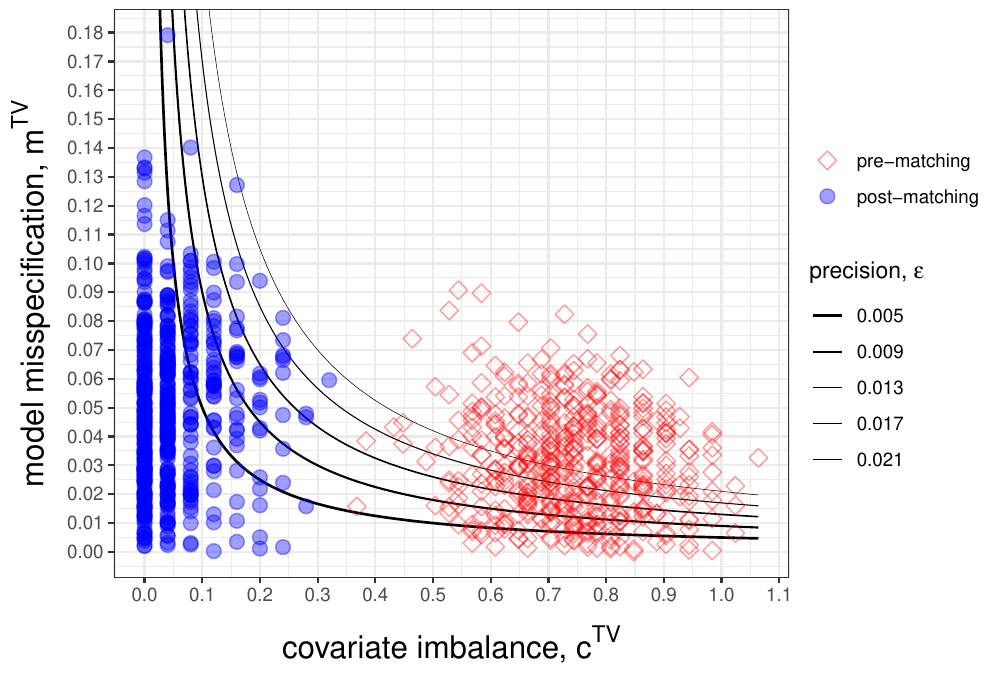}}
\caption{Total variation bound ($\mathsf N_1=50$, $\mathsf N_0=125$)}
\label{fig:total-variation-125}\bigskip
\subfloat[Specification A]{\includegraphics[width=0.33\textwidth]{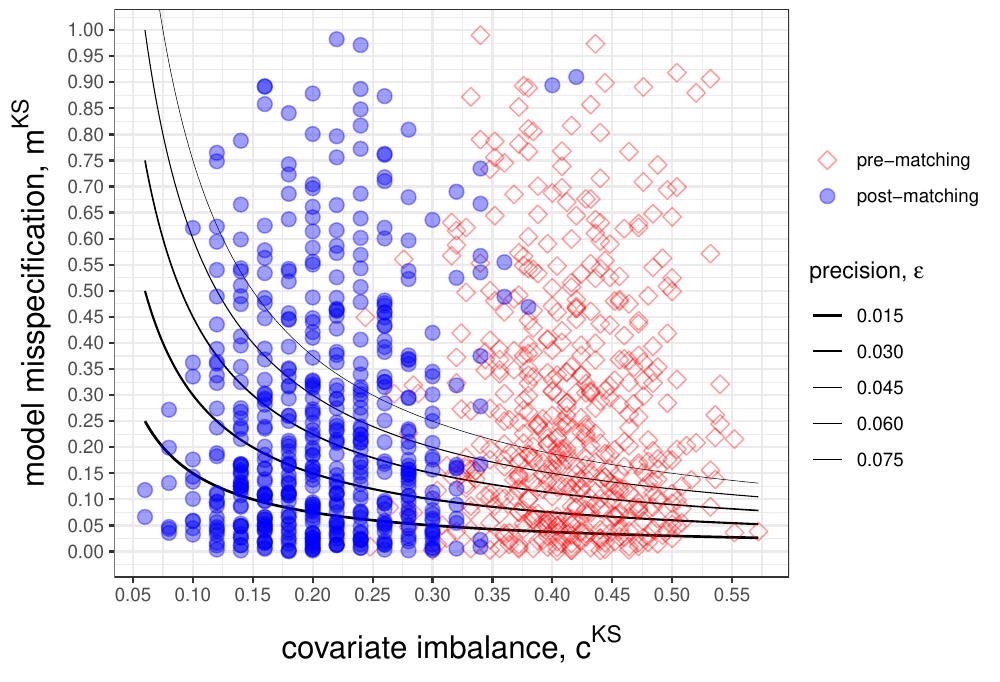}}
\subfloat[Specification B]{\includegraphics[width=0.33\textwidth]{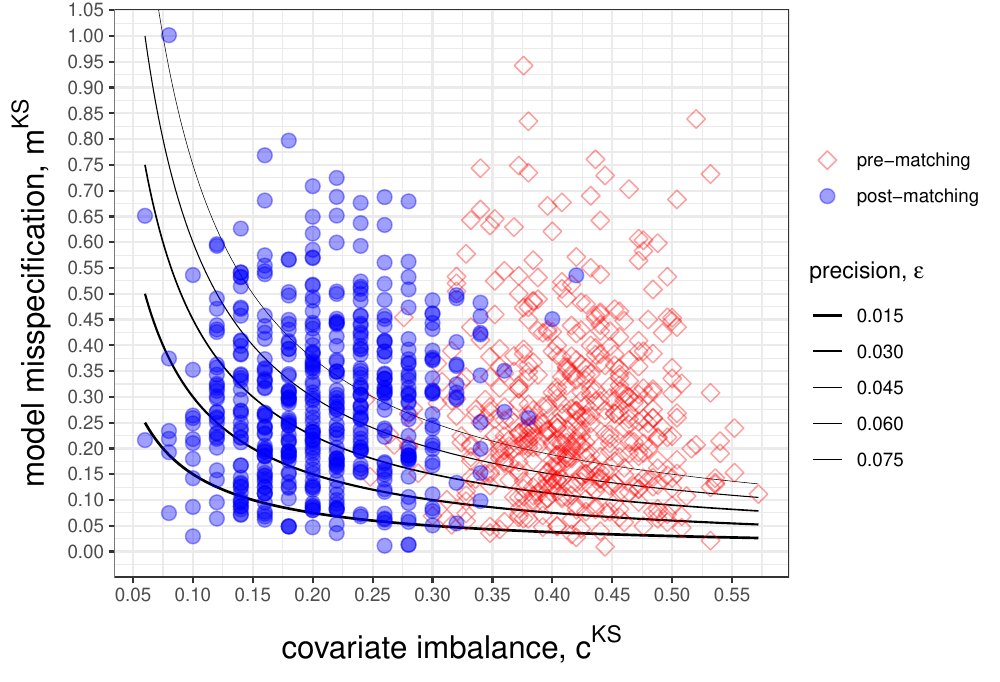}}
\subfloat[Specification C]{\includegraphics[width=0.33\textwidth]{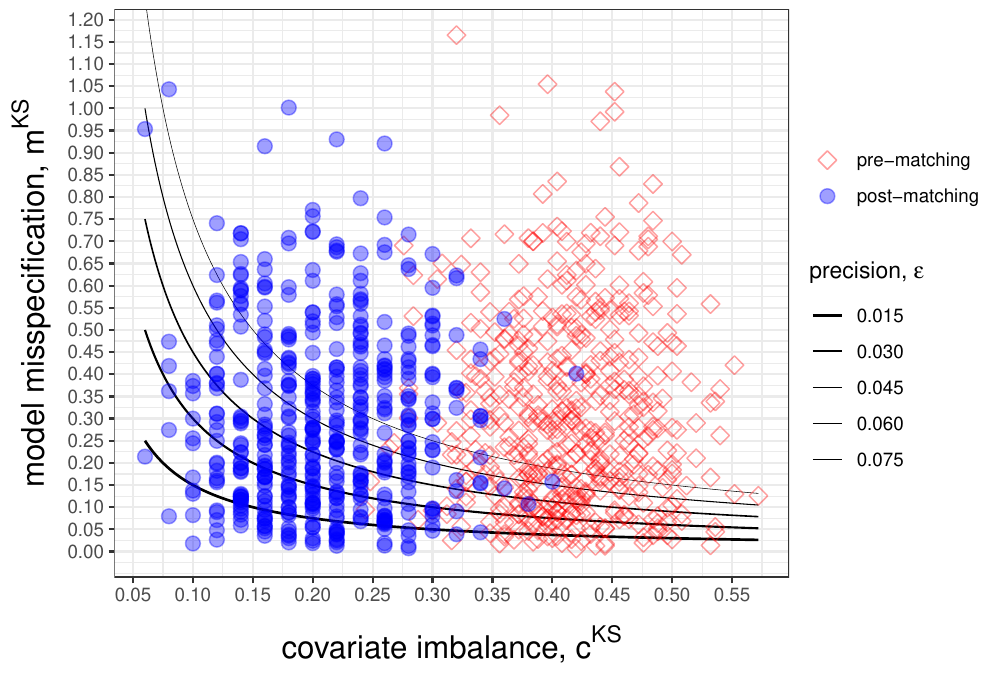}}
\caption{Kolmogorov-Smirnov bound ($\mathsf N_1=50$, $\mathsf N_0=125$)}
\label{fig:kolmogorov-smirnov-125}\bigskip
\subfloat[Without covariates]{\includegraphics[width=0.33\textwidth]{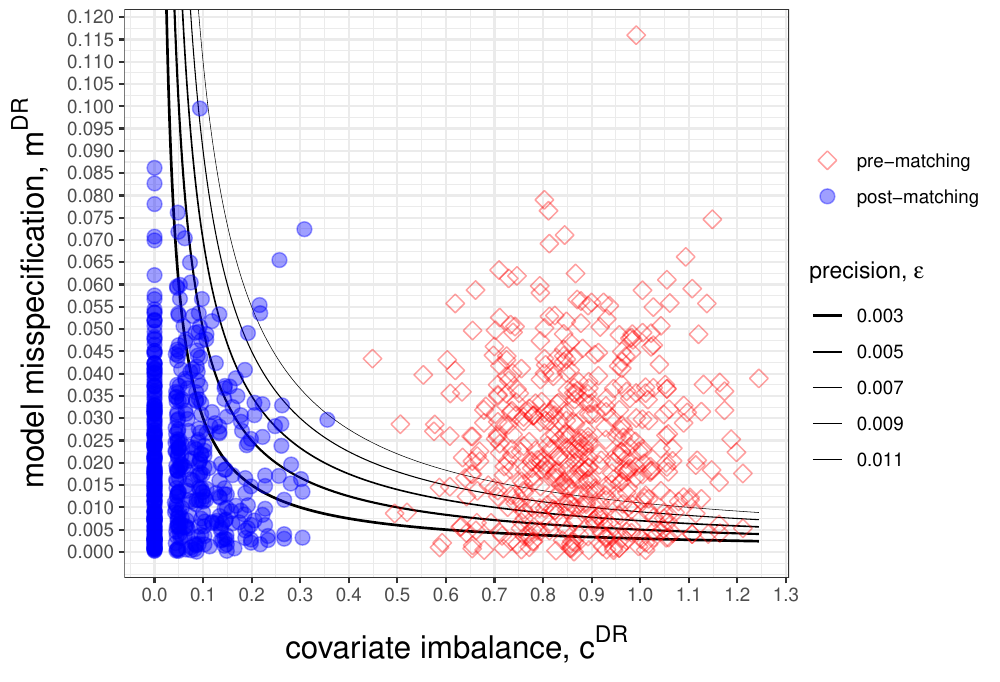}}
\subfloat[Saturated]{\includegraphics[width=0.33\textwidth]{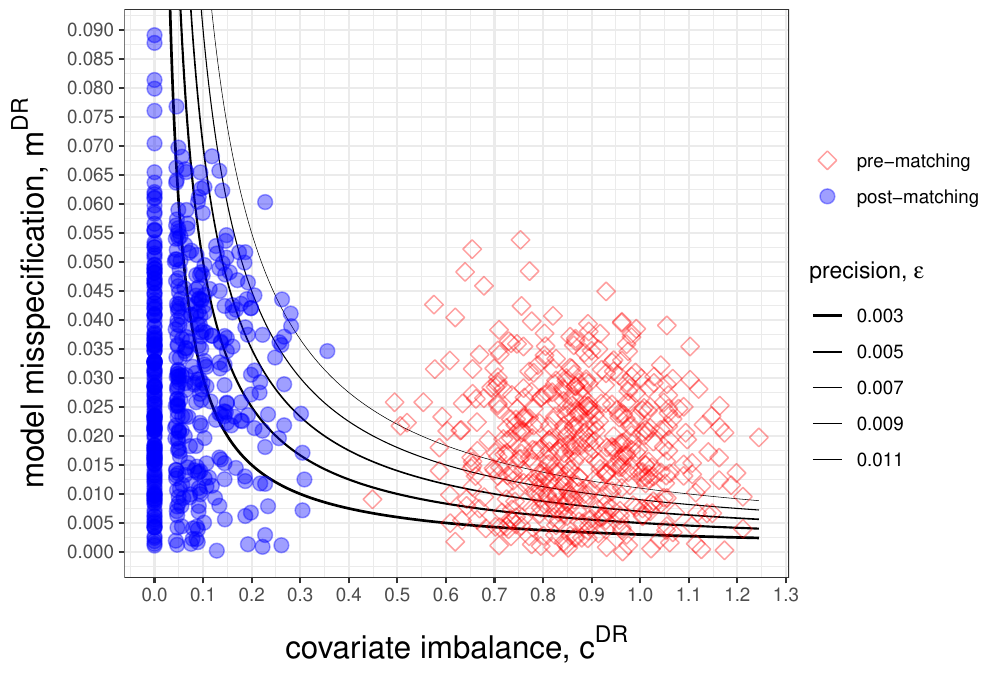}}
\caption{Density ratio bound ($\mathsf N_1=50$, $\mathsf N_0=125$)}
\label{fig:density-ratio-125}
\end{figure}

\begin{figure}
\centering
\subfloat[Specification A]{\includegraphics[width=0.33\textwidth]{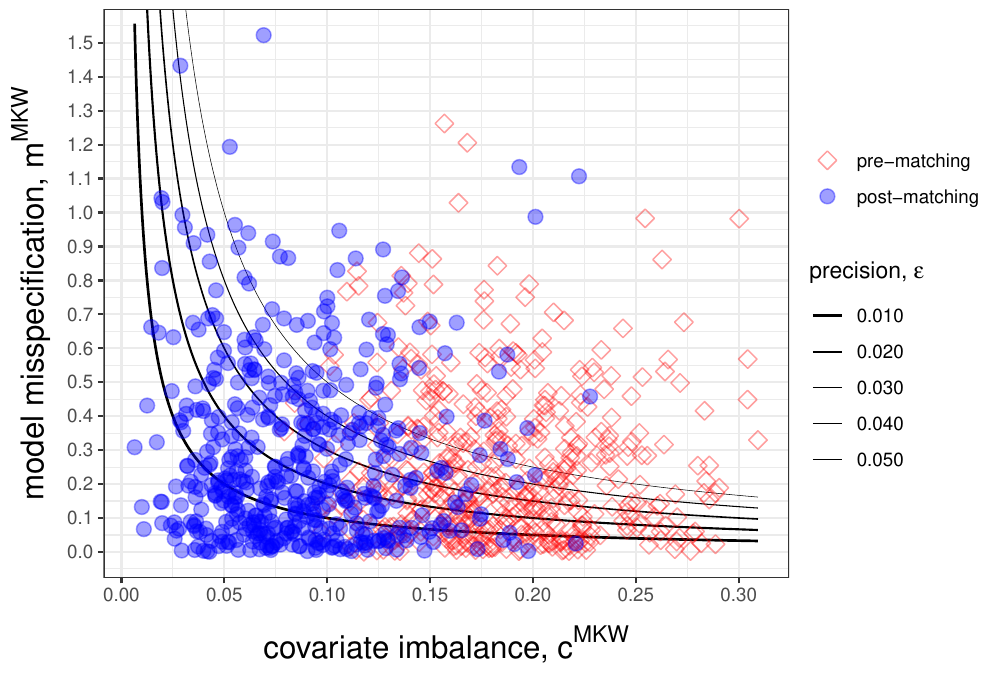}}
\subfloat[Specification B]{\includegraphics[width=0.33\textwidth]{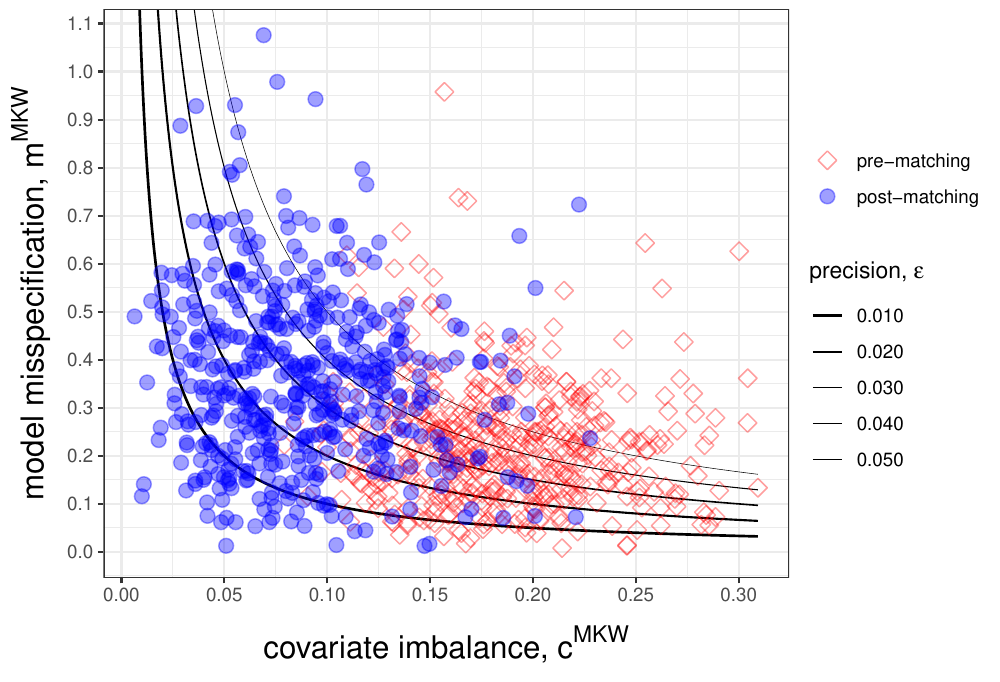}}
\subfloat[Specification C]{\includegraphics[width=0.33\textwidth]{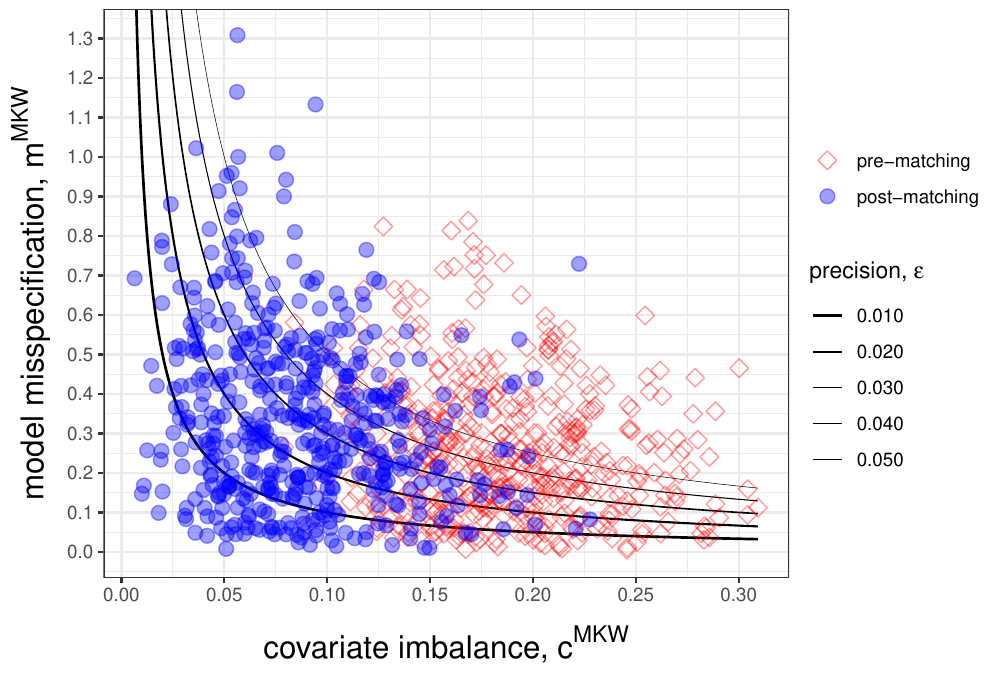}}
\caption{Monge-Kantorovich/Wasserstein Bound ($\mathsf N_1=50$, $\mathsf N_0=125$)}
\label{fig:monge-kantorovich/wasserstein-125}\bigskip
\centering
\subfloat[Specification A]{\includegraphics[width=0.33\textwidth]{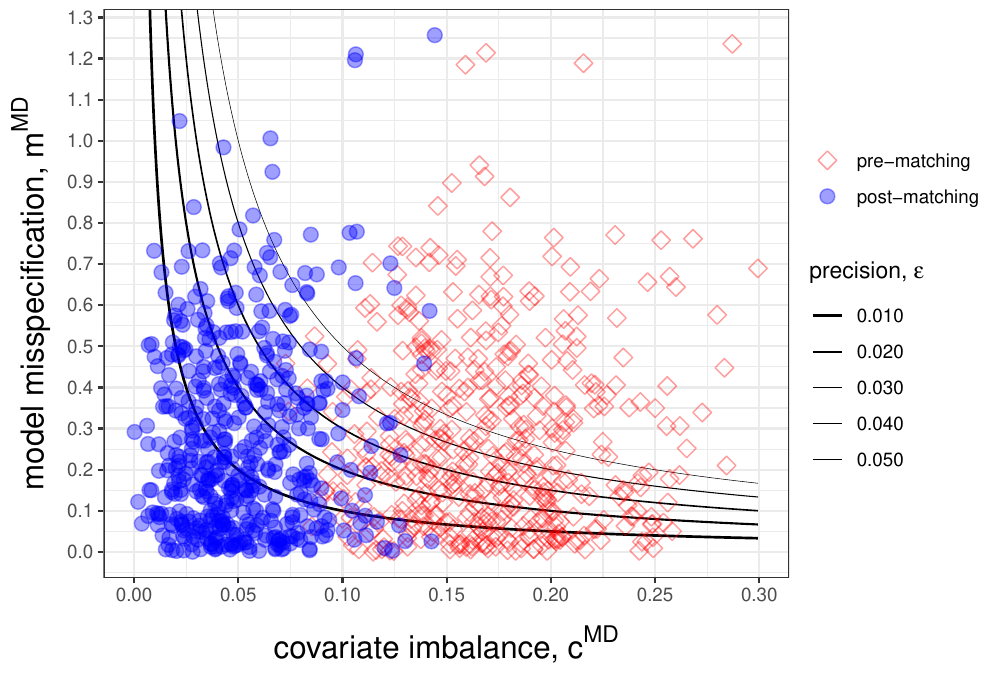}}
\subfloat[Specification B]{\includegraphics[width=0.33\textwidth]{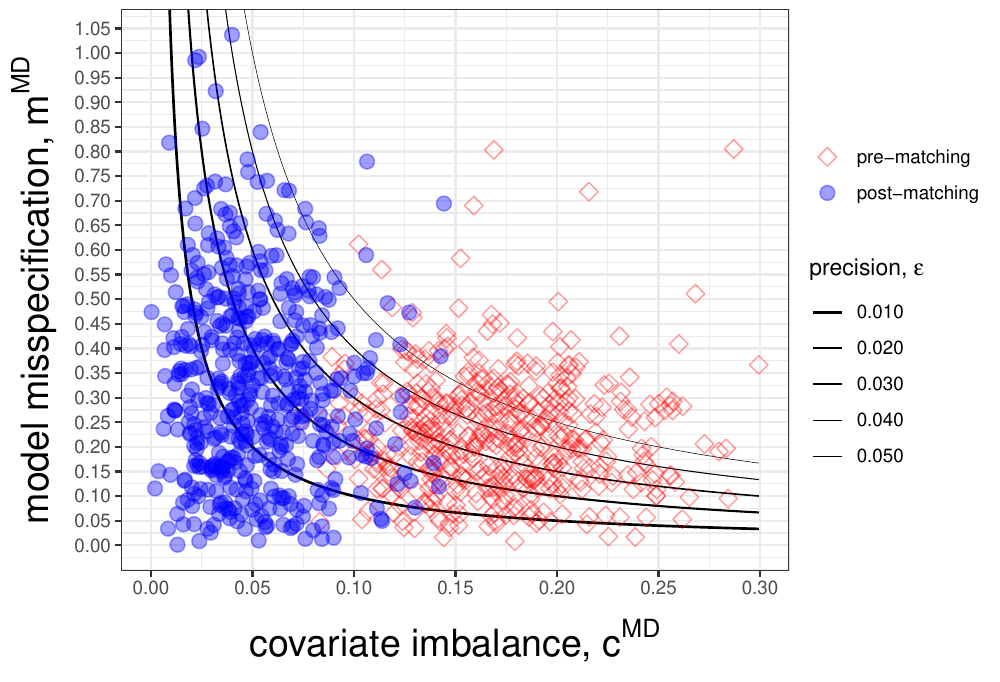}}
\subfloat[Specification C]{\includegraphics[width=0.33\textwidth]{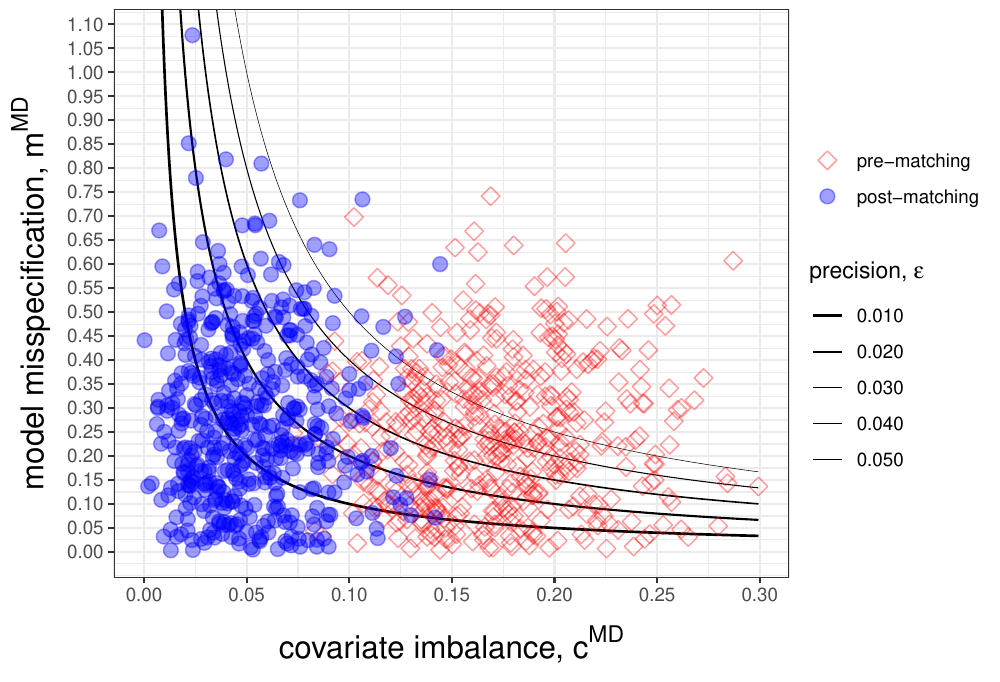}}
\caption{Mean difference Bound ($\mathsf N_1=50$, $\mathsf N_0=125$)}
\label{fig:mean-difference-125}
\end{figure}

\end{document}